%% file: arxiv.tex
\documentclass[acmsmall, screen, nonacm]{acmart}

\AtBeginDocument{%
  }

\setcopyright{acmcopyright}
\copyrightyear{2018}
\acmYear{2018}
\acmDOI{XXXXXXX.XXXXXXX}

\acmConference[Conference acronym 'XX]{Make sure to enter the correct
  conference title from your rights confirmation emai}{June 03--05,
  2018}{Woodstock, NY}
\acmPrice{15.00}
\acmISBN{978-1-4503-XXXX-X/18/06}

\usepackage{savesym}    
\usepackage{marvosym}
\usepackage{enumerate}
\usepackage{prftree}
\savesymbol{comment}
\usepackage[final,noding]{commenting}
\usepackage{listings}
\usepackage{multicol}
\usepackage{multirow}
\usepackage{tabularray}
\usepackage{wrapfig}
\usepackage{caption}
\usepackage{subcaption}
\usepackage{etoolbox}
\usepackage{mdframed}
\usepackage[inline,shortlabels]{enumitem}
\declareauthor{sr}{Steven}{teal}
\authorcommand{sr}{comment}
\declareauthor{cw}{Charlie}{red}
\authorcommand{cw}{comment}

\usepackage{graphicx}

\makeatletter
\providecommand*{\lmodels}{%
  \mathrel{%
    \mathpalette\@lmodels\models
  }%
}
\newcommand*{\@lmodels}[2]{%
  \reflectbox{$\m@th#1#2$}%
}
\makeatother

\font\stixarr=stix-mathsfit at 10pt

\input{macros}

\renewcommand{\bar}[1]{\vv{#1}}

\newcommand{\ch}[1]{#1}
\newcommand{\chA}[1]{#1}
\newcommand{\chB}[1]{#1}
\newcommand{\chC}[1]{#1}
\newcommand{\chD}[1]{#1}

\theoremstyle{remark}
\newtheorem{remark}{Remark}[section]

\newtoggle{supplementary}
\newtoggle{arxiv}
\toggletrue{arxiv}
\toggletrue{supplementary}

\begin{document} 

\title{Ill-Typed Programs Don't Evaluate}

\author{Steven Ramsay}
\email{steven.ramsay@bristol.ac.uk}
\affiliation{%
  \institution{University of Bristol}
  \country{UK}
}

\author{Charlie Walpole}
\email{op18921@bristol.ac.uk}
\affiliation{%
  \institution{University of Bristol}
  \country{UK}
}


\begin{abstract}
  We introduce two-sided type systems, which are sequent calculi for typing formulas.  Two-sided type systems allow for hypothetical reasoning over the typing of compound program expressions, and the refutation of typing formulas.  By incorporating a type of all values, these type systems support more refined notions of well-typing and ill-typing, guaranteeing both that well-typed programs don't go wrong and that ill-typed programs don't evaluate - that is, reach a value.  This makes two-sided type systems suitable for incorrectness reasoning in higher-order program verification, which we illustrate through an application to precise data-flow typing in a language with constructors and pattern matching.  Finally, we investigate the internalisation of the meta-level negation in the system as a complement operator on types.  This motivates an alternative semantics for the typing judgement, which guarantees that ill-typed programs don't evaluate, but in which well-typed programs may yet go wrong.
\end{abstract}





\maketitle

\input{intro}
\input{prog-lang} 
\input{okay}

\input{semantics}
\input{complements}
\input{constrained}
\input{consistency}

\input{related}

\begin{acks}
  The first author is very grateful for an award from Meta Research.  This paper was much improved due to  the advice of the POPL'24 anonymous reviewers, and we especially thank our shepherd Sophia Drossopoulou.  The work has also benefited from many useful discussions with colleagues in the Programming Languages Research Group at the University of Bristol.  
\end{acks}

\bibliographystyle{ACM-Reference-Format}
\bibliography{references}

\iftoggle{supplementary}{%
\clearpage

\appendix

\input{apx-language}
\input{apx-cbv}
\input{apx-complements}
\input{algorithmic}
\input{apx-soundness}
}{%
}

\end{document}
\endinput

%% file: macros.tex
\newcommand{\abbv}[1]{\textsf{\underline{\smash{#1}}}}

\newcommand{\succtm}{\texttt{succ}}
\newcommand{\predtm}{\mathsf{pred}}

\newcommand{\zerotm}{\texttt{zero}}
\newcommand{\projtm}{\mathsf{proj}}

\newcommand{\ifztm}{\mathsf{ifz}}
\newcommand{\iftm}[3]{\mathsf{ifz}\,#1\,\mathsf{then}\,#2\,\mathsf{else}\,#3}
\newcommand{\lettm}[3]{\mathsf{let}\,#1 = #2\,\mathsf{in}\,#3}
\newcommand{\fixtm}[2]{\mathsf{fix}\,#1.\,#2}
\newcommand{\pcfnum}[1]{\underline{#1}}
\newcommand{\pn}[1]{\pcfnum{#1}}
\newcommand{\ped}{\,\ \mathord{\rhd}\,\ }
\newcommand{\pedn}[1]{\,\ \mathord{\rhd}^{#1}\ \,}
\newcommand{\peds}{\,\ \mathord{\rhd}^{\!*}\ \,}
\newcommand{\pedm}{\,\ \mathord{\rhd}_{\scriptscriptstyle\calM}\ \,}

\newcommand{\evals}{\Downarrow}
\newcommand{\diverges}{\Uparrow}
\newcommand{\crash}{\text{\Lightning}}

\newcommand{\divtm}{\abbv{div}}

\newcommand{\Vals}{\mathsf{Vals}}
\newcommand{\FunVals}{\mathsf{FunV}}
\newcommand{\NatVals}{\mathsf{NatV}}
\newcommand{\PairVals}{\mathsf{ProdV}}

\newcommand{\matchtm}[2]{\mathsf{match}\,#1\,\mathsf{with}\,\{#2\}}

\newcommand{\calE}{\mathcal{E}}
\newcommand{\calF}{\mathcal{F}}
\newcommand{\calC}{\mathcal{C}}

\newcommand{\calT}{\mathcal{T}}
\newcommand{\calM}{\mathcal{M}}

\newcommand{\rlnm}[1]{\textsf{\footnotesize(#1)}}

\newcommand{\fv}{\mathsf{FV}}
\newcommand{\abs}[2]{\lambda #1.\,#2}

\newcommand{\types}{\vdash}
\newcommand{\intty}{\mathsf{Nat}}
\newcommand{\okty}{\mathsf{Ok}}

\newcommand{\from}{\mathrel{\text{\stixarr\char"B3}}}
\newcommand{\subtype}{\sqsubseteq}
\newcommand{\distype}{\mathrel{||}}
\newcommand{\subjects}{\mathsf{Subjects}}

\newcommand{\rulediv}[2]{\raisebox{.5ex}{\makebox[\linewidth]{\rule{\dimexpr50pt-#1\relax}{.5pt}\ \raisebox{-.5ex}{\makebox[2\dimexpr#1\relax]{\textsf{\small{#2}}}}\hspace{1ex}\hrulefill}}}
\newcommand{\semBrack}[1]{[\![#1]\!]}
\newcommand{\mng}[1]{\semBrack{#1}}

\newcommand{\ClVals}{\mathsf{Vals}_0}

\newcommand{\fixapprx}[3]{\mathsf{fix}^{#1}#2.\,#3}

\newcommand{\mngT}[1]{\mathcal{T}\mng{#1}}
\newcommand{\mngTb}[1]{\mathcal{T}_{\!\bot}\mng{#1}}
\newcommand{\mngTs}[1]{\mathcal{T}_{\!\bot\,\crash}\mng{#1}}

\newcommand{\constraints}{\mathsf{Con}}
\newcommand{\typings}{\mathsf{Typ}}

%% file: intro.tex
\section{Introduction}\label{sec:intro}

It is natural to think of a type system as a kind of proof system, whose purpose is reasoning about the behaviour of program expressions.  
In this view, the atomic formulas of the system are typings $M:A$, where $M$ is a term and $A$ a type and, in building a proof of a judgement $\Gamma \types M : A$, we aim to conclude an atomic formula $M:A$ under some assumptions $\Gamma$ on the types of its free variables.
Therefore, built into traditional type systems, but absent from most general proof systems, is a fundamental asymmetry: although we can conclude arbitrary typing formulas $M:A$, we may only make assumptions on variable typing formulas $x:B$.  
In this paper, we argue that type theory is enriched by removing this restriction, and there are interesting practical applications.
Our \emph{two-sided type systems} allow for making assumptions on the typing of arbitrary program expressions.  In fact, they are particular kinds of sequent calculi over formulas of shape $M:A$.  

The ability to reason hypothetically about the behaviour of compound terms gives us a new way to express how programs depend on their inputs.  
For example, we can prove that, if $x + 1$ evaluates to a natural number and $(\abs{fz}{f\,(f\,z)})\,(\abs{x}{x})\,y$ evaluates to a natural number, then $(x,y)$ is necessarily a pair of natural numbers:
\[
  x+1 : \intty,\,(\abs{fz}{f\,(f\,z)})\,(\abs{x}{x})\,y : \intty \types (x,y) : \intty \times \intty
\]

We also remove the restriction that one must conclude exactly one typing formula on the right of the turnstile, with multiple formulas understood disjunctively (as in most sequent calculi).  In particular, one is allowed to derive an empty conclusion $\Gamma \types $, meaning that the assumed typing formulas $\Gamma$ are inconsistent (the unit of disjunction is absurdity).  For example, one can prove:
\[
  (\abs{fz}{f\,(f\,z)})\,(\abs{x}{x})\,(y,\,y) : \intty \types
\]
stating that, no matter the value of $y$, the subject on the left \emph{doesn't} evaluate to a natural number.

Despite this, the term above \emph{does} evaluate, i.e. reduce to a value. By introducing a type $\okty$ to characterise all values, we can use hypothetical judgements like the above to \emph{prove} that terms \emph{do not evaluate} (at all), that is they either diverge or go wrong.  For example, we can prove:
\[
  (\abs{fz}{f\,(f\,z)})\,(\abs{x}{\predtm(x)})\,(y,\,y) : \okty \types
\]
stating that the program that applies the predecessor function twice in succession to a pair will fail to evaluate.  We say that such a term, to which we can assign the type $\okty$ on the left, is \emph{ill-typed}.  \chC{In this paper, ill-typed is not merely a synonym for \emph{untypable}: an ill-typed term can be typed (as one might expect from the suffix), but the type assignment is on the left of the turnstile.}

As well as the familiar rules for concluding typings on the right of the turnstile, our systems also have rules for refuting assumed typings on the left.  
The key to the whole enterprise is a new kind of function type, the necessity arrow $A \from B$, pronounced `$A$ \emph{only to} $B$'.  
The necessity arrow describes functions that produce a $B$ \emph{only if} they are given an $A$ as input.
It is, in a sense that we make precise, the contrapositive of the usual (sufficiency) arrow $A \to B$, which guarantees that a $B$ is produced \emph{if} an $A$ is given as input.
The rule, \rlnm{AbnR}, for deducing that a function $\abs{x}{M}$ is of type $A \from B$, has a premise that is symmetrical to the usual rule for $A \to B$:

  \[
    \prftree[l]{\rlnm{AppL}}{\Gamma \types M : A \from B,\,\Delta}{\Gamma,\,N:A \types \Delta}{\Gamma,\,(M\,N):B \types \Delta}
    \qquad 
    \prftree[l]{\rlnm{AbnR}}{\Gamma,\,M:B \types x:A,\,\Delta}{\Gamma \types (\abs{x}{M}):A \from B,\,\Delta}
  \]

On the other hand, the rule \rlnm{AppL} shows us how to decompose an application on the left, using necessity.  As is typical in sequent calculi, it is instructive to think of a left rule in terms of refuting its principal formula.  To \emph{refute} that an application $M\,N$ evaluates to an $B$, it suffices to \emph{affirm} that $M$ is a function that produces a $B$ only if given an $A$, and then to \emph{refute} that $N$ is an $A$.

The ability to refute that terms evaluate becomes very useful when we consider more sophisticated type systems, designed for verifying stronger, behavioural properties of terms.  In such systems, one faces the same difficulties as in general program verification, namely that proofs are difficult to construct automatically and, as a consequence, false positives -- that is, where a perfectly good program is not well-typed\footnote{A simple example that is rejected by many systems is: $\mathsf{if}\,\mathsf{true}\,\mathsf{then}\,1\,\mathsf{else}\,(1,1)$.} -- are common.  This is problematic because, in practice, a lot of the value of program verification tools derives from their use as a means to find bugs, i.e. true positives.  

One might argue that false positives are already a feature of the mainstream type systems, as e.g. used in compilers, and are thus evidently not too onerous (though proponents of untyped programming may forcefully disagree).  However, a significant difference is that, by and large, programmers understand enough of those systems to be able to predict false positives, or at least quickly diagnose and program around them.  By contrast, type systems for program verification are often significantly more complex and the success of automated tools is bound up with the efficacy of \emph{ad hoc} heuristics.  

With two-sided typing, we can obtain a `true positives' theorem: \emph{ill-typed programs don't evaluate} in addition to the usual `true negatives' theorem: \emph{well-typed programs don't go wrong}. 
Thus, like O'Hearn's Incorrectness Logic for imperative programs \cite{ohearn-popl2019}, our two-sided typing provides some foundation for the use of type systems to predict erroneous program behaviours.  Indeed, one of the original motivations for this work was to understand the basis for Erlang's highly effective Success Typing \cite{lindahl-sagonas-ppdp2006}.

We illustrate this in a two-sided, constrained type system designed for reasoning precisely about the shape of data structures.  Following \citet*{aiken-et-al-popl1994}, our system has a type constructor for every datatype constructor, so that, for example, $[]$ is both the term constructor for the empty list and the type of the empty list, and $A::B$ is the type of all non-empty (cons $::$) lists whose head element has type $A$ and whose tail is of type $B$.  Using the system we can prove, for example, that the head function requires a cons with head element of type $A$ in order to return an $A$, and that the map function requires non-emptiness of its list argument in order to provide a non-empty output:
\[
  \mathsf{head} : \forall a.\, (a::\okty) \from a 
  \qquad \mathsf{map} : \forall a\,b.\, (a \from b) \to (a::\okty) \from (b::\okty)
\] 
More precisely, the combination of \emph{to} and \emph{only to} arrows in $\mathsf{map}$'s type says that when given an $a$ only to $b$ function, it is guaranteed to behave like a function that necessarily requires a cons with element type $a$ in order to produce a cons with element type $b$.  Since our system can refute that the empty list is non-empty, taking the head of a list obtained by mapping the empty list is \emph{ill-typed}: 
\[
  \mathsf{head}\,(\mathsf{map}\, (\abs{x}{x})\, []) : \okty \types
\]
This program, which typically shows as a \emph{false negative} in mainstream type systems for functional programming (i.e. is reported as well-typed but actually goes wrong), is proven by our system to be a true positive: the term is guaranteed to either diverge or go wrong.

\begin{figure}
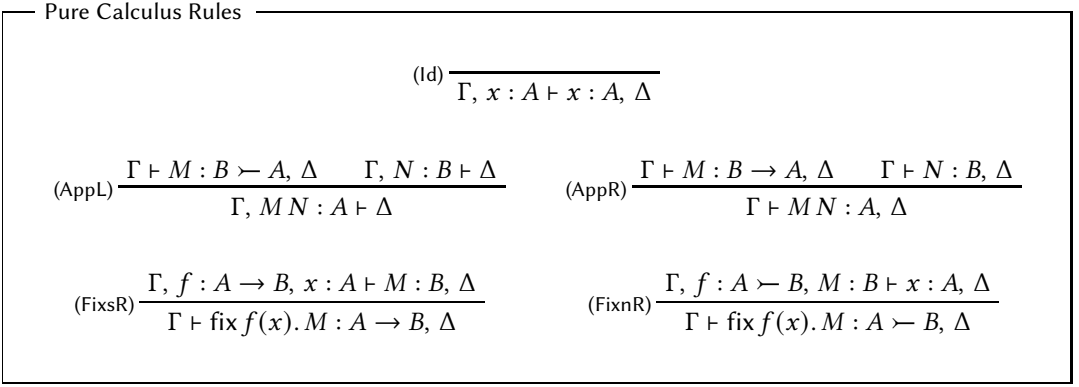

  \begin{mdframed}[topline=false,innertopmargin=-.84ex,innerleftmargin=-0.1ex,innerrightmargin=0ex,linewidth=.5pt,skipbelow=2ex]
  \chA{
    \rulediv{38pt}{Pure Calculus Rules}\\
  \[
    \begin{tblr}{Q[c,18em]Q[c,18em]}
      \SetCell[c=2]{c}\prftree[l]{\rlnm{Id}}{\Gamma,\,x:A \types x:A,\,\Delta}\\[5mm]
      \prftree[l]{\rlnm{AppL}}{\Gamma \types M : B \from A,\,\Delta}{\Gamma,\,N:B \types \Delta}{\Gamma,\,M\,N:A \types \Delta}
      &
      \prftree[l]{\rlnm{AppR}}{\Gamma \types M : B \to A,\,\Delta}{\Gamma \types N : B,\,\Delta}{\Gamma \types M\,N : A,\,\Delta}
      \\[5mm]
      \prftree[l]{\rlnm{FixsR}}{\Gamma,\,f:A \to B,\,x:A \types M : B,\,\Delta}{\Gamma \types \fixtm{f(x)}{M} : A \to B,\,\Delta}
      &
      \prftree[l]{\rlnm{FixnR}}{\Gamma,\,f:A \from B,\,M:B \types x : A,\,\Delta}{\Gamma \types \fixtm{f(x)}{M} : A \from B,\,\Delta}
    \end{tblr}
  \]\\[-2mm]
  }
  \end{mdframed}
  \caption{Typing rules for the pure calculus.}\label{fig:core-rules}
\end{figure}

\subsection{Contributions}
\chA{
We introduce the notion of \emph{two-sided type systems}, which are sequent calculi whose only formulas are typings $M:A$.  We are particularly interested in two-sided systems for call-by-value $\lambda$-calculi with fixpoints, for which we also introduce the necessity function type.

\paragraph{Two two-sided type systems.}
To show that the idea has both theoretical and practical interest, we formulate two distinct two-sided type systems.  Both systems have at their core the two-sided rules of Figure~\ref{fig:core-rules} for call-by-value (CBV), pure $\lambda$-calculus with fixpoints, but we add constants according to the topic of interest.  Note: throughout this paper, we will take fixpoint abstractions $\fixtm{f(x)}{M}$ as basic, and define $\abs{x}{M}$ as an abbreviation for the former in case $f$ does not occur free in $M$\footnote{We are grateful to one of the anonymous reviewers for this suggestion.}.
\begin{itemize}[left=2.5ex]
  \item In Sections~\ref{sec:prog-lang}-\ref{sec:complements} our goal is to illustrate the new kind of type-theoretic reasoning that two-sided systems enable, and to develop some of the metatheory.  To this end, we take the rules for typing pure terms with fixpoints and add constants amounting to a simple, PCF-like language.  
  \item In Section~\ref{sec:constrained-typing}, our goal is to show the practical value of two-sided systems, and  type inference.  To this end, we introduce a two-sided \emph{constrained} type system, which can be understood as a two-sided variant of certain systems in the literature that enable the automatic verification of pattern-match safety\footnote{That is, that no execution of a given term can trigger a pattern match non-exhaustiveness exception at runtime.}.  Starting from the typing rules for pure terms with fixpoints, we add top-level definitions, algebraic data types, and pattern matching.  To make type inference as simple as possible, we formulate this system as an \emph{intuitionistic} sequent calculus, in the sense that we restrict the consequent $\Delta$ in the rules to have size at most $1$.
\end{itemize}
}

\paragraph{The meaning of the necessity function type, and semantic soundness}
We introduce a CBV semantics for the two-sided type system for PCF.  \chD{This semantics highlights an important difference between the meaning of the necessity arrow and all the other type operators in the system.  In the absence of the necessity arrow, all of the types of the system can be thought of as defining safety properties (in a sense we make precise).  However, there are choices of $A$ and $B$ for which $A \from B$ is not a safety property (Theorem~\ref{thm:safety-cexs}).  By restricting the use of function types on the right of the necessity arrow, we obtain a fragment in which all types again describe safety properties (Theorem~\ref{thm:types-are-safety}).  For this fragment we can prove \emph{semantic soundness}: every provable judgement is true (Theorem~\ref{thm:strong-soundness}).  As a corollary, for this fragment we obtain both that \emph{well-typed programs don't go wrong} and also that \emph{ill-typed programs don't evaluate}.}

\paragraph{Complements, negation, and the success semantics}
A natural question is the extent to which a two-sided judgement can be simulated using a one-sided judgement with a type-level negation.  To answer this, we introduce a complement operator on types, $\_^c$, and show that it is inadequate, in the sense that $M:A \types$ and $\types M:A^c$ are not equivalent.  This motivates an alternative, \emph{success semantics} for two-sided judgements, in which showing that a term is typable does not preclude that the term goes wrong.  Under the success semantics, from our two-sided system with complements, we can derive a one-sided (traditional) type system, which actually subsumes the original (Theorem~\ref{thm:two-sided-subsumed}).  A preservation argument gives \emph{syntactic soundness} for the one-sided system, which implies that ill-typed programs don't evaluate, though well-typed programs may yet go wrong (Theorem \ref{thm:failures-soundness}).

\paragraph{Application to pattern-match safety, type inference, and syntactic soundness.}
The second of our two-sided systems allows for precise and automatic reasoning about the shape of data structures through expressive type inclusion constraints.  For this two-sided system, we describe a fully automatic type inference algorithm which, given a term and two environments as input, computes a type assignment that is principal in a suitable sense (Theorem~\ref{thm:inference-correctness}).  This algorithm exploits the fact that our two-sided system has similar characteristics to traditional one-sided systems that admit computable type inference.  By adapting the technique of \citet{wright-felleisen-infcomp1994} to two-sided systems, we prove \emph{syntactic soundness}, from which we obtain that well-typed programs don't go wrong and ill-typed programs don't evaluate (Theorem~\ref{thm:weak-soundness}).

\subsection{Outline}

The paper continues by introducing the basic ideas of two-sided typing through a simple, PCF-like programming language in Sections~\ref{sec:prog-lang} and \ref{sec:okay}, before considering the semantics of such systems in Section~\ref{sec:semantics}.  The relationship between the two-sided judgement and complement types is dealt with in Section \ref{sec:complements}, which introduces the success semantics and derives a one-sided system with a true positives theorem.  Then, in Section~\ref{sec:constrained-typing}, we then turn to an application of two-sided typing, namely automatic and precise reasoning in a functional programming language with constructors and pattern matching.  Finally, Section~\ref{sec:conclusion} concludes with a review of related work.

%% file: prog-lang.tex
\section{A Language and its Two-sided Typing}\label{sec:prog-lang}

\iftoggle{arxiv}{We start by introducing a PCF-like, call-by-value programming language.}{We start by introducing a PCF-like, call-by-value programming language which we shall use to illustrate some basic considerations of two-sided typing.}

\begin{definition}[Terms]
  Assume a denumerable set of term variables $x$, $y$, $z$.  
  We consider a simple functional programming language, whose \emph{values}, typically $V$, $W$, and whose \emph{terms}, typically $M$, $N$, $P$, $Q$, are defined by the following grammar:
    \[
        \begin{array}{rcl}
            V,\,W &\Coloneqq& x \mid \succtm^n(\zerotm) \mid (V,\,W) \mid \ch{\fixtm{f(x)}{M}}\\[1mm]
          M,N,P,Q &\Coloneqq& \zerotm \mid \succtm(M) \mid \predtm(M) \mid  \iftm{M}{N}{P} \mid\\[.75mm]
                   &        & \qquad x \mid M\,N \mid \ch{\fixtm{f(x)}{M}} \mid (M,N) \mid \lettm{(x,y)}{M}{N}
        \end{array}
    \]
    We consider terms identified up to renaming of bound variables and adopt the usual conventions regarding their treatment.  A term with no free variables is said to be \emph{closed}.
    A \emph{term substitution}, typically $\theta$, is finite map from term variables to terms.  We say that a substitution is \emph{closed} just if the terms in its range are closed.  We will often write concrete substitutions explicitly as a list of maplets as in $[M_1/x_1,\ldots,M_n/x_n]$, or $[\vv{M_i/x_i}]$ when indices are clear.

\end{definition}

\paragraph{Abbreviations}
It will be helpful to define some abbreviations.  For each $n\in\mathbb{N}$, we write $\pn{n}$ for the \emph{numeral} $\succtm^n\,(\zerotm)$.  \ch{We write $\abs{x}{M}$ for $\fixtm{f(x)}{M}$ when $f$ does not occur free in $M$}.  The let construct is used to obtain the two components of a pair by matching but, in examples, we will typically always use it on the argument of an abstraction.  \ch{Hence, it is helpful to define $\fixtm{f(x,y)}{M}$ as an abbreviation for $\fixtm{f(z)}{\lettm{(x,y)}{z}{M}}$.  We write $\divtm$ as an abbreviation for $(\fixtm{f(x)}{f\,x})\,\pn{0}$, and $\abbv{id}$ for $\abs{x}{x}$.}

The language is essentially an applied CBV $\lambda$-calculus with fixpoints.  Terms $M$ can be tested for equality with $\zerotm$ using the if-zero expression $\iftm{M}{N}{P}$, and can be deconstructed as a pair using the let expression $\lettm{(x,y)}{M}{N}$.  Closed values are numerals, pairs, and (fixpoint) abstractions.



\begin{definition}[Reduction]
    The \emph{evaluation contexts} are defined by the following grammar:
    \[
      \arraycolsep=1.4pt
  \begin{array}{rcl}
    \mathcal{E},\,\mathcal{F} &\Coloneqq& \Box \mid \ch{(\fixtm{f(x)}{M})\,\mathcal{E}} \mid \mathcal{E}\,N \mid \succtm(\mathcal{E}) \mid \predtm(\calE) \mid \\[1mm] 
    && \qquad  (\mathcal{E},M) \mid (V,\mathcal{E}) \mid \lettm{(x,y)}{\calE}{N} \mid \iftm{\mathcal{E}}{N}{P} 
  \end{array} 
  \]
  Given a context $\calE$, we write $\calE[M]$ for the term obtained by replacing the hole $\Box$ by $M$.  The \emph{one-step reduction relation}, written $M \ped N$, is the binary relation on (possibly open) terms obtained as the closure of the following schema under evaluation contexts.\\
  \[
    \arraycolsep=1.4pt
  \begin{array}{ccrcl}
    \rlnm{IfZ} & \phantom{woo} & \iftm{\pn{0}}{N}{P} &\ped& N \\[2mm]
    \rlnm{IfS} && \iftm{\succtm(\pn{n})}{N}{P} &\ped& P\\[2mm]   
    \rlnm{Let} && \lettm{(x,y)}{(V,W)}{M} &\ped& M[V/x,W/y] \\[2mm]
    \rlnm{PredZ} && \predtm(\pn{0}) &\ped& \pn{0}\\[2mm] 
    \rlnm{PredS} && \predtm(\pn{n+1}) &\ped& \pn{n}\\[2mm]
    \rlnm{Fix$\beta$} && \ch{(\fixtm{f(x)}{M})\,V} &\ped& \ch{M[V/x,\fixtm{f(x)}{M}/f]}
  \end{array}
\]\\
  The terms on the left-hand side of the above schema are said to be \emph{redexes}.
  We write $M \peds N$ for the reflexive transitive closure of the one-step relation. 
  A term $M$ that cannot make a step is said to be in \emph{normal form}.  A term with no normal form is said to \emph{diverge}.
  We write $M \evals V$ just if closed term $M \peds V$ and we say that $M$ \emph{evaluates}.
  A term is said to be \emph{stuck} just if it is a normal form that is not a value.  A term is said to \emph{get stuck} (or \emph{go wrong}) just if it reduces to a stuck term.
\end{definition}

\begin{figure}\vspace{1ex}
\begin{mdframed}[topline=false,innertopmargin=-.84ex,innerleftmargin=-.1ex,innerrightmargin=0ex,linewidth=.5pt]
\noindent
\rulediv{33pt}{Structural Rules}\\[3mm]
\[
      \prftree[l]{\rlnm{Id}}{\Gamma,\,x:A \types x:A,\,\Delta}
\]\\[2mm]
\rulediv{35pt}{Right Typing Rules}\\[3mm]
\[
  \begin{tblr}{Q[c,16em]Q[c,18em]}
    \prftree[l]{\rlnm{ZeroR}}{\Gamma \types \zerotm : \intty,\,\Delta}
    &
    \prftree[l,r]{$c \in \left\{\!\!\begin{array}{c}\succtm\\\predtm\end{array}\!\!\right\}$}{\rlnm{cR}}{\Gamma \types M : \intty,\,\Delta}{\Gamma \types c(M) : \intty,\,\Delta}
    \\[5mm]
    \ch{\prftree[l]{\rlnm{FixsR}}{\Gamma,\,f:A \to B,\,x:A \types M : B,\,\Delta}{\Gamma \types \fixtm{f(x)}{M} : A \to B,\,\Delta}}
    &
    \ch{\prftree[l]{\rlnm{FixnR}}{\Gamma,\,f:A \from B,\,M:B \types x : A,\,\Delta}{\Gamma \types \fixtm{f(x)}{M} : A \from B,\,\Delta}}
    \\[5mm] 
    \prftree[l]{\rlnm{AppR}}{\Gamma \types M : B \to A,\,\Delta}{\Gamma \types N : B,\,\Delta}{\Gamma \types M\,N : A,\,\Delta}
    &
    \prftree[l]{\rlnm{PairR}}{\Gamma \types M : A,\,\Delta}{\Gamma \types N : B,\,\Delta}{\Gamma \types (M,\,N) : A \times B,\,\Delta}
    \\
    \prftree[l]{\rlnm{LetR}}{\begin{array}{c}\\\Gamma \types M : B \times C,\,\Delta\\\Gamma,\,x:B,\, y: C \types N : A,\,\Delta\end{array}}{\Gamma \types \lettm{(x,y)}{M}{N} : A,\,\Delta}
    &
    \prftree[l]{\rlnm{IfzR}}{\begin{array}{c}\\\Gamma \types M:\intty,\,\Delta\end{array}}{\!\!\!\!\begin{array}{c}\Gamma \types P:A,\,\Delta\\\Gamma \types N:A,\,\Delta\end{array}}{\Gamma \types \iftm{M}{N}{P} : A,\,\Delta}
  \end{tblr}
\]\\[7mm]
\rulediv{35pt}{Left Typing Rules}\\[5mm]
\[
  \begin{tblr}{Q[c,16em]Q[c,18em]}
    \prftree[l,r]{$A \distype B$}{\rlnm{Dis}}{\Gamma \types M:B,\,\Delta}{\Gamma,\,M:A \types \Delta}
    &
    \prftree[l,r]{$i\in\{1,2\}$}{\rlnm{PairL}}{\Gamma,\,M_i:A_i \types \Delta}{\Gamma,\,(M_1,\,M_2) : A_1 \times A_2 \types \Delta}
    \\[5mm]
    \prftree[l]{\rlnm{AppL}}{\Gamma \types M : B \from A,\,\Delta}{\Gamma,\,N:B \types \Delta}{\Gamma,\,M\,N:A \types \Delta}
    &
      \prftree[l,r]{$c \in \left\{\!\!\begin{array}{c}\succtm\\\predtm\end{array}\!\!\right\}$}{\rlnm{cL}}{\Gamma,\,M:\intty \types \Delta}{\Gamma,\,c(M) : \intty \types \Delta}
    \\[5mm]
      \prftree[l]{\rlnm{IfzL1}}{\Gamma,\,M:\intty \types \Delta}{\Gamma,\,\iftm{M}{N}{P} : A \types \Delta}
    &
      \prftree[l]{\rlnm{IfzL2}}{\Gamma,\,N:A \types \Delta}{\Gamma,\,P:A \types \Delta}{\Gamma,\,\iftm{M}{N}{P} : A \types \Delta} 
    \\[5mm]
      \prftree[l]{\rlnm{LetL1}}{\Gamma,\,N:A \types \Delta}{\Gamma,\,\lettm{(x,y)}{M}{N} : A \types \Delta}
    &
    \prftree[l]{\rlnm{LetL2}}{\begin{array}{c}\Gamma,\,M:B_1 \times B_2 \types \Delta\\\Gamma,\,N:A\types x_i:B_i,\,\Delta\;(\forall i)\end{array}}{\Gamma,\,\lettm{(x_1,x_2)}{M}{N} : A \types \Delta}
  \end{tblr}
\]\\[5mm]
\end{mdframed}
\caption{Two-sided type assignment.}\label{fig:pcf-typing}
\end{figure}

\subsection{Two-Sided Type System}
We introduce a simple two-sided type system for this language.  \chA{Starting from the rules for pure terms of Figure~\ref{fig:core-rules}, we add rules for assigning types to terms headed by the constants of our language, and a rule allowing for reasoning by disjointness (which is discussed further in Section~\ref{subsec:discussion-one}).}  The types of the system are either the base type $\intty$, or constructed using product or the two arrows.

\chD{
\begin{definition}[Types]\label{def:simple-types}
    The \emph{types}, typically $A$, $B$ and so on, are defined by the following grammar: 
    \[
        \begin{array}{rcl}
             A,B &\Coloneqq& \intty \mid A \times B \mid A \to B \mid A \from B 
        \end{array}
    \]
\end{definition}
}

We will label the usual arrow type, $A \to B$, the \emph{sufficiency arrow} to distinguish it from the \emph{necessity arrow}, $A \from B$.  We pronounce the first of these types as `$A$ to $B$' and the second as `$A$ only to $B$'\footnote{Strictly speaking, for the English to be correct, one should pronounce it as `$A$ only, to $B$', but we do not recommend it.}.   We assume all arrows associate to the right and products bind tighter than the other operators.  

The idea of $A \from B$, which will be made precise in Section~\ref{sec:semantics}, is to classify those functions which return a $B$ \emph{only if} an $A$ was supplied as input (hence the pronunciation).  An example is the function $\abs{x}{(x,\pn{0})}$ which can be assigned the type $\intty \from \intty \times \intty$. If this function returns a pair of numerals, its input must have been a numeral.  On the other hand, $\abs{x}{(\pn{0},\pn{1})}$ cannot be assigned this type, because if it returns a pair of numerals, it is \emph{not necessary} that a numeral was input.


We shall not consider subtyping for this system, but it is useful to define a notion of type \emph{disjointness}.  \chB{We will discuss the motivation for this further in Section~\ref{subsec:discussion-one} but, intuitively, it is a weak axiomatisation of the assertion that two types do not have any values in common.}

\begin{definition}[Disjointness]\label{def:disjointness}
  We say that two types $A$ and $B$ are \emph{disjoint}, and write $A \distype B$, just if either (i) one is $\intty$ and the other is not, or (ii) one is an arrow and the other is not, or (iii) one is a product and the other is not.
\end{definition}

Now, we come to define the two-sided type system.  The idea is that the system is a simple kind of sequent calculus whose only formulas are typings $M:A$.  

\begin{definition}[Type Assignment]
  A \emph{typing formula}, or just \emph{typing}, is a pair $M:A$ of a term $M$ and type $A$.  The term $M$ is said to be the \emph{subject} of the formula.  
  A \emph{typing judgement} is a pair of finite sets of typings, written $\Gamma \types \Delta$.
  A judgement is said to be \emph{provable} (or \emph{derivable}) according to the rules in Figure~\ref{fig:pcf-typing}.  We make some additional requirements on the occurrences of free variables that are omitted from the rules for typesetting reasons: in every use of \rlnm{AbsR}, \rlnm{AbnR}, \rlnm{LetR}, \rlnm{LetL1}, \rlnm{LetL2} and \rlnm{FixR}, we require that the bound variables displayed do not occur freely in $\Gamma$ or $\Delta$.
\end{definition}

In a traditional, one-sided type system, a judgement like $x:A,\,y:B \types P:C$ is conventionally understood as `if x has type $A$ and $y$ has type $B$, then $P$ has type $C$'.  Here, we generalise this in two ways.  First, we allow arbitrary terms on the left and not only variables.  However, the judgement can be read in the same way.  So $M:A,\,N:B \types P:C$ should be understood as:
\begin{quote}
   `if $M$ has type $A$ and $N$ has type $B$, then $P$ has type $C$'.
\end{quote}  
Second, we allow for multiple typings (including none) on the right of the turnstile.  These are understood disjunctively, so  $M:A \types P:C,\,Q:D$ should be read as:
\begin{quote} 
  `if $M$ has type $A$, then \emph{either} $P$ has type $C$ \emph{or} $Q$ has type $D$'.  
\end{quote}
In particular, $M:A,\,N:B \types $, which has an empty conclusion, should be understood as `if $M$ has type $A$ and $N$ has type $B$ then false', i.e. either $M$ cannot have type $A$ or $N$ cannot have type $B$.

However, we must take care to say something about the meaning of `$M$ has type $A$' within a judgement.  Since our language is call-by-value, there is an asymmetry between the left- and right-hand side.  On the left-hand side of the turnstile `$M$ has type $A$' means \emph{$M$ evaluates to an $A$}, but on the right-hand side, it means \emph{either $M$ diverges or $M$ evaluates to an $A$}.  This is made precise in Section~\ref{sec:semantics}, but it is useful to have some intuitions already.

For example, suppose we have some term $\abbv{add}$ and we know that $\abbv{add} : \intty \times \intty \to \intty$.  Then, we should expect that the following judgement is provable:
\[
  \succtm(\succtm(z)) : \intty,\; \iftm{y}{\pn{0}}{\pn{1}} : \intty \types \abbv{add}\,(y,z) : \intty
\]
Intuitively, if $\succtm(\succtm(z))$ is a numeral, it must be that $z$ was a numeral, and if the term $\iftm{y}{\pn{0}}{\pn{1}}$ evaluates to a numeral, it must be that $y$ is a numeral too (recall that the conditional of the language requires the guard to be a numeral).  Therefore, $\abbv{add}\,(y,z)$ should  evaluate to a numeral.  This is the kind of reasoning that we are trying to capture with the rules of Figure~\ref{fig:pcf-typing}.

The rules have been divided into three kinds, \emph{structural rules}, \emph{right rules} and \emph{left rules}.  Typically, the rules are used bottom-to-top in order to construct a derivation.  The right rules should be familiar from simple type systems, and are for concluding typings on the right of the turnstile; they only differ in that they allow for multiple conclusions.  One can think of the purpose of left rules as for refuting typings on the left of the turnstile.  However, in practice, they are more often used as a means to extract information from assumptions.  

\paragraph{Structural Rules} The rule \rlnm{Id} should be familiar.  The only difference is that now we allow for multiple formulas in the conclusion, but since they are understood disjunctively, all is well.

\paragraph{Left Rules} The \rlnm{Dis} rule expresses a key mechanism for refuting that a term $M$ evaluates to an $A$, namely by affirming that $M$ either diverges or evaluates to a value of some disjoint type $B$.  \chB{The motivation for this rule will be further discussed in Section~\ref{subsec:discussion-one}.}

The rule \rlnm{SuccL} allows for refuting that a term of shape $\succtm(M)$ has type $\intty$.  The only way that $\succtm(M)$ may \emph{not} evaluate to a numeral is if $M$ does not evaluate to a numeral, so refuting $\succtm(M) : \intty$ reduces to refuting $M : \intty$.  The rule \rlnm{PredL} is similar.  

The rule \rlnm{AppL} says that, to refute that $M\,N$ evaluates to an $A$, we can do the following.  First show that $M$ either diverges or requires an input of type $B$ in order to produce an $A$, then refute that $N$ evaluates to a $B$.  Notice the combination of affirmation and refutation required in this rule.  

The rule \rlnm{PairL} says that to refute that $(M_1,M_2)$ evaluates to a value of type $A_1 \times A_2$ we need only refute that \emph{one} component of the pair evaluates accordingly.  

The rules \rlnm{LetL1} and \rlnm{LetL2} allow for reasoning about let expressions on the left.  To refute that $\lettm{(x,y)}{M}{N}$ evaluates to an $A$, one can refute that the body evaluates to an $A$ (independently of the values of $x$ and $y$); this is \rlnm{LetL1}.  Alternatively, one can refute that $M$ evaluates to a pair $(V,W)$ of the appropriate kind to allow $N[V/x,W/y]$ to evaluate to an $A$.  We use the `if ... then ...' reading of the judgement, in order to explain how the rule implements this strategy.  The first two premises of \rlnm{LetL2} say that, if $N$ evaluates to an $A$ then it must be that $x$ was a $B$ and that $y$ was a $C$.  In other words, in showing these two premises, we are showing that for $N$ to evaluate to an $A$ as in the conclusion, $(x,y)$ is necessarily a value of $B \times C$.  Then in the third premise we refute that $M$ evaluates to a pair of this form. 

Finally, the rules \rlnm{IfZL1} and \rlnm{IfZL2} are used for reasoning about the if-zero expression on the left.  To refute that $\iftm{M}{N}{P}$ evaluates to an $A$, we can refute that $M$ evaluates to a numeral, as in \rlnm{IfZL1}.  Alternatively, we can show that neither branch evaluates to an $A$, which is \rlnm{IfZL2}.

\paragraph{Right Rules.} The rule \rlnm{FixnR} allows for affirming the necessity function type of a fixpoint abstraction.  Otherwise, the right rules are mostly standard, excepting that they allow for the more general kind of judgement.  

\chB{
In sequent calculi defined over sequences or multisets of formulas, it is sometimes necessary to add structural rules such as exchange, weakening and contraction.  However, since our sequents consist of two \emph{sets} of formulas, only weakening is relevant. 
We do not have explicit weakening rules, but since all the rules allow for arbitrary contexts $\Gamma$ and $\Delta$ on the left and right respectively, weakening is admissible.

\begin{lemma}[Weakening]\label{lem:structural-admissibility}
  If $\Gamma \types \Delta$, then $\Gamma,\,M:A \types \Delta$ and $\Gamma \types M:A,\,\Delta$.
\end{lemma}
}



\paragraph{Abstractions} \ch{Recall that we treat $\abs{x}{M}$ as an abbreviation for $\fixtm{f(x)}{M}$, when $f$ does not occur free in $M$.  In light of Lemma~\ref{lem:structural-admissibility}, it is easy to see that the following typing rules are admissible: 
\[
  \prftree[l]{\rlnm{AbsR}}{\Gamma,\,x:B \types M:A,\,\Delta}{\Gamma \types \abs{x}{M}:B \to A,\,\Delta}
    \qquad
  \prftree[l]{\rlnm{AbnR}}{\Gamma,\,M:A \types x:B,\,\Delta}{\Gamma \types \abs{x}{M}:B \from A,\,\Delta}
\]
Each requires that $x$ does not occur free in $\Gamma$ or $\Delta$.
Recall that we use $\fixtm{f(x,y)}{M}$ as an abbreviation for $\fixtm{f(z)}{\lettm{(x,y)}{z}{M}}$.  In the following it will be efficient to use the derived typing rule:
\[
  \prftree[l]{\rlnm{PFixnR}}{\Gamma,\,f:\,B \times C \from A,\,M:A \types x : B,\,\Delta}{\Gamma,\,f:B \times C \from A,\,M:A \types y:C,\,\Delta}{\Gamma \types \fixtm{f(x,y)}{M} : B \times C \from A,\,\Delta}
\]
}

\begin{example}\label{ex:add-typing}
We can define addition $\abbv{add} = \fixtm{f(y,z)}{\iftm{y}{z}{\succtm\,(f\,(\predtm(y),z))}}$.
To return a numeral, both its arguments must be numerals.  First observe that, in the `else' branch of the conditional, if we know that the recursive call requires a pair of numerals in order to produce a numeral, then $z$ must be a numeral, we use $\Gamma$ to abbreviate $\{\, f : \intty \times \intty \from \intty \,\}$ in:
\[
  \prftree[l]{\rlnm{SuccL}}{
            \prftree[l]{\rlnm{AppL}}{
              \prftree[l]{\rlnm{Id}}{\Gamma \types f:\intty \times \intty \from \intty,\,z:\intty}
            }
            {
              \prftree[l]{\rlnm{PairL}}{
                \prftree[l]{\rlnm{Id}}{\Gamma,\,z:\intty \types z : \intty}
              }
              {
                \Gamma,\,(\predtm(y),z) : \intty \times \intty \types z : \intty
              }
            }
            {
               \Gamma,\,f\,(\predtm(y),z) : \intty \types z : \intty
            }
          }
          {
            \Gamma,\,\succtm(f\,(\predtm(y),z)) : \intty \types z : \intty
          }
\]
The same can be said in the `then' branch of the conditional, since the `then' branch actually returns $z$.  So we can use \rlnm{IfZL2} to conclude that $z$ is certainly a numeral whenever the conditional expression evaluates to a numeral:
\[
  \prftree[l]{\rlnm{IfZL2}}{
    \prftree[l]{\rlnm{Id}}{\Gamma,\,z:\intty \types z:\intty}
  }
  {
    \prftree{\cdots}{\Gamma,\,\succtm(f\,(\predtm(y),z)) : \intty \types z : \intty}
  }
  {
    \Gamma,\,\iftm{y}{z}{\succtm(f\,(\predtm(y),z))} : \intty \types z : \intty
  }
\]
Abbreviating the `else' branch for brevity, we can obtain $y:\intty$ from the fact that the conditional branches on $y$, using \rlnm{IfZ1}, and conclude using our derived rule \rlnm{PFixnR} for recursive functions requiring pairs as their input:
\[
    \prftree[l]{\rlnm{PFixnR}}{
    \prftree[l]{\rlnm{IfZL1}}{
      \prftree[l]{\rlnm{Id}}{
        \Gamma,\,y : \intty  \types y : \intty
      }
    }
    {
      \Gamma,\,\iftm{y}{z}{\ldots} : \intty  \types y : \intty
    }
  }
  {
    \prftree{\cdots}{\Gamma,\,\iftm{y}{z}{\ldots} : \intty  \types z : \intty}
  }
  {
    \types \abbv{add} : \intty \times \intty \from \intty
  }
\]
\end{example}

\begin{example}\label{ex:twice-typing}
We can define the higher-order combinator $\abbv{twice} = \abs{fx}{f\,(f\,x)}$.  It has the following property.  For any type $A$, it guarantees to map all functions of type $A \from A$ to an output function of the same type.  We abbreviate $\Gamma = \{f:A \from A\}$ in the proof:
\[
  \prftree[l]{\rlnm{AbsR}}{
    \prftree[l]{\rlnm{AbnR}}{
      \prftree[l]{\rlnm{Id}}{
        \prftree[l]{\rlnm{Id}}{\Gamma \types f : A \from A,\,x:A}
      }
      {
        \prftree[l]{\rlnm{AppL}}
        {
          \prftree[l]{\rlnm{Id}}{\Gamma \types f : A \from A,\,x:A}
        }
        {
          \prftree[l]{\rlnm{Id}}{\Gamma,\,x:A \types x : A}
        }
        {
          \Gamma,\,f\,x:A \types x : A
        }
      }
      {
        \Gamma,\,f\,(f\,x):A \types x : A
      }
    }
    {
      \Gamma \types \abs{x}{f\,(f\,x)} : A \from A
    }
  }
  {
    \types \abs{fx}{f\,(f\,x)} : (A \from A) \to A \from A
  }
\]
\end{example}

\begin{example}\label{ex:twice-pred-nat}
We can prove that applying predecessor twice to a pair does \emph{not} evaluate to a natural number.  From the foregoing example, we know that $\abbv{twice}\,(\abs{x}{\predtm(x)})$ is guaranteed to behave like a function that, to produce a numeral, requires a numeral as input:  
\[
  \prftree[l]{\rlnm{AppR}}{
    \prftree{
      \cdots
    }
    {
      \types \abbv{twice} : (\intty \from \intty) \to \intty \from \intty
    }
  }
  {
    \prftree[l]{\rlnm{AbnR}}{
      \prftree[l]{\rlnm{PredL}}{
        \prftree[l]{\rlnm{Id}}{x:\intty \types x:\intty}
      }
      {
        \predtm(x) : \intty \types x : \intty
      }
    }
    {
      \types (\abs{x}{\predtm(x)}) : \intty \from \intty
    }
  }
  {
    \types \abbv{twice}\;(\abs{x}{\predtm(x)}) : \intty \from \intty
  }
\]
Note, since our predecessor is a constant of fixed arity, we need to wrap it in an abstraction to provide it as an input. 
However, a pair is not a numeral, so we can use disjointness:  
\[
  \prftree[l]{\rlnm{AppL}}{
    \prftree{\cdots}{\types \abbv{twice}\;(\abs{x}{\predtm(x)}) : \intty \from \intty}
  }
  {
    \prftree[l]{\rlnm{Dis}}{
      \prftree{\cdots}{\types (\pn{0},\,\pn{1}) : \intty \times \intty}
    }
    {
      (\pn{0},\,\pn{1}) : \intty \types 
    }
  }
  {
    \abbv{twice}\;(\abs{x}{\predtm(x)})\,(\pn{0},\,\pn{1}) : \intty \types
  }
\]
\end{example}

In fact, not only does $\abbv{twice}\;(\abs{x}{\predtm(x)})\,(\pn{0},\,\pn{1})$ not evaluate a numeral, it does not evaluate at all -- it goes wrong.  In the next section we enrich our system with the means to express this

\chB{
\subsection{Further discussion of the rules}\label{subsec:discussion-one}
}

\begin{figure}
  \chB{
  \[
    \begin{tblr}{Q[c,18em]Q[c,18em]}
    \prftree[l,r]{$A \neq \intty$}{\rlnm{ZeroL}}{\Gamma,\,\zerotm : A \types \Delta}
    &
    \prftree[l,r]{$\begin{array}{l}A \neq \intty\\c \in \{\succtm,\predtm\}\end{array}$}{\rlnm{cL2}}{\Gamma,\,c(M) : A \types \Delta}
    \\[5mm]
    \prftree[l,r]{$A$ not a product}{\rlnm{PairL2}}{\Gamma,\,(M,N) : A \types \Delta}
    &
    \prftree[l,r]{$A$ not an arrow}{\rlnm{FixL}}{\Gamma,\,\fixtm{f(x)}{M} : A \types \Delta}
    \end{tblr}
  \]
  }
  \caption{Alternative disjointness rules.}\label{fig:disjointness-rules}
\end{figure}

\chB{
Notice that the \rlnm{ZeroR} rule and the two fixpoint rules \rlnm{FixsR} and \rlnm{FixnR} have no corresponding rule on the left-hand side, the left rules \rlnm{predL} and \rlnm{succL} only concern the case when the type is $\intty$, and the left rule \rlnm{PairL} only concerns the case when the type is a product\footnote{This can be considered unsatisfactory because, for example, unlike when we are affirming $\succtm(M) : A$ in the consequent, to refute $\succtm(M) : A$ (i.e. on the left), it need not be the case that $A=\intty$.  Similar remarks apply for $\predtm$ and pair values.}.  

To understand the situation better, let us consider each in turn.  To refute $\zerotm : A$, it must be that $A \neq \intty$, so we might expect a rule like \rlnm{ZeroL} in Figure~\ref{fig:disjointness-rules}.  Similarly, to refute a formula $c(M) : A$ (for $c \in \{\succtm,\,\predtm\}$), either (i) $A = \intty$ and then it must be that we can refute that $M$ is a number, as in the existing rule \rlnm{cL}, or (ii) it must be that $A \neq \intty$.  Hence, we appear to be missing a rule such as \rlnm{cL2} in Figure~\ref{fig:disjointness-rules} above.  A similar analysis leads us to conclude that we are missing a rule for refuting formulas $(M,N) : A$ when $A$ is not a product type, such as \rlnm{PairL2} of Figure~\ref{fig:disjointness-rules}.    

All the rules from Figure~\ref{fig:disjointness-rules} that we have noted as missing  express a form of reasoning based on disjointness of types: one can refute $\zerotm : A$, $\succtm(M) : A$ and $\predtm(M) : A$ whenever $A$ is a type disjoint from $\intty$, one can refute $(M,N) : A$ whenever $A$ is a type disjoint from any product type.  It just happens that, in this system, $A$ being disjoint from $\intty$ simply means that $A \neq \intty$.  

However, reasoning based on disjointness of types makes sense for all terms $M$, and not only those that are introduction forms.  In principle, it should be possible to refute any typing $M : A$ on the basis of affirming a typing $M : B$ for some type $B$ disjoint from $A$.  Indeed, this form of reasoning is closely related to that embodied in the success typing paradigm \cite{lindahl-sagonas-ppdp2006}.  In success typing, if a function $M$ can be assigned the \emph{success type} $A \to B$, then it means that the domain of $f$ is wholly contained within $A$.  Therefore, we can be sure that $f$ fails whenever it is applied to a value whose type is disjoint from $A$.

Thus, we formulated our system with a single rule \rlnm{Dis}, expressing a general notion of reasoning based on disjointness, instead of a number of special cases for each introduction form, as in Figure~\ref{fig:disjointness-rules}.  However, in expressive power, the two are actually incomparable.  On the one hand, the rules of Figure~\ref{fig:disjointness-rules} are applicable even when the subject $M$ of the principal formula would not be well-typed.  For example, $\succtm(\abbv{id}) : \intty \times \intty$ can be refuted using the disjointness rules of Figure~\ref{fig:disjointness-rules}, but not using \rlnm{Dis}.  On the other hand, \rlnm{Dis} allows for reasoning based on disjointness even when the subject of the principal formula is not an introduction form.  For example, $(\abs{xy}{x})\,2\,3 : \intty \times \intty$ can be refuted using \rlnm{Dis}, but not if we were to take the rules of Figure~\ref{fig:disjointness-rules} instead of \rlnm{Dis}.  Of course, if one is interested in the most powerful system possible, then one can throw all the rules in together.

Now, let us consider refuting a typing of shape $\fixtm{f(x)}{M} : A$.  By a similar analysis, we should expect that 
either (i) $A$ is not an arrow at all; or (ii) $A$ is a sufficiency arrow, say $B \to C$, but then there is a $B$-value $V$ such that $M[V/x]$ is not a $C$; or (iii) $A$ is a necessity arrow, say $B \from C$, but then there is a non-$B$ value $V$ such that $M[V/x]$ is a $C$.  For (i), it seems we should expect a rule like \rlnm{FixL} of Figure~\ref{fig:disjointness-rules}, but instead we use the \rlnm{Dis} rule, as discussed above.  For case (ii),  it seems we should have a rule:
\[
  \prftree{
    \Gamma \types V : B,\,\Delta
  }
  {
    \Gamma,\,f: B \to C,\,M[V/x] : C \types \Delta
  }
  {
    \Gamma,\,\fixtm{x}{M} : B \to C \types \Delta
  }
\]
However, this rule would be unsound.  Suppose we have some $B$ value $V$ such that $M[V/x]$ does not evaluate to a $C$, and suppose, for simplicity, that $M$ does not contain $f$ (so the fixpoint abstraction is just an ordinary $\lambda$-abstraction).  This rule would allow us to conclude that $\abs{x}{M}$ is not a $B \to C$ function.  However, it is not enough to know that $M[V/x]$ does not evaluate to a $C$, because this includes the possibility that $M[V/x]$ diverges, and this is a perfectly acceptable outcome for any $B \to C$ function.  For example, using this rule we could refute $\abs{x}{\divtm} : \intty \to \intty$, i.e. prove the judgement $\abs{x}{\divtm} : \intty \to \intty \types$, but this is absurd because $\abs{x}{\divtm}$ \emph{is} a $\intty \to \intty$ function, according to the usual understanding of this type (and our semantics in Section~\ref{sec:semantics}).   The situation for refuting typings of shape $\fixtm{f(x)}{M} : B \from C$, i.e. case (iii) above, is analogous.  

Finally, although we formulated our rule \rlnm{Id} rule with principal formula $x:A$ on both sides, it would be sound to take the more general $\Gamma,\,M:A \types M:A,\,\Delta$ as axiom instead.  This would be more powerful, e.g. $x\ y : A \types x\ y: A$, but it takes us further away from traditional one-sided type systems than we want to go in this paper.}\\






%% file: okay.tex
\section{The type of values and Ill Typedness}\label{sec:okay}
\chA{One may wonder whether there is any use in being able to reason hypothetically about typing.  A compelling answer can be given once we introduce a new type $\okty$, with meaning `evaluates'.}

\begin{definition}\label{defn:ill-typing-system}
  We extend the grammar of types with the constant $\okty$ and extend the typing rules to include those given in Figure~\ref{fig:ok-rules}.
\end{definition}

The rule \rlnm{OkVarR} expresses the fact that, in call-by-value, we may assume that all variables denote values.  The rules \rlnm{OkL} and \rlnm{OkR} express the fact that $\okty$ is the type of all values.  To conclude that $M:\okty$, it suffices to show $M:A$ since every type $A$ represents some non-empty set of values. To refute that $M$ evaluates to an $A$, it suffices to refute that it evaluates at all.  Rules \rlnm{OkApL1} and \rlnm{OkApL2} express the fact that, for an application $M\,N$ to evaluate, requires that $M$ evaluates to a function (see Remark~\ref{rmk:ok-semantics-val-sets}) and that $N$ evaluates.  Rules \rlnm{OkSL} and \rlnm{OkPL} express the fact that the successor and predecessor are only defined on numerals.  In combination with \rlnm{OkL}, these rules subsume \rlnm{SuccL} and \rlnm{PredL} respectively.  Finally, \rlnm{OkPrL} allows to refute that $(M_1,M_2)$ evaluates by refuting that one of the components evaluates.

\begin{remark}\label{rmk:ok-semantics-val-sets}
  Since $\okty$ is the set of all values, the set of all (fixpoint) abstractions can be represented by any type of shape $\okty \from A$.  Indeed, using \rlnm{FixnR} and \rlnm{OkVarR} we can prove $\types \fixtm{f(x)}{M} : \okty \from A$. 
\end{remark}

For example, if we use \rlnm{OkPL} instead of \rlnm{PredL} in the first derivation of Example~\ref{ex:twice-pred-nat} we can show that $\types \abbv{twice}\,(\abs{x}{\predtm(x)}) : \intty \from \okty$.  Then, using \rlnm{AppL} and \rlnm{Dis}, we can conclude $\abbv{twice}\,(\abs{x}{\predtm(x)})\,(\pn{0},\pn{1}) : \okty \types$.  

In Example~\ref{ex:add-typing}, we proved that \abbv{add} requires both its arguments to be numerals in order to return a numeral.  However, we cannot prove that it requires both its arguments to be numerals in order to evaluate.  To see why, suppose the body of the function, $\iftm{y}{z}{\succtm\,(f\,(\predtm(y),z))}$, evaluates for some actual parameters $y$ and $z$.
Clearly, $y$ is required to be a numeral since it appears in the guard.  However, $z$ is \emph{not} required to be a numeral by the branches of the conditional.  In particular, the `then' branch yields $z:\okty \not\types z:\intty$.  However, there is a good reason, it is possible for an application of the function to successfully evaluate even when the second argument is not a numeral!  For example, $\abbv{add}(\pn{0},\abs{x}{x}) \evals \abs{x}{x}$.
The best we can do is prove $\types \abbv{add} : \intty \times \okty \from \okty$.

\begin{figure}
  \begin{mdframed}[topline=false,innertopmargin=-.84ex,innerleftmargin=-0.1ex,innerrightmargin=0ex]
    \rulediv{40pt}{Typing Rules for Ok}\\[2mm]
    \[
      \begin{tblr}{ccc}
      \prftree[l]{\rlnm{OkVarR}}{\Gamma \types x : \okty,\,\Delta}
      &
      \prftree[l]{\rlnm{OkL}}{\Gamma,\,M:\okty \types \Delta}{\Gamma,\,M : A \types \Delta}
      &
      \prftree[l]{\rlnm{OkR}}{\Gamma \types M : A,\,\Delta}{\Gamma \types M : \okty,\,\Delta}
      \\[5mm]
      \SetCell[c=3]{c} \prftree[l]{\rlnm{OkApL1}}{\Gamma,\,M:\okty \from A \types \Delta}{\Gamma,\, M\,N : \okty \types \Delta}
      \qquad
      \prftree[l]{\rlnm{OkApL2}}{\Gamma,\,N:\okty \types \Delta}{\Gamma,\,M\,N:\okty \types \Delta}
      \\[5mm]
      \prftree[l]{\rlnm{OkSL}}{\Gamma,\,M:\intty \types \Delta}{\Gamma,\,\succtm(M) : \okty \types \Delta}
      &
      \prftree[l]{\rlnm{OkPL}}{\Gamma,\,M:\intty \types \Delta}{\Gamma,\,\predtm(M) : \okty \types \Delta}
      &
      \prftree[l]{\rlnm{OkPrL}}{\Gamma,\,M_i:\okty \types \Delta}{\Gamma,\,(M_1,\,M_2) : \okty \types \Delta}          
    \end{tblr} 
    \]\\
  \end{mdframed}
    \caption{Type Assignment with $\okty$}\label{fig:ok-rules}
  \end{figure}

\paragraph{Ill-typed}
Some authors use the phrase `ill-typed' as synonymous with untypable, but we will use it to mean that a term is provably not ok (as is suggestive of the suffix \emph{-typed}).
\begin{definition}[Ill-typed and Well-typed]
  Suppose $M$ is a closed term.
  \begin{itemize}
    \item We say that $M$ is \emph{ill-typed} just if $M : \okty \types$.
    \item We say that $M$ is \emph{well-typed} just if $\types M : \okty$.    
  \end{itemize}
\end{definition}

Of course, due to the \rlnm{OkR} rule, the definition of well-typed subsumes the usual one.  Note that $M$ being well-typed does not imply that $M$ evaluates, only that $M$ does not get stuck: having formula $M:A$ on the right of the turnstile only means that $M$ either diverges or $M$ evaluates to an $A$ (see Section~\ref{sec:semantics} for details).  On the other hand, proving that a closed term $M$ is ill-typed can be understood as proving that $M$ \emph{does not} evaluate: $M : \okty \types$ means that $M$ does not evaluate to a value in $\okty$, so $M$ either diverges or goes wrong.  
The following example is ill-typed only because it will diverge.  
We use $\Gamma$ to abbreviate $\{\,f:\intty \from \okty\,\}$:
\[
  \prftree[l]{\rlnm{AppL}}{
    \prftree[l]{\rlnm{FixnR}}
    {
        \prftree[l]{\rlnm{AppL}}{
          \prftree[l]{\rlnm{Id}}{\Gamma \types f:\intty \from \okty}
        }
        {
          \prftree[l]{\rlnm{Id}}{\Gamma,\,x:\intty \types x:\intty}
        }
        {
          \Gamma,\,f\,x : \okty \types x:\intty
        }
    }
    {
      \types \fixtm{f(x)}{f\,x} : \intty \from \okty
    }
  }
  {
    \prftree[l]{\rlnm{Dis}}
    {
      \prftree[l]{\rlnm{AbsR}}{
        \prftree[l]{\rlnm{Id}}{
          y:\intty \types y:\intty
        }
      }
      {
        \types \abs{y}{y} : \intty \to \intty
      }
    }
    {
      \abs{y}{y} : \intty \types
    }
  }
  {
    (\fixtm{f(x)}{f\,x})\,(\abs{y}{y}) : \okty \types
  }
\]\\[-4ex]



%% file: semantics.tex
\section{Semantics of type assignment}\label{sec:semantics}

In Sections~\ref{sec:prog-lang} and \ref{sec:okay}, we gave some intuitions about how we want to understand the \emph{meaning} of types and the typing judgement.  We now make this precise.  

\begin{definition}[Semantics of Types]\label{def:type-semantics}
  Let $\ClVals$ be the set of all closed values.
  We interpret types as certain sets of closed, well-behaved terms:
  \[
    \begin{array}{rcl}
      \mng{\okty} &=& \ClVals\\
      \mng{\intty} &=& \{\;\pn{0},\,\pn{1},\,\pn{2},\,\ldots\;\}\\
      \mng{A \times B} &=& \{\;(V,W) \mid V \in \mng{A},\,W \in \mng{B}\;\}\\
      \mng{A \to B} &=& \ch{\{\; \fixtm{f(x)}{M} \mid \forall V \in \ClVals.\, V \in \mng{A} \Rightarrow M[V/x][\fixtm{f(x)}{M}/f] \in \mngTb{B} \;\}}\\
      \mng{A \from B} &=& \ch{\{\; \fixtm{f(x)}{M} \mid \forall V \in \ClVals.\, M[V/x][\fixtm{f(x)}{M}/f] \in \mngT{B} \Rightarrow V \in \mng{A}\;\}}\\[3mm]
      \mngT{A} &=& \{\;M \mid M \evals V,\, V \in \mng{A}\;\}\\
      \mngTb{A} &=& \{\;M \mid M \in \mngT{A} \vee M \diverges\;\}
    \end{array}
  \]
\end{definition}

The idea is that $\mng{A}$ is the set of all closed values of type $A$, $\mngT{A}$ is the set of all terms that would evaluate to a value of type $A$ and $\mngTb{A}$ is the set of all terms that either evaluate to a value of type $A$ or diverge.
It follows that $\mngTb{\okty}$ is the set of closed terms that do not go wrong.  As usual, the meaning of the sufficiency arrow $A \to B$ is as the set of all functions that guarantee to map each $A$ value to a $B$ value or diverge in the process.  The meaning of the necessity arrow $A \from B$ is as the set of all functions that require a value from $A$ in order to eventually return a value from $B$.

\begin{remark}\label{rmk:arrows-contrapositive}
It seems there is an asymmetry between the semantics of the two arrows, but this is only because of the bias towards values induced by CBV.  The two implications are equivalent to:
\[
  \begin{array}{rcl}
    V \in \mngT{A} &\Rightarrow& M[V/x][\fixtm{f(x)}{M}/f] \in \mngTb{B}\\
    M[V/x][\fixtm{f(x)}{M}/f] \in \mngT{B} &\Rightarrow& V \in \mngTb{A}
  \end{array}
\]
The relationship between the two function types will become clearer in Section~\ref{sec:complements}, once we introduce the success semantics. \chB{We discuss the choice of CBV semantics in Section~\ref{sec:semantics-discussion}.}
\end{remark}

A consequence of the semantics of necessity is that it can be a little subtle to express properties of interest.  Given a function $\abs{x}{M}$, its membership in some type $\mng{A \from B}$ becomes trivial whenever there is no value $V$ such that $M[V/x] \in \mngT{B}$.  This can occur, for example, when $M$ simply never produces an $F$.  Thus we have:  $\abs{x}{\predtm(x)} \in \mng{\intty \times \intty \from \intty \times \intty}$, because $\abs{x}{\predtm(x)}$ does not return pairs.  This phenomenon can often be exploited in the proof system, using the disjointness rule.  For example, using \rlnm{Dis} we can prove $\types \abs{x}{\predtm(x)} : \intty \times \intty \from \intty \times \intty$.

Now, we turn to the semantics of the typing judgement.  \chB{In CBV, when analysing the body $M$ of an abstraction $\abs{x}{M}$, one can assume that any actual parameter that gets bound to $x$ must have evaluated.  Hence, in the denotational semantics of  CBV programming languages, it is standard to interpret a judgement $x:A \types M:B$ as a function from \emph{values} of type $A$ to \emph{computations} of type $B$ -- where the latter may diverge \cite{winskel-1993}.  In our two-sided system we can have terms on the left-hand side that are not themselves values, so types cannot be interpreted as sets of values, but we do want to retain this essential asymmetry of CBV.  In fact, as we discuss further in Section~\ref{sec:semantics-discussion}, the soundness of much of the system depends on it.  Thus, we will define the semantics so that $M:A$ on the left means that $M$ is a term that evaluates to an $A$, and $M:A$ on the right means that $M$ is a term that either evaluates to an $A$ or diverges.} 

\begin{definition}[Semantics of Judgements]
  A \emph{valuation}, $\theta$, is just a closed substitution.
  \begin{itemize}
    \item Given an atomic formula $M:A$, we say that a valuation $\theta$ \emph{satisfies the formula on the left} just if $M\theta \in \mngT{A}$.  We say that $\theta$ \emph{satisfies the formula on the right} just if $M\theta \in \mngTb{A}$.  
    \item We say that a valuation $\theta$ \emph{satisfies a set of formulas} $\Gamma$ \emph{on the left}, just if $\theta$ satisfies each formula in $\Gamma$ on the left.  A valuation $\theta$ \emph{satisfies a set of formulas} $\Delta$ \emph{on the right}, just if there is a formula in $\Delta$ that is satisfied on the right.
    \item We say that a judgement $\Gamma \types \Delta$ is \emph{true}, written $\Gamma \models \Delta$, just if all valuations $\theta$ that satisfy $\Gamma$ on the left, satisfy $\Delta$ on the right.
  \end{itemize}
  To keep the development smooth, whenever we talk about the satisfaction of formulas by some valuation $\theta$, we will assume that $\theta$ closes all the terms involved (i.e. its domain is sufficiently large).
\end{definition}

Notice from the first of the three clauses that divergence is always allowed on the right, but never on the left.  So, for example, valuation $[\pn{1}/x]$ satisfies $\iftm{x}{\pn{0}}{\divtm} : \intty$ on the right but not on the left.  Then, notice also the asymmetry in the semantics of the judgement: $\Gamma \models \Delta$ just if all valuations that satisfy \emph{all} formulas in $\Gamma$ on the left, satisfy \emph{some} formula from $\Delta$ on the right.  

\subsection{Semantic Soundness for Safety Properties}

By considering the semantics of types, we can draw a sharp distinction between the necessity arrow and all the other (more usual) types of the system.
To see this, we will need to look in more detail at the behaviour of fixpoints in the system, and so we introduce the following \emph{fixpoint approximants} \chC{(see, e.g. \citet{bradfield-walukiewicz-2018} for the use of approximants in mu-calculus)}:
\begin{definition}
  \ch{Given a fixpoint term $\fixtm{f(x)}{M}$, the following family of terms, indexed by $n\in\mathbb{N}$, are its \emph{fixpoint approximants}.
  \begin{align*}
    \fixapprx{0}{f(x)}{M} &= \abs{x}{\divtm} \\
    \fixapprx{n+1}{f(x)}{M} &= \abs{x}{M[\fixapprx{n}{f(x)}{M}/f]}
  \end{align*}}
\end{definition}
Note: we are not introducing a new piece of syntax, but merely an abbreviation that will be useful in the sequel.  Each approximant is simply a finite unfolding of the possibly infinite fixpoint computation.  The idea is to think of a term $C[\fixtm{f(x)}{M}]$ containing a fixpoint as being finitely approximated by each member of the family $C[\fixapprx{n}{f(x)}{M}]$.  

\chD{
As is common, we think of types as defining properties, with a closed term $M$ satisfying a property $A$ whenever $M \in \mngTb{A}$.  In the theory of program verification, it is standard to distinguish between two kinds of properties, safety and liveness.  The former is typically characterised by (i) being downwards closed in the termination order -- that is, if a program satisfies a safety property, then so does any program that terminates less often, and (ii) the existence of finite counterexamples -- whenever a program does not satisfy a safety property, there is a finite trace to act as witness.  We make precise this idea for our type system in the following: 

\begin{definition}
  For closed terms $N$ and $P$, write $P \lesssim N$ just if $P$ normalises implies $N$ normalises and then they have the same normal form.  Then we say that a type $A$ defines a \emph{safety property} just if the following are both true:
  \begin{enumerate}[({S}1)]
    \item If $\calC[N] \in \mngTb{A}$ and $P \lesssim N$, then $\calC[P] \in \mngTb{A}$.
    \item If $C[\fixtm{f(x)}{M}] \notin \mngTb{A}$, then there is some $k$ for which $C[\fixapprx{k}{f(x)}{M}] \notin \mngTb{A}$.
  \end{enumerate}
\end{definition}
}

It is possible to show that, if we exclude the necessity arrow, then all types of the system define safety properties.  However, with the necessity arrow, this is no longer the case.  In fact, we can find counterexamples to both aspects of safety.

\begin{theorem}\label{thm:safety-cexs}
  There are types $A$ and $B$ such that $A \from B$ is not a safety property.  In particular:
  \begin{itemize}
  \item The type $\intty \from \intty \to \intty$ does not satisfy (S1).
  \item The type $\intty \from \intty \to \intty \from \intty \to \intty$ does not satisfy (S2).
  \end{itemize}
\end{theorem}

However, by inspection of these proofs, it is evident that allowing function types on the right of the necessity arrow is a key feature.  We will now define a fragment of our type system in which types are restricted so as to always define safety properties.

\chD{
\begin{definition}[Safety Fragment]
  The \emph{safety fragment} consists of restricting the two-sided system to judgements only involving the following types:
  \[
        \begin{array}{rcl}
             F &\Coloneqq& \intty \mid \okty \mid F_1 \times F_2\\
             A,B &\Coloneqq& \intty \mid \okty \mid A \times B \mid A \to B \mid A \from F 
        \end{array}
  \]
\end{definition}

We call types $F$ the \emph{finitely verifiable} types, because we can show that whenever a term inhabits a finitely verifiable type, there is already a finite approximation of the term that inhabits the type.  In the safety fragment, only finitely verifiable types are allowed on the right-hand side of necessity arrows.  This is enough to recover safety:
}

\chD{
\begin{theorem}[Safety properties]\label{thm:types-are-safety}
  Every type in the safety fragment defines a safety property.
\end{theorem}
}

The fact that types define safety properties is very helpful for the purpose of analysing the system.  The contrapositive of (S2) gives us a principle that we can use in order to show that term satisfies a type: if $\fixapprx{n}{f(x)}{M} \in \mngTb{A}$ for all $n$, then it follows that $\fixtm{f(x)}{M} \in \mngTb{A}$.  We use this principle to show soundness, i.e. that every provable judgement is true.

\chD{
\begin{theorem}[Semantic Soundness]\label{thm:strong-soundness}
  In the safety fragment, if $\Gamma \types \Delta$ then $\Gamma \models \Delta$.
\end{theorem}
}

\chB{
This formulation of soundness is standard for a sequent calculi \cite{ebbinhaus-flum-thomas-1994}, but there is another, more syntactical notion of soundness that is standard in the literature for type systems, namely `progress and preservation' \cite{wright-felleisen-infcomp1994}.  
Hence, we call this result \emph{semantic soundness} and the latter kind, which we will obtain in Sections \ref{sec:complements} and \ref{sec:constrained-typing}, \emph{syntactic soundness}.}  As corollary, we obtain:
\begin{corollary}
  Suppose $M$ is a closed term.  In the safety fragment:
  \begin{itemize}
    \item If $M$ is well typed, then $M$ does not go wrong.
    \item If $M$ is ill typed, then $M$ does not evaluate.
  \end{itemize}
\end{corollary}

\chD{
At the time of writing, we do not have a proof of semantic soundness for the full system.  However, it seems clear that one could prove a syntactic soundness result for the full system, as we do for its extension with complements in the success semantics (Section~\ref{sec:complements}) and for the constrained type system (Section~\ref{sec:constrained-typing}).  From syntactic soundness (progress and preservation), one also obtains that well-typed programs don't go wrong and that ill-typed programs don't evaluate, but we note that these statements only concern closed terms $M$, whereas semantic soundness concerns the truth of judgements in general.
}

\chC{
  It is perhaps unsurprising to note that the system is incomplete: there are true judgements that are not provable.  For example, the term $\iftm{\pn{0}}{\pn{1}}{\abbv{id}}$ is neither well-typed nor ill-typed.  Semantically, this term evaluates to $\pn{1}$, but we cannot prove $\types \iftm{\pn{0}}{\pn{1}}{\abbv{id}} : \okty$ because our rules for conditionals do not allow for a sufficiently precise case analysis.  On a related note, one can observe that there are terms that are \emph{both} well-typed and ill-typed: namely certain terms that diverge.  The subject of the derivation at the end of Section~\ref{sec:okay} is one such example.
}

\chB{
\subsection{Further discussion of the semantics}\label{sec:semantics-discussion}

One may wonder if the theory would be better developed using a semantics that has more symmetrical judgements, such as call-by-name (CBN).  In a CBN semantics, in order for the standard (right) typing rules to be sound, we are forced to interpret $M:A$ on the left as `$M$ either evaluates to an $A$ \emph{or diverges}' rather than, as in CBV, `$M$ evaluates to an $A$'.  Thus, the meaning of a formula in the premise and the meaning in the consequent become identical, which may seem desirable.

However, $M:A$ as a premise in CBN is a much weaker assumption, and reasoning on the left becomes very weak as a result.  Consider a simple judgement of the form:
\[
  \iftm{M}{\succtm(x)}{\succtm(x)} : \intty \types x : \intty
\]
In CBV, this judgement is true independent of the choice of $M$, and is provable using the \rlnm{IfZL2} rule -- both branches require $x$ to be a $\intty$, so necessarily the whole term requires $x$ to be a $\intty$.  The reason that this is sound depends crucially on knowing that the if-expression terminates.  In CBV, $\iftm{M}{\succtm(x)}{\succtm(x)} : \intty$ on the left means that we can assume the if-expression evaluated to a $\intty$ (i.e. terminated), and hence one of the two branches evaluated to a $\intty$, and whichever branch that was, it used $x$ at type $\intty$, and so the conclusion follows.  

By contrast, in CBN, $\iftm{M}{\succtm(x)}{\succtm(x)} : \intty$ on the left means that we can assume only that the if-expression either evaluated to a $\intty$ \emph{or diverged}.  Therefore, we do not even know that execution of the guard $M$ necessarily terminated.  In the case that the guard $M$ diverged, \emph{neither} of the branches was ever executed.  In this situation, it can be that no execution of the program actually uses $x$ at type $\intty$, and so it is not true that $x:\intty$ \emph{necessarily} follows.  For example, the judgement $\iftm{\abbv{div}}{\succtm(x)}{\succtm(x)} : \intty \types x: \intty$ is \emph{not} true in CBN, a counterexample is $x := \abbv{id}$.  Therefore, the \rlnm{IfZL2} rule for reasoning about cases on the left is unsound in CBN.

In fact, in CBN, when we assume $\iftm{M}{P}{Q} : A$ on the left, we can conclude almost nothing about the necessary behaviour of $P$ and $Q$!  This prevents us from reasoning about almost any interesting program.  We do not know that either of $P$ and $Q$ is ever executed, so we can only use the behaviour of the guard $M$ to justify some conclusion.  The same is true of an application $MN : A$ on the left -- we only know that $MN$ either yielded an $A$ \emph{or diverged} and thus we cannot conclude anything about the necessary behaviour of $N$ or $M[N/x]$, neither of which may have been executed if the execution of $M$ did not terminate.  Hence, \rlnm{AppL} is unsound.  To restore soundness it seems likely that one would have to equip the system with an ability to reason about termination explicitly.  An approach like this was taken by \citet{Vazou:2014:RTH:2628136.2628161} to repair a similar unsoundness when adapting the CBV \emph{Liquid Types} system for Haskell.
}

%% file: complements.tex
\section{Complements and the Success Semantics}\label{sec:complements}

The two-sided type systems of this paper provide a form of negation at the meta level: to establish $M \notin \mngT{A}$ we can prove $M:A \types$.  \chA{Since there are already a number of one-sided (traditional) type systems that incorporate a negation or complement operator on types, a natural question regards the extent to which two-sided reasoning can be simulated in a one-sided system with complements; in other words if one can internalise the set of judgements of shape $M:A \types$ by some type $A^c$.}



\begin{definition}
  We extend the grammar of types in Definition~\ref{def:simple-types} by a complement operator, $A \Coloneqq \cdots \mid A^c$, and we extend the system of rules with the two below.
\end{definition}

\begin{wrapfigure}[5]{l}{8cm}
  \vspace{-4mm}
  \begin{mdframed}[topline=false,innertopmargin=-.84ex,innerleftmargin=-.1ex,innerrightmargin=0ex]
  \rulediv{40pt}{Complement Rules}\\[-2mm]
  \begin{center}
  $
  \prftree[l]{\rlnm{CompL}}{\Gamma \types M : A,\,\Delta}{\Gamma,\,M:A^c \types \Delta}
  \qquad
    \prftree[l]{\rlnm{CompR}}{
      \Gamma,\, M:A \types \Delta
    }
    {
      \Gamma \types M:A^c,\,\Delta
    }
  $\vspace{2mm}
  \end{center}
  \end{mdframed}
\end{wrapfigure}

\noindent 
The correspondence between these two rules and those for negation in sequent calculi supporting propositional logic should be evident.
Intuitively, the semantics of the new operator is to take the complement of the type with respect to the set of closed values.
\chB{
Rule \rlnm{CompL}, is closely related to our disjointness rule, and \rlnm{Dis} can be derived from \rlnm{CompL} if the system is augmented with a sufficiently strong subtyping theory, because $A \cap B = \emptyset$ iff $A \subseteq B^c$.   
}
This rule is sound for our call-by-value semantics, but rule \rlnm{CompR} is \emph{unsound}.  
The problem is that when $M:A^c$ is not satisfied on the right by some $\theta$, it does not imply that $M:A$ will necessarily be satisfied by $\theta$ on the left.  

\begin{wrapfigure}[7]{r}{4.5cm}
  \begin{center}\vspace{-4ex}
  $\prftree[l]{\rlnm{CompR}}{
    \prftree[l]{\rlnm{OkPL}}{
      \prftree[l]{\rlnm{Dis}}{
        \prftree{\cdots}{\types \abbv{id} : \intty \to \intty}
      }
      {
        \abbv{id} : \intty \types
      }
    }
    {
      \predtm(\abbv{id}) : \okty \types
    }
  }
  {
    \types \predtm(\abbv{id}) : \okty^c 
  }
$
  \end{center}
\end{wrapfigure}

For example, consider the derivation to the right.
If this were sound in the call-by-value semantics, $\predtm(\abbv{id})$ would be guaranteed to \emph{diverge}, since we would have that $\predtm(\abbv{id}) \in \mngTb{\okty^c}$, and $\mng{\okty^c} = \ClVals \setminus \ClVals = \emptyset$.  This is clearly absurd, because $\predtm(\abbv{id})$ goes wrong.  Thus our semantics of Section~\ref{sec:semantics} does not support \rlnm{CompR}.

\subsection{The success semantics}

\chC{
The key issue is that, if we are internalise the judgements of shape $M:A \types$ by some type $A^c$, in the way that we have indicated, then we should expect that $M:A \types$ and $\types M : A^c$ have the same meaning.  However, $M:A \types$ can hold of some term $M$ that gets stuck (goes wrong); but this is impossible for $M$ that satisfies $\types M : A^c$, because the semantics of the judgement ensures that any such term either diverges or evaluates to an $A$ value.
}
Thus, we have only the backward direction of the desired equivalence: $M \notin \mngT{A}$ \emph{iff} $M \in \mngTb{A^c}$.


Our approach will be to weaken the semantics of typing on the right of the judgement, so that a formula $M:A$ on the right \emph{can (always)} be satisfied by a term $M$ that goes wrong.
\begin{definition}
  The \emph{success semantics} is defined as follows.  First, we define, for each type $A$, the set of closed terms that may go wrong, diverge, or evaluate to a value in $\mng{A}$.
  \[
    \mngTs{A} = \{\; M \mid \text{$M \in \mngT{A}$ or $M$ diverges or $M$ goes wrong}  \;\}
  \]
  We redefine the meaning of the sufficiency arrow type to allow the body to go wrong when executing on an argument:
  \[
    \mng{A \to B} = \ch{\{\; \fixtm{f(x)}{M} \in \ClVals \mid \forall V \in \mng{A}.\ M[V/x][\fixtm{f(x)}{M}/f] \in \mngTs{B} \;\}}
  \]
  Finally, we redefine satisfaction on the right by saying that a formula $M:A$ is \emph{satisfied by $\theta$ on the right} just if $M\theta \in \mngTs{A}$.
\end{definition}

The success semantics has the strong point that the judgements are now symmetrical in the sense of the above equivalence.
Thus both \rlnm{CompL} and \rlnm{CompR} are sound for the success semantics.
In fact, we will show that all the typing rules presented in Sections~\ref{sec:prog-lang} and \ref{sec:okay} are (syntactically) sound for the success semantics too.


A considerable disadvantage of the success semantics is that we lose the true negatives theorem \emph{well-typed programs don't go wrong}.  Now, proving $\types M : \okty$ means that $M$ may either evaluate, diverge or go wrong.  In other words, it means nothing at all!  We do however retain \emph{ill-typed programs don't evaluate}, so this system is exclusively for proving that programs behave badly.  This puts this system in the same territory as Erlang's celebrated \emph{success types} \citep{lindahl-sagonas-ppdp2006}, which similarly provide no guarantees for terms that are well-typed.  

However, the symmetry in the system allows for a considerable saving.  Under the success semantics, $\mng{A \from B} = \mng{A^c \to B^c}$ -- the conditions on membership in these types are the contrapositive of each other.  Thus, there is the potential to simply define $A \from B$ as an abbreviation for $A^c \to B^c$.  Then, the symmetry between \rlnm{AppL} and \rlnm{AppR}, and the symmetry between \rlnm{FixsR} and \rlnm{FixnR} can be exploited to derive necessity from sufficiency + complement:
\[
  \prftree[l]{\rlnm{CompL}}{
    \prftree[l]{\rlnm{AppR}}{
        \Gamma \types M : B \from A,\,\Delta
    }
    {
      \prftree[l]{\rlnm{CompR}}{
        \Gamma,\,N:B \types \Delta
      }
      {
        \Gamma \types N : B^c,\,\Delta
      }
    }
    {
      \Gamma \types M\,N : A^c,\,\Delta
    }
  }
  {
    \Gamma,\,M\,N:A \types \Delta
  }
  \qquad
  \prftree[l]{\rlnm{FixsR}}{
    \prftree[l]{\rlnm{CompL}}{
       \prftree[l]{\rlnm{CompR}}{
         \Gamma,\,f:B \from A,\,M:B \types x:A,\,\Delta
       }
       {
        \Gamma,\,f:A \from B \types M:B^c,\,x:A,\,\Delta
       }
    }
    {
      \Gamma,\,f:A \from B,\,x:A^c \types M : B^c,\,\Delta
    }
  }
  {
    \Gamma \types \fixtm{f(x)}{M} : A \from B,\,\Delta
  }
\]
Thus, in the success system, the sufficiency arrow in combination with complements provides a complete treatment of the necessity arrow, and we could dispense with \rlnm{FixnR} and \rlnm{AppL} if desirable.  

\subsection{The one-sided system}

We can go further by exploiting a key symmetry of the typing rules that we have presented so far.  

\begin{definition}[Variable typings and type environments.]\label{def:type-envs}
  Let us say that a typing formula $M:A$ is a \emph{variable typing} just if $M$ is a variable, and otherwise it is a \emph{non-variable typing}.  Let us say that an environment $\Gamma$ is a \emph{type environment} just in case every formula therein is a variable typing.  
\end{definition}  

In each of the rules we have introduced in Sections~\ref{sec:prog-lang} and \ref{sec:okay}, and so far in \ref{sec:complements}, observe the following.  If the conclusion has at most one non-variable typing formula, then each of the hypotheses will have at most one non-variable typing formula too.
For example, one can see by inspection of Examples~\ref{ex:add-typing} and \ref{ex:twice-typing} of Section~\ref{sec:prog-lang} that, in each judgement of the respective proof trees, at most one typing formula is non-variable.  


When a judgment contains at most one non-variable typing then, in the success semantics, it is equivalent to a judgement in which a non-variable typing is the only typing formula on the right-hand side of the turnstile.  Thus, in the success type system, we are able to \emph{normalise} the typing rules by using \rlnm{CompL} and \rlnm{CompR} to exchange the positions of formulas between the two sides, until they form a traditional, one-sided type system.  In fact, the use of complements allows us to eliminate further redundancies, and so we obtain a smaller system.

\begin{figure}
  \[
    \begin{array}{c}
    \prftree[l]{\rlnm{Ok}}{\Gamma \types M : \okty}
    \qquad
    \prftree[l]{\rlnm{OkC1}}{\Gamma,\,x:\okty^c \types M : A}
    \qquad
    \prftree[l]{\rlnm{OkC2}}{\Gamma \types M : \okty^c}{\Gamma \types M : A}
    \\[5mm]
    \prftree[l]{\rlnm{Contra}}{\Gamma,\,x:A,\,x:A^c \types M : A}
    \qquad
    \prftree[l]{\rlnm{Var}}{\Gamma,\,x:A \types x:A}
    \qquad
    \prftree[l,r]{$B \distype A$}{\rlnm{Disj}}{\Gamma \types M : B}{\Gamma \types M : A^c}
    \\[5mm]
    \prftree[l]{\rlnm{Zero}}{\Gamma \types \zerotm : \intty}
    \qquad
    \prftree[l]{\rlnm{Succ1}}{\Gamma \types M : \intty}{\Gamma \types \succtm(M) : \intty}
    \qquad
    \prftree[l]{\rlnm{Succ2}}{\Gamma \types M : \intty^c}{\Gamma \types \succtm(M) : A}
    \\[5mm] 
    \prftree[l]{\rlnm{Pred1}}{\Gamma \types M : \intty}{\Gamma \types \predtm(M) : \intty}
    \qquad
    \prftree[l]{\rlnm{Pred2}}{\Gamma \types M : \intty^c}{\Gamma \types \predtm(M) : A}
    \\[5mm]
    \ch{\prftree[l]{\rlnm{Fix}}{\Gamma,\,f:A \to B,\,x:A \types M : B}{\Gamma \types \fixtm{f(x)}{M} : A \to B}}
    \qquad
    \prftree[l]{\rlnm{Let3}}{\Gamma \types N : A}{\Gamma \types \lettm{(x,y)}{M}{N} : A}
    \\[5mm]
    \prftree[l]{\rlnm{Let2}}{\begin{array}{c}\Gamma \types M:(B_1 \times B_2)^c\\\Gamma,\,x_i:B_i^c \types N :A\;(\forall i)\end{array}}{\Gamma \types \lettm{(x_1,x_2)}{M}{N} : A}
    \qquad
    \prftree[l]{\rlnm{Let1}}{\begin{array}{c}\Gamma \types M : B \times C\\\Gamma,\,x_1:B,\,x_2:C \types N:A\end{array}}{\Gamma \types \lettm{(x_1,x_2)}{M}{N} : A}
    \\[5mm]
    \prftree[l]{\rlnm{App1}}{\Gamma \types M : B \to A}{\Gamma \types N : B}{\Gamma \types M\,N : A}
    \qquad
    \prftree[l]{\rlnm{App2}}{\Gamma \types M : (\okty^c \to A)^c}{\Gamma \types M\,N : A}
    \qquad
    \prftree[l]{\rlnm{App3}}{\Gamma \types N : \okty^c}{\Gamma \types M\,N : A}
    \\[5mm]
    \prftree[l]{\rlnm{Pair1}}{\Gamma \types M : A}{\Gamma \types N : B}{\Gamma \types (M,\,N) : A \times B}
    \qquad
    \prftree[l]{\rlnm{Pair2}}{\Gamma \types M_i : \okty^c}{\Gamma \types (M_1,M_2) : A}
    \qquad
    \prftree[l]{\rlnm{Pair3}}{\Gamma \types M_i : A_i^c}{\Gamma \types (M_1,M_2) : (A_1 \times A_2)^c}
    \\[5mm]
    \prftree[l]{\rlnm{IfZ1}}{\Gamma \types M : \intty^c}{\Gamma \types \iftm{M}{N}{P} : A}
    \qquad
    \prftree[l]{\rlnm{IfZ2}}{\Gamma \types N:A}{\Gamma \types P:A}{\Gamma \types \iftm{M}{N}{P} : A}
    \end{array}
  \]
  \caption{One-Sided Type Assignment}\label{fig:one-sided-ta}
\end{figure}

\begin{definition}[One-sided type system] 
  A \emph{judgement} of the one-sided system is a pair $\Gamma \types M:A$ of a type environment $\Gamma$ and a typing $M:A$.  Provability is defined by the rules of Figure~\ref{fig:one-sided-ta}.  In rules \rlnm{Let1}, \rlnm{Let2}, \rlnm{Let3}, \rlnm{Abs} we require that the bound variables are not mentioned in $\Gamma$.
\end{definition}

In many cases, the rules of Figure~\ref{fig:one-sided-ta} derive from their two-sided counterparts simply by normalising the judgements involved, possibly by introducing complements.  Rule \rlnm{Disj} arises from \rlnm{Dis} this way.  Normalisation often introduces a complemented typing $M:A^c$ in the conclusion, but in most cases it is sound to generalise to $M:A$.  An example is deriving \rlnm{IfZ1} from \rlnm{IfZL1}.

The rule \rlnm{IfZ2} arises in this way, and it subsumes both \rlnm{IfZR} and \rlnm{IfZL2}.  However, the soundness of this one-sided rule is not obvious at first glance.  It is almost the same as the right side rule \rlnm{IfZR}, except that \rlnm{IfZR} has a third premise.  In \rlnm{IfZ2} this premise is absent, and we need only show that each branch has type $A$ in order to conclude that the whole conditional has type $A$ -- \emph{whether or not the guard is of type $\intty$}.  However, this \emph{is sound} for the success semantics, because if the guard normalises to something other than a numeral, then the conditional will go wrong, and thus satisfies any success type.

The rules \rlnm{Ok}, \rlnm{OkC} and \rlnm{Contra} represent structural features of the success semantics; namely that every term satisfies $\mngTs{\okty}$, that no term satisfies $\mngT{\okty^c}$ and that $\mngT{A} \cap \mngT{A^c}$ is empty for every type $A$.  
Rule \rlnm{Contra} is necessary in order to account for the possibility that a two-sided judgment is provable due to a contradiction among the variable typings.  For example, the two-sided judgment $x:\intty,\,M:A \types x:\intty$ is clearly provable using \rlnm{Id}, but when normalised to a one-sided judgement, it becomes $x:\intty,\,x:\intty^c \types M:A$ and \rlnm{Var} is not applicable.


\begin{theorem}\label{thm:two-sided-subsumed}
  Suppose $\Gamma$ and $\Delta$ are type environments.
  \begin{itemize}
    \item If $\Gamma,\,M:A \types \Delta$ in the two-sided system, then $\Gamma \cup \Delta^c \types M:A^c$ in the one-sided system.
    \item If $\Gamma \types M:A,\,\Delta$ in the two-sided system, then $\Gamma \cup \Delta^c \types M:A$ in the one-sided system.
  \end{itemize}
\end{theorem}

\begin{example}
Consider again the judgement $\types \abbv{twice} : (A \from A) \to A \from A$, which we proved as Example~\ref{ex:twice-typing} of Section~\ref{sec:prog-lang}.  As in that example, we use $\Gamma$ to abbreviate $\{f:A \from A\}$:
\[
  \prftree[l]{\rlnm{Abs}}{
    \prftree[l]{\rlnm{Abs}}{
      \prftree[l]{\rlnm{Var}}{
        \prftree[l]{\rlnm{Var}}{\Gamma,\,x:A^c \types f : A \from A}
      }
      {
        \prftree[l]{\rlnm{App}}
        {
          \prftree[l]{\rlnm{Var}}{\Gamma,\,x:A^c \types f : A \from A}
        }
        {
          \prftree[l]{\rlnm{Var}}{\Gamma,\,x:A^c \types x : A^c}
        }
        {
          \Gamma,\,x:A^c \types f\,x:A^c
        }
      }
      {
        \Gamma,\,x:A^c \types f\,(f\,x):A^c
      }
    }
    {
      \Gamma \types \abs{x}{f\,(f\,x)} : A \from A
    }
  }
  {
    \types \abs{fx}{f\,(f\,x)} : (A \from A) \to A \from A
  }
\]
\end{example}

We prove syntactic soundness for the one-sided system using a progress and preservation argument.  In fact, progress as usually understood is not necessary -- since we allow going wrong on the right, there is no requirement to show that a subject in normal form is a value, merely that values have sensible types.  We get the syntactic soundness of the two-sided system under the success semantics as a corollary via Theorem~\ref{thm:two-sided-subsumed}.  

Let us say that a closed term $M$ is \emph{ill-typed} in the one-sided system just if $\types M : \okty^c$.
\begin{theorem}[One-Sided Syntactic Soundness]\label{thm:failures-soundness}
  If $M$ is ill-typed, then $M$ does not evaluate.
\end{theorem}

%% file: constrained.tex
\section{A Constrained Type System}\label{sec:constrained-typing}

\chA{
We now depart from the study of our PCF-like language, to introduce a two-sided \emph{constrained} type system.
Our aim is to illustrate how one can take the core of an existing one-sided system, which is useful in practice, and extend it to a two-sided system, whilst retaining its good properties (such as its true negatives theorem\footnote{For this reason, we will use a standard CBV semantics and not the success semantics of the previous section.} and computable type inference).

Starting from the two-sided rules for pure terms with fixpoints of Figure~\ref{fig:core-rules}, we add type constraints, top-level definitions, term constants for algebraic datatypes and pattern-matching, and types corresponding to datatype constructors.  This gives us a two-sided system whose features are similar to those of \citet{aiken-et-al-popl1994,jones-ramsay-popl2021,marlow-wadler-icfp1997}, which are capable of precise reasoning about datatypes and matching.
}


\begin{definition}\label{def:pat-match-terms}
  We assume a denumerable set of \emph{top-level identifiers}, written typically as $\mathsf{f},\mathsf{g}$, and a denumerable set of \emph{local variables}, typically $f,g,x,y,z$.  We also fix a finite signature $\calC$ of ranked datatype \emph{constructors}, typically $c$.  We will always assume that the binary pair constructor, written mixfix as $(\,,\,)$ is in the signature.  We consider various kinds of program expressions:
  \[
    \begin{array}{rcl}
      p,\,q &\Coloneqq& c(x_1,\,\ldots,\,x_m) \\
      M,\,N,\,P,\,Q &\Coloneqq& x \mid c(M_1,\,\ldots,\,M_m) \mid M\,N \mid \ch{\fixtm{f(x)}{M}} \mid \matchtm{M}{|_{i=1}^k (p_i \mapsto P_i)}\\
      V,\,W &\Coloneqq& x \mid c(V_1,\ldots,V_n)\\
      \mathcal{M} &\Coloneqq& \epsilon \mid f = M;\,\mathcal{M}
    \end{array}
  \]
  As usual, we identify terms up to renaming of bound variables.
  We consider a term to be \emph{closed} just if it contains no free variables (but it may contain top-level function identifiers).
  \ch{As before, we consider $\abs{x}{M}$ as an abbreviation for $\fixtm{f(x)}{M}$ when $f$ is not free in $M$.}

  We make the following additional requirements: (i) in a pattern $c(x_1,\ldots,x_n)$ the arity of $c$ is $n$ and the $x_i$ are pairwise distinct, (ii) in a pattern-match term $\matchtm{M}{|_{i=1}^k (p_i \mapsto P_i)}$, the family $(p_i)_{i=1}^k$ is orthogonal (i.e. for every distinct $i$ and $j$, $p_i$ and $p_j$ are headed by distinct constructors) and \emph{all} bound variables throughout the family have distinct names, (iii) in module $f_1=M_1;\,\ldots;\,f_n=M_n$,  each $f_i$ is distinct (there are not two definitions for the same top-level identifier).  
  
  Thus, we will feel free to index pattern families by the set of constructors, say $I$, that head their alternatives $|_{c \in I}\, c(\bar{x_c})$, and sometimes treat a module $\mathcal{M}$ as a partial map from identifiers to terms.
  \end{definition}

  The first category are the \emph{patterns}, which are required to be \emph{simple} in the sense of being shallow and constructor headed.  
  The second category are the \emph{terms}.  As in our PCF-like language, we have an applied  $\lambda$-calculus with fixpoints, but this time we are interested in datatype constructors and pattern matching and thus we have constants accordingly.  Thirdly, we have \emph{modules} $\calM$, which are just sequences of definitions for top-level identifiers.  Finally, \emph{values} are either variables, fixpoint abstractions or datatype constructor terms built entirely from values.

  \begin{remark}\label{rem:simple-matching}
  The requirement for simple matching is to ease the presentation.  We could allow arbitrary patterns at the expense of more complex subtype and consistency checking.  Since this is not the focus of our work and is anyway well-covered elsewhere, we will be content with the simple case.  Note also that the rank (arity) of a constructor is built into the syntax, so we will never consider the possibility of supplying a constructor with an inappropriate number of arguments since it is not considered a well-formed term.
  \end{remark}

  \begin{example}\label{ex:head-map-defns}
  Let us assume a constructor signature $\calC$, which, in addition to the binary pair constructor $(\,,\,)$, contains the binary list constructor `cons', written infix $::$, and the nullary empty list constructor `nil', written $[]$.  We can define the \textsf{head} and \textsf{map} functions by:
  \[
    \renewcommand{\arraystretch}{1.1}
    \begin{array}{rcl}
      \mathsf{head} &=& \abs{xs}{\matchtm{xs}{y::ys \mapsto y}}\\
      \mathsf{map} &=& \fixtm{m(f)}{\abs{xs}{\matchtm{xs}{[] \mapsto [] \mid y::ys \mapsto f\,y :: m\,f\,ys}}}
    \end{array}
  \]
  \end{example}

  As before, we adopt a call-by-value reduction strategy and reuse the notation\footnote{Actually the one-step relation is parametrised by a module $\calM$, but we leave this implicit.} $M \ped N$.  The sequence of arguments to a datatype constructor is evaluated from left to right.  Due to the requirement for shallow and orthogonal patterns, matching can be executed by simply indexing into the pattern family.  A full definition is available in \iftoggle{supplementary}{Appendix~\ref{apx:syntactic-soundness}}{the long version of the paper \cite{ramsay2023illtyped}}.

\subsection{Types and type constraints}

Our type system is designed for reasoning about the shape of datatype constructions.  It is a \emph{constrained type system} (see e.g. \cite{odersky-sulzmann-wehr-TSPOS1999}), so each judgement is parametrised by a set of type constraints $C$, which restrict the possible instantiations of type variables.  

Our types are a cut-down version of the types defined in the constrained type system of \citet{aiken-et-al-popl1994}, in that for each datatype constructor $c$ of arity $n$, we can form a corresponding constructor type $c(A_1,\ldots,A_n)$.  Intuitively, this type represents the set of all values of shape $c(V_1,\ldots,V_n)$, where each $V_i$ has type $A_i$.  We allow for universal polymorphism through the construction of \emph{constrained type schemes}, $\forall \vv{a}.\,C \Rightarrow A$.  Intuitively, a top-level function that has a such a scheme is guaranteed to behaves like $A[\vv{B}/\vv{a}]$ for each instantiation $\vv{B}$ of the type variables $\vv{a}$ that satisfy the constraints $C$.

  \begin{definition}
    We assume a denumerable collection of \emph{type variables} $a$, $b$ and so on.  The types are stratified as follows:
  \[
    \begin{array}{rrcl}
    \rlnm{Sum Types} & K &\Coloneqq& \Sigma_{i=1}^k c_i(\bar{A_i}) \\
    \rlnm{Monotypes} & A,\,B &\Coloneqq& a \mid K \mid \okty \mid A \to B \mid A \from B\\
    \rlnm{Type Schemes} & S &\Coloneqq& A \mid \forall \bar{a}.\,C \Rightarrow A
    \end{array}
  \]
  Here $C$ (and sometimes $D$) is a finite set of \emph{type constraints} (also \emph{subtype formulas}), each of shape $A \subtype B$ (i.e. between monotypes).  We write $A \equiv B$ as an abbreviation for the two constraints $A \subtype B$ and $B \subtype A$.  We identify type schemes up to the renaming of bound type variables and we assume that arrows associate to the right.  We require that type schemes are \emph{closed} in the sense that they have no free type variables.  

  We consider the sum type $\Sigma_{i=1}^k c_i(\bar{A_i})$ as a finite set $\{c_1(\bar{A_1}),\ldots,c_k(\bar{A_k})\}$, and we require that the elements are orthogonal in the sense that $c_i = c_j$ implies $i=j$.  See also Remark~\ref{rem:simple-matching}.  In examples, We will often write a particular sum type $\Sigma_{i=1}^3\,c_i(\vv{A_i})$ as a list of summands $c_1(\vv{A_1}) + c_2(\vv{A_2}) + c_3(\vv{A_3})$.  In particular, when the sum is a singleton $\Sigma_{i=1}^1 c_i(\vv{A_i})$, which is quite typical, we just write $c_1(\vv{A_1})$.
  \end{definition}

  We don't have an explicit notion of recursive types with which to type recursively defined data structures.  However, as is well known, our constructor types together with (recursive) type constraints can capture the same notion implicitly.  For example, our \textsf{map} function from Example~\ref{ex:head-map-defns} can be assigned the type:
  \[
    \begin{array}{l}
      \mathsf{map} : \forall a\,b\,\ell_a\,\ell_b.\, C \Rightarrow (a \to b) \to \ell_a \to \ell_b\\
      \qquad \text{where}\; C = \{\; \ell_a \equiv [] + (a::\ell_a),\,\ell_b \equiv [] + (b :: \ell_b)  \;\}
    \end{array}
  \]
  This type says that $\mathsf{map}$ takes a function from $a$ to $b$ and a list of $a$ and returns a list of $b$.  Intuitively, the type `list of $a$' is described by the type variable $\ell_a$ under the constraint that $\ell_a \equiv [] + (a::\ell_a)$ and similarly for `list of $b$' with $\ell_b$ constrained so that $\ell_b \equiv [] + (b :: \ell_b)$.

  \begin{figure}
    \[
      \begin{array}{c}
        \prftree[l]{\rlnm{IdS}}{C,\,A \subtype B \types A \subtype B}
      \qquad
        \prftree[l]{\rlnm{TrS}}{C \types A_1 \subtype A_2}{C \types A_2 \subtype A_3}{C \types A_1 \subtype A_3}
      \\[5mm]
        \prftree[l]{\rlnm{ToS}}{C \types A' \subtype A}{C \types B \subtype B'}{C \types A \to B \subtype A' \to B'}
      \qquad
        \prftree[l]{\rlnm{FrS}}{C \types A \subtype A'}{C \types B' \subtype B}{C \types A \from B \subtype A' \from B'}
      \\[5mm]
      \prftree[l]{\rlnm{OkS}}{C \types A \subtype \okty}
      \qquad
      \prftree[l,r]{$I \subseteq J$}{\rlnm{SmS}}{C \types A_{i,c} \subtype B_{i,c} \; (\forall c \in I, \forall i \in [1..n_c])}{C \types \Sigma_{c \in I}\,c(A_{1,c},\ldots,A_{n_c,c}) \subtype \Sigma_{d \in J}\,d(B_{1,d},\ldots,B_{n_b,b})}
      \end{array}
    \]
    \caption{Constrained subtyping.}\label{fig:constrained-sub}
  \end{figure}

  In a constrained type system, it is usual to define subtyping with respect to a context containing subtyping formulas, and so we have the following.

  \begin{definition}
    A \emph{subtyping judgement} is a triple $C \types A \subtype B$ in which $C$ is a set of type constraints and $A \subtype B$ is a type constraint.  Provability is defined using the rules of Figure~\ref{fig:constrained-sub}.  We extend the notion of provability to sets of constraints, writing $C \types C'$ just if $C \types A \subtype B$ for every $A \subtype B \in C'$.
  \end{definition}

  The rule \rlnm{IdS} allows for justification with respect to the context, and the rule \rlnm{TrS} ensures closure under transitivity of subtyping.  Rule \rlnm{ToS} gives the familiar relationship between sufficiency arrow types and \rlnm{FrS} describes the dual relationship between necessity arrow types.  Rule \rlnm{OkS} puts $\okty$ at the top of the subtype ordering.  Finally, the rule \rlnm{SmS} allows for subtyping between sum types.  It says that a sum type $K_1$ is a subtype of $K_2$ just if whenever $c(A_1,\ldots,A_n)$ is a summand of $K_1$, then there is a summand of shape $c(B_1,\ldots,B_n)$ in $K_2$ and, moreover, the argument types are covariantly related.  For example, using the constructors from Example~\ref{ex:head-map-defns}, we have $\types [] \subtype [] + (\okty :: \okty)$ and $a \subtype b \types (a::[]) \subtype [] + (b :: [])$.

  \begin{figure}\vspace{1ex}
    \begin{mdframed}[topline=false,innertopmargin=-.84ex,innerleftmargin=-.1ex,innerrightmargin=0ex]
    \rulediv{20pt}{Structural}\\[1mm]
    \[
      \begin{tblr}{Q[c,10em]Q[c,13em]Q[c,13em]}
        \prftree[l]{\rlnm{VarK}}{\Gamma \types x : \okty}
        &
        \SetCell[c=2]{c} \prftree[l, r]{$\Gamma \types C[\bar{B}/\bar{a}]$}{\rlnm{GVar}}{\Gamma,\,f:\forall \bar{a}. C \Rightarrow \,A \types f:A[\bar{B}/\bar{a}]}
      \\[4mm]
      \prftree[l]{\rlnm{Id}}{\Gamma,\,x:A \types x:A}
      &
        \prftree[l,r]{$\Gamma \types A \subtype B$}{\rlnm{SubL}}{\Gamma,\,M:B \types \Delta}{\Gamma,\,M:A \types \Delta}
      &
        \prftree[l,r]{$\Gamma \types B \subtype A$}{\rlnm{SubR}}{\Gamma \types M : B}{\Gamma \types M : A}
      \end{tblr}
    \]
    \\[4mm]
    \rulediv{20pt}{Functions}\\[1mm]
    \[
      \begin{tblr}{Q[c,17em]Q[c,17em]}
        \ch{\prftree[l]{\rlnm{FixsR}}{\Gamma,\,f:A \to B,\,x:A \types M : B}{\Gamma \types \fixtm{f(x)}{M} : A \to B}}
      &
        \ch{\prftree[l]{\rlnm{FixnR}}{\Gamma,\,f:A \from B,\,M:B \types x:A}{\Gamma \types \fixtm{f(x)}{M}:A \from B}}
      \\[5mm]
        \prftree[l]{\rlnm{AppL}}{\Gamma \types M : B \from A}{\Gamma,\,N:B \types \Delta}{\Gamma,\,M\,N:A \types \Delta}
      &
        \prftree[l]{\rlnm{AppR}}{\Gamma \types M : B \to A}{\Gamma \types N : B}{\Gamma \types M\,N : A}
      \end{tblr}
    \]
    \\[4mm]
    \rulediv{25pt}{Constructors}\\[2mm]
    \[
      \begin{array}{c}
        \prftree[l]{\rlnm{CnsL}}{\Gamma,\, M_i : A_i \types \Delta}{\Gamma,\,c(M_1,...\,,M_m) : c(A_1,...\,,A_m) + K \types \Delta}
      \qquad
        \prftree[l]{\rlnm{CnsR}}{\Gamma \types M_i : A_i \;(\forall i)}{\Gamma \types c(M_1,...\,,M_m) : c(A_1,...\,,A_m)}
      \end{array}
    \]
    \\[3mm]
    \rulediv{20pt}{Matching}\\[2mm]
    \[
      \begin{array}{c}
        \prftree[l]{\rlnm{MchR}}{\Gamma \types M : \Sigma_{i=1}^k\, p_i[\vv{B_x/x}]}{\Gamma \cup \{x:B_x \mid x \in \fv(p_i)\} \types P_i : A \; (\forall i)}{\Gamma \types \matchtm{M}{|_{i=1}^k (p_i \mapsto P_i)} : A}
        \\[5mm]
        \prftree[l]{\rlnm{MchL}}{\Gamma,\,P_i : A \types x : B_x \; (\forall i.\forall x\in\fv(p_i))}{\Gamma,\,(M, P_i):(p_i[\vv{B_x/x}], A) \types \Delta \; (\forall i)}{\Gamma,\,\matchtm{M}{|_{i=1}^k (p_i \mapsto P_i)} : A \types \Delta}
      \end{array}
    \]
    \\[4mm]
    \rulediv{20pt}{Evaluation}\\[2mm]
    \[
      \prftree[l]{\rlnm{CnsK}}{\Gamma,\,M_i:\okty \types \Delta}{\Gamma,\,c(M_1,...\,,M_n) : \okty \types \Delta}
      \qquad
      \prftree[l]{\rlnm{FunK}}{\Gamma,\,M:\okty \from A \types \Delta}{\Gamma,\,M\,N:\okty \types \Delta}
    \]
    \\[3mm]
    \rulediv{24pt}{Disjointness}\\[1mm]
    \[
      \prftree[l,r]{\begin{tabular}{l}$A$ an arrow, or shape\\ $\Sigma_{d\in I}\,d(...)$ with $c \notin I$\end{tabular}}{\rlnm{CnsDL}}{\Gamma,\,c(M_1,...\,,M_n) : A \types \Delta}
      \quad
      \ch{\prftree[l]{\rlnm{FixDL}}{\Gamma,\,\fixtm{f(x)}{M} : K \types \Delta}}
    \]\\[-2mm]
    \end{mdframed}
    \caption{Constrained type assignment.}\label{fig:constrained-ta-terms}
  \end{figure}
  
  

  \subsection{Type assignment}
  Finally, we have the two-sided type system itself.  In the interests of making the syntactical soundness proof and type inference as straightforward as possible (by making the system as close to a traditional type system as possible), we present the type system as an \emph{intuitionistic} sequent calculus, in the sense that there will be allowed \emph{at most one} formula on the right hand side.  We choose not to have an explicit component in the judgement for constraints in order to simplify the notation.

  \begin{definition}[Typing Formulas and Type Assignment]
  A \emph{typing formula} is a pair of shape either $M : A$ or $\mathsf{f} : S$, with $M$ a term, $A$ a monotype, $\mathsf{f}$ a top-level identifier and $S$ a type scheme.
  A \emph{typing judgment} of the system consists of a pair $\Gamma \types \Delta$ in which $\Gamma$ is a finite set of typing and subtype formulas and $\Delta$ is a finite set of typing formulas of shape $M:A$ and whose size is at most 1.  For brevity, by some abuse, we will write $\Gamma$ even for the subset $\{A \subtype B \mid A \subtype B \in \Gamma\}$ of subtype constraints contained therein.
  
  The rules of the type system are given in Figure~\ref{fig:constrained-ta-terms}.  We additionally require that: (i) in the rules \rlnm{FixsR}, \rlnm{FixnR}, \rlnm{MchL}, \rlnm{MchR} and \rlnm{FixR}, the bound variables in the principal subject of the conclusion do not occur in $\Gamma$ or $\Delta$; and (ii) in the rule \rlnm{GVar}, the vector of types $\vv{B}$ has the same length as the vector of type variables $\vv{a}$; and (iii) the rules \rlnm{CnsK} and \rlnm{CnsL} require that $n > 0$ and $1 \leq i \leq n$.
\end{definition}

Many of the typing rules of Figure~\ref{fig:constrained-ta-terms} are recognisable from Sections~\ref{sec:prog-lang} and \ref{sec:okay}, so we will just comment on new aspects.  First, we have separate typing rules for local variables (introduced by abstractions and pattern-match cases) \rlnm{Id}, and top-level identifiers \rlnm{GVar}.  Top-level identifiers can be assumed to have polymorphic type schemes $\forall \vv{a}.\,C \Rightarrow A$, so the \rlnm{GVar} rule allows for the instantiation of quantified type variables $\vv{a}$ by a vector of monotypes $\vv{B}$, subject to the requirement that each of the constraints in $C$ is already derivable from the type constraints assumed in $\Gamma$.
The \rlnm{CnsL}, \rlnm{CnsR} and \rlnm{CnsK} rules allow for refuting and affirming the types of constructor-headed terms.  
\chB{
  Note that the \rlnm{CnsR} rule allows for concluding the singleton sum $c(A_1,...,A_m)$, whereas the \rlnm{CnsL} rule allows a more general sum of shape $c(A_1,...,A_m) + K$.  This is because the former may yet be generalised to the larger type $c(A_1,\ldots,A_m) + K$ using \rlnm{SubR} (which would be impossible on the left, since \rlnm{SubL} only allows for concluding smaller types from larger ones).
}

The rules \rlnm{MchL} and \rlnm{MchR} are used for typing pattern-matching on the left and right respectively.  The \rlnm{MchR} rule is relatively standard: one must choose a typing $B_x$ for each pattern-bound variable $x$ (recall from additional requirement (ii) of Definition~\ref{def:pat-match-terms}, that we assume all bound variables throughout the pattern-match term to have distinct names), in such a way that the scrutinee $M$ is inside the sum of the induced pattern types.  These are obtained by taking each pattern case $c(x_1,\ldots,x_n)$ and replacing the free variables $x_i$ by the corresponding type $B_{x_i}$, giving $c(B_{x_1},\ldots,B_{x_n})$. Then one must show that every branch can guarantee an $A$.  

In \rlnm{MchL}, one must first derive a necessary type $B_x$ for each pattern-bound variable $x$, based on how it is used in the corresponding branch of the match.  These types then give rise to a type for the scrutinee, as above.  Then one must show that, for each branch $i$, either the desired conclusion $\Delta$ follows from the fact that the body of the branch $P_i$ has type $A$, or already from the fact that the scrutinee has type $p_i[\vv{B_x/x}]$.  This disjunction is encoded by asserting that the pair $(M,P_i)$ has type $(p_i[\vv{B_x/x}],A)$ so as to maintain the invariant that at most one non-trivial typing is introduced in each premise.  An example of how this is used is at the end of this subsection.

\chB{
Instead of a single disjointness rule \rlnm{Dis}, in this system we have disjointness rules \rlnm{CnsDL} and \rlnm{FixDL} specific to each introduction form (see Section~\ref{subsec:discussion-one} for further discussion of these alternatives).  We make this choice here to simplify type inference: we avoid the non-syntax directed rule \rlnm{Dis}, and we avoid the need to infer disjointness constraints in addition to type constraints.
}

\begin{definition}[Top-level Function Typing]
  Top-level functions are typed according to the rule:
  \[
    \prftree[r]{$\bar{a} = \fv(C) \cup \fv(A)$}{
      \Gamma \cup C \types \calM(\mathsf{f}) : A
    }
    {
      \Gamma \types \mathsf{f} : \forall \vv{a}.\,C \Rightarrow A
    }
  \]
  We write $\types \calM : \Gamma$ just if, (i) every top-level function of $\calM$ appears as a subject of $\Gamma$, and (ii) every top-level function typing $\mathsf{f}:S \in \Gamma$ is properly justified $\Gamma \types \mathsf{f} : S$.
\end{definition}

Since each rule of the system allows for an arbitrary context $\Gamma$ on the left, it satisfies weakening on the left.  Consequently, rules \rlnm{AbsR} and \rlnm{AbnR} of Section~\ref{sec:prog-lang} are admissible.
Consider the \textsf{head} function of Example~\ref{ex:head-map-defns}.  We show that $\mathsf{head}$ requires a cons with an element of type $a$ as input in order to produce an $a$, i.e. $\mathsf{head} : \forall a.\,(a::\okty) \from a$.  The derivation starts:
\[
  \prftree[l]{\rlnm{AbnR}}{
    \matchtm{xs}{y::ys \mapsto y} : a \types xs : (a :: \okty)
  }
  {
    \types \abs{xs}{\matchtm{xs}{y::ys \mapsto y}} : (a::\okty) \from a
  }
\]
Then, according to \rlnm{MchL} we must derive types for the bound variables $y$ and $ys$ based on how they were necessarily used to obtain type $a$ in their branch.  In this case, their branch body is just $y$ and so we can assign $y$ the type $a$ and $ys$, which is not used in the branch, must be given $\okty$.
Thus the scrutinee must have type $a::\okty$.  According to \rlnm{MchL}, we must show that the conclusion, $xs: (a::\okty)$, either follows from this or from the body of the branch.  In this case, it is clear that it follows already from the type of the scrutinee.  This reasoning is captured as:
\[
  \prftree[l]{\rlnm{MchL}}{
    \prftree[l]{\rlnm{Id}}{y:a \types y:a}
  }
  {
    \prftree[l]{\rlnm{VarK}}{y:a \types ys:\okty}
  }
  {
    \prftree[l]{\rlnm{CnsL}}{
      \prftree[l]{\rlnm{Id}}{xs: (a::\okty) \types xs:(a::\okty)}
    }
    {
      (xs,y) : (a::\okty,a) \types xs:(a::\okty)
    }
  }
  {
    \matchtm{xs}{y::ys \mapsto y} : a \types xs : (a :: \okty)
  }
\]

Consider the \textsf{map} function from Example~\ref{ex:head-map-defns}.  We will show that, when \textsf{map} is given a function that requires an $a$ to produce a $b$, then to produce a list of $b$ requires it be given a list of $a$:
\[
  \begin{array}{l}
  \mathsf{map} : \forall a\,b\,\ell_a\,\ell_b.\, C \Rightarrow (a \from b) \to \ell_a \from \ell_b\\
  \qquad\text{where}\;\{\; [] + (a::\ell_a) \subtype \ell_a,\,\ell_b \subtype [] + (b::\ell_b) \;\}
  \end{array}
\]
The derivation begins as we have seen previously using \rlnm{FixsR} and \rlnm{AbnR}.
The key part is to show that, under $\Gamma = C \cup \{ f:a \from b,\,m:(a \from b) \to \ell_a \from \ell_b\}$:
\[
  \Gamma, \matchtm{xs}{[] \mapsto [] \mid y::ys \mapsto f\,y::m\,f\,ys} : \ell_b \types xs : \ell_a
\]
To use the \rlnm{MchL} rule, we first derive types that were necessary for $y$ and $ys$ to produce a value of type $\ell_b$ in the cons branch (we omit some standard right-side reasoning in the latter).  
\[
  \prftree[l]{\rlnm{SubL}}{
    \prftree[l]{\rlnm{CnsL}}{
      \prftree[l]{\rlnm{AppL}}{
        \prftree[l]{\rlnm{Id}}{\Gamma  \types f : a \from b}
      }
      {
        \prftree[l]{\rlnm{Id}}{\Gamma,y:a \types y:a}
      }
      {
        \Gamma,\,f\,y : b \types y : a
      }
    }
    {
     \Gamma,\,(f\,y::m\,f\,ys) : [] + (b::\ell_b) \types y : a
    }
  }
  {
    \Gamma,\,(f\,y::m\,f\,ys) : \ell_b \types y : a
  }
  \quad
  \prftree[l]{\rlnm{SubL}}{
    \prftree[l]{\rlnm{CnsL}}{
      \prftree[l]{\rlnm{AppL}}{
        \prftree{\cdots}
        {  \Gamma \types m\,f : \ell_a \from \ell_b}
      }
      {
        \prftree[l]{\rlnm{Id}}{\Gamma,\,ys:\ell_a \types ys:\ell_a}
      }
      {
        \Gamma,\,m\,f\,ys : \ell_b \types ys : \ell_a
      }
    }
    {
     \Gamma,\,(f\,y::m\,f\,ys) : [] + (b::\ell_b) \types y : \ell_a
    }
  }
  {
    \Gamma,\,(f\,y::m\,f\,ys) : \ell_b \types ys : \ell_a
  }
\]
Then, we show that $xs:\ell_a$ is a necessary consequence of each branch being of type $\ell_b$.  In fact, it follows directly from the induced type of the scrutinee in each (we omit the \rlnm{CnsL} at the root).
\[
    \prftree[l]{\rlnm{SubR}}{
      \prftree[l]{\rlnm{Id}}{\Gamma,\,xs: (a::\ell_a) \types xs : (a::\ell_a)}
    }
    {
      \Gamma,\,xs: (a::\ell_a) \types xs : \ell_a
    }
  \qquad\qquad
    \prftree[l]{\rlnm{SubR}}{
      \prftree[l]{\rlnm{Id}}{\Gamma,\,xs: [] \types xs : []}
    }
    {
      \Gamma,\,xs: [] \types xs : \ell_a
    }
\]

As one final example, we show that $\mathsf{map}$, given a function that requires an $a$ to obtain a $b$, returns a cons with element of type $b$ only if given a cons with element of type $a$, i.e:  
\[
  \mathsf{map} : \forall a\,b. (a \from b) \to (a::\okty) \from (b::\okty)
\]
The key part of the derivation is again the \rlnm{MchL} rule.  Under $\Gamma = \{ f:a \from b,\,m:(a \from b) \to (a::\okty) \from (b::\okty)\}$, we find, as above, that the cons case requires $y$ to be an $a$ (left), but we don't require anything special of $ys$ (right):
\[
  \prftree{\cdots}{\Gamma,\,(f\,y::m\,f\,ys) : \ell_b \types y : a}
  \qquad\qquad
  \prftree[l]{\rlnm{VarK}}{\Gamma,\,(f\,y::m\,f\,ys) : \ell_b \types ys : \okty}
\]
When showing that the desired conclusion follows from the two cases, we need to use the choice offered by the pair $(M,P_i) : (p_i[\vv{B_x/x}],A)$ on the left of the second family of premises of \rlnm{MchL}.  By using \rlnm{CnsL} to choose one component of the pair with which to continue the proof, we can effectively ignore irrelevant cases.
Here, the nil case is excluded by the type $b::\okty$ of the output (left), and the result follows from the type of the scrutinee in the cons case (abbreviating $M = f\,y::m\,f\,ys$ and $\Delta = \{\;xs: (a::\okty)\;\}$):
\[
  \prftree[l]{\rlnm{CnsL}}{
    \prftree[l]{\rlnm{CnsDL}}{
      \Gamma,\,[] : (b::\okty) \types \Delta
    }
  }
  {
    \Gamma,\,(xs,[]) : ([],b::\okty) \types \Delta
  }
  \qquad
  \prftree[l]{\rlnm{CnsL}}{
    \prftree[l]{\rlnm{Id}}{
      \Gamma,\,xs : (a::\okty) \types \Delta
    }
  }
  {
    \Gamma,\,(xs,M) : (a::\okty,b::\okty) \types \Delta
  }
\]
We can now prove the example of ill-typedness in the introduction, the term that takes the head of a list obtained by mapping over the empty list will not evaluate:  $\mathsf{head}\,(\mathsf{map}\, (\abs{x}{x})\, []) : \okty \types$.

\subsection{Constrained Type Inference}
For the purpose of inferring types and, in the next part, proving soundness, we are going to restrict ourselves to judgements that are as close to those in a traditional type system as possible: they will have at most one typing whose subject is not a variable.   The following definition is a variation of Definition~\ref{def:type-envs} suitable for our constrained type system.
\begin{definition}[Type Environments]
  A finite set of formulas $\Gamma$ is said to be a \emph{type environment} just if (i) all the typing formulas contained therein are variable typings $x:A$ and top-level identifier typings $\mathsf{f}:A$, and (ii) if $M:A \in \Gamma$ and $M:B \in \Gamma$ then $A=B$.  A type environment is said to be a \emph{top-level type environment} just if all the typing formulas are top-level identifier typings.
\end{definition}

Inference for the two-sided constrained type system works, broadly, like Hindley-Milner constrained type inference (see e.g. \citep{odersky-sulzmann-wehr-TSPOS1999}), but with two notable differences. The first is that most terms will have multiple `principal' types.  This is because there are several shapes of proof tree with the same subject in the conclusion, but whose inferred types are not necessarily related.  Hence, type inference infers sets of types\footnote{One can observe the same phenomenon in, e.g. intersection type systems.}.  
The other main difference is that our algorithm requires both the left environment $\Gamma$ and the right environment $\Delta$ as input when inferring the types of a term on the left.  
Hence, inference is split into two procedures, $\mathsf{InferL}$ and $\mathsf{InferR}$:
\begin{align*}
    &\mathsf{InferL} \in  \mathcal{P}(\mathsf{Variable} \times \mathsf{Type}) \to \mathsf{Term} \to \mathcal{P}(\mathsf{Variable} \times \mathsf{Type}) \to  \mathcal{P}(\mathsf{Judgement})\\
    &\mathsf{InferR} \in \mathcal{P}(\mathsf{Variable} \times \mathsf{Type}) \to \mathsf{Term} \to  \mathcal{P}(\mathsf{Judgement})
\end{align*}

The call $\mathsf{InferL}(\Gamma)(M)(\Delta)$ returns a finite set of judgements\footnote{We formalise it this way for convenience, but in practice it need only return $C$ and $A$.} $C \cup \Gamma,\,M:A \types \Delta$ that are principal in the sense that all provable judgements with the same left-environment $\Gamma$, subject $M$ and right-environment $\Delta$ arise as a type-substitution instance of one of those in the set, but with a possibly stronger set of assumed type constraints.  The call $\mathsf{InferR}(\Gamma)(M)$ acts similarly on the right.

\newcommand{\eff}{\mathsf{f}}

The algorithms terminate because, modulo subtyping, every proof tree has a maximum height that is determined by the shape of the term.  As usual, subtyping can be absorbed into the other rules in a so-called algorithmic system, and then inference consists of (implicitly) constructing the finite set of principal proof trees for the given term and environment(s).  Inferring on the right for the term $\abs{x}{\eff\,x}$ in the environment $\{\eff : [] \from \okty\}$, we will obtain:
\[
  \begin{array}{rcl}
    \okty \subtype a_1,\,a_1 \from a_2 \subtype a_3,\,\eff: [] \from \okty &\types& \abs{x}{\eff\,x} : a_3\\
    \okty \subtype a_1,\,[] \from \okty \subtype a_3,\,a_3 \subtype a_2 \from a_1,\,a_1 \from a_4 \subtype a_5,\,\eff: [] \from \okty &\types& \abs{x}{\eff\,x} : a_5\\
    a_2 \subtype a_1,\,[] \from \okty \subtype a_3,\,a_3 \subtype a_2 \from a_4,\,a_1 \from a_4 \subtype a_5,\,\eff: [] \from \okty &\types& \abs{x}{\eff\,x} : a_5\\
    a_1 \subtype a_2,\, [] \from \okty \subtype a_3,\, a_3 \subtype a_2 \to a_4,\,a_1 \to a_4 \subtype a_5,\, \eff: [] \from \okty &\types& \abs{x}{\eff\,x} : a_5
  \end{array}
\]

The first three correspond, modulo subtyping inferences, to proof trees rooted at \rlnm{AbnR}.  The first corresponds to the case when the tree is then immediately closed by the \rlnm{VarK} axiom.  The second corresponds to proceeding by \rlnm{AppL} before closing with \rlnm{GVar} in the left branch and \rlnm{VarK} in the right branch.  The third is similar but corresponds to closing the right branch with \rlnm{Id} instead.  The fourth corresponds to a tree rooted at \rlnm{AbsR}, after which the shape is completely determined (modulo subtyping).  However, notice that the constraints returned in this last case are inconsistent, we have $[] \from \okty \subtype a_3$ and $a_3 \subtype a_2 \to a_4$, but the former function type contains functions that go wrong after being given an input, whereas the latter function type does not.  Indeed, one cannot assign a sufficiency arrow type to this term in the given environment under consistent subtyping assumptions.
The full definition is given in \iftoggle{supplementary}{Appendix \ref{apx:algo}}{the long version \cite{ramsay2023illtyped}}.
\begin{theorem}[Correctness of the Type Inference Algorithm]\label{thm:inference-correctness}
    Let $\Gamma, \Delta$ be type environments without type constraints, $C$ and $C'$ be sets of  constraints, $M$ be a term, and $A$, $A'$ be types. Then:
    \begin{description}
        \item[\rlnm{Left Soundness}] If $(C \cup \Gamma,\,M : A \types \Delta) \in \mathsf{InferL}(\Gamma)(M)(\Delta)$, then $C \cup \Gamma,\,M : A \types \Delta$ is provable.
        \item[\rlnm{Left Completeness}] If $C' \cup \Gamma,\,M : A \types \Delta$ is provable, then there exists a type substitution $\sigma$ and a judgement $(C \cup \Gamma,\,M : A' \types \Delta)\in\mathsf{InferL}(\Gamma)(M)(\Delta)$ such that $A = A'\sigma$ and $C' \types C\sigma$.
      \end{description}
      \begin{description}
        \item[\rlnm{Right Soundness}] If $(C \cup \Gamma \types M : A)\in\mathsf{InferR}(\Gamma)(M)$, then $C \cup \Gamma \types M : A$ is provable.
        \item[\rlnm{Right Completeness}] If $C' \cup \Gamma \types M : A$ is provable, then there exists a type substitution $\sigma$ and a judgement $(C \cup \Gamma \types M : A')\in\mathsf{InferR}(\Gamma)(M)$ such that $A = A'\sigma$ and $C' \types C\sigma$.
    \end{description}
\end{theorem}

%% file: consistency.tex
\subsection{Syntactic Soundness}\label{sec:constrained-soundness}
Rather than proving semantic soundness, we take the opportunity to show how one can prove a syntactic soundness result in the sense of \citet{wright-felleisen-infcomp1994}, but first we need to generalise the definition of well-typed and ill-typed to account for top-level identifiers.


\begin{definition}
  Suppose $M$ is closed and $\Gamma$ is a consistent top-level type environment with $\types \calM : \Gamma$.
  \begin{itemize}
    \item We say that $M$ is \emph{well-typed in} $\Gamma$ just if $\Gamma \types M : \okty$.
    \item We say that $M$ is \emph{ill-typed in} $\Gamma$ just if $\Gamma,\,M:\okty \types$.
  \end{itemize}
\end{definition}

The force of the qualifier \emph{consistent} is to require that the type constraints in a type environment are not contradictory.  Several definitions are possible, and we use an adaptation of the syntactical definition given by \citet{eifrig-et-al-oopsla1995}.  Since it is orthogonal to the two-sided aspect of the work, the definition has been relegated to \iftoggle{supplementary}{Appendix~\ref{apx:syntactic-soundness}}{the long version of this paper \cite{ramsay2023illtyped}}.

An appropriate formulation of \emph{progress} for two-sided systems must take into account a non-trivial subject on the left of the turnstile as well as the right.  On the right, as usual, we have that terms can either make a step or are already values.  However, in a typing $M:A$ on the left, it is possible that $M$ can make a step, already be stuck or even be a value.  However, in the latter case, the value must not be in $A$.  We can state this succinctly as follows:

\begin{theorem}[Progress]\label{thm:progress}
  Suppose $M$ is closed and $\Gamma$ is a consistent, top-level type environment.
  \begin{itemize}
    \item If $\Gamma \types M:A$ then either $M$ can make a step, or $M$ is a value
    \item If $\Gamma,\,M:A \types$ then either $M$ can make a step, or $\Gamma \not\types M : A$
  \end{itemize}
\end{theorem}

The formulation uses a fact about typing in our two-sided systems that is peculiar to systems with the necessity arrow: all closed values are typable on the right.  We already remarked in Section~\ref{sec:okay} that every (fixpoint) abstraction is typable on the right with $\okty \from A$, and a simple induction shows that every constructor-headed term has a corresponding constructor-headed type.  Thus, $\Gamma \not\types M : A$ above in particular implies that $M$ is not a value of type $A$.

For \emph{preservation}, it is typical to prove some preliminary lemmas that show closure under well-typed substitutions.  In the two-sided case, there are many more of these substitution lemmas, since one has to take account of the possibilities that the substituted-for variable occurs on the opposite side of the turnstile to the subject (as usual), on the same side, or not at all (which also relies on the fact that closed values are typable).  However, preservation then follows directly.

\begin{theorem}[Preservation]
  Suppose $M$ is a closed term and $\Gamma$ is a consistent, top-level type environment in which all top-level identifiers have justified typings $\types \calM : \Gamma$.
  \begin{itemize}
    \item If $\Gamma \types M : A$ and $M \ped N$ then $\Gamma \types N : A$.
    \item If $\Gamma,\,M:A \types$ and $M \ped N$ then $\Gamma,\,N:A \types$.
  \end{itemize}
\end{theorem}

\noindent
Finally, syntactic soundness follows immediately.

\begin{theorem}[Syntactic Soundness]\label{thm:weak-soundness}
  Suppose $M$ is a closed term and $\Gamma$ is a consistent, top-level type environment in which all top-level identifiers have justified typings $\types \calM : \Gamma$.
  \begin{itemize}
    \item If $M$ is well-typed in $\Gamma$, then $M$ does not go wrong.
    \item If $M$ is ill-typed in $\Gamma$, then $M$ does not evaluate.
  \end{itemize}
\end{theorem}


%% file: related.tex
\section{Conclusion and Related Work}\label{sec:conclusion}

We have introduced \emph{two-sided type systems} which are sequent calculi for typing formulas, made possible through a new function type $A \from B$.  We have shown several ways in which these calculi can be considered sound, and illustrated how left-sided rules can be added to a constrained type system.  We also investigated the internalisation of negation, and a one-sided system without a true negatives theorem.  
Many of the basic constants of the programming languages we use today have necessary requirements on their inputs.  However, traditional type systems largely ignore this basic aspect of their behaviour.  Two-sided systems can use this additional dimension, and enable simple reasoning about program incorrectness as well as program correctness.

\paragraph{Typing contexts beyond variables.}
Many type systems for higher-order program verification generalise the typing context to allow logical formulae for the purpose of recording path conditions.  For example, Liquid Types \cite{rondon-et-al-pldi2008,vazou-et-al-esop2013,vazou-et-al-icfp2015} and the systems of \citet{terauchi-popl2010} and \citet{unno-kobayashi-PPDP2009} place the Boolean valued expression $M$ (or an equivalent formula) into the context when proving that the `then'-branch of a conditional, $\mathsf{if}\ {M}\ \mathsf{then}\ {N}\ \mathsf{else}\ {P}$, is correctly typed.  In a two-sided system, this can be managed very naturally by introducing separate types for $\mathsf{true}$ and $\mathsf{false}$ (as we have in our system of Section~\ref{sec:constrained-typing}) and assuming $M:\mathsf{true}$.
A rather different use of the typing context was made by \citet{curien-herbelin-icfp2000}.  Their system is a kind of sequent calculus in which the formulas on the left of the turnstile comprise variable typings and a stack of terms (cf. Krivine machines), which is thought of as an evaluation context.   

\paragraph{Constrained type systems.}
Constrained type systems and inclusion constraints go back at least to the work of \citet{mitchell-popl1984}.  A unified treatment is given by \citet{odersky-sulzmann-wehr-TSPOS1999}.  Our system takes inspiration from the one-sided system of \citet{aiken-et-al-popl1994} (see also \citet{marlow-wadler-icfp1997}), though their constraints are far more expressive.  Our approach to consistency of subtype constraints is an adaptation of the purely syntactic approach of \citet{eifrig-et-al-oopsla1995}.  It is not clear how well our algorithm would scale in practice, ideally one would like some guarantees on the size of types or the efficiency of inference, such as features in more recent works by \citet{dolan-mycroft-popl} and \citet{jones-ramsay-popl2021}.  None of these systems has a true positives theorem.  

\paragraph{Success Typing.}
Our original motivation was to try to better understand the very effective \emph{success typing} paradigm for Erlang \cite{lindahl-sagonas-ppdp2006}.  Work by \citet{jakob-thiemann-nasa2015} provides a useful view through their falsification type system, in which one constructs a logical formula used to describe inputs that guarantee to crash the function.  The formalisation of \citet{Lopez-Fraguas-et-al-lpar2018} is also enlightening and they moreover extend the system with polymorphism.

Success Types are perhaps the best known example of a type-based true positives theorem, `ill-typed programs always fail' \citep{sagonas-flp2010}.  They are so named because they over-approximate the successes of expressions: if $A$ is a success type and $M:A$ then $M$ may succeed in evaluating to an $A$, but may also go wrong or diverge.  Hence why we named our alternative semantics after this paradigm.  Note, the `ill-typed' of the slogan really means untypable in our sense.

A key feature is a bespoke function type.  It is defined in \citep{lindahl-sagonas-ppdp2006} as follows: \begin{quote}\emph{A success typing of a function $f$ is a type signature $(\bar{\alpha}) \to \beta$, such that whenever an application $f(\bar{p})$ reduces to a value $v$, then $v \in \beta$ and $\bar{p} \in \bar{\alpha}$}.\end{quote}

\noindent
Thus, the \emph{success arrow}, let us write $\to_s$, can be understood as having the right rule shown to the side.  That is, it contains a component of necessity, but restricted to $\okty$ on the right. Hence, with only this
\begin{wrapfigure}[4]{l}{6.5cm}\vspace{-1ex}
  \begin{center}
    $
      \prftree{
        \Gamma,\,M:\okty \types x : \okty 
      }
      {
        \Gamma,\,M:\okty \types M : B
      }
      {
        \Gamma \types \abs{x}{M} : A \to_s B,\,\Delta
      }
    $
  \end{center}
\end{wrapfigure} 
arrow, one can refute that an application $MN$ has a type $A$ only when $A=\okty$.
Yet, despite this, success types appear to be highly effective in practice.  One reason could be the following: if an application $MN$ fails to evaluate to an $A$, then there are three possibilities, it could go wrong, it could diverge, or it could evaluate to something else.  In the first two cases, one can attempt to refute that it has type $\okty$, and in the third, one can attempt to affirm that it has a type that is disjoint from $A$ instead.  Thus, it is likely that one can get quite far with only this form of necessity. 

\paragraph{Logics for incorrectness.}
Incorrectness logic of \citet{ohearn-popl2019} has sparked a new interest in systems for reasoning about program behaviour that enjoy true positives theorems.  Subsequent work has extended its scope \citep{raad-et-al-popl2022}, and demonstrated its real-world effectiveness \citep{le-et-al-oopsla2022}.  A key feature of these logics is \emph{under-approximation} as a means to achieve true positives, and this allows for a more specific notion of `positive', which excludes divergence.  In incorrectness logic, states in a postcondition must be \emph{reachable}.  This leads to the need to reason about termination, but it also allows for an analysis to dynamically drop disjunctions in order to scale well.  A interesting alternative is the Outcome Logic of \citet{zilberstein-et-al-oopsla2023}.  Here there is a distinction between Boolean disjunction and outcome disjunction.  The former has true as annihilator, but the latter does not, and this allows for the dropping of outcomes for efficiency.

\paragraph{Necessity.}
The foundation of our system is not under-approximation, but necessity.  When $M : A$ is provable, type $A$ is an \emph{over-approximation} of $M$.  
\chC{
  This makes the system closer to the works of, on the one hand, \citet{cousot-et-al-vmcai2011,coutsot-et-al-vmcai2013}, and on the other, \citet{mackay-et-al-oopsla2022}.  The former considers the problem of inferring \emph{necessary preconditions} in order to detect contract violations that will lead to assertion failures.  Specifications are given as code assertions and  various analyses based on abstract interpretation are showcased for efficient inference.  The latter uses necessity as a means to specify the \emph{robustness} of module specifications in object-oriented programming.  Three novel necessity operators are added to an assertion language equipped with certain object capabilities, and a logic is developed for proving that a module satisfies its necessity specifications.
}


\paragraph{Type systems with complement.}
A complement operator appears in the system of \citet{aiken-et-al-popl1994} and a negation in the work of \citet{parreaux-et-al-oopsla2022}, though neither has a true positives theorem.
In the highly expressive dependent type system of \citet{unno-et-al-popl2018}, none of the programs can go wrong, but the system is nevertheless very well-equipped for incorrectness reasoning, and includes sophisticated means for reasoning about recursion and tracking divergence.  Since their system can express the complements of arbitrary types, it should be possible to encode a version of our necessity arrow.  One can think of our one-sided system of Section~\ref{sec:complements} as a kind of sweet spot for partial correctness, in which the theory and automation become particularly straightforward.

%% file: apx-language.tex
\section{Additional Material in Support of Section~\ref{sec:prog-lang}}

\subsection{Proof of Lemma~\ref{lem:structural-admissibility}}

\chB{
We prove the following result:
\begin{quote}    
Let $\Gamma$ and $\Delta$ be sets of formulas, $M$ be a term, and $A$ be a type. If $\Gamma \types \Delta$ then both, $\Gamma,\, M : A \types \Delta$ and $\Gamma \types M : A,\, \Delta$.
\end{quote}
\begin{proof}
    Let $\Gamma, \Delta, M$, and $A$ be as above, such that $\Gamma \types \Delta$. we then proceed by induction on the derivation of $\types$.\\
    \begin{description}
        \item[$\rlnm{ID}$] {
            We have $\Gamma,\, x : A \types x : A,\, \Delta$. Then, clearly, by $\rlnm{ID}$ we also have, $\Gamma,\, M : A',\, x : A \types x : A,\, \Delta$, and by the same reasoning, $\Gamma,\, x : A \types x : A,\, M : A',\, \Delta$ as required.
            }
        \item[$\rlnm{ZeroR}$] {
            We have $\Gamma \types \zerotm : \intty,\, \Delta$. Then, clearly, by $\rlnm{ZeroR}$ we also have, $\Gamma,\, M : A \types \zerotm : \intty,\, \Delta$ and $\Gamma \types \zerotm : \intty,\, M : A,\, \Delta$.
            }
        \item[$\rlnm{cR}$] {
            We have $\Gamma \types c(N) : \intty,\, \Delta$ and $\Gamma \types N : \intty,\, \Delta$ for $c\in\{\succtm,\, \predtm\}$. By the inductive hypothesis, we have $\Gamma,\, M : A \types N : \intty,\, \Delta$ and $\Gamma \types N : \intty,\, M : A,\, \Delta$. Then, by $\rlnm{cR}$, we have $\Gamma,\, M : A \types c(N) : \intty,\, \Delta$ and $\Gamma \types c(N) : \intty,\, M : A,\, \Delta$ as required.
            }
        \item[$\rlnm{FixsR}$] {
            We have $\Gamma \types \fixtm{f(x)}{N} : A \to B,\, \Delta$ and $\Gamma,\, f : A \to B,\, x : A \types N : B,\, \Delta$. By the inductive hypothesis, we have $\Gamma,\, M : A',\, f : A \to B,\, x : A \types N : B,\, \Delta$ and $\Gamma,\, f : A \to B,\, x : A \types N : B,\, M : A',\, \Delta$. Then, by $\rlnm{FixsR}$, we have $\Gamma,\, M : A' \types \fixtm{f(x)}{N} : A \to B,\, \Delta$ and $\Gamma \types \fixtm{f(x)}{N} : A \to B,\, M : A',\, \Delta$ as required.
            }
        \item[$\rlnm{FixnR}$] {
            We have $\Gamma \types \fixtm{f(x)}{N} : A \from B,\, \Delta$ and $\Gamma,\, f : A \from B,\, x : A \types N : B,\, \Delta$. By the inductive hypothesis, we have $\Gamma,\, M : A',\, f : A \from B,\, x : A \types N : B,\, \Delta$ and $\Gamma,\, f : A \from B,\, x : A \types N : B,\, M : A',\, \Delta$. Then, by $\rlnm{FixsR}$, we have $\Gamma,\, M : A' \types \fixtm{f(x)}{N} : A \from B,\, \Delta$ and $\Gamma \types \fixtm{f(x)}{N} : A \from B,\, M : A',\, \Delta$ as required.
            }
        \item[$\rlnm{AppR}$] {
            We have $\Gamma \types P\ Q : B,\, \Delta$ and $\Gamma \types P : A \to B,\, \Delta$ and $\Gamma \types Q : A,\, \Delta$. By the inductive hypothesis, we have $\Gamma,\, M : A' \types P : A \to B,\, \Delta$, $\Gamma,\, M : A' \types Q : A,\, \Delta$, $\Gamma \types P : A \to B,\, M : A' \Delta$ and $\Gamma \types Q : A,\, M : A' \Delta$. Then, by $\rlnm{AppR}$, we have $\Gamma,\, M : A' \types P\ Q : B,\, \Delta$ and $\Gamma \types P\ Q : B,\, M : A' \Delta$ as required.
            }
        \item[$\rlnm{PairR}$] {
            We have $\Gamma \types (P, Q) : A \times B,\, \Delta$ and $\Gamma \types P : A,\, \Delta$ and $\Gamma \types Q : B,\, \Delta$. By the inductive hypothesis, we have $\Gamma,\, M : A' \types P : A,\, \Delta$, $\Gamma,\, M : A' \types Q : B,\, \Delta$, $\Gamma \types P : A,\, M : A' \Delta$ and $\Gamma \types Q : B,\, M : A' \Delta$. Then, by $\rlnm{PairR}$, we have $\Gamma,\, M : A' \types (P, Q) : A \times B,\, \Delta$ and $\Gamma \types (P, Q) : A \times B,\, M : A' \Delta$ as required.
            }
        \item[$\rlnm{LetR}$] {
            We have $\Gamma \types \lettm{(x,y)}{P}{Q} : A,\, \Delta$, $\Gamma \types P : B \times C,\, \Delta$, and $\Delta,\, x : B,\, y : C \types Q : A,\, \Delta$. By the inductive hypothesis, we have $\Gamma,\, M : A' \types P : B \times C,\, \Delta$, $\Delta,\, M : A',\, x : B,\, y : C \types Q : A,\, \Delta$, $\Gamma \types P : B \times C,\, M : A',\, \Delta$, and $\Delta,\, x : B,\, y : C \types Q : A,\, M : A',\, \Delta$. Then by $\rlnm{LetR}$, we have $\Gamma,\, M : A' \types \lettm{(x,y)}{P}{Q} : A,\, \Delta$ and $\Gamma \types \lettm{(x,y)}{P}{Q} : A,\, M : A',\, \Delta$ as required.
            }
        \item[$\rlnm{IfzR}$] {
            We have $\Gamma \types \ifztm{R}{P}{Q} : A,\, \Delta$, $\Gamma \types R : \intty,\, \Delta$, $\Gamma \types P : A,\, \Delta$, and $\Gamma \types Q : A,\, \Delta$. By the inductive hypothesis, we have $\Gamma,\, M : A' \types R : \intty,\, \Delta$, $\Gamma,\, M : A' \types P : A,\, \Delta$, $\Gamma,\, M : A \types Q : A,\, \Delta$, $\Gamma \types R : \intty,\, M : A',\, \Delta$, $\Gamma \types P : A,\, M : A',\, \Delta$, and $\Gamma \types Q : A,\, M : A',\, \Delta$. Then by $\rlnm{IfzR}$, we have $\Gamma,\, M : A' \types \ifztm{R}{P}{Q} : A,\, \Delta$ and $\Gamma \types \ifztm{R}{P}{Q} : A,\, M : A',\, \Delta$ as required.
            }
        \item[$\rlnm{Dis}$] {
            We have $\Gamma,\, N : A \types \Delta$, $\Gamma \types N : B,\, \Delta$, and $A \distype B$. By the inductive hypothesis, we have $\Gamma,\, M : A',\, \types N : B,\, \Delta$ and $\Gamma \types N : B,\, M : A',\, \Delta$. By $\rlnm{Dis}$, we have $\Gamma,\, M : A'\, N : A \types \Delta$ and $\Gamma,\, N : A \types M : A',\, \Delta$.
            }
        \item[$\rlnm{PairL}$] {
            We have $\Gamma,\, (M_1, M_2) : B_1 \times B_2 \types \Delta$ and $\Gamma,\, M_i : B_i \types \Delta$ for some $i \in \{1, 2\}$. By the inductive hypothesis, we have $\Gamma,\, M : A,\, M_i : B_i \types \Delta$ and $\Gamma,\, M_i : B_i \types M : A,\, \Delta$. Then, by $\rlnm{PairL}$, we have $\Gamma,\, M : A,\, (M_1, M_2) : B_1 \times B_2 \types \Delta$ and $\Gamma,\, (M_1, M_2) : B_1 \times B_2 \types M : A,\, \Delta$ as required.
            }
        \item[$\rlnm{AppL}$] {
            We have $\Gamma,\, P\ Q : B \types \Delta$, $\Gamma \types P : A \from B,\, \Delta$, and $\Gamma,\, Q : A \types \Delta$. By the inductive hypothesis, we have $\Gamma,\, M : A' \types P : A \from B,\, \Delta$, $\Gamma,\, M : A',\, Q : A \types \Delta$, $\Gamma \types P : A \from B,\, M : A',\, \Delta$, and $\Gamma,\, Q : A \types M : A',\, \Delta$. Then by $\rlnm{AppL}$, we have $\Gamma,\, M : A',\, P\ Q : B \types \Delta$ and $\Gamma,\, P\ Q : B \types M : A',\, \Delta$ as required.
            }
        \item[$\rlnm{cL}$] {
            We have $\Gamma,\, c(N) : \intty \types \Delta$ and $\Gamma,\, N : \intty \types \Delta$ for $c\in\{\succtm,\, \predtm\}$. By the inductive hypothesis, we have $\Gamma,\, M : A,\, N : \intty \types \Delta$ and $\Gamma,\, N : \intty \types M : A,\, \Delta$. Then, by $\rlnm{cL}$, we have $\Gamma,\, M : A,\, c(N) : \intty \types \Delta$ and $\Gamma,\, c(N) : \intty \types M : A,\, \Delta$ as required.
            }
        \item[$\rlnm{IfzL1}$] {
            We have $\Gamma,\, \ifztm{R}{P}{Q} : A \types \Delta$ and $\Gamma,\, R : \intty \types \Delta$. By the inductive hypothesis, we have $\Gamma,\, M : A',\, R : \intty \types \Delta$ and $\Gamma,\, R : \intty \types M : A',\, \Delta$. Then by $\rlnm{IfzL1}$ we have $\Gamma,\, M : A',\, \ifztm{R}{P}{Q} : A \types \Delta$ and $\Gamma,\, \ifztm{R}{P}{Q} : A \types M : A',\, \Delta$ as required.
            }
        \item[$\rlnm{IfzL2}$] {
            We have $\Gamma,\, \ifztm{R}{P}{Q} : A \types \Delta$, $\Gamma,\, P : A \types \Delta$, and $\Gamma,\, Q : A \types \Delta$. By the inductive hypothesis, we have $\Gamma,\, M : A',\, P : A \types \Delta$, $\Gamma,\, M : A',\, Q : A \types \Delta$, $\Gamma,\, P : A \types M : A',\, \Delta$, and $\Gamma,\, Q : A \types M : A',\, \Delta$. Then by $\rlnm{IfzL2}$, we have $\Gamma,\, M : A',\, \ifztm{R}{P}{Q} : A \types \Delta$ and $\Gamma,\, \ifztm{R}{P}{Q} : A \types M : A',\, \Delta$ as required.
            }
        \item[$\rlnm{LetL1}$] {
            We have $\Gamma,\, \lettm{(x, y)}{P}{Q} : A \types \Delta$ and $\Gamma,\, Q : A \types \Delta$. By the inductive hypothesis, we have $\Gamma,\, M : A',\, Q : A \types \Delta$ and $\Gamma,\, Q : A \types M : A',\, \Delta$. Then by $\rlnm{LetL1}$, we have $\Gamma,\, M : A',\, \lettm{(x, y)}{P}{Q} : A \types \Delta$ and $\Gamma,\, \lettm{(x, y)}{P}{Q} : A \types M : A',\, \Delta$ as required.
            }
        \item[$\rlnm{LetL2}$] {
            We have $\Gamma,\, \lettm{(x_1, x_2)}{P}{Q} : A \types \Delta$, $\Gamma,\, P : B_1 \times B_2 \types \Delta$, and $(\forall i)\, \Gamma,\, Q : A \types x_i : B_i,\, \Delta$. By the inductive hypothesis, we have $\Gamma,\, M : A',\, P : B_1 \times B_2 \types \Delta$, $(\forall i)\, \Gamma,\, M : A',\, Q : A \types x_i : B_i,\, \Delta$, $\Gamma,\, P : B_1 \times B_2 \types M : A',\, \Delta$, and $(\forall i)\, \Gamma,\, Q : A \types x_i : B_i,\, M : A',\, \Delta$. Then, by $\rlnm{LetL2}$, we have $\Gamma,\, \lettm{(x_1, x_2)}{P}{Q} : A,\, M : A' \types \Delta$ and $\Gamma,\, \lettm{(x_1, x_2)}{P}{Q} : A \types M : A',\, \Delta$ as required.
            }
    \end{description}
\end{proof}

\begin{corollary}
    Both left and right weakening hold.
\end{corollary}
}

%% file: apx-cbv.tex
\section{Additional Material in Support of Section~\ref{sec:semantics}}\label{sec:apx-cbv}

First, we give an explicit characterisation of getting stuck.

\begin{definition}[Getting Stuck]
  We separate the values into different classes:
    \begin{multicols}{2}
      \begin{itemize}
        \item $\Vals$ -- all values
        \item $\FunVals$ -- those values of shape \ch{$\fixtm{f(x)}{M}$}
        \item $\NatVals$ -- those values of shape $\pn{n}$
        \item $\PairVals$ -- those values of shape $(V,\,W)$
      \end{itemize}
    \end{multicols}
    A term is said to be a \emph{stuckex} if it is of one of the following forms:
    \begin{multicols}{2}
    \begin{itemize}
      \item $\succtm(V)$ or $\predtm(V)$ with $V \notin \NatVals$
      \item $V\,M$ with $V \notin \FunVals$
      \item $\iftm{V}{N}{P}$ with $V \notin \NatVals$
      \item $\lettm{(x,y)}{V}{M}$ with $V \notin \PairVals$
    \end{itemize}
  \end{multicols}
    A term is said to be \emph{stuck} just if it has shape $\calE[M]$ with $M$ a stuckex.  A term $M$ is said to \emph{get stuck} just if $M \peds N$ with $N$ stuck.  It is easy to see that a term in normal form is either stuck or is a value.  Hence, every term either diverges, gets stuck or evaluates.
\end{definition}

Next, we have a series of lemmas building up some basic results about the combinatorics of reduction.  The overall goal is to show that replacing a subterm by another that either reduces to the same normal form or diverges leads to an essentially similar outcome.

In the following lemma, we use the phrase `$M$ is not blocked by $z$' to mean that $M$ does not have shape $\calE[z]$ for any evaluation context $\calE$.

\begin{lemma}\label{lem:contexts-mod-subs}
  If $M[N/z] = \calE[P]$ and $M$ is not blocked by $z$, then there is $\calF$ and $P'$ such that $M = \calF[P']$ and $\calE = \calF[N/z]$ and $P=P'[N/z]$.
\end{lemma}
\begin{proof}
  By induction on $\calE$.
  \begin{itemize}
    \item Suppose $\calE = \Box$ and $M$ is not blocked on $z$.  Then $M[N/z] = P$ so let $\calE' = \Box$ and $P' = M$.
    \item Suppose $\calE$ is of shape $\succtm(\calE')$ and $M$ is not blocked by $z$, so $M[N/z]$ is of shape $\succtm(\calE'[P])$.  Since $M$ is not blocked by $z$, it must have shape $\succtm(Q)$ for some $Q$ again not blocked by $z$, also $Q[N/z] = \calE'[P]$.  Hence, it follows from the induction hypothesis that there is some $\calE''$ and $P'$ such that $P'[N/z] = P$ and $\calE' = \calE''[N/z]$.  Hence, let $\calF := \succtm(\calE'')$.
    \item Suppose $\calE$ is of shape $\predtm(\calE')$, we reason analogously to the previous case.
    \item Suppose $\calE$ is of shape $(\calE',Q_2)$ and $M$ is not blocked by $z$, so $M[N/z]$ is of shape $(\calE'[P],Q_2)$.  Since $M$ is not blocked by $z$, it has shape $(Q_1,Q_2')$ and $Q_1$ is not blocked by $z$ and $Q_1[N/z] = \calE'[P]$ and $Q_2'[N/z] = Q_2$.  It follows by induction that there is $\calE''$ and $P'$ such that $\calE''[N/z] = \calE'$ and $P'[N/z] = P$.  Hence, let $\calF := (\calE'',\,Q_2')$.
    \item Suppose $\calE = (V,\calE')$ and $M$ is not blocked by $z$.  So, $M[N/z]$ is of shape $(V,\calE'[P])$.  Since $M$ is not blocked by $z$, $M$ must be of shape $(V',Q)$ and $Q$ is again not blocked by $z$, $V'[N/z] = V$ and $Q[N/z] = \calE'[P]$.  Then it follows from the induction hypothesis that there are $\calE''$ and $P'$ such that $P'[N/z] = P$ and $\calE''[N/z] = \calE'$.  Hence, let $\calF$ be $(V',\,\calE'')$ and the result is immediate.
  
    \item Suppose $\calE$ is of shape $\lettm{(x,y)}{calE'}{Q}$ and $M$ is not blocked by $z$. Then $M[N/z]$ is of shape $\lettm{(x,y)}{calE'[P]}{Q}$. Since $M$ is not blocked by $z$, it must be that $M$ has shape $\lettm{(x,y)}{R}{Q'}$ and $R[N/z] = \calE'[P]$ and $Q'[N/z] = Q$ and $R$ is not blocked by $z$.  Hence, it follows from the induction hypothesis that there are $P'$ and $\calE''$ such that $R=\calE''[P']$ and $P'[N/z] = P$ and $\calE''[N/z] = \calE'$.  Thus, let $\calF$ be $\lettm{(x,y)}{\calE''}{Q'}$.
    
    \item Suppose $\calE$ is of shape $\iftm{\calE'}{Q_1}{Q_2}$ and $M$ is not blocked by $z$.  Then $M[N/z] = \iftm{\calE'[P]}{Q_1}{Q_2}$. Since $M$ is not blocked by $z$, it must be that $M$ has shape\\ $\iftm{Q_0}{Q_1'}{Q_2'}$ and $Q_0[N/z] = \calE'[P]$, $Q_1'[N/z] = Q_1$, $Q_2'[N/z] = Q_2$.  Therefore, it follows from the induction hypothesis that there is $P'$ and $\calE''$ such that $P'[N/z] = P$ and $\calE''[N/z] = \calE'$.  Thus, take $\calF$ to be $\iftm{\calE''}{Q_1'}{Q_2'}$ and the result follows.
    
    
    \item \ch{Suppose $\calE$ is of shape $(\fixtm{f(x)}{Q})\,\calE'$ and $M$ is not blocked by $z$.  Then $M[N/z]$ is equal to $(\fixtm{f(x)}{Q})\,\calE'[P]$.  Since $M$ is not blocked by $z$, $M$ must have shape $(\fixtm{f(x)}{Q'})\,R$, with $Q'[N/x] = Q$, $R[N/z] = \calE'[P]$ and $R$ not blocked by $z$.  Hence, it follows from the induction hypothesis that there is $\calE''$ and $P'$ such that $R = \calE''[P']$, $\calE''[N/z] = \calE'$ and $P'[N/z] = P$.  Therefore, take witness $\calF = (\fixtm{f(x)}{Q'})\,\calE''$ and the result follows.}
    
    \item Suppose $\calE$ is of shape $\calE'\,Q$ and $M$ is not blocked by $z$. Then $M[N/z] = \calE'[P]\,Q$.  Then, since $M$ is not blocked by $z$, $M$ must have shape $R\,Q'$ with $R$ not blocked by $z$ and $R[N/z] = \calE'[P]$ and $Q'[N/z]= Q$.  Then it follows from the induction hypothesis that there is $P'$ and $\calE''$ with $P'[N/z] =P$ and $\calE''[N/z] = \calE'$.  Thus, take $\calF$ to be $\calE''\,Q'$ and the result follows.
  \end{itemize}
\end{proof}

\begin{lemma}\label{lem:one-step-inv-subs}
  If $M[N/z] \ped P$ then $M$ is blocked by $z$ or there is $M'$ such that $M \ped M'$ and $P=M'[N/z]$
\end{lemma}
\begin{proof}
  We prove the following by induction on $\calE$.
  \begin{quote}
    For all $\calE$, for all $M$, if $M[N/z]$ is of shape $\calE[P_1]$ and $P_1 \ped P_2$ and $M$ is not blocked by $z$, then there is some $M'$ such that $M \ped M'$ and $M'[N/z] = \calE[P_2]$.
  \end{quote}
  \begin{itemize}
    \item Suppose $\calE = \Box$.  Then assume $M[N/z]$ decomposes to $\calE[P_1]$ and $P_1 \ped P_2$ and $M$ is not blocked by $z$.  Then we analyse cases on the axiom $P_1 \ped P_2$.
      \begin{itemize}
        \item (PZ). Then $M[N/z]$ is of shape $\predtm(\pn{0})$ and $P_2 = \pn{0}$.  Since $M$ is not blocked by $z$, it must have shape $\predtm(\pn{0})$, so let $M'$ be $\pn{0}$.
        \item (PS). Then $M[N/z]$ is of shape $\predtm(\pn{k+1})$ and $P_2 = \pn{k}$.  Since $M$ is not blocked by $z$, it must have shape $\predtm(\pn{k+1})$, so let $M'$ be $\pn{k}$.
        \item (IfZ).  Then $M[N/z]$ is of shape $\iftm{\pn{0}}{Q_1}{Q_2}$ and $P_2 = Q_1$.  Since $M$ is not blocked by $z$, it must have shape $\iftm{\pn{0}}{Q_1'}{Q_2'}$ with $Q_i=Q_i'[N/z]$.  So, let $M'$ be $Q_1'$ and the result follows immediately.
        \item (IfS).  Then $M[N/z]$ is of shape $\iftm{\pn{n+1}}{Q_1}{Q_2}$ and $P_2=Q_2[\pn{n}/x]$.  Since $M$ is not blocked by $z$, $M$ must have shape $\iftm{\pn{n+1}}{Q_1'}{Q_2'}$ and $Q_i=Q_i'[N/z]$.  Let $M'$ be $Q_2'[\pn{n}/x]$.  Then $M'[N/z] = (Q_2'[\pn{n}/x])[N/z] = (Q_2'[N/z])[\pn{n}/x]$, as required.
        \item (Fix$\beta$).  \ch{Then $M[N/z]$ is of shape $(\fixtm{f(x)}{Q})\,V$ and $P_2 = Q[V/x,\fixtm{f(x)}{Q}/f]$.  Since $M$ is not blocked by $z$, it follows that $M$ is of shape $(\fixtm{f(x)}{Q'})\,V'$ and $Q=Q'[N/z]$ and $V=V'[N/z]$.  Let $M'$ be $Q'[V'/x,\fixtm{f(x)}{Q'}/f]$.  Then it follows that: 
        \[
          M'[N/z] = (Q'[N/z])[V'[N/z]/x,\fixtm{f(x)}{Q'[N/z]}/f]
        \] }
        \item (Let).  Then $M[N/z]$ is of shape $\lettm{(x,y)}{(V_1,V_2)}{Q}$ and $P_2=Q[V_1/x,V_2/y]$.  Since $M$ is not blocked by $z$, it follows that $M$ has shape $\lettm{(x,y)}{(V_1',V_2')}{Q'}$ with $V_i=V_i'[N/z]$ and $Q = Q'[N/z]$.  Then let $M'$ be $Q'[V_1'/x,V_2'/y]$ and the result follows immediately.
      \end{itemize}
    \item Suppose $\calE$ is of shape $\succtm(\calE')$ and $M[N/z]$ is of shape $\succtm(\calE'[P_1])$ and $P_1 \ped P_2$ and $M$ is not blocked by $z$.  Thus $M$ must have shape $\succtm(Q)$ and $Q$ is again not blocked by $z$, also $Q[N/z] = \calE'[P_1]$.  Hence, it follows from the induction hypothesis that there is some $Q'$ such that $Q \ped Q'$ and $Q'[N/z] = \calE'[P_2]$.  Hence, let $M'$ be $\succtm(Q')$.
    \item Suppose $\calE$ is of shape $\predtm(\calE')$, then we reason analogously to the previous case.
    \item Suppose $\calE$ is of shape $(\calE',Q_2)$ and $M[N/z]$ is of shape $(\calE'[P_1],Q_2)$ and $P_1 \ped P_2$ and $M$ is not blocked by $z$.  Thus $M$ has shape $(Q_1,Q_2')$ and $Q_1$ is not blocked by $z$ and $Q_1[N/z] = \calE'[P_1]$ and $Q_2'[N/z] = Q_2$.  It follows by induction that there is $Q_1'$ and $Q_1 \ped Q_1'$ and $Q_1'[N/z] = \calE'[P_2]$.  So, let $M'$ be $(Q_1',Q_2')$.
    \item Suppose $\calE$ is of shape $(V,\calE')$ and $M[N/z]$ is of shape $(V,\calE'[P_1])$ and $P_1 \ped P_2$ and $M$ is not blocked by $z$.  So, $M$ must be of shape $(V',Q)$ and $Q$ is again not blocked by $z$, $V'[N/z] = V$ and $Q[N/z] = \calE'[P_1]$.  Then it follows from the induction hypothesis that there is $Q'$ such that $Q \ped Q'$ and $\calE'[P_1]=Q'[N/z]$.  Hence, let $M'$ be $(V',\,Q')$ and the result is immediate.
    \item Suppose $\calE$ is of shape $\projtm_i(\calE')$ and $M[N/z]$ is of shape $\projtm_i(\calE'[P_1])$ and $P_1 \ped P_2$ and $M$ is not blocked by $z$. So, it must be that $M$ has shape $\projtm_i(Q)$ and $Q[N/z] = \calE'[P_1]$ and $Q$ is not blocked by $z$.  Hence, it follows from the induction hypothesis that there is $Q'$ and $Q \ped Q'$ and $Q'[N/z] = \calE'[P_2]$.  Thus, let $M'$ be $\projtm_i(Q')$.
    \item Suppose $\calE$ is of shape $\iftm{\calE'}{Q_1}{Q_2}$, $M[N/z] = \iftm{\calE'[P_1]}{Q_1}{Q_2}$ and $P_1 \ped P_2$ and $M$ is not blocked by $z$. So, it must be that $M$ has shape $\iftm{Q_0}{Q_1'}{Q_2'}$ and $Q_0[N/z] = \calE'[P_1]$, $Q_1'[N/z] = Q_1$, $Q_2'[N/z] = Q_2$.  Therefore, it follows from the induction hypothesis that there is $Q_0'$ such that $Q_0 \ped Q_0'$ and $Q_0'[N/z] = \calE'[P_2]$.  Thus, take $M'$ to be $\iftm{Q_0'}{Q_1'}{Q_2'}$ and the result follows.
    \item \ch{Suppose $\calE$ has shape $(\fixtm{f(x)}{Q})\,\calE'$ and $M[N/z] = (\fixtm{f(x)}{Q})\,\calE'[P_1]$ and $P_1 \ped P_2$ and $M$ is not blocked by $z$. So, $M$ must have shape $(\fixtm{f(x)}{Q'})\,R$ and $Q'[N/z] = Q$ and $R[N/z] = \calE'[P_1]$ and $R$ is not blocked by $z$.  Thus, it follows from the induction hypothesis that there is $R'$ such that $R \ped R'$ and $R'[N/z] = \calE'[P_2]$.  Then let $M'$ be $(\fixtm{f(x)}{Q'})\,R'$ and the result follows.}
    \item Suppose $\calE$ is of shape $\calE'\,Q$ and $M[N/z] = \calE'[P_1]\,Q$ and $P_1 \ped P_2$ and $M$ is not blocked by $z$.  Then $M$ must have shape $R\,Q'$ with $R$ not blocked by $z$ and $R[N/z] = \calE'[P_1]$ and $Q'[N/z]= Q$.  Then it follows from the induction hypothesis that there is $R'$ with $R \ped R'$ and $R'[N/z] = \calE'[P_2]$.  Thus, take $M'$ to be $R'\,Q'$ and the result follows.
  \end{itemize}
\end{proof}

\begin{lemma}\label{lem:peds-mod-sub}
  If $M[N/x] \peds P$ with $P$ a normal form and $Q \lesssim N$, then either $M[Q/x]$ diverges or there is some $M'$ such that $P=M'[N/x]$ and $M[Q/x] \peds M'[Q/x]$.
\end{lemma}
\begin{proof}
  By lexicographic induction on the length of the sequence and the number of free occurrences of $x$ in $M$.
  \begin{itemize}
    \item Suppose $M[N/x] = P$.  Then take $M':=M$.
    \item Suppose $M[N/x] \pedn{k} P$ for some $k>0$.  If $M$ is $\calE[x]$, then $M[Q/x] = (\calE[Q/x])[Q]$.  If $Q$ diverges then so does $M[Q/x]$.  Otherwise, $N$ and $Q$ normalise to the same normal form, say $N'$.  Hence, $M[N/x] \peds (\calE[N/x])[N']$ and $M[Q/x] \peds (\calE[Q/x])[N']$, then $(\calE[N/x])[N'] \pedn{m} P$ with $m \leq k$ and the number of free occurrences of $x$ in $\calE[N']$ is decreased by 1.  Hence, it follows from the induction hypothesis that $(\calE[Q/x])[N']$ either diverges or there is some $M'$ such that $P=M'[N/x]$ and $(\calE[Q/x])[N'] \peds M'[Q/x]$.  The result follows immediately. Otherwise, we may assume $M''$ is not blocked by $z$ and we can apply Lemma~\ref{lem:one-step-inv-subs} to obtain some $M''$ such that $M \ped M''$.  Then $M''[N/x] \pedn{k-1} P$, and it follows from the induction hypothesis that either $M''[P/x]$ diverges, or there is some $M'$ such that $P=M'[N/x]$ and $M''[Q/x] \peds M'[Q/x]$.  Then the result follows immediately since $M[Q/x] \ped M''[Q/x]$.
  \end{itemize}
\end{proof}


\begin{lemma}\label{lem:vals-mod-sub}
  If $M[N/z]$ reduces to a value $V$ and $P \lesssim N$, then either $M[P/z]$ diverges or $M[P/z]$ reduces to a value $W$ and all of the following are true: 
  \begin{enumerate}[(a)]
    \item if $V$ does not contain a \ch{fixpoint abstraction}, then $W=V$.
    \item if $V$ has shape $(V_1,V_2)$, then $W$ has shape $(W_1[P/z],W_2[P/z])$, $W_1[N/z] = V_1$ and $W_2[N/z] = V_2$.
    \item \ch{if $V$ has shape $\fixtm{f(x)}{P}$, then $W$ has shape $\fixtm{f(x)}{Q[P/z]}$ with $Q[N/z] = P$}.
  \end{enumerate}
\end{lemma}
\begin{proof}
  The proof is by induction on $V$.  Spp $M[N/z] \peds V$.  Then, by the previous result, $M[P/z]$ either diverges or there is some value $U$ such that $V = U[N/z]$ and $M[P/z] \peds U[P/z]$.  If $U = z$, then $N = V$ and $M[P/z] \peds P$.  If $P$ diverges, then $M[P/z]$ diverges, as required.  Otherwise $P \peds V$ and $W=V$, (a), (b) and (c) follow immediately. Otherwise $U \neq z$ and we proceed by inspecting the cases.
  \begin{itemize}
    \item If $V$ is a variable $x$, then since $U \neq z$, $x$ must be different from $z$ and $U=x$.
    \item If $V$ is $\zerotm$, then since $U \neq z$, $U = \zerotm$.
    \item If $V$ is of shape $\succtm(\pn{n})$, then since $U \neq z$, $U$ has shape $\succtm\,W'$, and $W'[N/z] = \pn{n}$.  Therefore, it follows from the induction hypothesis that $W'[P/z]$ reduces to $\pn{n}$ and the conclusion follows.
    \item If $V$ is of shape $(V_1,V_2)$, then since $U \neq z$, $U$ is of shape $(M_1,M_2)$ and $M_1[N/z] = V_1$ and $M_2[N/z] = V_2$.  Therefore, $M_1$ and $M_2$ are themselves values.  Moreover, if $V$ did not contain an abstraction, then neither do $V_1$ or $V_2$ and the induction hypothesis yields that $M_1[P/z]=V_1$ and $M_2[P/z]=V_2$.   Thus $W = V$.
    \item \ch{If $V$ is of shape $\fixtm{f(x)}{P}$, then since $U \neq z$, $U$ has shape $\fixtm{f(x)}{Q}$, with $Q[N/z] = P$ as required.}
  \end{itemize}
\end{proof}

\begin{lemma}\label{lem:stuck-mod-sub}
  If $M[N/z]$ gets stuck and $P \lesssim N$, then either $M[P/z]$ diverges or $M[P/z]$ gets stuck.
\end{lemma}
\begin{proof}
  Suppose $M[N/z]$ gets stuck.  Then $M[N/z] \peds P$ with $P$ stuck.  It follows that there is some $R$ such that $P=R[N/z]$ and $M[P/z]$ either diverges or reduces to $R[P/z]$.  Then $R[N/z]$ has shape $\calE[Q]$ for some stuckex $Q$.  By the forgoing lemma, $R$ has shape $\calF[Q']$ with $\calF[N/z] = \calE$ and $Q'[N/z] = Q$.  If $Q'=z$, then $Q=N$ and by $P \lesssim N$ either $R[P/z]$ diverges or $R[P/z] \peds (\calF[P/z])[N']$ for some stuckex $N'$, and thus also gets stuck.  Otherwise, $Q' \neq z$ and we show in each case that $Q'[P/z]$ is a stuckex.
  \begin{itemize}
    \item If $Q$ is of shape $x$, then since $Q' \neq z$, it must be that $x \neq z$ and $Q'=x$.
    \item If $Q$ is of shape $V\,R$ with $V \notin \FunVals$, then since $Q' \neq z$, it must be that $Q'$ has shape $R_1\,R_2$ with $R_1[N/z] = V$.  Hence, it follows from the previous lemma that either $R_1[P/z]$ diverges, in which case $M[P/z]$ diverges, or $R_1[P/z]$ is also a value not in $\FunVals$.
  \end{itemize}
  The remaining cases are similar.
\end{proof}

\subsection{Proof of Theorem~\ref{thm:types-are-safety}}

Now we show the first part of Theorem\ref{thm:types-are-safety}.

\begin{theorem}\label{thm:types-down-closed}
  Let $M[N/z]$ be closed, $Q \lesssim N$.
  \begin{enumerate}[(i)]
    \item If $M[N/z] \in \mngTb{A}$, then $M[Q/z] \in \mngTb{A}$.
    \item If $M[N/z] \notin \mngT{F}$, then $M[Q/z] \notin \mngT{F}$.
  \end{enumerate}
\end{theorem}
\begin{proof}
  The proof is by induction on $A$.
  \begin{itemize}
    \item When $A = \intty$, we reason as follows.  
      \begin{enumerate}[(i)]
        \item Suppose $M[N/z] \in \mngTb{\intty}$, so either $M[N/z]$ diverges or $M[N/z] \peds \pn{n}$.  Then it follows from Lemma~\ref{lem:vals-mod-sub} that either $M[Q/z]$ diverges or $M[Q/z]$ evaluates to $\pn{n}$.
        \item Suppose $M[N/z] \notin \mngT{\intty}$, so either (i) $M[N/z] \peds V$ and $V \notin \mng{\intty}$ or (ii) $M[N/z]$ gets stuck or (iii) $M[N/z]$ diverges.  In cases (i) it follows from Lemma~\ref{lem:vals-mod-sub} that $M[Q/z]$ either diverges or reduces to the same kind of value.  In case (ii), it follows from Lemma~\ref{lem:stuck-mod-sub} that $M[Q/z]$ either diverges or gets stuck.  In case (iii), $M[Q/z]$ also diverges.  Hence, $M[Q/z] \notin \mngT{\intty}$.
      \end{enumerate}
      \item When $A$ is of shape $B \times C$:
      \begin{enumerate}[(i)]
        \item Suppose $M[N/z] \in \mngTb{B \times C}$, so either $M[N/z]$ diverges, or $M[N/z] \peds (V_1,V_2)$ with $V_1 \in \mng{B}$ and $V_2 \in \mng{C}$.  In the former case, also $M[Q/z]$ diverges.  Otherwise, it follows from Lemma~\ref{lem:vals-mod-sub} that either $M[Q/z]$ diverges or evaluates to $(W_1[Q/z],W_2[Q/z])$ with $W_1[N/z] = V_1$ and $W_2[N/z] = V_2$.  Hence, we have $W_1[N/z] \in \mng{B}$ and thus also $\in \mngTb{B}$; similarly $W_2[N/x] \in \mngTb{C}$.  Consequently, it follows from the induction hypotheses that $W_1[Q/z] \in \mngTb{B}$ and $W_2[Q/z] \in \mngTb{C}$.  If either diverge, then $M[Q/z]$ diverges, as required.  Otherwise, $W_1[Q/z]$ evaluates to a value in $\mng{B}$ and $W_2[Q/z]$ evaluates to a value in $\mng{C}$, as required.
        \item Suppose $M[N/z] \notin \mngT{B \times C}$.  Here we can further assume that $B$ and $C$ are finitely verifiable.   So either (i) $M[N/z] \peds V$ with $V \notin \mng{B \times C}$ or (ii) $M[N/z]$ gets stuck, or (iii) $M[N/z]$ diverges.  In case (i), it follows from Lemma~$\ref{lem:vals-mod-sub}$ that either $M[Q/z]$ diverges or it evaluates to $V$.  In case (ii) it follows from Lemma~\ref{lem:stuck-mod-sub} that $M[Q/z]$ either diverges or gets stuck.  In case (iii) it follows that $M[Q/z]$ diverges.  Hence, $M[Q/z] \notin \mngT{B \times C}$.
      \end{enumerate}
        \item When $A$ is of shape $B \to C$, we need only prove (i).  Suppose $M[N/z] \in \mngTb{B \to C}$, so either: 
          \begin{enumerate}[(i)]
            \item $M[N/z]$ diverges 
            \item \ch{or $M[N/z] \peds \fixtm{f(x)}{P}$ with $P[V/x,\fixtm{f(x)}{P}/f] \in \mngTb{C}$ for all $V \in \mng{B}$.}
          \end{enumerate}
          In case (i), $M[Q/z]$ diverges.  \ch{In case (ii), it follows from Lemma~\ref{lem:vals-mod-sub} that either $M[Q/z]$ diverges or there is some $Q'$ such that $M[Q/z] \peds \fixtm{f(x)}{Q'[Q/z]}$ and $Q'[N/z] = P$.
          Let $V \in \mng{B}$.  We have $(Q'[V/x,\fixtm{f(x)}{Q'}/f])[N/z] = (Q'[N/z])[V/x,\fixtm{f(x)}{Q'[N/z]}] \in \mngTb{C}$ ($V$ is closed) and so it follows from the induction hypothesis that we also have $(Q'[V/x,\fixtm{f(x)}{Q'}/f])[Q/z] \in \mngTb{C}$ too.}
        \item When $A$ is of shape $B \from F$, we need only prove (i).  \ch{Suppose $M[N/z] \in \mngTb{B \from F}$, so either $M[N/z]$ diverges or $M[N/z] \peds \fixtm{f(x)}{P}$ with $P[V/x,\fixtm{f(x)}{P}/f] \notin \mngT{F}$ if $V \notin \mng{B}$.  In the former case, $M[Q/z]$ diverges.  Otherwise, it follows from Lemma~\ref{lem:vals-mod-sub} that either $M[Q/z]$ diverges or there is $Q'$ such that $M[Q/z] \peds \fixtm{f(x)}{Q'[Q/z]}$ and $Q'[N/z] = P$.  Let $V$ be a closed value and $V \notin \mng{B}$.  Hence, as in the previous case, $(Q'[V/x,\fixtm{f(x)}{Q'}/f])[N/z] = (Q'[N/z])[V/x,\fixtm{f(x)}{Q'[N/z]}] \notin \mngT{F}$ and it follows from the induction hypothesis that therefore $(Q'[Q/z])[V/x] \notin \mngT{F}$ either.  }
        \item When $A$ is of shape $\okty$ we reason as follows.  For (i) suppose $M[N/z] \in \mngTb{\okty}$, then we reason as above, according to whichever value $M[N/z]$ reduces to.  For (ii) suppose $M[N/z] \notin \mngT{\okty}$, then either $M[N/z]$ gets stuck or $M[N/z]$ diverges.  It follows from Lemma~\ref{lem:stuck-mod-sub} that, in both cases, $M[N'/z] \notin \mng{T}$ too.
  \end{itemize}
\end{proof}

And then the second part of Theorem~\ref{thm:types-are-safety}.

\ch{
\begin{theorem}\label{thm:types-finitely-refutable}
  For all types $A$ and finitely verifiable types $F$:
  \begin{enumerate}[(i)]
    \item If $C[\fixtm{f(x)}{M}] \notin \mngTb{A}$, then there is some $n_0$, such that for all $n \geq n_0$: $C[\fixapprx{n}{f(x)}{M}] \notin \mngTb{A}$.
    \item If $C[\fixtm{f(x)}{M}] \in \mngT{F}$, then there is some $n_0$, such that for all $n \geq n_0$: $C[\fixapprx{n}{f(x)}{M}] \in \mngT{F}$.
  \end{enumerate}
\end{theorem}
\begin{proof}
  The proof is by induction on $A$.
  \begin{itemize}
    \item When $A$ is $\intty$, we reason as follows.
    \begin{enumerate}[(i)]
      \item Suppose $C[\fixtm{f(x)}{M}] \notin \mngTb{A}$.  Either $C[\fixtm{f(x)}{M}]$ evaluates to a value that is not a numeral or $C[\fixtm{f(x)}{M}]$ crashes.  In both cases, it requires only finitely many reduction steps and, among them, only finitely many applications of fix$\beta$ reduction, say $k$.  Then we can take $n_0 := k$.
      \item Suppose $C[\fixtm{f(x)}{M}] \in \mngT{A}$.  Then $C[\fixtm{f(x)}{M}]$ evaluates to a numeral, in finitely many steps.  Among them, a certain number, say $k$, are applications of the fix$\beta$ reduction.  Therefore, for any $m \geq k$, $C[\fixapprx{m}{f(x)}{M}] \in \mngT{A}$.
    \end{enumerate} 
    \item When $A$ is $\okty$, we reason analogously.
    \item When $A$ is of shape $B \times C$ we reason as follows.  
    \begin{enumerate}[(i)]
      \item Suppose $C[\fixtm{f(x)}{M}] \notin \mngTb{B \times C}$.  Then either (i) $C[\fixtm{f(x)}{M}]$ evaluates to a value that is not a pair, or (ii) $C[\fixtm{f(x)}{M}]$ gets stuck, or (iii) $C[\fixtm{f(x)}{M}]$ evaluates to a pair value $(V_1,V_2)$, but either $V_1 \notin \mng{B}$ or $V_2 \notin \mng{C}$.  In all cases, the reduction sequences involved are finite, and so contain only finitely many applications of the fix$\beta$ reduction.  Hence, we may obtain the required bound as above. In cases (i) and (ii) the result follows immediately.  In case (iii), first observe that this evaluation used only finitely many applications of the fix$\beta$ reduction, say $k_0$, and so for any $k \geq k_0$, $C[\fixapprx{k}{f(x)}{M}]$ will reduce to a pair value $(W_1,W_2)$.  Now let us assume that it is $V_1 \notin \mng{B}$ because the other case is symmetrical.  Note that, since $V_1$ is a value, we know that $V_1 \notin \mngTb{B}$.  Next, write $V_1$ as $D[\fixtm{f(x)}{M}]$, for some context $D$, in which all descendants of occurrences of $\fixtm{f(x)}{M}$ in $C[\fixtm{f(x)}{M}]$ are replaced by holes in $D$ (for example, by introducing labelled subterms).  Then it follows from the induction hypothesis that there is $m_0$ and for all $m \geq m_0$, $D[\fixapprx{m}{f(x)}{M}] \notin \mngTb{B}$.  Since each $D[\fixapprx{m}{f(x)}{M}]$ is a value, it follows that $D[\fixapprx{m}{x}{M}] \notin \mng{B}$.  Hence, for all $n \geq m_0+k_0$, $C[\fixapprx{n}{f(x)}{M}] \notin \mngTb{B \times C}$.
      \item In this case, we may assume that $B$ and $C$ are finitely verifiable.  Suppose $C[\fixtm{f(x)}{M}] \in \mngT{B \times C}$.  Then $C[\fixtm{f(x)}{M}]$ evaluates to a pair $(V_1,V_2)$ (and this pair must not contain any abstraction, since $B$ and $C$ are finitely verifiable).  Then this reduction sequence is finite, and applies the fix$\beta$ reduction rule only a certain number, say $n_0$ of times.  We can use this as the bound required by the conclusion.
    \end{enumerate}
    \item When $A$ is $B \to C$, we need only consider (i).  So suppose $C[\fixtm{f(x)}{M}] \notin \mngTb{B \to C}$.  Then either (i) $C[\fixtm{f(x)}{M}]$ evaluates to a value that is not a fixpoint abstraction, or (ii) $C[\fixtm{f(x)}{M}]$ gets stuck, or (iii) $C[\fixtm{f(x)}{M}]$ evaluates to a fixpoint abstraction $\fixtm{g(y)}{P}$, but there is a value $V \in \mng{B}$ such that $P[V/y,\fixtm{g(y)}{P}/g] \notin \mngTb{C}$.  In cases (i) and (ii) we can reason as above to obtain the bound.  In case (iii), we can first observe that evaluating to an abstraction requires only finitely many applications of the fix$\beta$ reduction rule, say $m_0$.  So, for any $m \geq m_0$, $C[\fixapprx{m}{f(x)}{M}]$ will also evaluate to an abstraction, say $\fixtm{g(y)}{Q}$.  Then we note that we can write $P[V/y,\fixtm{g(y)}{P}/g]$ as $D[\fixtm{f(x)}{M}]$, in which all descendants of occurrences of $\fixtm{f(x)}{M}$ in $C[\fixtm{f(x)}{M}]$ are replaced by holes in $D$ (for example, by introducing labelled subterms).  Then it follows from the induction hypothesis that there is some $p_0$ such that, for all $p \geq p_0$, $D[\fixapprx{p}{f(x)}{M}] \notin \mngTb{C}$.  Hence, by construction, we have that, for all $n \geq m_0+p_0$, $C[\fixapprx{n}{f(x)}{M}] \notin \mngTb{B \to C}$.
    \item When $A$ is $B \from F$, we need only consider (i).  So suppose $C[\fixtm{f(x)}{M}] \notin \mngTb{B \from F}$.  Then either (i) $C[\fixtm{f(x)}{M}]$ evaluates to a value that is not a fixpoint abstraction, or (ii) $C[\fixtm{f(x)}{M}]$ gets stuck, or (iii) $C[\fixtm{f(x)}{M}]$ evaluates to a fixpoint abstraction $\fixtm{g(y)}{P}$, but there is a value $V$ such that $V \notin \mng{B}$ and yet $P[V/y,\fixtm{g(y)}{P}/g] \in \mngT{F}$.  In cases (i) and (ii) we can find the bound as above.  In case (iii), we first observe that evaluating to an abstraction requires only finitely many applications of the fixpoint reduction, say $m_0$.  So, for any $m \geq m_0$,  $C[\fixapprx{m}{f(x)}{M}]$ will also evaluate to an fixpoint abstraction.  Next, note that we can write $P[V/y,\fixtm{g(y)}{P}/g]$ as $D[\fixtm{f(x)}{M}]$ by factoring out all descendants of $\fixtm{f(x)}{M}$ in the original term.  Then, it follows from the induction hypothesis that there is some $k_0$ such that, for all $k \geq k_0$, $D[\fixapprx{k}{f(x)}{M}] \in \mngT{F}$.  Then, it follows from this construction that, for all $n \geq m_0+k_0$, $C[\fixapprx{n}{f(x)}{M}] \notin \mngTb{B \from F}$. 
  \end{itemize}
\end{proof}
}

Note, the contrapositive of part (i) of this result gives us the following principle: if $C[\fixapprx{n}{f(x)}{M}] \in \mngTb{A}$ for all $n$, then $C[\fixtm{f(x)}{M}] \in \mngTb{A}$.

\paragraph{Theorem~\ref{thm:types-are-safety}} follows from Theorems~\ref{thm:types-down-closed} and \ref{thm:types-finitely-refutable} by observing that, since $N$ and $P$ are closed in the statement, we can always find some fresh variable $z$ in order to write $C[N]$ as $(C[z])[N/z]$.\\

\begin{theorem}[Fixpoint Closure.]\label{thm:types-fix-closed}
  \ch{In the safety fragment: 
  \begin{itemize}
    \item if $\abs{fx}{M} \in \mng{(A \to B) \to A \to B}$, then $\fixtm{f(x)}{M} \in \mngTb{A \to B}$.
    \item if $\abs{fx}{M} \in \mng{(A \from B) \to A \from B}$, then $\fixtm{f(x)}{M} \in \mngTb{A \from B}$.
  \end{itemize}
  }
\end{theorem}

\ch{
\begin{proof}
  Suppose $\abs{fx}{M} \in \mng{(A \to B) \to A \to B}$.  Then, by definition, for all $V \in \mng{A \to B}$, $\abs{x}{M[V/f]} \in \mngTb{A \to B}$ (*).  
  We show that $\fixtm{f(x)}{M} \in \mngTb{A \to B}$ using the above principle, by induction on $n \in \mathbb{N}$.
  \begin{itemize}
    \item When $n=0$, $\fixapprx{n}{f(x)}{M} = \abs{x}{\divtm}$ and $\abs{x}{\divtm} \in \mngTb{A \to B}$ by definition.
    \item When $n$ is of shape $k+1$, we assume $\fixapprx{k}{x}{M} \in \mngTb{A \to B}$.  Since this is a value, it follows from (*) that $\abs{x}{M[\fixapprx{k}{x}{M}/x]} \in \mngTb{A \to B}$.  By definition, therefore $\fixapprx{k+1}{x}{M} \in \mngTb{A \to B}$, as required.
  \end{itemize}
  The case when $\abs{fx}{M} \in \mng{(A \to B) \to A \to B}$ is analogous, since also $\abs{x}{
  \divtm} \in \mngTb{A \from B}$.
\end{proof}
}

\subsection{Proof of Theorem~\ref{thm:strong-soundness}}

Finally, we can show semantic soundness, Theorem~\ref{thm:strong-soundness}.

\begin{proof}
  The proof is by induction on the derivation.  
  \begin{description}
    \item[\rlnm{Id}] Suppose $\theta$ satisfies $\Gamma,\,x:A$ on the left, then $x\theta \in \mngT{A}$.  Hence, $x\theta \in \mngTb{A}$ and the conclusion follows.
    \item[\rlnm{Dis}] Suppose $A \distype B$ and that $\theta$ satisfies $\Gamma,\,M:A$ on the left, so that $M\theta \in \mngT{A}$.  It follows that $M\theta \notin \mngTb{B}$ (*).   Then it follows from the induction hypothesis that $\theta$ satisfies $M:B,\,\Delta$ on the right, but by (*), $\theta$ does not satisfy $M:B$ on the right.  Therefore, $\theta$ must satisfy $\Delta$ on the right, as required.
    \item[\rlnm{ZeroR}] By definition $\zerotm\theta = \zerotm \in \mngTb{\intty}$.
    \item[\rlnm{SuccR}] Suppose $\theta$ satisfies $\Gamma$ on the left, then it follows from the induction hypothesis that $\theta$ satisfies $M:\intty\,\Delta$ on the right.  Suppose $\theta$ does not satisfy $\Delta$ on the right.  Then $\theta$ satisfies $M:\intty$ on the right, so either $M\theta$ diverges or $M\theta$ evaluates to a number.  Hence $\theta$ satisfies $\succtm(M):\intty$ on the right.
    \item[\rlnm{PredR}] Analogous to the previous case.
    \item[\rlnm{FixsR}] \ch{Suppose $\Gamma,\,f:A \to B,\,x:A \models M : A \to B,\,\Delta$ and let $\theta$ be a valuation satisfying $\Gamma$ on the left.  To see that $\theta$ will satisfy $\fixtm{f(x)}{M}:A \to B,\,\Delta$ on the right we suppose $\theta$ does not satisfy $\Delta$ on the right, and show that $\theta$ satisfies $\fixtm{f(x)}{M}:A \to B$ on the right, i.e. $\fixtm{f(x)}{M\theta} \in \mngTb{A \to B}$.  By this assumption and our original supposition, we have that, for all $V \in \mng{A \to B}$, for all $W \in \mng{A}$, $(M\theta)[V/f,W/x] \in \mngTb{A \to B}$, so $\abs{fx}{M\theta} \in \mng{(A \to B) \to A \to B}$.  Then the result follows since typing is closed under taking fixpoints.}
    \item[\rlnm{FixnR}] \ch{Suppose $\Gamma,\,f:A \from B,\,x:A \models M : A \from B,\,\Delta$ and let $\theta$ be a valuation satisfying $\Gamma$ on the left.  To see that $\theta$ will satisfy $\fixtm{f(x)}{M}:A \from B,\,\Delta$ on the right we suppose $\theta$ does not satisfy $\Delta$ on the right, and show that $\theta$ satisfies $\fixtm{f(x)}{M}:A \from B$ on the right, i.e. $\fixtm{f(x)}{M\theta} \in \mngTb{A \from B}$.  By this assumption and our original supposition, we have that, for all $V \in \mng{A \from B}$, for all $W \in \mng{A}$, $(M\theta)[V/f,W/x] \in \mngTb{A \from B}$, so $\abs{fx}{M\theta} \in \mng{(A \from B) \to A \from B}$.  Then the result follows since typing is closed under taking fixpoints.}
    \item[\rlnm{LetR}] Suppose $\theta$ satisfies $\Gamma$ on the left, then it follows from the induction hypotheses that $\theta$ satisfies $M:B \times C,\,\Delta$ on the right.  Suppose $\theta$ does not satisfy $\Delta$.  Hence, $\theta$ satisfies $M:B \times C$ on the right so $M\theta$ either diverges or evaluates to a pair $(V,W)$ with $V \in \mng{B}$ and $W \in \mng{C}$.  In the former case, the let expression diverges under $\theta$ and we are done.  Otherwise $\theta \cup [V/x,W/y]$ satisfies $\Gamma,\,x:B,\,x:C$ on the left and it follows from the induction hypothesis that $\theta$ satisfies $N:A$ on the right, so $(N[V/x,W/y])\theta \in \mngTb{A}$.  The conclusion follows since $(\lettm{(x,y)}{M}{N})\theta \peds (N[V/x,W/y])\theta$.
    \item[\rlnm{AppR}] Suppose $\theta$ satisfies $\Gamma$ on the left.  Then it follows from the induction hypotheses that $\theta$ satisfies $M:B \to A,\,\Delta$ and $N:B,\,\Delta$ on the right.  Suppose $\theta$ does not satisfy $\Delta$.  Then $M\theta \in \mngTb{B \to A}$ and $N\theta \in \mngTb{B}$.  If either diverges, then so does $(M\,N)\theta$ and we are done.  \ch{Otherwise, $M\theta$ evaluates to some abstraction $\fixtm{f(x)}{P} \in \mng{B \to A}$ and $N\theta$ evaluates to a value $V$ in $\mng{B}$ and hence it follows from the definition that $(M\,N)\theta \peds (\fixtm{x}{P})\,V \ped P[V/x,\fixtm{f(x)}{P}/f] \in \mngTb{A}$, as required.}
    \item[\rlnm{PairR}] Suppose $\theta$ satisfies $\Gamma$ on the left, then it follows from the induction hypotheses that $\theta$ satisfies $M:A,\,\Delta$ and $N:B,\,\Delta$ on the right.  Suppose $\theta$ does not satisfy $\Delta$ on the right, then $M\theta \in \mngTb{A}$ and $N\theta \in \mngTb{B}$.  If either diverge, then so does $(M,N)\theta$ and the conclusion is immediate.  Otherwise, $M\theta$ evaluates to some $V \in \mng{A}$ and $N\theta$ evaluates to some $W \in \mng{B}$ and thus $(M,N)\theta$ evaluates to $(V,W) \in \mng{A \times B}$.
    \item[\rlnm{IfZR}] Suppose $\theta$ satisfies $\Gamma$ on the left.  Then it follows from the induction hypothesis that $\theta$ satisfies (i) $M:\intty,\,\Delta$, (ii) $N:A,\,\Delta$ and (iii) $P:A,\,\Delta$ on the right.  Suppose that $\theta$ does not satisfy $\Delta$.  Then $M\theta$ either diverges or evaluates to a numeral, say $\pn{n}$.  In the former case, the if expression diverges too and we conclude.  Otherwise, if $n=0$, it follows that $\iftm{M\theta}{N\theta}{P\theta} \peds N\theta$ and the result follows from (ii).  If $n>0$, then $\iftm{M\theta}{N\theta}{P\theta} \peds P\theta$ and the result follows from (iii).
    \item[\rlnm{AppL}] Suppose $\theta$ satisfies $\Gamma,\,(M\,N):A$ on the left.  Then $(M\,N)\theta \in \mngT{A}$.  It follows from the induction hypothesis that $\theta$ satisfies $M:B\from A,\,\Delta$ on the right.  If $\theta$ satisfies $\Delta$ on the right, then we are done.  If $\theta$ satisfies $M:B\from A$ on the right, then $M\theta \in \mngTb{B \from A}$.  \ch{Since $(M\,N)\theta$ is guaranteed to terminate, it follows that $M\theta$ and $N\theta$ are too, and so $M\theta$ evaluates to some fixpoint abstraction $\abs{f(x)}{P} \in \mng{B \from A}$ and $N\theta$ to some value $V$.  It follows from the definition that since, therefore, $(M\,N)\theta \peds P[V/x,\fixtm{f(x)}{P}/f] \in \mngT{A}$, $V \in \mng{B}$ and hence $N\theta \in \mngT{B}$.  Then $\theta$ satisfies $\Gamma,\,N:B$ on the left and the result follows from the induction hypothesis.}
    \item[\rlnm{LetL1}] Suppose $\theta$ satisfies $\Gamma,\,\lettm{(x,y)}{P}{N}:A$ on the left, so $\lettm{(x,y)}{P\theta}{N\theta} \in \mngT{A}$.  Therefore, it must be that $P\theta$ evaluates to a pair $(V,W)$ and $(N\theta)[V/x,W/y] \in \mngT{A}$.  Then $\theta \cup [V/x,W/y]$ satisfies $\Gamma,\,N:A$ on the left, and it follows from the induction hypothesis that it also satisfies $\Delta$ on the right.  Since $x$ and $y$ are bound in the let, we may assume they do not occur free in $\Delta$, so also $\theta$ satisfies $\Delta$ on the right.
    \item[\rlnm{LetL2}] Suppose $\theta$ satisfies $\Gamma,\,\lettm{(x,y)}{P}{N}:A$ on the left, so $\lettm{(x,y)}{P\theta}{N\theta} \in \mngT{A}$.  It must be that $P\theta$ evaluates to some pair $(V,W)$ and $(N\theta)[V/x,W/y] \in \mngT{A}$.  Hence, $\theta \cup [V/x,W/y]$ satisfies $\Gamma,\,N:A$ on the left and it follows from the induction hypotheses that it also satisfies $x:B,\,\Delta$ and $y:C,\,\Delta$ on the right.  If it satisfies $\Delta$ on the right, then $\theta$ satisfies $\Delta$ on the right since we mat assume the bound variables $x$ and $y$ do not occur in $\Delta$.  Otherwise, $V \in \mng{B}$ and $W \in \mng{C}$ and so we have that $P\theta \peds (V,w) \in \mng{B \times C}$.  Therefore, it follows from the induction hypothesis that $\theta$ satisfies $\Delta$ on the right.
    \item[\rlnm{SuccL}] Suppose $\theta$ satisfies $\Gamma,\,\succtm(P):\intty$ on the left, so $\succtm(P\theta) \in \mngT{\intty}$ and hence $P\theta$ must evaluate to a number.  Then the conclusion follows from the induction hypothesis.
    \item[\rlnm{PredL}] Analogous to the previous case.
    \item[\rlnm{IfZL1}] Suppose $\theta$ satisfies $\Gamma,\,\iftm{Q}{N}{P}:A$, so that $\iftm{Q\theta}{N\theta}{P\theta}$ evaluates to an $A$.  Then it must be that $Q\theta$ evaluates to a numeral, and then the result follows from the induction hypothesis.
    \item[\rlnm{IfZL2}] Suppose $\theta$ satisfies $\Gamma,\,\iftm{Q}{N}{P}:A$, so that $\iftm{Q\theta}{N\theta}{P\theta}$ evaluates to an $A$.  Then it must be that $Q\theta$ evaluates to a numeral, say $\pn{n}$.  If $n=0$, then the if expression reduces to $N\theta \in \mngT{A}$ and the result follows from the first induction hypothesis.  Otherwise, the if expression reduces to $P\theta$ and the result follows from the second induction hypothesis.
    \item[\rlnm{OkVarR}] It is immediate from the definition that $x\theta$ is a value.
    \item[\rlnm{OkL}] Suppose $\theta$ satisfies $\Gamma,\,M:A$ on the left.  Then $M\theta \in \mngT{A}$.  Since $\mngT{A} \subseteq \mngT{\okty}$, the conclusion follows from the induction hypothesis.
    \item[\rlnm{OkR}] Symmetrical to the previous case.
    \item[\rlnm{OkSL}] Suppose $\theta$ satisfies $\Gamma,\,\succtm(M):\okty$ on the left.  So $\succtm(M\theta)$ evaluates.  The only values headed by $\succtm$ are numerals, so it follows that $M\theta$ evaluates to a numeral and then the conclusion follows from the induction hypothesis.
    \item[\rlnm{OkPL}] Suppose $\theta$ satisfies $\Gamma,\,\predtm(M):\okty$ on the left.  So $\predtm(M\theta)$ evaluates.  Since $\predtm(M\theta)$ can only reduce to a numeral, and this requires that $M\theta$ evaluate to a numeral, the conclusion follows from the induction hypothesis.
    \item[\rlnm{OkPrL}] Suppose $\theta$ satisfies $\Gamma,\,(M_1,M_2) : \okty$.  Then $(M_1\theta,M_2\theta)$ evaluates and it must be that it evaluates to a pair of values $(V_1,V_2)$.  Then it follows that each $M_i$ evaluates and so the conclusion follows from the induction hypothesis.
    \item[\rlnm{OkApL1}] Suppose $\theta$ satisfies $\Gamma,\,M\,N:\okty$ on the left, so $M\theta\,N\theta$ evaluates.  Then it must be that $M\theta$ evaluates to a fixpoint abstraction.  Since, for any $A$, $\mng{\okty \from A}$ is the set of all such, the conclusion follows from the induction hypothesis.
    \item[\rlnm{OkApL2}] Suppose $\theta$ satisfies $\Gamma,\,M\,N : \okty$, so $M\theta\,N\theta$ evaluates.  Then, in particular, $N\theta$ evaluates, and the conclusion follows from the induction hypothesis. 
  \end{description}
\end{proof}

\subsection{Proof of Theorem~\ref{thm:safety-cexs}}

\begin{proof}
  We prove each separately.
\begin{itemize}
  \setlength\itemsep{1mm}
  \item The term $\abs{xy}{(\abs{z}{\predtm(z)})\,x} \in \mngTb{\intty \from \intty \to \intty}$ because, reasoning according to the contrapositive of the definition, if a given value $V$ is not a numeral, then it follows that $\abs{y}{(\abs{z}{\predtm(z)})\,V} \notin \mngT{\intty \to \intty}$, since the body will get stuck on any argument.  However, replacing $\abs{z}{\predtm(z)}$ by $\divtm$, we have that $\abs{xy}{\divtm\,x} \notin \mngTb{\intty \from \intty \to \intty}$. Even when a given value $V$ is not a numeral, we yet have $\abs{y}{\divtm\,V} \in \mngT{\intty \to \intty}$, since the body diverges on any input.
  \item First observe that \ch{$\abs{wxyz}{(\fixtm{f(u)}{\iftm{u}{\abbv{add}\,(y,z)}{f\,(\predtm(u))}})\,x}$} is not a value of the type $\intty \from \intty \to \intty \from \intty \to \intty$.  To see why, let us abbreviate the conditional subterm by $M$ and show that when applied to a value $\abbv{id}$ that is not a numeral, it may yet behave like a function of type $\intty \to \intty \from \intty \to \intty$.  To see that $\abs{xyz}{(\fixtm{f(u)}{M})\,x}$ is a value in $\mng{\intty \to \intty \from \intty \to \intty}$, we let $n \in \mathbb{N}$ and show that $\abs{yz}{(\fixtm{f(u)}{M})\,\pn{n}}$ is a value in $\mng{\intty \from \intty \to \intty}$.  To see this, we use the contrapositive of the definition and show that when applied to any non-numeral value, the function returned does not satisfy $\intty \to \intty$.  So let $V \notin \mng{\intty}$ be a non-numeral value.  Then the function $\abs{z}{(\fixtm{f(u)}{\iftm{u}{\abbv{add}(V,z)}{f\,(\predtm(u))}})\,\pn{n}}$ is not in $\mng{\intty \to \intty}$ because, when applied to $\pn{0}$ it will reduce to $\abbv{add}(V,\pn{0})$, but $V$ is not a numeral.\\[1mm]
  Next, observe that any term $\abs{wxyz}{(\fixapprx{k}{f(u)}{\iftm{u}{\abbv{add}\,(y,z)}{f\,(\predtm(u))}})\,x}$, in which we have replaced the fixpoint by a particular finite approximant, \emph{is} a value of the type $\intty \from \intty \to \intty \from \intty \to \intty$.  To see this, let us abbreviate the conditional subterm again by $M$ and let $W \notin \mng{\intty}$.  We show that, therefore, $\abs{xyz}{(\fixapprx{k}{f(u)}{M})\,x}$ is not in $\mng{\intty \to \intty \from \intty \to \intty}$.  Take $x \coloneqq \pn{k+1}$ as witness.  Then $\abs{yz}{(\fixapprx{k}{f(u)}{M})\,\pn{k+1}}$ is not in $\mng{\intty \from \intty \to \intty}$.  To see why, take $y \coloneqq \abbv{id} \notin \mng{\intty}$ as the witness, and we find that yet $\abs{z}{(\fixapprx{k}{f(u)}{\iftm{u}{\abbv{add}\,(\abbv{id},z)}{f\,(\predtm(u))}})\,\pn{k+1}}$ is actually in the type $\intty \to \intty$, because when given any numeral as input, it will diverge.
\end{itemize}
\end{proof}

%% file: apx-complements.tex
\section{Additional Material in Support of Section~\ref{sec:complements}}

In this appendix, we give the proofs of Theorem~\ref{thm:two-sided-subsumed}, showing that two-sided judgements of the specified form are subsumed by one-sided judgements, and Theorem~\ref{thm:failures-soundness}, showing syntactic soundness of the one-sided system.  The latter requires first establishing a kind of inversion lemma for values, that shows that they can only be given the types one would expect, and then substitution lemmas, and preservation.

\subsection{Proof of Theorem~\ref{thm:two-sided-subsumed}}

\begin{proof}
  The proof is by induction on $\Gamma,\,M:A \types \Delta$.
  \begin{description}
    \item[\rlnm{Id}] In this case, there are two cases. If $M$ is some variable $x$ and $x:A \in \Delta$, then we may conclude $\Gamma\cup\Delta^c \types x:A^c$ by \rlnm{Id}.  Otherwise, $M$ is not a variable, but there is a variable typing $x:A$ in both $\Gamma$ and $\Delta$.  In that case, $\Gamma\cup\Delta^c$ will contain a contradiction, and the result follows by \rlnm{Contra}.
    \item[\rlnm{LetL1}] In this case $M$ is of shape $\lettm{(x,y)}{N}{P}$ and we may assume (i) $\Gamma\cup\Delta^c \types P:A^c$.  Then the result follows from \rlnm{Let3} with $A$ set to $A^c$
  \item[\rlnm{LetR}] In this case $M$ is $\lettm{(x_1,x_2)}{N}{P}$ and we can assume (i) $\Gamma\cup\Delta^c \types N : B \times C$ and (ii) $\Gamma\cup\Delta^c,\,x:B,\,y:C \types P:A$.  Then the result follows immediately from \rlnm{Let1}.
  \item[\rlnm{LetL2}] In this case $M$ is of shape $\lettm{(x,y)}{N}{P}$ and we can assume (i) $\Gamma,\,\Delta^c\,x:B^c \types P:A^c$, (ii) $\Gamma\cup\Delta^c,\,y:C^c \types P:A^c$ and (iii) $\Gamma\cup\Delta^c \types M : (B \times C)^c$.  The conclusion follows immediately by instantiating $A$ in \rlnm{Let2} by $A^c$.
  \item[\rlnm{IfZR}] In this case, $M$ is of shape $\iftm{N}{P_1}{P_2}$ and we may assume (i) $\Gamma\cup\Delta^c \types N:\intty$, (ii) $\Gamma\cup\Delta^c \types P_1:A$, (iii) $\Gamma\cup\Delta^c,\,x:\intty \types P_2:A$.  Then the desired conclusion follows immediately from (ii) and (iii) by \rlnm{IfZ1}.
  \item[\rlnm{IfZL1}] In this case, $M$ is of shape $\iftm{N}{P_1}{P_2}$ and we may assume (i) $\Gamma\cup\Delta^c \types N : \intty^c$.  Then the result follows immediately by \rlnm{IfZ2}.
  \item[\rlnm{IfZL2}] In this case, $M$ is of shape $\iftm{N}{P_1}{P_2}$ and we may assume (i) $\Gamma\cup\Delta^c \types P_1 : A^c$ and (ii) $\Gamma\cup\Delta^c,\,x:\intty \types P_2 : A^c$.  Then the result follows immediately by \rlnm{IfZ1}.
  \item[\rlnm{CompL}] In this case, we may assume $\Gamma\cup\Delta^c \types M:A$ and then the result is immediate.
  \item[\rlnm{CompR}] In this case, we may assume $\Gamma\cup\Delta^c \types M:A^c$ and again, the result is immediate.
  \item[\rlnm{Disj}] In this case, we may assume $\Gamma\cup\Delta^c \types M:B$ and $A \distype B$.  Then the result follows from \rlnm{Dis}.
  \item[\rlnm{AppL}] In this case, $M$ is of shape $P\,N$ and we may assume (i) $\Gamma\cup\Delta^c \types P : B \from A$ and (ii) $\Gamma\cup\Delta^c \types N:B^c$.  Recall that $B \from A$ is just an abbreviation for $B^c \to A^c$ in the one-sided system.  Hence, the result follows immediately from \rlnm{App}.
  \item[\rlnm{OkApL1}] In this case, $M$ is of shape $P\,N$ and we may assume (i) $\Gamma\cup\Delta^c \types P : (\okty \from A)^c = (\okty^c \to A^c)^c$.  Therefore, the result follows from \rlnm{App2}.
  \item[\rlnm{OkApL2}] In this case, $M$ is of shape $P\,N$ and we may assume (i) $\Gamma\cup\Delta^c \types N : \okty^c$.  Then the desired conclusion follows immediately by \rlnm{App3}. 
  \item[\rlnm{OkL}] In this case we may assume $\Gamma\cup\Delta^c \types M : \okty$.  Hence, the result follows from \rlnm{OkC2}.
  \item[\rlnm{OkR}] In this case the result follows immediately from \rlnm{Ok}.
  \item[\rlnm{OkSL},\rlnm{SuccL}] In these cases, $M$ is of shape $\succtm(N)$ and we may assume $\Gamma\cup\Delta^c \types N:\intty^c$.  Therefore, the results follow from \rlnm{Succ2}.
  \item[\rlnm{OkPrL}] In this case, $M$ is of shape $(N_1,N_2)$ and we may assume $\Gamma\cup\Delta^c \types M_i : \okty^c$.  Then the result follows immediately by \rlnm{Pair2}.
  \item[\rlnm{PairL}] In this case, $M$ is of shape $(N_1,N_2)$, $A$ is of shape $A_1 \times A_2$.  Let us consider the $i=1$ case, the other is symmetrical.  We may assume that $\Gamma\cup\Delta^c \types N_1 : A_1^c$.  Then the result follows from \rlnm{Pair3}.
  \item[\rlnm{FixsR}] The result follows from the induction hypothesis and \rlnm{Fix}.
  \item[\rlnm{FixnR}] In this case, $M$ is of shape $\abs{x}{P}$ and $A$ is of shape $B \from A$.  We may assume $\Gamma\cup\Delta^c,\,x:B^c \types M:A^c$.  Since, in the one-sided system, $B \from A$ is an abbreviation for $B^c \to A^c$, the result follows from \rlnm{Fix}.
  \item[\rlnm{AppR}] The result follows from the induction hypotheses and \rlnm{App1}.
  \item[\rlnm{SuccR}] The result follows from the induction hypotheses and \rlnm{Succ1}.
  \item[\rlnm{ZeroR}] The result follows from the induction hypotheses and \rlnm{Zero}.     
  \end{description}
\end{proof}

\subsection{Proof of Soundness}

\begin{lemma}[Value Inversion]
  Suppose $\types V : A$.
  \begin{itemize}
    \item If $V$ is of shape $\zerotm$ then $A = \intty$ or $A$ is of shape $B^c$ with $B \distype \intty$
    \item If $V$ is of shape $\succtm(W)$ then $A = \intty$ or $A$ is of shape $B^c$ with $B \distype \intty$.
    \item If $V$ is of shape $(W_1,W_2)$ then either:
    \begin{itemize}
      \item $A$ is of shape $B \times C$
      \item or, $A$ is of shape $B^c$ and $B$ is not a pair type nor $\okty$.
      \item or, $A$ is of shape $(B_1 \times B_2)^c$ and there is $i$ with $\types W_i : B_i^c$
    \end{itemize}
    \item If $V$ is of shape $\fixtm{f(x)}{M}$ then either:
    \begin{itemize}
      \item \ch{there are types $B_1$ and $B_2$ such that $A = B_1 \to B_2$ and $f:B_1 \to B_2,\,x:B_1 \types M:B_2$}
      \item or, $A$ is of shape $B^c$ with $B$ not an arrow type nor $\okty$.
    \end{itemize}
  \end{itemize}
\end{lemma}
\begin{proof}
  The proof is by induction on the derivation.  In cases \rlnm{OkC}, \rlnm{Contra} and \rlnm{Var}, the conclusion is immediate because the rules only apply to judgements with a non-empty environment.  In cases \rlnm{Let1}, \rlnm{Let2}, \rlnm{App1}, \rlnm{App2}, \rlnm{App3}, \rlnm{IfZ1} and \rlnm{IfZ2} the conclusion is immediate since the subject is not a value.
  \begin{description}
    \item[\rlnm{OkC2}] We suppose the induction hypothesis, but then we obtain a contradiction because, by case analysis on $V$, it follows that $\okty^c$ cannot be $\okty^c$.  Hence, the result follows vacuously.
    \item[\rlnm{Dis}] In this case, $A$ has shape $A'^c$ and we assume $B \distype A'$.  We proceed by case analysis on $V$.  Note that, $B$ cannot be of shape $B'^c$.
    \begin{itemize}
      \item If $V$ is a numeral $\pn{n}$, then it follows from the induction hypothesis that $B$ must be $\intty$.  Therefore, $A$ is of shape $A'^c$ with $A' \distype \intty$, as required.
      \item If $V$ is a pair $(W_1,W_2)$, then it follows from the induction hypothesis that $B$ must have shape $B_1 \times B_2$.  Therefore, $A'$ is not a pair type nor $\okty$.
      \item If $V$ is an abstraction $\abs{x}{P}$, then it follows from the induction hypothesis that $B$ must have shape $B_1 \to B_2$.  Therefore $A'$ is not an arrow type nor $\okty$.
    \end{itemize}
    \item[\rlnm{Zero},\rlnm{Succ1}]  In these cases, $V$ is a numeral and $A$ is $\intty$.
    \item[\rlnm{Succ2}] In this case, $V$ is a numeral.  We assume the induction hypothesis, but then we obtain a contradiction because $\intty^c$ is not of shape $B^c$ with $B \distype \intty$.
    \item[\rlnm{Fix}] In this case, $V$ is a fixpoint abstraction and $A$ has shape $B_1 \to B_2$.
    \item[\rlnm{Pair1}] In this case, $V$ is of shape $(W_1,W_2)$ and $A$ is of shape $A \times B$ and $\Gamma \types W_1 : A$ and $\Gamma \types W_2 : B$
    \item[\rlnm{Pair2}] In this case, we suppose the induction hypothesis, but then we obtain a contradiction because the type of $W_i$ cannot be $\okty^c$.
    \item[\rlnm{Pair3}] In this case, $V$ is of shape $(W_1,W_2)$ and $A$ is of shape $(B_1 \times B_2)$ and there is an $i$ such that $\Gamma \types W_i : B_i^c$.
  \end{description}
\end{proof}

\begin{lemma}
  Suppose $\Gamma$ is a type environment and $V$ a closed value and $\Gamma \types M : A$.
  \begin{itemize}
    \item If $x:B \in \Gamma$ and $\Gamma \setminus \{x:B\} \types V : B$ then $\Gamma \setminus \{x:B\} \types M[V/x] : A$.
    \item If $x \notin \subjects(\Gamma)$, then $\Gamma \types M[V/x] : A$.
  \end{itemize}
\end{lemma}
\begin{proof}
  The proof is by induction on the derivation.
  \begin{description}
    \item[\rlnm{Ok}] In this case, $\Gamma \types M[V/x] : \okty$ by \rlnm{Ok}.
    \item[\rlnm{OkC1}] In this case, $\Gamma$ contains $y:\okty^c$ and we consider two cases.  If $x=y$, then we may assume $\Gamma \setminus \{x:B\} \types V : \okty^c$ and this gives the result since $x[V/x] = V$.  Otherwise, the result follows from \rlnm{OkC1}.
    \item[\rlnm{Zero}] In this case, the result follows again from \rlnm{Zero}.
    \item[\rlnm{Contra}] This rule does not apply when $\Gamma$ is a type environment.
    \item[\rlnm{Var}] In this case, $\Gamma$ contains $y:A'$.  If $x=y$ then $A=A'$ and we may assume $\Gamma - \{x:B\} \types V : A$.  Then this gives the result since $x[V/x] = V$.  Otherwise, the result follows from \rlnm{Var} since $y[V/x] = y$.
  \end{description}
  In the remaining cases, the result follows from the induction hypotheses and the fact that bound variables in the conclusion may be assumed distinct from $x$.
\end{proof}

\begin{theorem}
  If $\types M : A$ and $M \ped N$ then $\types N : A$.
\end{theorem}
\begin{proof}
  The proof is by induction on the derivation.  The cases \rlnm{OkC1}, \rlnm{Contra} and \rlnm{Var} do not apply since there are no assumptions.
  \begin{description}
    \item[\rlnm{Ok}] In this case, the result follows immediately from \rlnm{Ok}.
    \item[\rlnm{OkC2}] It follows from the induction hypothesis that $\types N : \okty^c$ and the result follows from \rlnm{OkC2}.
    \item[\rlnm{Dis}] In this case, $A$ is of shape $A'^c$.  Suppose $A \distype B$.  It follows from the induction hypothesis that $\types N : B$ and hence the result follows from \rlnm{Dis}.
    \item[\rlnm{Zero}] In this case, $M$ is $\zerotm$ and we obtain a contradiction from $M \ped N$.
    \item[\rlnm{Succ1}] In this case, $M$ is of shape $\succtm(P)$ and $A=\intty$.  Thus it must be that $P$ makes a step to some $Q$.  It follows from the induction hypothesis that $\types Q : \intty$ and the result follows from \rlnm{Succ1}.
    \item[\rlnm{Succ2}] In this case, $M$ is of shape $\succtm(P)$ and thus it must be that $P$ makes a step to some $Q$.  It follows from the induction hypothesis that $\types Q : \intty^c$ and thus the result follows from \rlnm{Succ2}.
    \item[\rlnm{Fix}] \ch{In this case, $M$ does not make a step.}
    \item[\rlnm{Let3}] In this case, $M$ has shape $\lettm{(x,y)}{Q}{P}$ and there are two cases.  If $Q$ makes a step to some $Q'$ then the result follows from the induction hypothesis.  Otherwise, it must be that $Q$ is a pair $(V,W)$ and $N=P[V/x,W/y]$.  Since $x$ and $y$ do not occur in $\Gamma$, the result follows from the substitution lemma.
    \item[\rlnm{Let2}] In this case, $M$ has shape $\lettm{(x,y)}{Q}{P}$ and there are two cases.  If $Q$ makes a step to some $Q'$ then the result follows from the induction hypothesis.  Otherwise, it must be that $Q$ is a pair $(V,W)$ and $N=P[V/x,W/y]$.  Hence, it follows by value inversion on the first premise that either $\Gamma \types V:B_1^c$ or $\Gamma \types W:B_2^c$.  In either case, we can use the corresponding second premise and the substitution lemma (both parts) to conclude.
    \item[\rlnm{Let1}] In this case, $M$ has shape $\lettm{(x,y)}{Q}{P}$ and there are two cases.  If $Q$ makes a step to some $Q'$ then the result follows from the induction hypothesis.  Otherwise, it must be that $Q$ is a pair $(V,W)$ and $N=P[V/x,W/y]$.  It follows from value inversion that $\Gamma \types V : B$ and $\Gamma \types W : C$ and so we can conclude using the substitution lemma.
    \item[\rlnm{App1}] In this case, $M$ has shape $PQ$.  There are two cases.  If $P$ or $Q$ makes a step, then the result follows from the induction hypothesis.  Otherwise, it must be that $P$ is an abstraction $\abs{x}{P'}$ and $Q$ is a value.  By value inversion, it must be that $x:B \types P':A$ and then the result follows from the substitution lemma.
    \item[\rlnm{App2}] In this case, $M$ has shape $PQ$.  There are three cases.  If $P$ makes a step, then the result follows from the induction hypothesis.  If $Q$ makes a step, the result follows immediately from \rlnm{App2}.  Otherwise $P$ must be an abstraction $\abs{x}{P'}$ but it follows from value inversion that this is impossible with type $(\okty^c \to A)^c$.
    \item[\rlnm{App3}] In this case, $M$ has shape $PQ$.  There are three cases.  If $Q$ makes a step, then the result follows from the induction hypothesis.  If $P$ makes a step, the result follows immediately from \rlnm{App3}.  Otherwise, it must be that $Q$ is a value, but it follows from value inversion that this is impossible with type $\okty^c$.
    \item[\rlnm{Pair1}] In this case, $M$ is a pair $(P,Q)$ and $A$ has shape $B_1 \times B_2$.  It can only be that either $P$ makes a step or $Q$ makes a step, and then the result follows from the induction hypothesis in each case.
    \item[\rlnm{Pair2}] In this case, $M$ is a pair $(P_1,P_2)$ and $i \in \{1,2\}$.  If $P_i$ makes a step, then the result follows from the induction hypothesis; otherwise the result follows immediately from \rlnm{Pair2}.
    \item[\rlnm{Pair3}] In this case, $M$ is a pair $(M_1,M_2)$ and $i \in \{1,2\}$ and $A$ has shape $(A_1 \times A_2)^c$.  If $M_i$ makes a step, then the result follows from the induction hypothesis; otherwise the result follows immediately from \rlnm{Pair3}.
    \item[\rlnm{IfZ1}] In this case, $M$ is of shape $\iftm{Q}{P_1}{P_2}$.  There are two cases.  If $Q$ makes a step, then the result follows from the induction hypothesis.  Otherwise, $M$ is a numeral.  However, it follows from value inversion that this is impossible with type $\intty^c$.
    \item[\rlnm{IfZ2}] In this case, $M$ is of shape $\iftm{Q}{P_1}{P_2}$.  There are two cases.  If $Q$ makes a step, then the result follows from the induction hypothesis.  Otherwise, $M$ is a numeral.  If $M = \pn{0}$, then $M \ped P_1$ and the result follows from the first premise.  Otherwise, $M \ped P_2$ and the result follows from the second premise.
  \end{description}
\end{proof}
  
\subsubsection{Proof of Theorem~\ref{thm:failures-soundness} (One-Sided Syntactic Soundness)}

\begin{proof}
  Suppose $\types M: \okty^c$.  For the purpose of obtaining a contradiction, suppose $M$ reaches a value $V$.  Then it follows from preservation that $\types V : \okty^c$, but this contradicts the value inversion lemma.
\end{proof}

%% file: algorithmic.tex
\section{Additional Material in Support of Section~\ref{sec:constrained-typing}}\label{apx:algo}

This section contains: a statement (and proof) of (left) weakening for the constrained type system given in Section~\ref{sec:constrained-typing}, a definition for an algorithmic version of the constrained type system rules (where the sub-type rules have been distributed throughout the other rules in the standard way), the proof of its correctness in the form of a soundness and completeness proof (relative to the original constraint type system rules), a type inference algorithm, and its correctness in the form of a soundness and completeness proof (relative to the algorithmic typing rules).

\chB{
\begin{lemma}[Left Weakening]
    If $\Gamma \types \Delta$ then $\Gamma,\, M : A \types \Delta$.
\end{lemma}
\begin{proof}
    Let $\Gamma,\, \Delta$ be environments, $N$ be a term and $\tau$ be a type such that $\Gamma \types \Delta$. We now proceed by induction on $\types$: \\
    \begin{description}
        \item[$\rlnm{GVar}$] This case is trivial by the use of $\rlnm{GVar}$.
        \item[$\rlnm{Id}$] This case is trivial by the use of $\rlnm{Id}$.
        \item[$\rlnm{VarK}$] This case is trivial by the use of $\rlnm{VarK}$.
        \item[$\rlnm{SubL}$] {
            Let $\Gamma,\, M : A \types \Delta$,\enspace $\Gamma,\, M : B \types \Delta$, and $A \subtype B$. \\
            By the inductive hypothesis, $\Gamma,\, N : \tau,\, M : B \types \Delta$. \\
            As $A \subtype B$, by $\rlnm{SubL}$, we have $\Gamma,\, N : \tau,\, M : A \types \Delta$ as required.
            }
        \item[$\rlnm{SubR}$] {
            Let $\Gamma \types M : A$,\enspace $\Gamma \types M : B$, and $B \subtype A$. \\
            By the inductive hypothesis, we have $\Gamma,\, N : \tau \types M : B$. \\
            As $B \subtype A$, by $\rlnm{SubR}$, we have $\Gamma,\, N : \tau \types M : A$.
            }
        \item[$\rlnm{FixsR}$] \ch{
            Let $\Gamma \types (\fixtm{f(x)}{M}) : B \to A$ and $\Gamma,\, f : B \to A,\, x : B \types M : A$. \\
            By the inductive hypothesis, we have $\Gamma,\, N : \tau,\, f : B \to A,\, x : B \types M : A$.\\
            By $\rlnm{FixsR}$, we have $\Gamma,\, N : \tau \types (\fixtm{f(x)}{M}) : B \to A$ as required.
            }
        \item[$\rlnm{FixnR}$] \ch{
            Let $\Gamma \types (\fixtm{f(x)}{M}) : B \from A$ and $\Gamma,\, f : B \from A,\, M : A \types x : B$. \\
            By the inductive hypothesis, we have $\Gamma,\, N : \tau,\, f : B \from A,\, M : A \types x : B$.\\
            By $\rlnm{FixnR}$, we have $\Gamma,\, N : \tau \types (\fixtm{f(x)}{M}) : B \from A$ as required.
            }
        \item[$\rlnm{AppL}$] {
            Let $\Gamma,\, (M\ N') : A \types \Delta$,\enspace $\Gamma \types M : B \from A$, and $\Gamma,\, N' : B \types \Delta$. \\
            By the inductive hypothesis, we have $\Gamma,\, N : \tau \types M : B \from A$ and $\Gamma,\, N : \tau,\, N' : B \types \Delta$.\\
            By $\rlnm{AppL}$, we have $\Gamma,\, N : \tau,\, (M\ N') : A \types \Delta$ as required.
            }
        \item[$\rlnm{AppL}$] {
            Let $\Gamma \types (M\ N') : A$,\enspace $\Gamma \types M : B \to A$, and $\Gamma \types N' : B$. \\
            By the inductive hypothesis, $\Gamma,\, N : \tau \types M : B \to A$ and $\Gamma,\, N : \tau \types N' : B$.\\
            By $\rlnm{AppR}$, we have $\Gamma,\, N : \tau \types (M\ N') : A$ as required.
            }
        \item[$\rlnm{CnsL}$] {
            Let $\Gamma,\, c(M_1,\ldots,M_n) : c(A_1,\ldots,A_n) + \Sigma_{d\in\mathcal{C}\backslash\{c\}}(d(\vec{A_d})) \types \Delta$ and $\Gamma,\, M_i : A_i \types \Delta$ for some $i$.\\
            By the inductive hypothesis, we have $\Gamma,\, N : \tau,\, M_i : A_i \types \Delta$.\\
            By $rlnm{CnsL}$, we have $\Gamma,\, N : \tau,\, c(M_1,\ldots,M_n) : c(A_1,\ldots,A_n) + \Sigma_{d\in\mathcal{C}\backslash\{c\}}(d(\vec{A_d})) \types \Delta$ as required.
            }
        \item[$\rlnm{CnsR}$] {
            Let $\Gamma \types c(M_1,\ldots,M_n) : c(A_1,\ldots,A_n)$ and $\Gamma \types M_i : A_i$ for all $i$. \\
            By the inductive hypothesis, $\Gamma,\, N : \tau \types M_i : A_i (\forall i)$.\\
            By $\rlnm{CnsR}$, we have $\Gamma,\, N : \tau \types c(M_1,\ldots,M_n) : c(A_1,\ldots,A_n)$ as required.
            }
        \item[$\rlnm{MchL}$] {
            Let $\Gamma,\,\matchtm{M}{|_{i=1}^k (p_i \mapsto P_i)} : A \types \Delta$,\enspace $\Gamma,\,P_i : A \types x : B_x \; (\forall i.\forall x\in\fv{(p_i)})$, and $\Gamma,\,(M, P_i):(p_i[B_x/x \mid x\in\fv{(p_i)}], A) \types \Delta \; (\forall i)$. \\
            By the inductive hypothesis we have:
            \begin{enumerate}
                \item $\Gamma,\, N : \tau,\, P_i : A \types x : B_x \; (\forall i.\forall x\in\fv{(p_i)})$
                \item $\Gamma,\, N : \tau,\, (M, P_i):(p_i[B_x/x \mid x\in\fv{(p_i)}], A) \types \Delta \; (\forall i)$
            \end{enumerate}
            By $\rlnm{MchL}$, we have $\Gamma,\, N : \tau,\, \matchtm{M}{|_{i=1}^k (p_i \mapsto P_i)} : A \types \Delta$ as required.
            }
        \item[$\rlnm{MchR}$] {
            Let $\Gamma \types \matchtm{M}{|_{i=1}^k (p_i \mapsto P_i)} : A$,\enspace $\Gamma \types M : \Sigma_{i=1}^k p_i[B_x/x \mid x\in\fv(p_i)]$, and $\Gamma \cup \{x:B_x \mid x \in \fv(p_i)\} \types P_i : A (\forall i)$. \\
            By the inductive hypothesis we have:
            \begin{enumerate}
                \item $\Gamma,\, N : \tau \types M : \Sigma_{i=1}^k p_i[B_x/x \mid x\in\fv(p_i)]$
                \item $\Gamma,\, N : \tau \cup \{x:B_x \mid x \in \fv(p_i)\} \types P_i : A (\forall i)$
            \end{enumerate}
            By $\rlnm{MchR}$, we have $\Gamma,\, N : \tau \types \matchtm{M}{|_{i=1}^k (p_i \mapsto P_i)} : A$ as required.
            }
        \item[$\rlnm{CnsK}$] {
            Let $\Gamma,\, c(M_1, \cdot, M_n) : \okty \types \Delta$ and $\Gamma,\, M_i : \okty \types \Delta$ for some $i$.\\
            By the inductive hypothesis, we have $\Gamma,\, N : \tau,\, M_i : \okty \types \Delta$ for some $i$.\\
            By $\rlnm{CnsK}$, we have $\Gamma,\, N : \tau,\, c(M_1, \cdot, M_n) : \okty \types \Delta$ as required.
            }
        \item[$\rlnm{FunK}$] {
            Let $\Gamma,\, M\ N' : \okty \types \Delta$ and $\Gamma,\, M : \okty \from A \types \Delta$.\\
            By the inductive hypothesis, we have $\Gamma,\, N : \tau,\, M : \okty \from A \types \Delta$.\\
            By $\rlnm{FunK}$, we have $\Gamma,\, N : \tau,\, M\ N' : \okty \types \Delta$ as required.
            }
        \item[$\rlnm{CnsDl}$] This case is trivial by the use of $\rlnm{CnsDL}$.
        \item[$\rlnm{FixDL}$] This case is trivial by the use of $\rlnm{FixDL}$.
    \end{description}
\end{proof}
}

\begin{definition}
    Given a set of typing and subtype formulas $\Gamma$, write $\constraints(\Gamma)$ for the subset of subtype formulas, and $\typings(\Gamma)$ for the typing formulas.
    Given two sets of typing and subtype formulas $\Gamma$ and $\Gamma'$ write $\Gamma \subtype \Gamma'$ just if for all typings $M:A \in \Gamma$, there is a typing $M:B \in \Gamma'$ and $\Gamma \types A \subtype B$.
\end{definition}

\begin{figure}
  \begin{mdframed}[topline=false,innertopmargin=-.84ex,innerleftmargin=-.1ex,innerrightmargin=0ex]
  \rulediv{25pt}{Structural}
  \[
    \begin{array}{c}
      \prftree[l, r]{$C \types C'[\bar{B}/\bar{a}]$ and $C \types A [ \bar{B} / \bar{a} ] \subtype A'$}{\rlnm{Inst2}}{\Gamma,\, f:\forall \bar{a}. C' \Rightarrow A \mid C \types_2 f : A'}
      \\[5mm]
      \prftree[l,r]{$C \types A \subtype B$}{\rlnm{Var2}}{\Gamma,\, x:A \mid C \types_2 x:B}
      \qquad
      \prftree[l,r]{$C \types \okty \subtype A$}{\rlnm{VarK2}}{\Gamma \mid C \types_2 x : A}
    \end{array}
  \]
  \\[2mm]
  \rulediv{25pt}{Functions}\\[1ex]
    \[
    \begin{array}{c}
        \ch{
        \prftree[l,r]{$C \types B_1 \to B_2 \subtype A$}{\rlnm{FixsR2}}{\Gamma,\, f : A,\, x:B_1 \mid C \types_2 M : B_2}{\Gamma \mid C \types_2 (\fixtm{f(x)}{M}) : A}
        }
        \\[5mm]
        \ch{
        \prftree[l,r]{$C \types B_2 \from B_1 \subtype A$}{\rlnm{FixnR2}}{\Gamma,\, f : A,\, M:B_1 \mid C \types_2 x : B_2}{\Gamma \mid C \types_2 (\fixtm{f(x)}{M}) : A}
        }
        \\[5mm]
        \prftree[l,r]{$C \types B_1 \subtype B_2 \to A$}{\rlnm{AppR2}}{\Gamma \mid C \types_2 M : B_1}{\Gamma \mid C \types_2 N : B_2}{\Gamma \types_2 (M\ N) : A}
        \\[5mm]
        \prftree[l,r]{$C \types B_1 \subtype B_2 \from A$}{\rlnm{AppL2}}{\Gamma \mid C \types_2 M : B_1}{\Gamma,\, N : B_2 \types_2 \Delta}{\Gamma,\, (M\ N):A \mid C \types_2 \Delta}
    \end{array}
    \]
  \\[2mm]
  \rulediv{30pt}{Constructors}
  \[
    \begin{array}{c}
      \prftree[l,r]{$C \types A \subtype \Sigma_{d \in \calC \setminus \{c\}}\,d(\vv{A_d}) + c(A_1,\ldots,A_n)$}{\rlnm{CnsL2}}{\Gamma,\,M_i : A_i \mid C \types_2 \Delta}{\Gamma,\,c(M_1,\ldots,M_n) : A \mid C \types_2 \Delta}
    \\[5mm]
      \prftree[l,r]{$C \types c(A_1,\ldots,A_n) \subtype A$}{\rlnm{CnsR2}}{\Gamma \mid C \types_2 M_i : A_i \;(\forall i)}{\Gamma \mid C \types_2 c(M_1,\ldots,M_n) : A}
    \end{array}
  \]
  \\[2mm]
  \rulediv{30pt}{Pattern Matching}\\[1ex]
  \[
    \begin{array}{c}
      \prftree[l,r]{$\begin{array}{l}C \types A \subtype A_i \; (\forall i)\\C \types (B_i, A_i) \subtype A_i' \;(\forall i)\\C \types A \subtype A_x \; (\forall i.\forall x \in \fv(p_i))\\C \types p_i[B_x/x \mid x \in \fv(p_i)] \subtype B_i \; (\forall i)\end{array}$}{\rlnm{MchL2}}{\begin{array}{c}\Gamma,\,(M, P_i):A'_i \mid C \types_2 \Delta \; (\forall i)\\[1mm]\Gamma,\,P_i : A_x \mid C \types_2 x : B_x \; (\forall i. \forall x\in\fv{(p_i)})\end{array}}{\Gamma,\,\matchtm{M}{|_{i=1}^k (p_i \mapsto P_i)} : A \mid C \types_2 \Delta}
      \\[5mm]
      \prftree[l,r]{$\begin{array}{l}C \types B \subtype \Sigma_{i=1}^k p_i[B_x/x \mid x \in \fv(p_i)]\\C \types A_i \subtype A\;(\forall i)\end{array}$}{\rlnm{MchR2}}{\begin{array}{c}\Gamma \mid C \types_2 M : B\\[1mm]\Gamma \cup \{x:B_x \mid x \in \fv(p_i)\} \mid C \types_2 P_i : A_i \; (\forall i)\end{array}}{\Gamma \mid C \types_2 \matchtm{M}{|_{i=1}^k (p_i \mapsto P_i)} : A}
      \\[5mm]
    \end{array}
  \]
  \\[2mm]
\end{mdframed}
\caption{Algorithmic constrained type assignment}\label{fig:alg-con-ta}
\end{figure}

\begin{figure}
    \begin{mdframed}[topline=false,innertopmargin=-.84ex,innerleftmargin=-.1ex,innerrightmargin=0ex]
    \rulediv{30pt}{Evaluation}\\[1ex]
    \[
        \begin{array}{c}
            \prftree[l,r]{$C \types \okty \subtype B$}{\rlnm{CnsK2}}{\Gamma,\, M_i : B \mid C \types_2 \Delta}{\Gamma,\, c(M_1,\ldots,M_n) : A \mid C \types_2 \Delta}
        \\[5mm]
            \prftree[l,r]{$C \types \okty \from A \subtype B$}{\rlnm{FunK2}}{\Gamma,\, M : B \mid C \types_2 \Delta}{\Gamma,\, M\ N : A \mid C \types_2 \Delta}
    \end{array}
    \]
    \\[2mm]
    \rulediv{25pt}{Disjointness}\\[1ex]
    \[
        \begin{array}{c}
      \prftree[l,r]{$CA \subtype B_1 \to B_2$}{\rlnm{CnsDL21}}{\Gamma,\,c(M_1,\ldots,M_n) : A \mid C\types \Delta}
      \\[5mm]
      \prftree[l,r]{$C \types A \subtype B_1 \from B_2$}{\rlnm{CnsDL22}}{\Gamma,\,c(M_1,\ldots,M_n) : A \mid C \types \Delta}
      \\[5mm]
      \prftree[l,r]{$C \types A \subtype \Sigma_{d\in \calC \setminus \{c\}}\,d(\bar{A_d})$}{\rlnm{CnsDL23}}{\Gamma,\,c(M_1,\ldots,M_n) : A \mid C \types \Delta}
      \\[5mm]
      \ch{\prftree[l,r]{$C \types A \subtype \Sigma_{c \in \calC} c(\bar{A_c})$}{\rlnm{FixDL2}}{\Gamma,\,\fixtm{f(x)}{M} : A \mid C\types \Delta}}
        \end{array}
    \]\\
    \end{mdframed}
\caption{Algorithmic constrained type assignment (continued)}
\end{figure}

\begin{theorem}[Soundness of Algorithmic Type Assignment]
    For all $\Gamma, \Delta$, if $\typings(\Gamma) \mid \constraints(\Gamma) \types_2 \Delta$ then $\Gamma \cup C \types \Delta$.
\end{theorem}
\begin{proof}
    We proceed by induction on the derivation of $\types_2$.\\
    \begin{description}
        \item[$\rlnm{Inst2}$] This case is trivial by the use of $\rlnm{GVar}$.
        \item[$\rlnm{Var2}$] Let $\Gamma,\, x : A$ be the context with $A \subtype B$. By $\rlnm{Id}$, we have $\Gamma,\, x : A \types x : A$. By $\rlnm{SubR}$, we have $\Gamma,\, x : A \types x : B$ as required.
        \item[$\rlnm{FixsR2}$]\ch{
            Let $\typings(\Gamma) \mid \constraints(\Gamma) \types_2 (\fixtm{f(x)}{M}) : A$ and $\typings(\Gamma),\, f : A,\, x : B_1 \mid \constraints(\Gamma) \types_2 M : B_2$ with $\Gamma \types B_1 \to B_2 \subtype A$. \\
            By the inductive hypothesis, $\Gamma,\, f : A,\, x : B_1 \types M : B_2$. \\
            By $\rlnm{SubL}$ we have $\Gamma,\, f : B_1 \to B_2,\, x : B_1 \types M : B_2$. \\
            By $\rlnm{FixsR}$ we have $\Gamma \types (\fixtm{f(x)}{M}) : B_1 \to B_2$. \\
            By $\rlnm{SubR}$ we have $\Gamma,\,\types \fixtm{f(x)}{M} : A$ as required.
            }
        \item[$\rlnm{FixnR2}$] \ch{
            Let $\typings(\Gamma) \mid \constraints(\Gamma) \types_2 (\fixtm{f(x)}{M}) : A$ and $\typings(\Gamma),\, f : A,\, M : B_2 \mid \constraints(\Gamma) \types_2 x : B_1$ with $B_1 \from B_2 \subtype A$. \\
            By the inductive hypothesis, $\Gamma,\, f : A,\, M : B_2 \types x : B_1$. \\
            By $\rlnm{SubL}$ we have $\Gamma,\, f : B_1 \from B_2,\, M : B_2 \types x : B_1$. \\
            By $\rlnm{FixnR}$ we have $\Gamma \types_2 (\fixtm{f(x)}{M}) : B_1 \from A_1$. \\
            By $\rlnm{SubR}$ we have $\Gamma \types \fixtm{f(x)}{M} : A$ as required.
            }
        \item[$\rlnm{AppR2}$] {
            Let $\typings(\Gamma) \mid \constraints(\Gamma) \types_2 (M\ N) : \tau_3$,\enspace $\typings(\Gamma) \mid \constraints(\Gamma) \types_2 M : \tau_1$, and $\typings(\Gamma) \mid \constraints(\Gamma) \types_2 N : \tau_2$ with $\tau_1\subtype\tau_2\to\tau_3$. \\
            By the inductive hypothesis, we have $\Gamma \types M : \tau_1$ and $\Gamma \types N : \tau_2$.
            By $\rlnm{SubR}$ on $M$, we have $\Gamma \types M : \tau_2 \to \tau_3$. \\
            By $\rlnm{AppR}$ we have $\Gamma \types (M\ N) : \tau_3$ as required.
            }
        \item[$\rlnm{AppL2}$] {
            Let $\typings(\Gamma),\, (M\ N) : \tau_3 \mid \constraints(\Gamma) \types_2 \Delta$,\enspace $\typings(\Gamma) \mid \constraints(\Gamma) \types_2 M : \tau_1$, and $\typings(\Gamma),\, N : \tau_2 \mid \constraints(\Gamma) \types_2 \Delta$ with $\tau_1 \subtype \tau_2 \from \tau_3$.\\
            By the inductive hypothesis, we have $\Gamma \types M : \tau_1$ and $\Gamma,\, N : \tau \types_2 \Delta$. \\
            By $\rlnm{SubR}$ on $M$ we have $\Gamma \types M : \tau_2 \from \tau_3$. \\
            By $\rlnm{AppL}$, we have $\Gamma,\, (M\ N) : \tau_3 \types \Delta$ as required.
            }
        \item[$\rlnm{ConsL2}$] {
            Let $\typings(\Gamma),\,c(M_1,\ldots,M_n) : \tau \mid \constraints(\Gamma) \types_2 \Delta$ and $\typings(\Gamma),\,M_i : A_i \mid \constraints(\Gamma) \types_2 \Delta$ with $\tau \subtype \Sigma_{d\in\mathcal{C}\backslash\{c\}}(d(\vec{A_d})) + c(A_1,\ldots,A_n)$. \\
            By the inductive hypothesis, $\Gamma,\,M_i : A_i \types \Delta$. \\
            By $\rlnm{CnsL}$, we have $\Gamma,\,c(M_1,\ldots,M_n) : c(A_1,\ldots,A_n) + \Sigma_{d\in\mathcal{C}\backslash\{c\}}(d(\vec{A_d})) \types \Delta$. \\
            By $\rlnm{SubL}$, we have $\Gamma,\,c(M_1,\ldots,M_n) : \tau \types \Delta$ as required.
            }
        \item[$\rlnm{ConsR2}$] {
            Let $\typings(\Gamma) \mid \constraints(\Gamma) \types_2 c(_1,\ldots,M_n) : \tau$ and $\typings(\Gamma) \mid \constraints(\Gamma) \types_2 M_i : A_i \;(\forall i)$ with $c(A_1,\ldots,A_n) \subtype \tau$. \\
            By the inductive hypothesis, we have $\Gamma \mid C \types M_i : A_i \;(\forall i)$. \\
            By $\rlnm{CnsR}$, we have $\Gamma \types_2 c(_1,\ldots,M_n) : c(A_1,\ldots,A_n)$. \\
            By $\rlnm{SubR}$, we have $\Gamma \types c(_1,\ldots,M_n) : \tau$ as required.
            }
        \item[$\rlnm{MchL2}$] {
            Let $\typings(\Gamma),\,\matchtm{M}{|_{i=1}^k (p_i \mapsto P_i)} : A \mid \constraints(\Gamma) \types_2 \Delta$,\enspace $\typings(\Gamma),\,(M, P_i):A'_i \mid \constraints(\Gamma) \types_2 \Delta \; (\forall i)$ and $\typings(\Gamma),\, P_i : A_x \mid \constraints(\Gamma) \types_2 x : B_x \; (\forall i.\forall x\in\fv{p_i})$ where $A \subtype A_x \; (\forall i.\forall x\in\fv{(p_i)})$,\enspace $A \subtype A_i \; (\forall i.)$,\enspace $(B_i, A_i) \subtype A'_i$, and $p_i[B_x/x \mid x \in \fv{(p_i)}]\subtype B_i \; (\forall i.)$. \\
            By the inductive hypothesis, $\Gamma,\,(M, P_i):A'_i \types \Delta \; (\forall i)$ and $\Gamma,\, P_i : A_x \types x : B_x \; (\forall i.\forall x\in\fv{p_i})$.\\
            By $\rlnm{SubL}$ (applied to $(M, P_i)$ twice and $P_i$ once), we have $\Gamma,\,(M, P_i):(p_i[B_x/x \mid x\in\fv{(p_i)}], A) \types \Delta \; (\forall i)$ and $\Gamma,\, P_i : A \types x : B_x \; (\forall i.\forall x\in\fv{p_i})$.\\
            By $\rlnm{MchL}$, we have $\Gamma,\,\matchtm{M}{|_{i=1}^k (p_i \mapsto P_i)} : A \types \Delta$ as required.
            }
        \item[$\rlnm{MchR2}$] {
            Let $\typings(\Gamma) \mid \constraints(\Gamma) \types_2 \matchtm{M}{|_{i=1}^k (p_i \mapsto P_i)} : A$,\enspace $\typings(\Gamma) \cup \{x:B_x \mid x \in \fv(p_i)\} \mid \constraints(\Gamma) \types_2 P_i : A_i \; (\forall i)$, and $\typings(\Gamma) \mid \constraints(\Gamma) \types_2 M : B$ where $B \subtype \Sigma_{i=1}^k p_i[B_x/x \mid x\in\fv(p_i)]$ and $A_i \subtype A \; (\forall i)$. \\
            By the inductive hypothesis we have $\Gamma \cup \{x:B_x \mid x \in \fv(p_i)\} \types P_i : A_i \; (\forall i)$ and $\Gamma \types M : B$. \\
            By $\rlnm{SubR}$ (applied to $M$ and $P_i$), we have $\Gamma \cup \{x:B_x \mid x \in \fv(p_i)\} \types P_i : A \; (\forall i)$ and $\Gamma \types M : p_i[B_x/x \mid x\in\fv(p_i)]$. \\
            By $\rlnm{MchR}$, we have $\Gamma \types \matchtm{M}{|_{i=1}^k (p_i \mapsto P_i)} : A$.
            }
        \item[$\rlnm{VarK2}$] This case is trivial by the use of $\rlnm{VarK}$.
        \item[$\rlnm{CnsK2}$] This case is trivial by the use of $\rlnm{CnsK}$.
        \item[$\rlnm{FunK2}$] {
            Let $\typings(\Gamma),\, M\ N : A \mid \constraints(\Gamma) \types_2 \Delta$,\enspace $\typings(\Gamma),\, M : B \types_2 \Delta$, and $\okty \from A \subtype B$. \\
            By the inductive hypothesis, we have $\Gamma,\, M : B \types \Delta$. \\
            By $\rlnm{SubL}$, we have $\Gamma,\, M : \okty \from A \types \Delta$. \\
            By $\rlnm{FunK}$, we have $\Gamma,\, M\ N : \okty \types \Delta$. \\
            By $\rlnm{SubL}$, we have $\Gamma,\, M\ N : A \types \Delta$ as required.
            }
        \item[$\rlnm{CnsDL21}$] This case is trivial by the use of $\rlnm{CnsDL}$ where the type $A$ is an arrow.
        \item[$\rlnm{CnsDL22}$] This case is trivial by the use of $\rlnm{CnsDL}$ where the type $A$ is an arrow.
        \item[$\rlnm{CnsDL23}$] This case is trivial by the use of $\rlnm{CnsDL}$ where the type $A$ is a sum of constructor types that does not include $c$.
        \item[$\rlnm{FixDL2}$] This case is trivial by the use of $\rlnm{FixDL}$ where the type $A$ is a sum of constructor types.
    \end{description}
\end{proof}

\begin{theorem}[Completeness of Algorithmic Type Assignment]
    For all $\Gamma, \Gamma', \Delta, \Delta'$, if the following hold:
    \begin{enumerate}
        \item $\Gamma \types \Delta$
        \item $\Delta \subtype \Delta'$
        \item $\Gamma' \subtype \Gamma$
    \end{enumerate}
    Then $\typings(\Gamma') \mid \constraints(\Gamma') \types_2 \Delta'$.
\end{theorem}
\begin{proof}
    We proceed by induction on the derivation of $\types$.\\
    \begin{description}
        \item[$\rlnm{GVar}$] This case is trivial by the use of $\rlnm{Inst2}$.
        \item[$\rlnm{Id}$] {
            Let $\Gamma,\, x : A \types x : A$,\enspace $\Gamma',\, x:C \subtype \Gamma,\, x:A$, and $A \subtype B$.\\
            Then by $\rlnm{Var2}$, $\typings(\Gamma'),\, x : C \mid \constraints(\Gamma),\, C \subtype B \types_2 x : B$.
            }
        \item[$\rlnm{SubL}$] {
            Let $\Gamma,\, M : A \types \Delta$,\enspace $\Gamma,\, M : B \types \Delta$, and $A \subtype B$. \\
            By the inductive hypothesis, for all $\Gamma', \Delta'$ such that $\Gamma' \subtype \Gamma,\, M : A$ and $\Delta \subtype \Delta'$ we have $\typings(\Gamma') \mid \constraints(\Gamma') \types_2 \Delta'$. \\
            As $A \subtype B$, $M : A$ is an included entry in an instance of $\Gamma'$. \\
            Thus, for all $\Gamma'' \subtype \Gamma$, $A' \subtype A$, and $\Delta \subtype \Delta''$, we have $\typings(\Gamma''),\, M : A' \mid \constraints(\Gamma'') \types_2 \Delta''$.
            }
        \item[$\rlnm{SubR}$] {
            Let $\Gamma \types M : A$,\enspace $\Gamma \types M : B$, and $B \subtype A$. \\
            By the inductive hypothesis, for all $\Gamma', B'$ such that $\Gamma' \subtype \Gamma$ and $B \subtype B'$ we have $\typings(\Gamma') \mid \constraints(\Gamma') \types_2 M : B'$. \\
            Thus, as $B \subtype A$ we have for all $A'$ such that $A \subtype A'$, $B \subtype A'$ and so $\typings(\Gamma') \mid \constraints(\Gamma') \types_2 M : A'$ as required.
            }
        \item[$\rlnm{FixsR}$] \ch{
            Let $\Gamma \types (\fixtm{f(x)}{M}) : B \to A$ and $\Gamma,\, f : B \to A,\, x : B \types M : A$. \\
            By the inductive hypothesis, for all $\Gamma', B', A'$ such that$\Gamma' \subtype \Gamma$, $B' \subtype B$, and $A \subtype A'$, we have $\typings(\Gamma'),\, f : B' \to A',\, x : B' \mid \constraints(\Gamma') \types_2 M : A'$. \\
            Let $\Gamma'' \subtype \Gamma$, $B' \subtype B$, and $A \subtype A'$. Then $B \to A \subtype B' \to A'$ and $\typings(\Gamma''),\, f : B' \to A',\, x : B' \mid \constraints(\Gamma'') \types_2 M : A'$. \\
            Then, by $\rlnm{FixsR2}$, we have $\typings(\Gamma'') \mid \constraints(\Gamma'') \types_2 \fixtm{f(x)}{M} : B' \to A'$ as required.
            }
        \item[$\rlnm{FixnR}$] \ch{
            Let $\Gamma \types (\fixtm{f(x)}{M}) : B \from A$ and $\Gamma,\, f : B \from A,\, M : A \types x : B$. \\
            By the inductive hypothesis, for all $\Gamma', A', B'$ such that $\Gamma' \subtype \Gamma$, $B \subtype B'$ and $A' \subtype A$, we have $\typings(\Gamma'),\, f : B' \from A',\, M : A' \mid \constraints(\Gamma') \types_2 x : B'$. \\
            Let $\Gamma'' \subtype \Gamma$, $A' \subtype A$, $B \subtype B'$. Then $B \from A \subtype B' \from A'$ and $\typings(\Gamma''),\, f : B' \from A',\, M : A' \mid \constraints(\Gamma'') \types_2 x : B'$.\\
            Then, by $\rlnm{FixnR2}$, we have $\typings(\Gamma'') \mid \constraints(\Gamma'') \types_2 (\fixtm{f(x)}{M}) : B' \from A'$ as required.
            }
        \item[$\rlnm{AppL}$] {
            Let $\Gamma,\, (M\ N) : A \types \Delta$,\enspace $\Gamma \types M : B \from A$, and $\Gamma,\, N : B \types \Delta$. \\
            By the inductive hypothesis, for all $\Gamma', \Delta', A', B'$ such that $\Gamma' \subtype \Gamma$, $\Delta \subtype \Delta'$, $B \from A \subtype B' \from A'$, and $B \subtype B'$ we have $\typings(\Gamma') \mid \constraints(\Gamma') \types_2 M : B' \from A'$ and $\typings(\Gamma'),\, N : B' \mid \constraints(\Gamma') \types_2 \Delta'$. \\
            Let $\Gamma'' \subtype \Gamma$, $A' \subtype A$, and $\Delta \subtype \Delta'$. Then, for all $B'$ such that $B \from A \subtype B' \from A'$ and $B \subtype B'$, we have $\typings(\Gamma'') \mid \constraints(\Gamma'') \types_2 M : B' \from A'$ and $\typings(\Gamma''),\, N : B' \mid \constraints(\Gamma'') \types_2 \Delta''$.\\
            Then, by $\rlnm{AppL2}$, as the side condition is trivially satisfied, we have $\typings(\Gamma''),\, (M\ N) : A' \mid \constraints(\Gamma'') \types_2 \Delta''$.
            }
        \item[$\rlnm{AppR}$] {
            Let $\Gamma \types (M\ N) : A$,\enspace $\Gamma \types M : B \to A$, and $\Gamma \types N : B$. \\
            By the inductive hypothesis, for all $\Gamma', A', B'$ we have $\Gamma' \subtype \Gamma$, $B \subtype B'$, $B \to A \subtype B' \to A'$, $\typings(\Gamma') \mid \constraints(\Gamma') \types_2 M : B' \to A'$, and $\typings(\Gamma') \mid \constraints(\Gamma') \types_2 N : B'$. \\
            Let $\Gamma'' \subtype \Gamma$ and $A \subtype A'$. Then, for all $B'$ such that $B \subtype B'$ and $B \to A \subtype B' \to A'$, we have $\typings(\Gamma'') \mid \constraints(\Gamma'') \types_2 M : B' \to A'$, and $\typings(\Gamma'') \mid \constraints(\Gamma'') \types_2 N : B'$. \\
            Then, by $\rlnm{AppR2}$, as the side condition is trivially satisfied, we have $\typings(\Gamma'') \mid \constraints(\Gamma'') \types_2 (M\ N) : A'$ as required.
            }
        \item[$\rlnm{CnsL}$] {
            Let $\Gamma,\, c(M_1,\ldots,M_n) : c(A_1,\ldots,A_n) + \Sigma_{d\in\mathcal{C}\backslash\{c\}}(d(\vec{A_d})) \types \Delta$ and $\Gamma,\, M_i : A_i \types \Delta$ for some $i$.\\
            By the inductive hypothesis, for all $\Gamma', A'_i, \Delta'$ we have $\Gamma' \subtype \Gamma$, $\Delta \subtype \Delta'$, $A'_i \subtype A_i$, and $\typings(\Gamma'),\, M_i : A'_i \mid \constraints(\Gamma') \types_2 \Delta'$. \\
            Let $\Gamma'' \subtype \Gamma$, $\Delta \subtype \Delta''$ and $A'_j \subtype A_j (\forall j)$. Then $c(A'_1,\ldots,A'_n) + \Sigma_{d\in\mathcal{C}\backslash\{c\}}(d(\vec{A'_d})) \subtype c(A_1,\ldots,A_n) + \Sigma_{d\in\mathcal{C}\backslash\{c\}}(d(\vec{A_d}))$ and $\typings(\Gamma''),\, M_i : A'_i \mid \constraints(\Gamma'') \types_2 \Delta''$.\\
            Then, by $\rlnm{ConsL2}$, we have $\typings(\Gamma''),\, c(M_1,\ldots,M_n) : c(A'_i,\ldots,A'_n) + \Sigma_{d\in\mathcal{C}\backslash\{c\}}(d(\vec{A'_d})) \mid \constraints(\Gamma'') \types_2 \Delta''$.
            }
        \item[$\rlnm{CnsR}$] {
            Let $\Gamma \types c(M_1,\ldots,M_n) : c(A_1,\ldots,A_n)$ and $\Gamma \types M_i : A_i$ for all $i$. \\
            By the inductive hypothesis, for all $\Gamma', A'_i$ we have $\Gamma' \subtype \Gamma$, $A_i \subtype A'_i$, and $\typings(\Gamma') \mid \constraints(\Gamma') \types_2 M_i : A'_i$ for all $i$. \\
            Let $\Gamma'' \subtype \Gamma$ and $A_i \subtype A'_i \; (\forall i)$. Then $c(A_i,\ldots,A_n) \subtype c(A'_i,\ldots,A'_n)$ and $\typings(\Gamma'') \mid \constraints(\Gamma'') \types_2 M_i : A'_i \; (\forall i)$.\\
            Then, by $\rlnm{ConsR2}$, we have $\typings(\Gamma'') \mid \constraints(\Gamma'') \types_2 c(M_i,\ldots,M_n) : c(A'_i,\ldots,A'_n)$ as required.
            }
        \item[$\rlnm{MchL}$] {
            Let $\Gamma,\,\matchtm{M}{|_{i=1}^k (p_i \mapsto P_i)} : A \types \Delta$,\enspace $\Gamma,\,P_i : A \types x : B_x \; (\forall i.\forall x\in\fv{(p_i)})$, and $\Gamma,\,(M, P_i):(p_i[B_x/x \mid x\in\fv{(p_i)}], A) \types \Delta \; (\forall i)$. \\
            By the inductive hypothesis we have:
            \begin{enumerate}
                \item $\forall \Gamma', A_x, B'_x. \Gamma'\subtype\Gamma \wedge A_x\subtype A \wedge B_x \subtype B'_x \Rightarrow \typings(\Gamma'),\, P_i : A_x \mid \constraints(\Gamma') \types_2 x : B'_x$
                \item $\forall \Gamma', \tau, \Delta'. \Gamma'\subtype\Gamma \wedge \tau \subtype (p_i[B_x/x \mid x\in\fv{(p_i)}], A) \wedge \Delta\subtype\Delta' \Rightarrow \typings(\Gamma'),\, (M, P_i) : (p_i[B'_x/x \mid x\in\fv{(p_i)}], A') \mid \constraints(\Gamma') \types_2 \Delta'$
            \end{enumerate}
            Let $\Gamma'' \subtype \Gamma$,\enspace $A'' \subtype A$, and $\Delta \subtype \Delta''$.\\
            By instantiating (1): $\Gamma'$ with $\Gamma''$, $A_x$ with $A''$, and $B'_x$ with $B_x$ gives $\typings(\Gamma''),\, P_i : A'' \mid \constraints(\Gamma'') \types_2 x : B_x$.\\
            By instantiating (2): $\Gamma'$ with $\Gamma''$, $\tau$ with $(B_x, A'')$, and $\Delta'$ with $\Delta''$ gives $\typings(\Gamma''),\, (M, P_i) : (p_i[B_x/x \mid x\in\fv{(p_i)}], A'') \mid \constraints(\Gamma'') \types_2 \Delta''$.\\
            Finally, by $\rlnm{MchL2}$, as the side conditions are trivially satisfied, we have $\typings(\Gamma''),\, \matchtm{M}{|_{i=1}^k (p_i \mapsto P_i)} : A'' \mid \constraints(\Gamma'') \types_2 \Delta''$ as required.
            }
        \item[$\rlnm{MchR}$] {
            Let $\Gamma \mid C \types \matchtm{M}{|_{i=1}^k (p_i \mapsto P_i)} : A$,\enspace $\Gamma \types M : \Sigma_{i=1}^k p_i[B_x/x \mid x\in\fv(p_i)]$, and $\Gamma \cup \{x:B_x \mid x \in \fv(p_i)\} \types P_i : A (\forall i)$. \\
            By the inductive hypothesis, for all $\Gamma', A', B, B'_x$ such that $\Gamma' \subtype \Gamma$, $\Sigma_{i=1}^k p_i[B_x/x \mid x\in\fv(p_i)] \subtype B$, $(\forall i) B'_x \subtype B_x$, and $A \subtype A'$, we have
            \begin{enumerate}
                \item $\typings(\Gamma') \mid \constraints(\Gamma') \types_2 M : B$
                \item $\typings(\Gamma') \cup \{x:B'_x \mid x \in \fv(p_i)\} \mid \constraints(\Gamma') \types_2 P_i : A' (\forall i)$
            \end{enumerate}
            Let $\Gamma'' \subtype \Gamma$ and $A \subtype A''$. We want to show $\typings(\Gamma'') \mid \constraints(\Gamma'') \types_2 \matchtm{M}{|_{i=1}^k (p_i \mapsto P_i)} : A''$. \\
            Instantiating (1): $\Gamma'$ with $\Gamma''$ and $B$ with $\Sigma_{i=1}^k p_i[B_x/x \mid x\in\fv(p_i)]$ gives $\typings(\Gamma'') \mid \constraints(\Gamma'') \types_2 M : \Sigma_{i=1}^k p_i[B_x/x \mid x\in\fv(p_i)]$. \\
            Instantiating (2): $\Gamma'$ with $\Gamma''$, $B'_x$ with $B_x$, and $A'$ with $A''$, gives $\typings(\Gamma'') \cup \{x:B_x \mid x \in \fv(p_i)\} \mid \constraints(\Gamma'') \types_2 P_i : A'' (\forall i)$.\\
            Finally, by $\rlnm{MchR2}$, we have $\typings(\Gamma'') \mid \constraints(\Gamma'') \types_2 \matchtm{M}{|_{i=1}^k (p_i \mapsto P_i)} : A''$ as required.
            }
        \item[$\rlnm{VarK}$] This case is trivial by the use of $\rlnm{VarK2}$.
        \item[$\rlnm{CnsK}$] This case is trivial by the use of $\rlnm{CnsK2}$ with $A$ as $\okty$.
        \item[$\rlnm{FunK}$] {
            Let $\Gamma,\, M\ N : \okty \types \Delta$ and $\Gamma,\, M : \okty \from A \types \Delta$.\\
            By the inductive hypothesis, for all $\Gamma', B, \Delta'$ such that $\Gamma' \subtype \Gamma'$, $B \subtype \okty \from A$, and $\Delta \subtype \Delta'$, we have $\typings(\Gamma'),\, M : B \mid \constraints(\Gamma') \types_2 \Delta'$.\\
            Let $\Gamma'' \subtype \Gamma$ and $\Delta' \subtype \Delta''$. \\
            By instantiating $\Gamma'$ with $\Gamma''$, $\Delta'$ with $\Delta''$, and $B$ with $\okty \from \okty$, we have $\typings(\Gamma''),\, M : \okty \from \okty \mid \constraints(\Gamma'') \types_2 \Delta''$.\\
            By $\rlnm{FunK2}$, we have $\typings(\Gamma''),\, M\ N : \okty \mid \constraints(\Gamma'') \types_2 \Delta''$ as required.
            }
        \item[$\rlnm{CnsDL}$] {
            Let $\Gamma,\, c(M_1,\ldots,M_n) : A \types \Delta$ where $A$ is and arrow type or of the shape $\Sigma_{d\in I}$ with $c\notin I$.\\
            Let $\Gamma' \subtype \Gamma$, $A' \subtype A$, and $\Delta \subtype \Delta'$.\\
            If $A$ is a sum of constructor types that does not include $c$, then so is $A'$. Thus, by $\rlnm{CnsDL23}$, we have $\Gamma',\, c(M_1,\ldots,M_n) : A' \types \Delta'$.\\
            If $A$ is a sufficiency arrow ($\tau_1 \to \tau_2$), then so is $A'$. Thus, by $\rlnm{CnsDL21}$, we have $\Gamma',\, c(M_1,\ldots,M_n) : A' \types \Delta'$.\\
            If $A$ is a necessity arrow ($\tau_1 \from \tau_2$), then so is $A'$. Thus, by $\rlnm{CnsDL22}$, we have $\Gamma',\, c(M_1,\ldots,M_n) : A' \types \Delta'$.
            }
    \end{description}
\end{proof}

\subsection{Inference Algorithm}

This (sub-)section contains a definition for the constrained type system inference algorithm in Figures~\ref{fig:tyInferR} and \ref{fig:inferL}, and a proof of its correctness in the form of a soundness and correctness proof (relative to the algorithmic typing rules).


\begin{figure}
\begin{lstlisting}[language=Haskell]
    InferR($\Gamma$, M) = case M of
    Var x -> { ($\Gamma$ | {a $\subtype$ b} $\types$ M : b)  |  b = freshVar, 
            a = $\begin{cases}t & (x : t)\in \Gamma\\ \okty & \text{otherwise}\end{cases}$ 
        } $\cup$ { ($\Gamma$ | C[$\overline{b}/\overline{a}$] $\cup$ {a[$\overline{b}/\overline{a}$] $\subtype$ a'} $\types$ M : a')  | $\overline{b}$ = $\overline{\text{freshVar}}$, a' = freshVar,
            (x : $\forall$ $\overline{a}$. C $\Rightarrow$ a) $\in$ $\Gamma$
        }
    App P Q -> { ($\Gamma$ | c1 $\cup$ c2 $\cup$ {b $\subtype$ c -> a} $\types$ M : a)  |  a = freshVar,
            ($\Gamma$ | c1 $\types$ P : b) $\in$ inferR($\Gamma$, P),
            ($\Gamma$ | c2 $\types$ Q : c) $\in$ inferR($\Gamma$, Q)
        }
    Fix f x N -> { ($\Gamma$ | c $\cup$ {t -> b $\subtype$ a} $\types$ M : a)  |  a = freshVar, t = freshVar,
            ($\Gamma$ $\cup$ {x : t, f : a} | c $\types$ N : b) $\in$ inferR($\Gamma$ $\cup$ {x : t, f : a}, N)
        } $\cup$ { ($\Gamma$ | c $\cup$ {t $\from$ b $\subtype$ a} $\types$ M : a)  |  a = freshVar, t = freshVar,
            ($\Gamma$ $\cup$ {N : b, f : a} | c $\types$ x : t) $\in$ inferL($\Gamma$ $\cup$ {f : a}, N, {x : t})
        }
    Cons c ms -> { ($\Gamma$ | $\bigcup_{i=1}^n(C_i)$ $\cup$ {c($a_1,\ldots,a_n$) $\subtype$ a} $\types$ M : a)  |  a = freshVar,
            ($\Gamma$ | $C_i \types m_i : a_i$)$_{i=1}^n$ $\in \Pi_{m\in \text{ms}}$(inferR($\Gamma$, m))
        } $\cup$ { $\Gamma$ | {c $\subtype$ a} $\types$ M : a  | a = freshVar, ms = { } }
    Match Q {|$_{i=1}^k$ p$_i$ -> P$_i$} -> { ($\Gamma$ | $\bigcup_{i=1}^k(C_i \cup \{b \subtype p_i[b_x/x \mid x\in\fv(p_i)], a_i \subtype a\})$ $\cup$ C $\types$ M : a)  |
            a = freshVar, $\overline{b_x}$ = $\overline{\text{freshVar}}$,
            ($\Gamma$ | C $\types$ Q : b) $\in$ inferR($\Gamma$, Q),
            ($\Gamma$ $\cup$ {$\overline{(x : b_x)}$} | $C_i \types P_i : a_i$)$_{i=1}^k$ $\in \Pi_{i=1}^k$(inferR($\Gamma$ $\cup$ {$\overline{(x : b_x)}$}, $P_i$))
        }
\end{lstlisting}
   \caption{Type Inference Algorithm - InferR}
   \label{fig:tyInferR}
\end{figure}

\begin{figure}
\clearpage    
\begin{lstlisting}[language=Haskell]
    InferL($\Gamma$, M, $\Delta$) = { $\Gamma$ $\cup$ {M : a} | {Ok $\subtype$ b} $\types$ $\Delta$  | a = freshVar, (x : b) $\in$ d } $\cup$ case M of
            Var x -> { $\Gamma$ $\cup$ {M : a} | {a $\subtype$ b} $\types$ $\Delta$  | a = freshVar, (x : b) $\in$ $\Delta$ }
            App P Q -> { ($\Gamma$ $\cup$ {M : a} | c1 $\cup$ c2 $\cup$ {b $\subtype$ c $\from$ a} $\types$ d)  |  a = freshVar,
                    ($\Gamma$ | c1 $\types$ P : b) $\in$ inferR($\Gamma$, P),
                    ($\Gamma$ $\cup$ {Q : c} | c2 $\types$ $\Delta$) $\in$ inferL($\Gamma$, Q, $\Delta$)
                } $\cup$ { ($\Gamma$ $\cup$ {M : a} | c $\cup$ {Ok $\from$ a $\subtype$ b} $\types$ $\Delta$)  |  a = freshVar,
                    ($\Gamma$ $\cup$ {P : b} | c $\types$ $\Delta$) $\in$ inferL($\Gamma$, P, $\Delta$)
                }
            Abs x N -> { ($\Gamma$ $\cup$ {M : a} | {a $\subtype$ $\Sigma_{c\in\mathcal{C}}$(c($\overline{a}$))} $\types$ $\Delta$)  | a = freshVar, $\overline{a}$ = $\overline{\text{freshVar}}$ }
            Cons $\kappa$ ms -> { ($\Gamma$ $\cup$ {M : a} | c $\cup$ {a $\subtype$ $\Sigma_{\kappa'\in\mathcal{C}\backslash\{\kappa\}}(\kappa'(\overline{a_{\kappa'}})) + \kappa(a_1,\ldots,b,\ldots,a_n)$} $\types$ $\Delta$)  |
                    a = freshVar, $\overline{a_{\kappa'}}$ = $\overline{\text{freshVar}}$, $(a)_{i=1}^n$ = $\overline{\text{freshVar}}$,
                    $m_i \in$ ms,
                    ($\Gamma$ $\cup$ {$m_i$ : b} | c $\types$ $\Delta$) $\in$ inferL($\Gamma$, m, $\Delta$)
                } $\cup$ { ($\Gamma$ $\cup$ {M : a} | c $\cup$ {Ok $\subtype$ b} $\types$ $\Delta$)  |  a = freshVar,
                    $m_i \in$ ms,
                    ($\Gamma$ $\cup$ {$m_i$ : b} | c $\types$ d) $\in$ inferL($\Gamma$, m, $\Delta$)
                } $\cup$ { ($\Gamma$ $\cup$ {M : a} | {$a \subtype b_1 \to b_2$} $\types$ $\Delta$)  | a = freshVar, $b_1$ = freshVar, $b_2$ = freshVar } 
                $\cup$ { ($\Gamma$ $\cup$ {M : a} | {$a \subtype b_1 \from b_2$} $\types$ $\Delta$)  | a = freshVar, $b_1$ = freshVar, $b_2$ = freshVar } 
                $\cup$ { ($\Gamma$ $\cup$ {M : a} | {a $\subtype$ $\Sigma_{c\in\mathcal{C}\backslash \kappa}(c(\overline{a}))$} $\types$ $\Delta$)  | a = freshVar, $\overline{a}$ = $\overline{\text{freshVar}}$ }
            Match q {|$_{i=1}^k$ p$_i$ -> P$_i$} -> { ($\Gamma$ $\cup$ {M : a} | c $\types$ $\Delta$)  |
                    a = freshVar,  $(a_i, b_i)_{i=1}^k$ = $\overline{\text{freshVar}}$,  $(b_{(i,x)})_{(i, x)\in [1..k]\times\fv{(p_i)}}$ = $\overline{\text{freshVar}}$,
                    ($\Gamma$ $\cup$ {(q, $P_i$) : $a'_i$} | $c_i$ $\types$ $\Delta$)$_{i=1}^k$ $\in \Pi_{i=1}^k$(inferL($\Gamma$, (q, $P_i$), $\Delta$)),
                    ($\Gamma$ $\cup$ {$P_i$ : $a_{(i, x)}$} | $c'_{(i, x)}$ $\types$ $\Delta$)$_{(i, x)}$ $\in \Pi_{(i, x) \in [1..k]\times\fv{(p_i)}}$(inferL($\Gamma$, $P_i$, {x : $b_{(i,x)}$})),
                    c = $\bigcup_{i=1}^k(c_i \cup \{a \subtype a_i, (a_i, b_i) \subtype a'_i, p_i[b_{(i,x)}/x \mid x\in\fv{(p_i)}] \subtype b_i\})$ $\cup$ 
                        $\bigcup_{(i, x) \in [1..k]\times\fv{(p_i)}}(c'_{(i, x)} \cup \{a \subtype a_{(i,x)}\})$
                }
\end{lstlisting}
   \caption{Type Inference Algorithm - InferL}
   \label{fig:inferL}
\end{figure}

\begin{theorem}[Soundness of the Inference Algorithm]
    Let $\Gamma$ and $\Delta$ be strongly consistent variable environments, $M$ be a term, $A$ be a type, $\sigma$ be a substitution from type variables to types, and $C, C'$ be sets of constraints. Then:
    \begin{enumerate}
        \item $(\Gamma,\, M : A \mid C \types \Delta) \in$ InferL($\Gamma$, $M$, $\Delta$) $\Rightarrow$ $\forall \sigma. (\Gamma\sigma,\, M : A\sigma \mid C\sigma \types_2 \Delta\sigma)$
        \item $(\Gamma \mid C \types M : A) \in$ InferR($\Gamma$, $M$) $\Rightarrow$ $\forall \sigma. (\Gamma\sigma \mid C\sigma \types_2 M : A\sigma)$
    \end{enumerate}
\end{theorem}
\begin{proof}
    We proceed by induction on the shape of the term $M$.\\
    \begin{description}
        \item[Specially treated rule] {
            In \lstinline|InferL| there is a specially generated set that does not fall inline nicely with the others. The section of interest is:
            \begin{lstlisting}
                { $\Gamma$ $\cup$ {M : a} | {Ok $\subtype$ b} $\types$ $\Delta$  |
                    a = freshVar,
                    (x : b) $\in$ $\Delta$
                }
            \end{lstlisting}
            This set is the judgment(s) provable by $\rlnm{VarK2}$. This case is proven trivially but it is important to note that this judgment is produced by all terms $M$ if the delta of the judgment is a variable. As this case is trivially sound, it will be ignored in future cases despite being generated in all of them.
            }
        \item[$\rlnm{Var}$] {
            First, \lstinline|InferR|. The section of interest is:
            \begin{lstlisting}[language=Haskell]
                { $\Gamma$ | {a $\subtype$ b} $\types$ M : b  |
                    b = freshVar,
                    a = $\begin{cases}t & (x : t)\in \Gamma\\ \okty & \text{otherwise}\end{cases}$
                } $\cup$ { $\Gamma$ | C[$\overline{b}/\overline{a}$] $\cup$ {a[$\overline{b}/\overline{a}$] $\subtype$ a'} $\types$ M : a'  |
                    (x : $\forall$ $\overline{a}$. C $\Rightarrow$ a) $\in$ $\Gamma$,
                    $\overline{b}$ = $\overline{\text{freshVar}}$,
                    a' = freshVar
                }
            \end{lstlisting}
            The second set creates a judgment for every typing of the variable at a (constrained) type scheme.\\
            Fix a generated judgment, $\Gamma \mid C[\overline{b}/\overline{a}] \cup {a[\overline{b}/\overline{a}] \subtype a'} \types M : a'$.\\
            Let $\sigma$ be a type variable substitution and $C'$ be the inferred constraint set, which is defined by $C[\overline{b}/\overline{a}] \cup {a[\overline{b}/\overline{a}] \subtype a'}$.\\
            Then, by $\rlnm{Inst2}$, we have $\Gamma\sigma \mid C'\sigma \types_2 x : a'\sigma$.\\[1.5mm]
            The first set generates judgments for the $\rlnm{Var2}$ and $\rlnm{VarK2}$ rules.\\
            Let $\sigma$ be a type variable substitution.\\
            By $\Gamma$ being a strongly consistent variable environment, we have either $(x : t) \in \Gamma$ or $(x : t) \notin \Gamma$.\\
            If $(x : t) \in \Gamma$ then, by $\rlnm{Var2}$, we have $\Gamma\sigma,\, x : t\sigma \mid \{t\sigma \subtype b\sigma\} \types_2 x : b\sigma$.\\[1.25mm]
            If $(x : t) \notin \Gamma$ then, by $\rlnm{VarK2}$, we have $\Gamma\sigma \mid \{\okty \subtype b\sigma\} \types_2 x : b$.\\[2mm]
            Now, \lstinline|InferL|. The section of interest is:
            \begin{lstlisting}[language=Haskell]
                { $\Gamma$ $\cup$ {M : a} | {a $\subtype$ b} $\types$ $\Delta$  |
                    a = freshVar,
                    (x : b) $\in$ $\Delta$
                }
            \end{lstlisting}
            Let $\sigma$ be a type variable substitution. Then, by $\rlnm{Var2}$, we have that $\Gamma\sigma,\, x : a\sigma \mid \{a\sigma \subtype b\sigma\} \types_2 x : b \sigma$.
            }
        \item[$\rlnm{App}$] {
            First, \lstinline|InferR|. The section of interest is:
            \begin{lstlisting}[language=Haskell]
                { $\Gamma$ | c1 $\cup$ c2 $\cup$ {b $\subtype$ c -> a} $\types$ M : a  |
                    ($\Gamma$ | c1 $\types$ P : b) $\in$ inferR($\Gamma$, P),
                    ($\Gamma$ | c2 $\types$ Q : c) $\in$ inferR($\Gamma$, Q),
                    a = freshVar
                }
            \end{lstlisting}
            By the inductive hypothesis (applied to lines 2 and 3), we have:
            \begin{enumerate}
                \item $\forall \sigma. \Gamma\sigma \mid C_1\sigma \types_2 P : b\sigma$
                \item $\forall \sigma. \Gamma\sigma \mid C_2\sigma \types_2 Q : c\sigma$ 
            \end{enumerate}
            Let $\sigma$ be a type variable substitution and the constraint set $C$ be defined by $C_1\sigma \cup C_2\sigma \cup \{b\sigma \subtype c\sigma \to a\sigma\}$.\\
            Then, by $\rlnm{AppR2}$, we have $\Gamma\sigma \mid C' \types_2 (P\ Q) : a\sigma$.\\[1.5mm]
    
            Now, \lstinline|InferL|. The section of interest is:
            \begin{lstlisting}[language=Haskell]
                { $\Gamma$ $\cup$ {M : a} | c1 $\cup$ c2 $\cup$ {b $\subtype$ c $\from$ a} $\types$ $\Delta$  |
                    a = freshVar,
                    ($\Gamma$ | c1 $\types$ P : b) $\in$ inferR($\Gamma$, P),
                    ($\Gamma$ $\cup$ {Q : c} | c2 $\types$ d) $\in$ inferL($\Gamma$, Q, $\Delta$)
                } $\cup$ { $\Gamma$ $\cup$ {M : a} | c $\cup$ {Ok $\from$ a $\subtype$ b} $\types$ $\Delta$  |
                    a = freshVar,
                    ($\Gamma$ $\cup$ {P : b} | c $\types$ $\Delta$) $\in$ inferL($\Gamma$, P, $\Delta$)
                }
            \end{lstlisting}
            The first set is the judgments provable by the $\rlnm{AppL2}$ rule.\\
            By the inductive hypothesis (applied to lines 3 and 4), we have:
            \begin{enumerate}
                \item $\forall \sigma. \Gamma\sigma \mid C_1\sigma \types_2 P : b\sigma$
                \item $\forall \sigma. \Gamma\sigma,\, Q : c\sigma \mid C_2\sigma \types_2 \Delta\sigma$
            \end{enumerate}
            Let $\sigma$ be a type variable substitution and the constraint set $C$ be the inferred constraints defined by $C_1 \cup C_2 \cup \{b \subtype c \from a\}$.\\
            Then, by $\rlnm{AppL2}$, we have $\Gamma\sigma,\, (P\ Q) : a\sigma \mid C\sigma \types_2 \Delta\sigma$.\\[1.25mm]
            The second set is the judgments provable by the $\rlnm{FunK2}$ rule.\\
            By the inductive hypothesis (applied to line 7), we have: $\forall \sigma. \Gamma\sigma,\, P : b\sigma \mid C_1\sigma \types_2 \Delta\sigma$.\\
            Let $\sigma$ be a type variable substitution and constraint set $C$ be the inferred constraints defined by $C_1 \cup \{\okty \from a \subtype b\}$.\\
            Then, by $\rlnm{FunK2}$, we have $\Gamma\sigma,\, (P\ Q) : a\sigma \mid C\sigma \types_2 \Delta\sigma$.
            }
        \item[$\rlnm{Fix}$] {
            First, \lstinline|InferR|. The section of interest is:
            \begin{lstlisting}[language=Haskell]
                { ($\Gamma$ | c $\cup$ {t -> b $\subtype$ a} $\types$ M : a)  |  
                    a = freshVar, 
                    t = freshVar,
                    ($\Gamma$ $\cup$ {x : t, f : a} | c $\types$ N : b) $\in$ inferR($\Gamma$ $\cup$ {x : t, f : a}, N)
                } $\cup$ { ($\Gamma$ | c $\cup$ {t $\from$ b $\subtype$ a} $\types$ M : a)  |  
                    a = freshVar, 
                    t = freshVar,
                    ($\Gamma$ $\cup$ {N : b, f : a} | c $\types$ x : t) $\in$ inferL($\Gamma$ $\cup$ {f : a}, N, {x : t})
                }
            \end{lstlisting}
            \ch{
                The first set is the judgments provable by the $\rlnm{FixsR2}$ rule.\\
                By the inductive hypothesis (applied to line 4), we have $\forall \sigma. \Gamma\sigma,\, x : t\sigma,\, f : a\sigma \mid C_1\sigma \types_2 N : b\sigma$.\\
                Let $\sigma$ be a type variable substitution and constraint set $C$ be the inferred constraints defined by $C_1 \cup \{t \to b \subtype a\}$.\\
                Then, by $\rlnm{FixsR2}$, we have $\Gamma\sigma \mid C\sigma \types_2 \fixtm{f(x)}{n} : a\sigma$.\\[1.25mm]
                The second set is the judgments provable by the $\rlnm{FixnR2}$ rule.\\
                By the inductive hypothesis (applied to line 8), we have $\forall \sigma. \Gamma\sigma,\, N : b\sigma,\, f : a \mid C_1\sigma \types_2 x : t\sigma$.\\
                Let $\sigma$ be a type variable substitution and witness constraint set $C$ be the inferred constraints defined by $C_1 \cup \{t \from b \subtype a\}$.\\
                Then, by $\rlnm{FixnR2}$, we have $\Gamma\sigma \mid C\sigma \types_2 \fixtm{(x)}{N} : a\sigma$.\\[1.5mm]
            }
            Now, \lstinline|InferL|. The section of interest is:
            \begin{lstlisting}[language=Haskell]
                { $\Gamma$ $\cup$ {M : a} | {a $\subtype$ $\Sigma_{c\in\mathcal{C}}$(c($\overline{a}$))} $\types$ $\Delta$  |
                    a = freshVar,
                    $\overline{a}$ = $\overline{\text{freshVar}}$
                }
            \end{lstlisting}
            This case is trivially provable by $\rlnm{FixDL2}$.
            }
        \item[$\rlnm{Cons}$] {
            First, \lstinline|InferR|. The section of interest is:
            \begin{lstlisting}[language=Haskell]
                { $\Gamma$ | $\bigcup_{i=1}^n(c_i)$ $\cup$ {c($a_1,\ldots,a_n$) $\subtype$ a} $\types$ M : a  |
                    a = freshVar,
                    ($\Gamma$ | $c_i \types m_i : a_i$)$_{i=1}^n$ $\in \Pi_{m\in \text{ms}}$(inferR($\Gamma$, m))
                } \cup { $\Gamma$ | {c $\subtype$ a} $\types$ M : a  |
                    a = freshVar,
                    ms = { }
                }
            \end{lstlisting}
            This is the set of judgments provable with the $\rlnm{CnsR2}$ rule.\\
            The first set is for constructors with more than zero arguments (e.g. Cons and Succ), and the second set is for nullary constructors (e.g. Nil and Zero).\\
            By the inductive hypothesis (applied to line 3), we have 
            \[
                \forall 1 \leq i \leq n.\ \forall \sigma.\ \Gamma\sigma \mid C_i\sigma \types_2 m_i : a_i\sigma
            \]
            Let $\sigma$ be a type variable substitution and constraint set $C$ be the inferred constraints defined by $\bigcup_{i=1}^n(C_i) \cup \{c(a_1,\ldots,a_n) \subtype a\}$.\\
            Then, for every combination of typings for the constructor's arguments, by $\rlnm{CnsR2}$, we have $\Gamma\sigma \mid C\sigma \types_2 c(M_1,\ldots,M_n) : a\sigma$.\\[1.25mm]
            The second set is proven sound trivially by $\rlnm{CnsR2}$.\\[2mm]

            Now, \lstinline|InferL|. The section of interest is:
            \begin{lstlisting}[language=Haskell]
                { $\Gamma$ $\cup$ {M : a} | c $\cup$ {a $\subtype$ $\Sigma_{\kappa'\in\mathcal{C}\backslash\{\kappa\}}(\kappa'(\overline{a_{\kappa'}})) + \kappa(a_1,\ldots,b,\ldots,a_n)$} $\types$ $\Delta$  |
                    a = freshVar,  $\overline{a_{\kappa'}}$ = $\overline{\text{freshVar}}$,  $(a)_{i=1}^n$ = $\overline{\text{freshVar}}$,
                    $m_i \in$ ms,
                    ($\Gamma$ $\cup$ {$m_i$ : b} | c $\types$ $\Delta$) $\in$ inferL($\Gamma$, m, $\Delta$)
                } $\cup$ { $\Gamma$ $\cup$ {M : a} | c $\cup$ {Ok $\subtype$ b} $\types$ $\Delta$  |
                    a = freshVar,
                    $m_i \in$ ms,
                    ($\Gamma$ $\cup$ {$m_i$ : b} | c $\types$ $\Delta$) $\in$ inferL($\Gamma$, m, $\Delta$)
                } $\cup$ { ($\Gamma$ $\cup$ {M : a} | {$a \subtype b_1 \to b_2$} $\types$ $\Delta$)  | a = freshVar, $b_1$ = freshVar, $b_2$ = freshVar } 
                $\cup$ { ($\Gamma$ $\cup$ {M : a} | {$a \subtype b_1 \from b_2$} $\types$ $\Delta$)  | a = freshVar, $b_1$ = freshVar, $b_2$ = freshVar } 
                $\cup$ { ($\Gamma$ $\cup$ {M : a} | {a $\subtype$ $\Sigma_{c\in\mathcal{C}\backslash \kappa}(c(\overline{a}))$} $\types$ $\Delta$)  |  a = freshVar, $\overline{a}$ = $\overline{\text{freshVar}}$ }
            \end{lstlisting}
            These sets are the judgments provable by the: $\rlnm{CnsL2}$, $\rlnm{CnsK2}$, $\rlnm{CnsDL21}$, $\rlnm{CnsDL22}$, and $\rlnm{CnsDL23}$, respectively.\\
            The last three sets are proven to be sound trivially by their respective rules.\\[1.25mm]
            
            For the first set, fix an inferred judgment $\Gamma,\, M : a \mid C_i \cup \{a \subtype \Sigma_{\kappa'\in\mathcal{C}\backslash\{\kappa\}}(\kappa'(\overline{a_{\kappa'}})) + \kappa(a_1,\ldots,b,\ldots,a_n) \types \Delta\}$.\\
            Then there exists an argument to the constructor, $m_i$ such that $\Gamma,\, m_i : b \mid C_i \types \Delta$ is inferred by \lstinline|InferL|($\Gamma$, $m_i$, $\Delta$).\\
            By the inductive hypothesis, we have $\forall \sigma. \Gamma\sigma\,\, m_i : b\sigma \mid C_i\sigma \types_2 \Delta\sigma$.\\
            Let $\sigma$ be a type variable substitution and constraint set $C$ be the inferred constraints defined by $C_i \cup \{a \subtype \Sigma_{\kappa'\in\mathcal{C}\backslash\{\kappa\}}(\kappa'(\overline{a_{\kappa'}})) + \kappa(a_1,\ldots,b,\ldots,a_n)\}$.\\
            Then, by $\rlnm{CnsL2}$, we have that $\Gamma\sigma,\, M : a\sigma \mid C\sigma \types_2 \Delta\sigma$.\\[1.25mm]
            
            For the second set, fix an inferred judgment $\Gamma,\, M : a \mid C_i \cup \{\okty \subtype b\} \types \Delta$.\\
            Then there exists an argument to the constructor, $m_i$ such that $\Gamma,\, m_i : b \mid C_i \types \Delta$ is inferred by \lstinline|InferL|($\Gamma$, $m_i$, $\Delta$).\\
            By the inductive hypothesis, we have $\forall \sigma. \Gamma\sigma,\, m_i : b\sigma \mid C_i\sigma \types_2 \Delta\sigma$ and $C'_i \types C_i\sigma$.\\
            Let $\sigma$ be a type variable substitution and constraint set $C$ be the inferred constraints defined by $C_i \cup \{\okty \subtype b\}$.\\
            Then, by $\rlnm{CnsK2}$, we have $\Gamma\sigma,\, M : a\sigma \mid C\sigma \types_2 \Delta\sigma$.
            }
        \item[$\rlnm{Match}$] {
            First, \lstinline|InferR|. The section of interest is:
            \begin{lstlisting}[language=Haskell]
                { $\Gamma$ | $\bigcup_{i=1}^k(c_i \cup \{b \subtype p_i[b_x/x \mid x\in\fv(p_i)], a_i \subtype a\})$ $\cup$ c $\types$ M : a  |
                    a = freshVar,
                    $\overline{b_x}$ = $\overline{\text{freshVar}}$,
                    ($\Gamma$ | c $\types$ q : b) $\in$ inferR($\Gamma$, Q),
                    ($\Gamma$ $\cup$ {$\overline{(x : b_x)}$} | $c_i \types P_i : a_i$)$_{i=1}^k$ $\in \Pi_{i=1}^k$(inferR($\Gamma$ $\cup$ {$\overline{(x : b_x)}$}, $P_i$))
                }
            \end{lstlisting}
            By the inductive hypothesis (applied to lines 4 and 5), we have:
            \begin{enumerate}
                \item $\forall \sigma. \Gamma\sigma \mid C_0\sigma \types_2 q : b\sigma$
                \item $\forall 1 \leq i \leq n. \forall \sigma. \Gamma\sigma,\, \overline{(x : b')} \mid C_i\sigma \types_2 p_i : a_i\sigma$
            \end{enumerate}
            Let $\sigma$ be a type variable substitution and constraint set $C$ be the inferred constraints defined by $\bigcup_{i=1}^k(C_i \cup \{b \subtype p_i[b_x/x \mid x\in\fv(p_i)], a_i \subtype a\}) \cup C_0$.\\[0.5mm]
            Then, for all combinations of typings of case bodies, by $\rlnm{MchR2}$, we have $\Gamma\sigma \mid C\sigma \types_2 M : a\sigma$.\\[2mm]
            
            Now, \lstinline|InferL|. The section of interest is:
            \begin{lstlisting}[language=Haskell]
                { $\Gamma$ $\cup$ {M : a} | c $\types$ $\Delta$  |
                    a = freshVar,
                    $(a_i, b_i)_{i=1}^k$ = $\overline{\text{freshVar}}$,
                    $(b_{(i,x)})_{(i, x)\in [1..k]\times\fv{(p_i)}}$ = $\overline{\text{freshVar}}$,
                    ($\Gamma$ $\cup$ {(q, $P_i$) : $a'_i$} | $c_i$ $\types$ $\Delta$)$_{i=1}^k$ $\in \Pi_{i=1}^k$(inferL($\Gamma$, (Q, $P_i$), $\Delta$)),
                    ($\Gamma$ $\cup$ {$P_i$ : $a_{(i, x)}$} | $c'_{(i, x)}$ $\types$ $\Delta$)$_{(i, x)}$ $\in \Pi_{(i, x) \in [1..k]\times\fv{(p_i)}}$(inferL($\Gamma$, $P_i$, {x : $b_{(i,x)}$})),
                    c = $\bigcup_{i=1}^k(c_i \cup \{a \subtype a_i, (a_i, b_i) \subtype a'_i, p_i[b_{(i,x)}/x \mid x\in\fv{(p_i)}] \subtype b_i\})$ $\cup$ 
                        $\bigcup_{(i, x) \in [1..k]\times\fv{(p_i)}}(c'_{(i, x)} \cup \{a \subtype a_{(i,x)}\})$
                }
            \end{lstlisting}
            By the inductive hypothesis (applied to lines 5 and 6), we have:
            \begin{enumerate}
                \item $\forall 1 \leq i \leq k. \forall \sigma. \Gamma\sigma,\, (Q, P_i) : a'_i\sigma \mid C_i\sigma \types_2 \Delta\sigma$
                \item $\forall 1 \leq i \leq k. \forall x\in\fv(p_i). \forall \sigma. \Gamma\sigma,\, P_i : a_{(i, x)}\sigma \mid C'_{(i, x)}\sigma \types_2 \Delta\sigma$
            \end{enumerate}
            Let $\sigma$ be a type variable substitution and constraint set $C$ be defined by:
            \begin{align*}
                C := &\bigcup_{i=1}^k(C_i \cup \{a \subtype a_i, (a_i, b_i) \subtype a'_i, p_i[b_{(i,x)}/x \mid x\in\fv{(p_i)}] \subtype b_i\}) \cup\\
                      &\bigcup_{(i, x) \in [1..k]\times\fv{(p_i)}}(C'_{(i, x)} \cup \{a \subtype a_{(i,x)}\})
            \end{align*}
            Then, for all combinations of typings, by $\rlnm{MchL2}$, we have $\Gamma\sigma,\, M : a\sigma \mid C\sigma \types_2 \Delta\sigma$.
            }
    \end{description}
\end{proof}

\begin{theorem}[Completeness of the Inference Algorithm]
    Let $\tau, \sigma$ be substitutions from type variables to types such that $\sigma$ extends $\tau$, $\Gamma\tau, \Delta\tau$ be strongly consistent variable environments, $M$ be a term, $A$ be a type, and $C, C'$ be sets of constraints. Then:
    \begin{enumerate}
        \item $\Gamma\tau,\, M : A \mid C \types_2 \Delta\tau$ $\Rightarrow$ $\exists \sigma. \exists (\Gamma,\, M : A' \mid C' \types \Delta) \in$ InferL($\Gamma$, $M$, $\Delta$). $(A = A'\sigma) \wedge (C \types C'\sigma)$
        \item $\Gamma\tau \mid C \types_2 M : A$ $\Rightarrow$ $\exists \sigma. \exists (\Gamma \mid C' \types M : A') \in$ InferR($\Gamma$, $M$). $(A = A'\sigma) \wedge (C \types C'\sigma)$
    \end{enumerate}
\end{theorem}
\begin{proof}
    We proceed by induction on the derivation of $\types_2$.\\
    \begin{description}
        \item[$\rlnm{Inst2}$] {
            For some $\tau$, we have $\Gamma\tau,\, f : \forall \vec{a}. C' \Rightarrow A \mid C \types_2 f : A'$, $C \types C'[\vec{B}/\vec{a}]$, and $C \types A[\vec{B}/\vec{a}] \subtype A'$.\\
            Observe:
            \begin{lstlisting}[language=Haskell]
                { g | c[$\overline{B}/\overline{a}$] $\cup$ {A[$\overline{B}/\overline{a}$] $\subtype$ a'} $\types$ M : a'  |
                    (x : $\forall$ $\overline{a}$. c $\Rightarrow$ A) $\in$ g,
                    $\overline{B}$ = $\overline{\text{freshVar}}$,
                    a' = freshVar
                }
            \end{lstlisting}
            Let $\sigma$ be a type variable substitution defined by $\tau \cup [\overline{B} \mapsto \vec{B}, a' \mapsto A'$. As the inferred set is a singleton, the witness judgment is clear. Then, by definition of $\sigma$, we have $A' = a'\sigma$ and $C \types (C'[\overline{B}/\overline{a}] \cup \{A[\overline{B}/\overline{a}] \subtype a'\})\sigma$, as required.
            }
        \item[$\rlnm{Var2}$] {
            For some $\tau$, we have $\Gamma\tau,\, x : A \mid C \types_2 x : B$.\\
            By the symmetry of the rule, we need to show both cases of the the rule (both \lstinline|InferL| and \lstinline|InferR|).\\
            In \lstinline|InferR|, observe:
            \begin{lstlisting}[language=Haskell]
                { g | {a $\subtype$ b} $\types$ M : b  |
                    b = freshVar,
                    a = $\begin{cases}t & (x : t)\in g\\ \okty & \text{otherwise}\end{cases}$
                }
            \end{lstlisting}
            As $x : A$ is in the context, the variable on line 3 is $a = A$. So the inferred judgment is: $\Gamma,\, x : A \mid \{A \subtype b\} \types x : b$. Thus, taking $\sigma := \tau \cup [b \mapsto B]$ trivially satisfies the requirements.\\
            In \lstinline|InferL|, observe:
            \begin{lstlisting}[language=Haskell]
                { g $\cup$ {M : a} | {a $\subtype$ b} $\types$ d  |
                    a = freshVar,
                    (x : b) $\in$ d
                }
            \end{lstlisting}
            As $x : B$ is in $\Delta$, the judgment inferred is: $\Gamma,\, x : a \mid \{a \subtype B\} \types x : B$. Thus, taking $\sigma := \tau \cup [a \mapsto A]$ trivially satisfies the requirements.
            }
        \item[$\rlnm{VarK2}$] {
            For some $\tau$, we have $\Gamma\tau \mid C \types_2 x : A$ and $C \types \okty \subtype A$.\\
            As this case only requires a variable typing in the delta position, the non-trivial term can be either on the left or the right of the judgment.\\
            In \lstinline|InferR|, observe:
            \begin{lstlisting}[language=Haskell]
                { g | {a $\subtype$ b} $\types$ M : b  |
                    b = freshVar,
                    a = $\begin{cases}t & (x : t)\in g\\ \okty & \text{otherwise}\end{cases}$
                }
            \end{lstlisting}
            If $(x : t) \in \Gamma$ then the inferred judgment is: $\Gamma \mid \{t \subtype b\} \types x : b$. Taking $\sigma := \tau \cup [b \mapsto A]$ trivially satisfies the requirements.\\
            If $(x : t) \notin \Gamma$ then the inferred judgment is: $\Gamma \mid \{\okty \subtype b\} \types x : b$. Taking $\sigma := \tau \cup [b \mapsto A]$ trivially satisfies the requirements.\\[1.5mm]
    
            In \lstinline|InferL|, observe:
            \begin{lstlisting}[language=Haskell]
                { g $\cup$ {M : a} | {Ok $\subtype$ b} $\types$ d  |
                    a = freshVar,
                    (x : b) $\in$ d
                }
            \end{lstlisting}
            As $x : A$ is the delta, the inferred judgment (for all terms $M$) is: $\Gamma,\, M : a \mid \{\okty \subtype b\} \types x : b$.\\
            As $M : a$ is included in the context of the judgment, $a\tau$ is a concrete type.\\
            Thus, taking $\sigma := \tau \cup [b \mapsto A]$ trivially satisfies the requirements.
            }
        \item[$\rlnm{FixsR2}$] {
            \ch{
                For some $\tau$, we have $\Gamma\tau \mid C \types_2 \fixtm{f(x)}{M} : A$,\enspace $\Gamma\tau,\, f : A,\, x : B_1 \mid C \types_2 M : B_2$ and $C \types B_1 \to B_2 \subtype A$.\\
                By the inductive hypothesis, there exists a $\sigma$ that extends $\tau$ and an inferred judgment $(\Gamma,\, f : a,\, x : b_1 \mid C'_0 \types_2 M : b_2)\in \text{InferR}(\Gamma \cup \{x : b_1\},\, M)$ such that $B_2 = b_2\sigma$ and $C \types C'_0\sigma$.\\
            }
            Observe this extract from \lstinline|InferR| in the abstraction case:
            \begin{lstlisting}[language=Haskell]
                { g | c $\cup$ {t -> b $\subtype$ a} $\types$ M : a  |
                    a = freshVar,
                    t = freshVar,
                    (g $\cup$ {x : t, f : a} | c $\types$ n : b) $\in$ inferR(g $\cup$ {x : t, f : a}, n)
                }
            \end{lstlisting}
            With this inferred judgment, extending $\sigma$ to the witness $\sigma' = \sigma \cup [b_1 \mapsto B_1, a \mapsto A]$ clearly satisfies the requirements of $A = a\sigma'$ and $C \types (C'_0 \cup \{b_1 \to b_2 \subtype a\})\sigma'$. 
            }
        \item[$\rlnm{FixnR2}$] {
            \ch{
                For some $\tau$, we have $\Gamma\tau \mid C \types_2 \fixtm{f(x)}{M} : A$,\enspace $\Gamma\tau,\, f : A,\, M : B_1 \mid C \types_2 x : B_2$ and $C \types B_2 \from B_1 \subtype A$.\\
                By the inductive hypothesis, there exists a $\sigma$ that extends $\tau$ and an inferred judgment $(\Gamma,\, f : a,\, M : b_1 \mid C'_0 \types_2 x : b_2) \in \text{InferL}(\Gamma,\, M,\, \{x : b_2\})$ such that $B_1 = b_1\sigma$ and $C \types C'_0\sigma$.\\
            }
            Observe this extract from \lstinline|InferR| in the abstraction case:
            \begin{lstlisting}[language=Haskell]
                { g | c $\cup$ {t $\from$ b $\subtype$ a} $\types$ M : a  |
                    a = freshVar,
                    t = freshVar,
                    (g $\cup$ {n : b, f : a} | c $\types$ x : t) $\in$ inferL(g $\cup$ {f : a}, n, {x : t})
                }
            \end{lstlisting}
            With this inferred judgment, extending $\sigma$ to the witness $\sigma' = \sigma \cup [b_2 \mapsto B_2, a \mapsto A]$ clearly satisfies the requirements of $A = a\sigma'$ and $C \types (C'_0 \cup \{b_2 \from b_1 \subtype a\})\sigma'$.
            }
        \item[$\rlnm{FixDL2}$] {
            For some $\tau$, we have $\Gamma\tau,\, \fixtm{f(x)}{M} : A \mid C \types_2 \Delta\tau$ and $C \types A \subtype \Sigma_{c\in\mathcal{C}}(c(\vec{A_c}))$.\\
            Observe this extract from \lstinline|InferL| in the abstraction case:
            \begin{lstlisting}[language=Haskell]
                { g $\cup$ {M : a} | {a $\subtype$ $\Sigma_{c\in\mathcal{C}}$(c($\overline{a}$))} $\types$ d  |
                    a = freshVar,
                    $\overline{a}$ = $\overline{\text{freshVar}}$
                }
            \end{lstlisting}
            Define the witness $\sigma$ by $\sigma := \tau \cup [\overline{a} \mapsto \vec{A_c}, a \mapsto A]$.\\
            Then, $A = a\sigma$ and $\{a \subtype \Sigma_{c\in\mathcal{C}}$(c($\overline{a}))\}\sigma = \{A \subtype \Sigma_{c\in\mathcal{C}}(c(\vec{A_c}))\}$ and so $C \types C'\sigma$ as required.
            }
        \item[$\rlnm{AppR2}$] {
            For some $\tau$, we have $\Gamma\tau \mid C \types_2 (l\ r) : A$,\enspace $\Gamma\tau \mid C \types_2 l : B_1$,\enspace $\Gamma\tau \mid C \types_2 r : B_2$, and $C \types B_1 \subtype B_2\to A$.\\
            By the inductive hypothesis, we have:
            \begin{enumerate}
                \item $\exists \sigma_1. \exists (\Gamma \mid C'_1 \types l : b_1)\in\text{InferR}(\Gamma,\, l). B_1 = b_1\sigma_1 \wedge C \types C'_1\sigma_1$
                \item $\exists \sigma_2. \exists (\Gamma \mid C'_2 \types r : b_2)\in\text{InferR}(\Gamma,\, r). B_2 = b_2\sigma_2 \wedge C \types C'_2\sigma_2$
            \end{enumerate}
            Observe this extract from \lstinline|InferR| in the application case:
            \begin{lstlisting}[language=Haskell]
                { g | c1 $\cup$ c2 $\cup$ {b $\subtype$ c -> a} $\types$ M : a  |
                    (g | c1 $\types$ l : b) $\in$ inferR(g, l),
                    (g | c2 $\types$ r : c) $\in$ inferR(g, r),
                    a = freshVar
                }
            \end{lstlisting}
            With this inferred judgment, and the observation that $\sigma_1$ and $\sigma_2$ only share variables given by $\tau$ (i.e. $\sigma_1\cap \sigma_2 = \tau$), we define a witness $\sigma = \sigma_1\cup\sigma_2\cup [a \mapsto A]$. Thus, $A = a\sigma$ and $C \types (C_1\cup C_2 \cup \{b_1 \subtype b_2 \to a\})\sigma$ as required.
            }
        \item[$\rlnm{AppL2}$] {
            For some $\tau$, we have $\Gamma\tau,\, (l\ r) : A \mid C \types_2 \Delta\tau$,\enspace $\Gamma\tau \mid C \types_2 l : B_1$,\enspace $\Gamma\tau,\, r : B_2 \mid C \types_2 \Delta\tau$, and $C \types B_1 \subtype B_2 \from A$.\\
            By the inductive hypothesis, we have:
            \begin{enumerate}
                \item $\exists \sigma_1. \exists (\Gamma \mid C'_1 \types l : b_1)\in\text{InferR}(\Gamma,\, l). B_1 = b_1\sigma_1 \wedge C \types C'_1\sigma_1$
                \item $\exists \sigma_2. \exists (\Gamma,\, r : b_2 \mid C'_2 \types \Delta)\in\text{InferL}(\Gamma,\, r,\, \Delta). B_2 = b_2\sigma_2 \wedge C \types C'_2\sigma_2$
            \end{enumerate}
            Observe this extract from \lstinline|InferL| in the application case:
            \begin{lstlisting}[language=Haskell]
                { g $\cup$ {M : a} | c1 $\cup$ c2 $\cup$ {b $\subtype$ c $\from$ a} $\types$ d  |
                    a = freshVar,
                    (g | c1 $\types$ l : b) $\in$ inferR(g, l),
                    (g $\cup$ {r : c} | c2 $\types$ d) $\in$ inferL(g, r, d)
                }
            \end{lstlisting}
            With this inferred judgment, and the observation that $\sigma_1$ and $\sigma_2$ only share variables given by $\tau$ (i.e. $\sigma_1\cap \sigma_2 = \tau$), we define a witness $\sigma = \sigma_1\cup\sigma_2\cup [a \mapsto A]$. Thus, $A = a\sigma$ and $C \types (C_1\cup C_2 \cup \{b_1 \subtype b_2 \from a\})\sigma$ as required.
            }
        \item[$\rlnm{AppK2}$] {
            For some $\tau$, we have $\Gamma\tau,\, (l\ r) : A \mid C \types_2 \Delta\tau$,\enspace $\Gamma\tau,\, l : B \mid C \types_2 \Delta\tau$, and $C \types \okty \from A \subtype B$.\\
            By the inductive hypothesis, there exists a $\sigma_1$ extending $\tau$ and an inferred judgment $(\Gamma,\, l : b \mid C'_1 \types \Delta)\in\text{InferL}(\Gamma,\, l,\, \Delta)$ such that $B = b\sigma_1$ and $C \types C'_1\sigma_1$.\\
            Observe this extract from \lstinline|InferL| in the application case:
            \begin{lstlisting}[language=Haskell]
                { g $\cup$ {M : a} | c $\cup$ {Ok $\from$ a $\subtype$ b} $\types$ d  |
                    a = freshVar,
                    (g $\cup$ {l : b} | c $\types$ d) $\in$ inferL(g, l, d)
                }
            \end{lstlisting}
            With this inferred judgment, we define a witness $\sigma = \sigma_1\cup [a \mapsto A]$. Thus, $A = a\sigma$ and $C \types (C_1\cup \{\okty \from a \subtype b\})\sigma$ as required.
            }
        \item[$\rlnm{CnsDL}$] {
            All of the $\rlnm{CnsDL21}$, $\rlnm{CnsDL22}$, and $\rlnm{CnsDL23}$ cases are all very similar.\\
            Observe this extract from \lstinline|InferL| in the constructor case:
            \begin{lstlisting}[language=Haskell]
                { g $\cup$ {M : a} | {$a \subtype b_1 \to b_2$} $\types$ d  |
                    a = freshVar,
                    $b_1$ = freshVar,
                    $b_2$ = freshVar
                } $\cup$ { g $\cup$ {M : a} | {$a \subtype b_1 \from b_2$} $\types$ d  |
                    a = freshVar,
                    $b_1$ = freshVar,
                    $b_2$ = freshVar
                } $\cup$ { g $\cup$ {M : a} | {a $\subtype$ $\Sigma_{c\in\mathcal{C}\backslash \kappa}(c(\overline{a}))$} $\types$ d  |
                    a = freshVar,
                    $\overline{a}$ = $\overline{\text{freshVar}}$
                }
            \end{lstlisting}
            Each of the inferred sets above are singletons and correspond to the $\rlnm{CnsDL21}$, $\rlnm{CnsDL22}$, and $\rlnm{CnsDL23}$ rules, respectively. Taking the respective witnesses to be defined by:
            \begin{enumerate}
                \item $\sigma := \tau \cup [b_1 \mapsto B_1, b_2 \mapsto B_2, a \mapsto A]$
                \item $\sigma := \tau \cup [b_1 \mapsto B_1, b_2 \mapsto B_2, a \mapsto A]$
                \item $\sigma := \tau \cup [\overline{a} \mapsto \vec{a}, a \mapsto A]$
            \end{enumerate}
            Each of the witness substitutions and judgments trivially satisfy the requirements.
            }
        \item[$\rlnm{CnsR2}$] {
            For some $\tau$, we have $\Gamma\tau \mid C \types_2 c(M_1,\ldots,M_n) : A$,\enspace $\forall 1\leq i \leq n. \Gamma\tau \mid C \types_2 M_i : A_i$, and $C \types c(A_1,\ldots,A_n) \subtype A$.\\
            By the inductive hypothesis, for all $1 \leq i \leq n$, there exists a $\sigma_i$ extending $\tau$ and an inferred judgment $(\Gamma \mid C'_i \types M_i : a_i)\in\text{InferR}(\Gamma,\, M_i)$ such that $A_i = a_i\sigma_i$ and $C \types C'_i\sigma_i$.\\
            Observe this extract from \lstinline|InferR| in the constructor case:
            \begin{lstlisting}[language=Haskell]
                { g | $\bigcup_{i=1}^n(c_i)$ $\cup$ {c($a_1,\ldots,a_n$) $\subtype$ a} $\types$ M : a  |
                    a = freshVar,
                    (g | $c_i \types m_i : a_i$)$_{i=1}^n$ $\in \Pi_{m\in \text{ms}}$(inferR(g, m))
                }
            \end{lstlisting}
            As there is clearly a tuple of length $n$ with the collection of inferred judgments that correspond to those from the inductive hypothesis on line 3, we take the inferred judgment $\Gamma \mid \bigcup_{i=1}^n(C'_i) \cup \{c(a_1,\ldots,a_n) \subtype a\} \types c(M_1,\ldots,M_n) : a$ as the required witness. We then define the witness $\sigma := \bigcup_{i=1}^n(\sigma_i) \cup [a \mapsto A]$, which is well defined as $\sigma_i \cap \sigma_j = \tau$ for all $i \neq j$ by each recursive call creating fresh variables. Thus, $A = a\sigma$ and $C \types (\bigcup_{i=1}^n(C'_i) \cup \{c(a_1,\ldots,a_n) \subtype a\})\sigma$ as required.
            }
        \item[$\rlnm{CnsL2}$] {
            For some $\tau$, we have $\Gamma\tau,\, c(M_1,\ldots,M_n) : A \mid C \types_2 \Delta\tau$,\enspace $\exists i. \Gamma\tau,\, M_i : A_i \mid C \types_2 \Delta\tau$, and $C \types A\subtype \Sigma_{d\in\mathcal{C}\backslash\{c\}} + c(A_1,\ldots,A_n)$.\\
            By the inductive hypothesis, there exists some $1 \leq i \leq n$ such that, there exists a $\sigma_i$ extending $\tau$ and an inferred judgment $(\Gamma,\, M_i : a_i \mid C'_i \types \Delta)\in\text{InferL}(\Gamma,\, M_i,\, \Delta)$ such that $A_i = a_i\sigma_i$ and $C \types C'_i\sigma_i$.\\
            Observe this extract from \lstinline|InferL| in the constructor case:
            \begin{lstlisting}[language=Haskell]
                { g $\cup$ {M : a} | c $\cup$ {a $\subtype$ $\Sigma_{\kappa'\in\mathcal{C}\backslash\{\kappa\}}(\kappa'(\overline{a_{\kappa'}})) + \kappa(a_1,\ldots,b,\ldots,a_n)$} $\types$ d  |
                    a = freshVar,
                    $\overline{a_{\kappa'}}$ = $\overline{\text{freshVar}}$,
                    $(a)_{i=1}^n$ = $\overline{\text{freshVar}}$,
                    $m_i \in$ ms,
                    (g $\cup$ {$m_i$ : b} | c $\types$ d) $\in$ inferL(g, m, d)
                }
            \end{lstlisting}
            Line 5 shows that an inferred judgment is generated for all $i$, thus it is possible to pick the correct $i$ and correct judgment in line 6 to produce the required witness judgment.\\
            Defining the witness $\sigma := \sigma_i \cup [\overline{a_\kappa'} \mapsto \vec{A_c}, (a_i)_{i\in [1..n]\backslash\{i\}} \mapsto (A_i)_{i\in [1..n]\backslash\{i\}}, b \mapsto A_i, a \mapsto A]$ satisfies the requirements of $A = a\sigma$ and $C \types (C'_i \cup \{a \subtype \Sigma_{\kappa'\in\mathcal{C}\backslash\{\kappa\}}(\kappa'(\overline{a_{\kappa'}})) + \kappa(a_1,\ldots,b,\ldots,a_n)\})\sigma$ as required.
            }
        \item[$\rlnm{CnsK2}$] {
            For some $\tau$, we have $\Gamma\tau,\, c(M_1,\ldots,M_n) : A \mid C \types_2 \Delta\tau$,\enspace $\exists i. \Gamma\tau,\, M_i : B \mid C \types_2 \Delta\tau$, and $C \types \okty \subtype B$.\\
            By the inductive hypothesis, there exists some $1 \leq i \leq n$ such that, there exists a $\sigma_i$ extending $\tau$ and an inferred judgment $(\Gamma,\, M_i : b \mid C'_i \types \Delta)\in\text{InferL}(\Gamma,\, M_i,\, \Delta)$ such that $B = b\sigma_i$ and $C \types C'_i\sigma_i$.\\
            Observe this extract from \lstinline|InferL| in the constructor case:
            \begin{lstlisting}[language=Haskell]
                { g $\cup$ {M : a} | c $\cup$ {Ok $\subtype$ b} $\types$ d  |
                    a = freshVar,
                    $m_i \in$ ms,
                    (g $\cup$ {$m_i$ : b} | c $\types$ d) $\in$ inferL(g, m, d)
                }
            \end{lstlisting}
            Line 3 shows that an inferred judgment is generated for all $i$, thus it is possible to pick the correct $i$ and correct judgment in line 4 to produce the required witness judgment.\\
            Defining the witness $\sigma := \sigma_i \cup [a \mapsto A]$ satisfies the requirements of $A = a\sigma$ and $C \types (C'_i \cup \{Ok \subtype b\})\sigma$.
            }
        \item[$\rlnm{MchR2}$] {
            For some $\tau$, we have:
            \begin{enumerate}
                \item $\Gamma\tau \mid C \types_2 \matchtm{Q}{|_{i=1}^k (p_i \mapsto P_i)} : A$
                \item $\gamma\tau \mid C \types_2 Q : B$
                \item $\Gamma\tau \cup \{x : B_x \mid x\in\fv{(p_i)}\} \mid C \types_2 P_i : A_i \; (\forall i)$
                \item $C \types B \subtype \Sigma_{i=1}^k(p_i[B_x/x \mid x\in\fv(p_i)])$
                \item $\forall i. C \types A_i \subtype A$
            \end{enumerate}
            By the inductive hypothesis, we have:
            \begin{enumerate}
                \item $\exists \sigma_0. \exists (\Gamma \mid C'_0 \types Q : b)\in\text{InferR}(\Gamma,\, Q). B = b\sigma_0 \wedge C \types C'_0 \sigma_0$
                \item $\forall 1\leq i \leq k. \exists \sigma_i. \exists (\Gamma \cup \{x : B_x \mid x\in\fv(p_i)\} \mid C'_i \types P_i : a_i)\in\text{InferR}(\Gamma\cup\{x : B_x \mid x\in\fv(p_i)\},\, P_i). A_i = a_i\sigma_i \wedge C \types C'_i\sigma_2$
            \end{enumerate}
            Observe this extract from \lstinline|InferR| in the match case:
            \begin{lstlisting}[language=Haskell]
                { g | $\bigcup_{i=1}^k(c_i \cup \{b \subtype p_i[b_x/x \mid x\in\fv(p_i)], a_i \subtype a\})$ $\cup$ c $\types$ M : a  |
                    a = freshVar,
                    $\overline{b_x}$ = $\overline{\text{freshVar}}$,
                    (g | c $\types$ q : b) $\in$ inferR(g, q),
                    (g $\cup$ {$\overline{(x : b_x)}$} | $c_i \types P_i : a_i$)$_{i=1}^k$ $\in \Pi_{i=1}^k$(inferR(g $\cup$ {$\overline{(x : b_x)}$}, $P_i$))
                }
            \end{lstlisting}
            By choosing the correct inferred judgment(s) in lines 4 and 5 makes it clear that the correct witness judgment can be inferred.\\
            Defining the witness $\sigma := \sigma_0 \cup \bigcup_{i=1}^k (\sigma_i \cup [(b_x)_{x\in\fv(p_i)} \mapsto (B_x)_{x\in\fv(p_i)}]) \cup [a \mapsto A]$ satisfies the requirements of $A = a\sigma$ and $C \types (\bigcup_{i=1}^k(C'_i \cup \{b \subtype p_i[b_x/x \mid x\in\fv(p_i)], a_i \subtype a\}) \cup C'_0)\sigma$.
            }
        \item[$\rlnm{MchL2}$] {
            For some $\tau$, we have:
            \begin{enumerate}
                \item $\Gamma\tau,\, \matchtm{Q}{|_{i=1}^k (p_i \mapsto P_i)} : A \mid C \types_2 \Delta\tau$
                \item $\forall i. \Gamma\tau,\, (M, P_i) : A'_i \mid C \types_2 \Delta\tau$
                \item $\forall i. \forall x\in\fv(p_i). \Gamma,\, P_i : A_x \mid C \types_2 x : B_x$
                \item $\forall i. C \types A \subtype A_i$
                \item $\forall i. C \types (B_i, A_i) \subtype A'_i$
                \item $\forall i. \forall x\in\fv(p_i). C \types A \subtype A_x$
                \item $\forall i. C \types p_i[B_x/x \mid x\in\fv(p_i)] \subtype B_i$
            \end{enumerate}
            By the inductive hypothesis, we have:
            \begin{enumerate}
                \item $\forall i. \exists \sigma'_i. \exists (\Gamma,\, (M, P_i) : a'_i \mid C'_i \types \Delta)\in\text{InferL}(\Gamma,\, (M, P_i),\, \Delta). A'_i = a'_i\sigma'_i \wedge C \types C'_i\sigma'_i$
                \item $\forall i. \forall x\in\fv(p_i). \exists \sigma_{(i, x)}. \exists (\Gamma,\, P_i : a_x \mid C_{(i, x)} \types x : b_x)\in\text{InferL}(\Gamma,\, P_i,\, \{x : b_x\}). A_x = a_x\sigma_{(i, x)} \wedge C \types C_{(i,x)}\sigma_{(i, x)}$
            \end{enumerate}
            Observe this extract from \lstinline|InferL| in the match case:
            \begin{lstlisting}[language=Haskell]
                { g $\cup$ {M : a} | c $\types$ d  |
                    a = freshVar,
                    $(a_i, b_i)_{i=1}^k$ = $\overline{\text{freshVar}}$,
                    $(b_{(i,x)})_{(i, x)\in [1..k]\times\fv{(p_i)}}$ = $\overline{\text{freshVar}}$,
                    (g $\cup$ {(q, $P_i$) : $a'_i$} | $c_i$ $\types$ d)$_{i=1}^k$ $\in \Pi_{i=1}^k$(inferL(g, (q, $P_i$), d)),
                    (g $\cup$ {$P_i$ : $a_{(i, x)}$} | $c'_{(i, x)}$ $\types$ d)$_{(i, x)}$ $\in \Pi_{(i, x) \in [1..k]\times\fv{(p_i)}}$(inferL(g, $P_i$, {x : $b_{(i,x)}$})),
                    c = $\bigcup_{i=1}^k(c_i \cup \{a \subtype a_i, (a_i, b_i) \subtype a'_i, p_i[b_{(i,x)}/x \mid x\in\fv{(p_i)}] \subtype b_i\})$ $\cup$ 
                        $\bigcup_{(i, x) \in [1..k]\times\fv{(p_i)}}(c'_{(i, x)} \cup \{a \subtype a_{(i,x)}\})$
                }
            \end{lstlisting}
            Line 5 generates a tuple of length $k$ by enumerating every arrangement of typings for the $k$ pairs $(Q, P_i)$. Hence, the correct sequence of types can be picked according to the inductive hypothesis (point 1).\\
            Line 6 generates a tuple by enumerating every arrangement of typings for each (left) typing of $P_i$ with delta $x : b_{(i, x)}$ (for each $i$ and $x$). Hence, the correct sequence of types can be picked according to the inductive hypothesis (point 2).\\
            Hence, the required witness judgment is generated.\\
            As each recursive call produces its own, independent, fresh variables, the pairwise intersection of any two distinct $\sigma$ given by the inductive hypotheses is $\tau$, we can define the witness:
            \[
                \sigma := \bigcup_{i=1}^k(\sigma_i \cup [a_i \mapsto A_i, b_i \mapsto B_i] \cup \bigcup_{x\in\fv(p_i)}(\sigma_{(i, x)} \cup [b_{(i, x)} \mapsto B_x])) \cup [a \mapsto A]
            \]
            This satisfies $A = a\sigma$ and $C \types C'\sigma$ (where $C'$ is the generated constraints on line 7/8 of the extract) as required.
            }
    \end{description}
\end{proof}

%% file: apx-soundness.tex
\section{Additional Material in Support of SubSection \ref{sec:constrained-soundness}}\label{apx:syntactic-soundness}

In this appendix, we give the definition of reduction for the programming language of the constrained system and we give the definition of consistency of a constraint set.  We then establish a number of useful lemmas regarding the type system:
\begin{itemize}
  \item left and right inversion lemmas
  \item typed substitution lemmas
\end{itemize}
This enables us to prove progress and preservation, and finally conclude syntactic soundness.

\begin{definition}[Reduction]
  The \emph{evaluation contexts} are defined by the following grammar:
  \[
    \calE \Coloneqq \square \mid c(\bar{V},\calE,\bar{M}) \mid \calE\,N \mid (\fixtm{f(x)}{M})\,\calE \mid \matchtm{\calE}{\mid_{i=1}^k (p_i \mapsto P_i)}
  \]
  Fix a module $\mathcal{M}$, then the one-step reduction relation wrt $\calM$, written $M \pedm N$, is the contextual closure of the following axiom schema:
  \[
  \begin{array}{rrcl}
    \rlnm{Delta} & \mathsf{f} &\pedm& \calM(\mathsf{f})\\[2mm]
    \rlnm{Fix$\beta$} & \ch{(\fixtm{f(x)}{M})\,V} &\pedm& M[V/x,\fixtm{f(x)}{M}/f]\\[2mm]
    \rlnm{Match} & \matchtm{c(\bar{V})}{|_{d \in I} (d(\bar{x_{d}}) \mapsto P_d)} &\pedm& P_c[\bar{V/x_{c}}] \qquad \text{(when $c \in I$)}
  \end{array}
  \]
  We will more often simply write $M \ped N$ whenever the context is unimportant or otherwise understood, and the reflexive-transitive closure by $M \peds N$.
\end{definition}

\begin{definition}
  A constraint set $C$ is said to be \emph{closed} just if the following conditions are met:
  \begin{enumerate}[(1)]
    \item If $A \subtype A' \in C$ and $A' \subtype B \in C$ then $A \subtype B \in C$.
    \item If $K \subtype K' \in C$, $c(A_1,\ldots,A_n) \in K$ and $c(B_1,\ldots,B_n) \in K'$, then $A_i \subtype B_i \in C$ for each $i \in [1..n]$.
    \item If $A \to B \subtype A' \to B' \in C$ then $A' \subtype A \in C$ and $B \subtype B' \in C$.
    \item If $A \from B \subtype A' \from B' \in C$ then $A \subtype A' \in C$ and $B' \subtype B \in C$.
  \end{enumerate}
\end{definition}

\begin{definition}
  We say that a constraint $A \subtype B$ is said to be \emph{syntactically consistent} just if either:
  \begin{itemize}
    \item At least one of $A$ or $B$ is a type variable
    \item or, $A$ and $B$ are both sufficiency arrows or are both necessity arrows 
    \item or, $B$ is $\okty$
    \item or, $A$ is of shape $\Sigma_{c \in I} c(\bar{A})$, $B$ is of shape $\Sigma_{d \in J} d(\bar{B})$ and $I \subseteq J$
  \end{itemize}
  We say that a constraint set $C$ is \emph{consistent} just if it is closed and contains only syntactically consistent constraints.  We say that a type environment $\Gamma$ is \emph{consistent} just if $\constraints(\Gamma)$ is consistent.
\end{definition}

\begin{definition}
  A judgement of the constrained type system is said to be a \emph{single-subject (1S)} just if it either has shape:
  \begin{enumerate}[(I)]
    \item $\calC \mid \Gamma,\,M:A \types \Delta$, $\Gamma$ and $\Delta$ are consistent variable environments with disjoint subjects
    \item or $\calC \mid \Gamma \types M:A$, with $\Gamma$ a consistent variable environment
  \end{enumerate}
\end{definition}

Note, a consequence of the 1S definition is that a 1S judgment has at most non-variable subject -- here we will consider top-level identifiers also as (global) variables, which is the $M:A$ above.  If a 1S judgment has a non-variable subject then it is either of type (I) or type (II), but cannot be both.  Otherwise, if a 1S judgement has two formulas on the left-hand side with the same variable subject, then it must be of type (I), with $M$ being one of the two.  Otherwise, a 1S judgement with all distinct variable subjects on the left-hand side can be seen as either type (I) or type (II).  1S variable judgments are important because well-typing and ill-typing are both 1S judgments, and a proof of a 1S judgment only involves 1S judgments.

\begin{lemma}
  A proof tree $\calT$ of a 1S conclusion contains only 1S judgments at each node.
\end{lemma}
\begin{proof}
  By induction on the height of $\calT$.  When the tree is of height 0, i.e. it is only a \rlnm{GVar}, \rlnm{Id}, \rlnm{CnsDL}, \rlnm{FixDL} or \rlnm{VarK} and the result follows by assumption.  Otherwise, suppose the tree has height $k$ and we argue by case analysis on the root:
  \begin{itemize}
    \item If the root is concluded by \rlnm{SubL} or \rlnm{SubR}, then clearly the judgement in the premise is also 1S and the result follows from the induction hypothesis.
    \item \ch{If the root is concluded by \rlnm{AbsR} or \rlnm{AbnR}, then the 1S conclusion has shape $\Gamma \types \fixtm{f(x)}{M} : A$.  Thus $\Gamma$ is a consistent variable environment, and we assume by the side condition that the bound variables $f$ and $x$ do not occur in $\Gamma$.  Thus the premises can be seen as 1S type (II) and type (I) respectively, and the result follows from the induction hypothesis.}
    \item If the root is concluded by \rlnm{AppL} or \rlnm{AppR}, then the 1S conclusion has shape $\Gamma,\,M\,N:A \types \Delta$ or $\Gamma \types M\,N:A$ respectively.  In the former case, $\Gamma$ and $\Delta$ are consistent and disjoint variable environments, and therefore $\Gamma \types M : B \from A$ is of type (II) and $\Gamma,\,N:B \types \Delta$ is of type (I).  In the latter case, $\Gamma$ is a consistent variable environment and hence $\Gamma \types M : B \to A$ and $\Gamma \types N : B$ are both of type (II).  In both cases, the result follows from the induction hypothesis.
    \item If the root is concluded by \rlnm{CnsL}, then the 1S conclusion has shape $\Gamma,\,C(M_1,\ldots,M_m) : C(A_1,\ldots,A_m) \types \Delta$ and therefore $\Gamma$ and $\Delta$ are disjoint, consistent variable environments.  Hence, $\Gamma,\,M_i:A_i \types \Delta$ is a 1S judgment of type (I) and the result follows from the induction hypothesis.
    \item If the root is concluded by \rlnm{CnsR}, then the 1S conclusion has shape $\Gamma \types c(M_1,\ldots,M_m) : c(A_1,\ldots,A_m)$ and $\Gamma$ is a consistent variable environment.  Then $\Gamma \types M_i : A_i$ is of type (II) for each $i$ and the result follows from the induction hypothesis.
    \item If the root is concluded by \rlnm{MchL}, then the 1S conclusion is of shape $\Gamma,\,\matchtm{M}{\mid^k_{i=1} (p_i \mapsto P_i)} : A \types \Delta$ and so $\Gamma$ and $\Delta$ are consistent, disjoint variable environments.  We may assume by the implicit side condition that the variables bound in the patterns are fresh for the context.  Hence, it follows immediately that each $\Gamma,\,P_i : A \types x : B_x$ is type (I).  Similarly, $\Gamma,\,(M, P_i):(p_i[B_x/x \mid x \in \fv(p_i)], A) \types \Delta$ is type (I).  Hence, the result follows from the induction hypothesis.
    \item If the root is concluded by \rlnm{MchR}, then the 1S conclusion is of shape $\Gamma \types \matchtm{M}{|_{i=1}^k (p_i \mapsto P_i)} : A$ and $\Gamma$ is a consistent variable environment.  Moreover, we may assume by the implicit side condition that the variables bound by the patterns are fresh for the context.  Therefore, both of the premises are type (II) and the result follows from the induction hypothesis.
    \item If the root is concluded by \rlnm{CnsK}, then the 1S conclusion is of shape $\Gamma,\,c(M_1,\ldots,M_n): \okty \types \Delta$ and $\Gamma$ and $\Delta$ are disjoint, consistent variable environments.  Therefore, the premise is type (I) and the result follows from the induction hypothesis.
    \item If the root is concluded by \rlnm{FunK}, then the 1S conclusion has shape $\Gamma,\,M\,N:\okty \types \Delta$ and $\Gamma$ and $\Delta$ are disjoint, consistent variable environments.  Therefore, the premise is type (I) and the result follows from the induction hypothesis.
  \end{itemize}
\end{proof}

From now on we will consider only 1S judgements.

\begin{lemma}[Left Inversion]\label{lem:left-inversion}
  Suppose $\Gamma,\,M:A \types \Delta$ has $\Gamma$ and $\Delta$ consistent, disjoint variable environments.  Then either (a) $\Delta$ is of shape $x:\okty$, or (b) one of the following is true:
  \begin{itemize}
    \item $M$ is of shape $x$ and there is some $B$ such that $x:B = \Delta$ and $\Gamma \types A \subtype B$.
    \item $M$ is of shape $f$ and there is some type $B$ such that $x:B = \Delta$ and $\Gamma \types A \subtype B$
    \item $M$ is of shape $PQ$ and either:
      \begin{enumerate}[(i)]
        \item there is a type $B$ such that $\Gamma \types M : B \from A$ and $\Gamma,\,N:B \types \Delta$
        \item or, $\Gamma,\,M:\okty \from A \types \Delta$.
      \end{enumerate}
    \item $M$ is of shape $c(M_1,\ldots,M_m)$ and either:
    \begin{enumerate}[(i)]
      \item there are types $B_1,\ldots,B_m$ and $i$ such that $A \subtype c(B_1,\ldots,B_m)$ and $\Gamma,\,M_i:B_i \types \Delta$.
      \item or, $\Gamma,\,M_i : \okty \types \Delta$ for some $i$
      \item or, there is some type $A'$ which is either an arrow or sum type not including a $c$-headed case, such that $\Gamma \types A \subtype A'$ and $\Gamma,\,c(M_1,\ldots,M_n) : A' \types \Delta$
    \end{enumerate}
    \item $M$ is of shape $\matchtm{Q}{\mid_{i=1}^k (p_i \mapsto P_i)}$ and there is a family of types $B_x$ (one for each pattern-bound variable $x$) such that $\Gamma,\,P_i : A \types x : B_x$ (for each $i$ and each $x \in \fv(p_i)$), and either $\Gamma,\,Q : p_i[B_x/x \mid x\in\fv(p_i)] \types \Delta$ or $\Gamma,\,P_i:A \types \Delta$.
    \item \ch{$M$ is of shape $\fixtm{f(x)}{P}$ and there is some sum type $A'$ such that $\Gamma \types A \subtype A'$ and $\Gamma,\,\fixtm{f(x)}{P} : A' \types \Delta$}
  \end{itemize}
\end{lemma}
\begin{proof}
  The proof is by cases on the shape of $M$.
  \begin{itemize}
    \item If $M$ is a variable $x$ then, according to the assumptions, all the formulas have variables as subjects.  
    By the completeness of algorithmic type assignment, we have $\typings(\Gamma),\,x:A \mid \constraints(\Gamma) \types_2 \Delta$.  Given that $\Gamma$ and $\Delta$ are disjoint, the only rules that can have concluded such a judgement are \rlnm{VarK2}, \rlnm{LVar2}.  The former gives rise to (a).  In case \rlnm{LVar2}, since $\Gamma$ and $\Delta$ are disjoint, it must be that there is some $B$ such that $x:B \in \Delta$ and $\Gamma \types A \subtype B$, as required by (c).
    \item If $M$ is a top-level identifier, then the reasoning is as in the former case, since $A$ is a monotype.
    \item If $M$ is an application $PQ$, then we reason as follows.  By the completeness of algorithmic type assignment, $\typings(\Gamma),\,P\,Q:A \mid \constraints(\Gamma) \types_2 \Delta$ and the only rules that could have concluded this judgement are \rlnm{VarK2}, \rlnm{AppL2} or \rlnm{FunK2}.  The former case gives rise to (a).  In case \rlnm{AppL2}, it follows that there are types $B_1$ and $B$ such that $\constraints(\Gamma) \types B_1 \subtype B \from A$, $\typings(\Gamma) \mid \constraints(\Gamma) \types_2 M : B_1$ and $\typings(\Gamma),\,N : B \mid \typings(\Gamma) \types_2 \Delta$.  It follows from the soundness of algorithmic type assignment that $\Gamma \types M : B_1$ and hence, by \rlnm{SubR}, $\Gamma \types M : B \from A$.  Similarly, $\Gamma,\,N:B \types \Delta$.  This is what is required by (i).  In the latter case it follows that there is a type $B$ such that $\constraints(\Gamma) \types_2 \okty \from A \subtype B$ and $\typings(\Gamma),\,M : B \mid \constraints(\Gamma) \types_2 \Delta$.  It follows from the soundness of algorithmic type assignment that $\Gamma,\,M:B \types \Delta$ and hence, by \rlnm{SubR}, $\Gamma,\,M:\okty \from A \types \Delta$, as required by (ii).
    \item If $M$ is a constructor term $c(M_1,\ldots,M_n)$, then we reason as follows.  By the completeness of algorithmic type assignment, $\typings(\Gamma),\,c(M_1,\ldots,M_n) : A \mid \constraints(\Gamma) \types_2 \Delta$.  Since $\Gamma$ and $\Delta$ are disjoint, the only rules that can conclude such a judgement are \rlnm{VarK2}, \rlnm{ConsL2}, \rlnm{CnsK2} and \rlnm{CnsDL}.  The first possibility give rise to (a).  The second implies that there are types $A_1,\ldots,A_n$ and $\constraints(\Gamma) \types A \subtype c(A_1,\ldots,A_n)$ and $\typings(\Gamma),\,M_i : A_i \mid \constraints(\Gamma) \types_2 \Delta$ for some $i$.  By the soundness of algorithmic type assignment, $\Gamma,\,M_i : A_i \types \Delta$, as required by (i).  In the third, we have for some $i$, $\typings(\Gamma),\,M_i : \okty \mid \constraints(\Gamma) \types_2 \Delta$.  Then (ii) follows by the soundness of algorithmic type assignment.  Finally, when the conclusion is by \rlnm{CnsDL}, there is some type $A'$ which is either an arrow or a sum not including $c$, such that $\Gamma \types A \subtype A'$ and $\Gamma,\,M:A' \types \Delta$.
    \item If $M$ is a match expression $\matchtm{Q}{\mid_{i=1}^k (p_i \mapsto P_i)}$ then we reason as follows.  It follows from the completeness of algorithmic type assignment that $\typings(\Gamma),\,\matchtm{Q}{\mid_{i=1}^k (p_i \mapsto P_i)} : A \mid \constraints(\Gamma) \types_2 \Delta$.  Given the disjointness of the environments, the only rules that could conclude this judgement are \rlnm{VarK2} and \rlnm{MchL2}.  In the former case we obtain (a), so let us assume that $\Delta$ is not of this shape.  In the latter case, we have that there are families of types $B_i$, $A_i$, $A_i'$, $A_x$ and $B_x$, such that (1) $\typings(\Gamma),\,(Q,P_i) : A_i' \mid \constraints(\Gamma) \types_2 \Delta$ for all $i$, (2) $\typings(\Gamma),\,P_i : A_x \mid \constraints(\Gamma) \types_2 x : B_x$ for all $i$ and all $x \in \fv(p_i)$, $\constraints(\Gamma) \types A \subtype A_i$ for all $i$, $\constraints(\Gamma) \types A \subtype A_x$ for all $i$ and all $x \in \fv(p_i)$, $\constraints(\Gamma) \types p_i[B_x/x \mid x \in \fv(p_i)] \subtype B_i$ for all $i$, and $\constraints(\Gamma) \types (B_i,A) \subtype A_i'$ for all $i$.  It follows from the soundness of algorithmic typing that $\Gamma,\,(Q,P_i) : A_i' \types \Delta$ for all $i$, and $\Gamma,\,P_i : A_x \types x:B_x$ for all $i$ and all $x \in \fv(p_i)$.  Then it follows from the previous case that either (1) $A_i'=\okty$ and either $\Gamma,\,Q:\okty \types \Delta$ or $\Gamma,\,P_i:\okty \types \Delta$, or (2) $A_i'$ is an arrow type or (3) there are types $B''_1$ and $B''_2$ such that $\Gamma \types A_i' \subtype (B''_1,B''_2)$ and $\Gamma,\,Q : B''_1 \types \Delta$ or $\Gamma,\,P_i:B''_2 \types \Delta$.  In case (1), it follows from $\rlnm{SubL}$ that either $\Gamma,\,Q : p_i[B_x/x \mid x \in \fv(p_i)] \types \Delta$ or $\Gamma,\,P_i:A \types \Delta$, as required.  Case (2) is impossible by the consistency of $\Gamma$.  In case (3), we have $\Gamma \types (p_i[B_x/x \mid x \in \fv(p_i)],A) \subtype (B_i,A) \subtype A_i' \subtype (B''_1,B''_2)$ and thus, by the closedness of $\Gamma$, we have $\Gamma \types p_i[B_x/x \mid x \in \fv(p_i)] \subtype B''_1$ and $\Gamma \types A \subtype B''_2$.  Then it follows from \rlnm{SubL} that either $\Gamma,\,Q : p_i[B_x/x \mid x \in \fv(p_i)] \types \Delta$ or $\Gamma,\, P_i :A \types \Delta$
    \item \ch{When $M$ is a fixpoint abstraction $\fixtm{f(x)}{P}$, it follows from the completeness of algorithmic type assignment that $\Gamma,\,\fixtm{f(x)}{P} : A \types_2 \Delta$ and then the only applicable rules are \rlnm{VarK2} and \rlnm{FixDL2}.  The former case gives rise to (a).  In the latter case, it is immediate that there is some sum type $A'$ such that $\Gamma \types A \subtype A'$ and $\Gamma,\,M:A' \types \Delta$.}
  \end{itemize}
\end{proof}

\begin{lemma}[Right Inversion]\label{lem:right-inversion}
  Suppose $\Gamma \types M : A$ and $\Gamma$ a consistent variable environment.  Then one of the following is true:
  \begin{itemize}
    \item $M$ is of shape $x$ and $A=\okty$ or there is some $B$ such that $x:B \in \Gamma$ and $\Gamma \types B \subtype A$.
    \item $M$ is of shape $f$ and there are types $B$, $a_1,\ldots,a_n$ and $B_1,\ldots,B_n$, and constraints $D$ such that $f:\forall\bar{a}.\,D \Rightarrow B \in \Gamma$ and $\Gamma \types B[\bar{B}/\bar{a}] \subtype A$
    \item $M$ is of shape $PQ$ and there is a type $B$ such that $\Gamma \types P : B \to A$ and $\Gamma \types N : B$
    \item \ch{$M$ is of shape $\fixtm{f(x)}{P}$ and there are types $B_1$ and $B_2$ such that either: 
    \begin{enumerate}[(i)]
      \item $\Gamma,\,f:B_1 \to B_2,\,x:B_1 \types M : B_2$ and $\Gamma \types B_1 \to B_2 \subtype A$
      \item or, $\Gamma,\,f:B_1 \from B_2,\,M:B_2 \types x:B_1$ and $\Gamma \types B_1 \from B_2 \subtype A$
    \end{enumerate}}
    \item $M$ is of shape $c(M_1,\ldots,M_n)$ and there are types $B_1,\ldots,B_n$ such that $\Gamma \types M_i : B_i$ for each $i$, and $\Gamma \types c(B_1,\ldots,B_n) \subtype A$.
    \item $M$ is of shape $\matchtm{Q}{\mid_{i=1}^k (p_i \mapsto P_i)}$ and there is a family of types $B_x$ (for each $i$ and each $x \in \fv(p_i)$) such that $\Gamma \types Q : \Sigma_{i=1}^k\,p_i[B_x/x \mid x \in \fv(p_i)]$ and $\Gamma \cup \{x:B_x \mid x \in \fv(p_i)\} \types P_i : A$ (for each $i$).
  \end{itemize}
\end{lemma}
\begin{proof}
  By cases on the shape of $M$.
  \begin{itemize}
    \item When $M$ is a local variable $x$ we reason as follows.  By the completeness of algorithmic type assignment, we have $\typings(\Gamma) \mid \constraints(\Gamma) \types_2 x : A$.  Given that $\Gamma$ is consistent, the only rules that can conclude the judgement are \rlnm{Var2} and \rlnm{VarK2}.  In the latter case, $A = \okty$.  In the former case, there is some $B$ such that $x:B \in \typings(\Gamma)$ and $\constraints(\Gamma) \types B \subtype A$.
    \item When $M$ is an identifier $f$, by the completeness of type assignment we have that $\typings(\Gamma) \mid \constraints(\Gamma) \types_2 f : A$.  The only applicable rule was \rlnm{Inst2} and hence there are type variables $\bar{a}$ and types $\bar{B}$ and $B$, and constraints $D$ such that $f:\forall\bar{a}.\,D \Rightarrow B \in \typings(\Gamma)$ and $\constraints(\Gamma) \types D[\bar{B}/\bar{a}]$ and $\constraints(\Gamma) \types B[\bar{B}/\bar{a}] \subtype A$.
    \item When $M$ is an application $PQ$, by the completeness of type assignment, we have that $\typings(\Gamma) \mid \constraints(\Gamma) \types_2 PQ : A$.  The only applicable rule is \rlnm{AppR2} and so there are types $B_1$ and $B$ such that $\typings(\Gamma) \mid \constraints(\Gamma) \types_2 P : B_1$, $\typings(\Gamma) \mid \constraints(\Gamma) \types_2 Q : B$ and $\constraints(\Gamma) \types B_1 \subtype B \to A$.  It follows from the soundness of algorithmic type assignment that $\Gamma \types P : B_1$ and $\Gamma \types Q : B$ and by \rlnm(SubR) also $\Gamma \types P : B \to A$, as required.
    \item \ch{When $M$ is a fixpoint abstraction $\fixtm{f(x)}{P}$, by the completeness of algorithmic type assignment, we have that $\typings(\Gamma) \mid \constraints(\Gamma) \types_2 \fixtm{f(x)}{P} : A$.  Only the rules \rlnm{FixsR2} and \rlnm{FixnR2} could have concluded.  In the former case, we have that there are types $B_1$ and $B_2$ such that $\typings(\Gamma),\,f:B_1 \to B_2,\,x:B_1 \mid \constraints(\Gamma) \types M : B_2$ and $\constraints(\Gamma) \types B_1 \to B_2 \subtype A$, as required by (i).  The latter case is analogous.}
    \item When $M$ is a constructor term $c(M_1,\ldots,M_n)$, by the completeness of algorithmic type assignment, we have that $\typings(\Gamma) \mid \constraints(\Gamma) \types_2 c(M_1,\ldots,M_n) : A$.  The only rule that could conclude this is \rlnm{ConsR2}, and hence there must be types $B_1,\ldots,B_n$ such that $\typings(\Gamma) \mid \constraints(\Gamma) \types_2 M_i : B_i$ for each $i$ and $\constraints(\Gamma) \types c(B_1,\ldots,B_n) \subtype A$.  It follows from the soundness of algorithmic type assignment that $\Gamma \types M_i : B_i$ for each $i$, as required.
    \item When $M$ is of shape $\matchtm{Q}{\mid_{i=1}^k (p_i \mapsto P_i)}$, it follows from the completeness of algorithmic type assignment that $\typings(\Gamma) \mid \constraints(\Gamma) \types \matchtm{Q}{\mid_{i=1}^k (p_i \mapsto P_i)} : A$.  The only applicable rule is \rlnm{MchR} and so there is a type $B$ and families of types $A_i$ (for each $i$) and $B_x$ (for each $i$ and each $x \in \fv(p_i)$) such that $\typings(\Gamma) \mid \constraints(\Gamma) \types_2 Q : B$ and $\typings(\Gamma) \cup \{x:B_x \mid x \in \fv(p_i)\} \mid \constraints(\Gamma) \types_2 P_i : A_i$ (for each $i$) and $\constraints(\Gamma) \types B \subtype \Sigma_{i=1}^k p_i[B_x/x \mid x \in \fv(p_i)]$ and $\constraints(\Gamma) \types A_i \subtype A$ (for all $i$).  It follows from the soundness of algorithmic type assignment that $\Gamma \types Q : B$ and, by \rlnm{SubR}, $\Gamma \types \Sigma_{i=1}^k p_i[B_x/x \mid x \in \fv(p_i)]$.  Similarly, $\Gamma \cup \{x:B_x \mid x \in \fv(p_i)\} \types P_i : A_i$ and, by \rlnm{SubR}, $\Gamma \cup \{x:B_x \mid x \in \fv(p_i)\} \types P_i : A$.
    \item When $M$ is of shape $\fixtm{x}{P}$, it follows from the completeness of algorithmic type assignment that $\typings(\Gamma) \mid \constraints(\Gamma) \types_2 \fixtm{x}{P} : A$.  The only rule that could conclude this judgement is \rlnm{FixR2} and so there is some type $B$ such that $\typings(\Gamma),\,x:A \mid \constraints(\Gamma) \types P : B$ and $\constraints(\Gamma) \types B \subtype A$.  Hence, it follows from soundness of algorithmic type assignment and \rlnm{SubR} that $\Gamma,\,x:A \types M : A$, as required.
  \end{itemize}
\end{proof}

\begin{lemma}\label{lem:subtype-chain-in-closure}
  If $C \types A_0 \subtype A_n$ then there is a sequence of constraints:
  \[
    A_1 \subtype A_2, A_3 \subtype A_4, \ldots{} A_{n-2} \subtype A_{n-1}
  \]
  all of which are in $C$ and such that:
  \begin{enumerate}[(i)]
    \item for all $i \in [0..n/2]$ either $A_{2i+1}$ is $\okty$ or otherwise:
    \begin{enumerate}[(a)] 
      \item $A_{2i}$ is a type variable and $A_i = A_{2i+1}$
      \item or, $A_{2i}$ is a sufficiency arrow and so is $A_{2i+1}$
      \item or, $A_{2i}$ is a necessity arrow and so is $A_{2i+1}$
      \item or, $A_{2i}$ is of shape $\Sigma_{c \in I}\,c(\bar{A})$ and $A_{2i+1}$ is of shape $\Sigma_{d \in J}\,d(\bar{A})$ and $I \subseteq J$.
    \end{enumerate}
  \end{enumerate}
\end{lemma}
\begin{proof}
  The proof is by induction on the derivation of $C \types A \subtype B$.
  \begin{description}
    \item[\rlnm{IdS}] Clearly this satisfies the conclusion with a sequence of length 1.
    \item[\rlnm{TrS}] It follows from the induction hypotheses that there are sequences
    \[
      A_1 \subtype A_2, A_3 \subtype A_4, \ldots{} A_{n-2} \subtype A_{n-1} \qquad\text{and}\qquad B_1 \subtype B_2, B_3 \subtype B_4, \ldots{} B_{m-2} \subtype B_{m-1}
    \]
    satisfying the appropriate properties.  The middle witness is $A_n = B_0$.  If it is a type variable, then $A_{n-1} = A_n = B_0 = B_1$.  If it is an arrow, then $A_{n-1}$ and $B_1$ are arrows of the same kind too.  If it is a sum $\Sigma_{d \in J}\,d(\ldots)$, then $A_{n-1}$ is a sum $\Sigma_{c \in I}\,c(\ldots)$ with $I \subseteq J$ and $B_1$ is a sum $\Sigma_{k \in K}\,k(\ldots)$ with $J \subseteq K$. If it is $\okty$, then there is no requirement on $A_{n-1}$ but $B_1$ must be $\okty$.  Hence, in each case, $A_{n-1}$ and $B_1$ satisfy the required relationship and thus the witness is just the concatenation of the sequences.
    \item[\rlnm{ToS},\rlnm{FrS},\rlnm{SmS},\rlnm{CnS},\rlnm{OkS}] The conclusion is obtained using a sequence of length 0 because $A_0$ and $A_1$ already satisfy the required relationship.
  \end{description}
\end{proof}

Consistent constraint sets do not derive inconsistencies.
\begin{lemma}\label{lem:consistent-cons-no-incons}
  Suppose $C$ is syntactically consistent.  
  \begin{itemize}
    \item If $C \types A \subtype B$ then $A \subtype B$ is syntactically consistent.
    \item If $C \types A \distype B$ then $A \distype B$ is syntactically consistent.
  \end{itemize}
\end{lemma}
\begin{proof}
  \begin{itemize}
    \item If $C \types A \subtype B$, then by the forgoing lemma there is a sequence of inequalities in $C$ satisfying the required conditions:
    \[
      A_1 \subtype A_2, A_3 \subtype A_4, \ldots{} A_{n-2} \subtype A_{n-1}
    \]  
    Suppose for the purpose of obtaining a contradiction that $A \subtype B$ is not syntactically consistent.  Then there is a consecutive subsequence of the above, of shape:
    \[
      B_{i} \subtype a_1, a_1 \subtype a_2, a_2 \subtype a_3 \ldots, a_k \subtype B_{2k+1}
    \]
    (i.e. intermediate types are all type variables) of constraints in $C$ in which $B_i \subtype B_{2k+1}$ is already not syntactically consistent.  However, it follows that $B_i \subtype B_{2k+1}$ is contained in the closure of $C$, which is assumed consistent.
    \item The proof is by induction on the derivation of $C \types A \distype B$.  In cases \rlnm{ConD}, \rlnm{ToD} and \rlnm{FromD} the result is immediate.  In case \rlnm{RfId} the result follows from the assumption.  In case \rlnm{SymD} the result follows from the induction hypothesis since the characterisation of syntactic consistency is symmetrical.  Finally, in case \rlnm{SubD} we can extend Lemma~\ref{lem:subtype-chain-in-closure} to show that there are two sequences of subtype constraints in $C$ each of which gives rise to one side of the disjointness constraint.  Then it follows similarly, that if a disjointness constraint is concluded that is not syntactically consistent, then already there would be a non-syntactically consistent constraint in the closure.
  \end{itemize}
\end{proof}

\begin{lemma}\label{lem:values-have-correct-types-on-right}
  Suppose $\constraints(\Gamma)$ is closed and syntactically consistent. If $\Gamma \types V : A$ with $V$ a closed value and $A$ not a type variable, then either:
  \begin{itemize}
    \item $A$ is a sum type containing a type of shape $c(A_1,\ldots,A_n)$ and $V$ a constructor term of shape $c(W_1,\ldots,W_n)$ and $\Gamma \types W_i : A_i$ for each $i$,
    \item \ch{or, $A$ is an arrow and $V$ is a fixpoint abstraction}
  \end{itemize}
\end{lemma}
\begin{proof}
  We proceed by cases on the shape of $V$.  
  \begin{itemize}
    \item If $V$ is of shape $c(W_1,\ldots,W_n)$ then, by inversion, there are types $B_1,\ldots,B_n$ such that $\Gamma \types M_i : B_i$ for each $i$, and $\Gamma \types c(B_1,\ldots,B_n) \subtype A$.  Therefore, it follows from consistency of subtyping (Lemma~\ref{lem:consistent-cons-no-incons}) that either $A$ is a sum type including a $c$-headed type or $A$ is a type variable.
    \item If $V$ is of shape $\fixtm{f(x)}{M}$ then, by inversion, there are types $B_1$ and $B_2$ such that either $\Gamma \types B_1 \to B_2 \subtype A$ or $\Gamma \types B_1 \from B_2 \subtype A$.
    By consistency of subtyping, in both cases, $A$ is either a type variable or an arrow.
  \end{itemize}
\end{proof}

\subsubsection{Proof of Theorem~\ref{thm:progress} (Progress)}
\begin{proof}
  The first is by induction on the derivation.  Due to the form of the judgement, every subject in $\Gamma$ can make a step, so we focus on the principal formula in each derivation step and exclude cases that cannot occur.
  \begin{description}
    \item[\rlnm{GVar}] in this case $M$ is a top-level identifier, so $M$ can make a step.
    \item[\rlnm{SubR}] It follows from the induction hypothesis that $M$ can make a step, or $M$ is a value.
    \item[\rlnm{FixsR}] It is immediate that $M$ is a value.
    \item[\rlnm{FixnR}] It is immediate that $M$ is a value.
    \item[\rlnm{AppR}] In this case, $M$ has shape $PQ$ and it follows from the induction hypothesis that either $P$ can make a step or is a value and $Q$ can make a step or is a value.  If $P$ can make a step, then $M$ can make a step, so assume to the contrary.  Then $P$ is a value.  It follows that from Lemma~\ref{lem:values-have-correct-types-on-right} that $P$ must be an abstraction.  Now, if $Q$ can make a step, then $M$ can make a step and if, not, then it is a value and $M$ can make a step by contracting the redex.
    \item[\rlnm{CnsR}] In this case, $M$ is of shape $c(P_1,\ldots,P_n)$ and it follows from the induction hypothesis that, for each $i$, $P_i$ can either make a step or is a value.  From any configuration it follows that $M$ can either make a step or is a value.
    \item[\rlnm{MchR}] In this case, $M$ is of shape $\matchtm{Q}{\mid^k_{i=1} (p_i \mapsto P_i)}$.  It follows from the induction hypothesis that $Q$ can either make a step or is a value.  If $Q$ can make a step, then $M$ can make a step.  Otherwise, $Q$ is a value of a sum type and it follows from Lemma~\ref{lem:values-have-correct-types-on-right} that $Q$ must be a constructor term $c(W_1,\ldots,W_n)$ and there is a $c$-headed type in the sum.  Hence, there is a pattern $p_i$ of shape $c(x_1,\ldots,x_n)$ and so $M$ can make a step.
  \end{description}
  The second is also by induction on the derivation.  Similar remarks apply.  Note that, if $M$ is stuck, then it follows from the previous result that $\Gamma \not\types M : A$ for any $A$.
  \begin{description}
    \item[\rlnm{SubL}] In this case we assume $\Gamma \types A \subtype B$.  It follows from the induction hypothesis that either $M$ makes a step or $\Gamma \not\types M : B$.  In the former case the result is immediate, in the latter case, we must have $\Gamma \not\types M : A$ too by the contrapositive of \rlnm{SubR}.  
    \item[\rlnm{AppL}] In this case, $M$ is of shape $PQ$.  Since $M$ is not a value, either $M$ can make a step or $M$ is stuck and hence the desired conclusion follows in both cases.
    \item[\rlnm{CnsL}] In this case, $M$ is of shape $c(M_1,\ldots,M_n)$ and $A$ is of shape $c(A_1,\ldots,A_n)$.  It follows from the induction hypothesis that either $M_i$ can make a step or $\Gamma \not\types M_i : A_i$.  In the former case, either $M$ can make a step, or $M$ is stuck.  In either case, the result follows.  If $\Gamma \not\types M_i : A_i$ then we have that $\Gamma \not\types M : A$.  To see this, suppose the contrary.  Then it follows from inversion that there are types $B_1,\ldots,B_n$ such that $\Gamma \types M_i:B_i$ and $\Gamma \types c(B_1,\ldots,B_n) \subtype c(A_1,\ldots,A_n)$.  Since $\constraints(\Gamma)$ is closed, it follows that $\Gamma \types B_i \subtype A_i$ and so $\Gamma \types M_i:A_i$ follows by \rlnm{SubR}, which contradicts our supposition.
    \item[\rlnm{MchL}] In this case, $M$ is of shape $\matchtm{Q}{\mid^k_{i=1} (p_i \mapsto P_i)}$.  Since $M$ is not a value, either $M$ can make a step or is stuck.  The desired conclusion follows in both cases.
    \item[\rlnm{CnsK}] In this case, $M$ is of shape $c(M_1,\ldots,M_n)$ and $A$ is $\okty$.  It follows from the induction hypothesis that there is some $M_i$ that can either make a step or for which $\Gamma \not\types M_i : \okty$.  In the former case, either $M$ can make a step or $M$ is stuck and the result follows.  In the latter case, since all values are well-typed, it follows that $M_i$ is stuck and so the result follows.
    \item[\rlnm{FunK}] In this case, $M$ is of shape $M\,N$ and $A$ is $\okty$.  Since $M$ is not a value, it can either make a step or is stuck and in both cases the desired conclusion follows.
    \item[\rlnm{CnsDL}, \rlnm{FixDL}] In these cases, suppose $\Gamma \types M : A$, and then we obtain a contradiction from inversion and subtype consistency.
  \end{description}
\end{proof}

\begin{lemma}[Substitution on the right]\label{lem:subst-right}
  Assume $\Gamma$ is a consistent variable environment, disjoint from $\Delta$ and that does not contain $x$.
  \begin{enumerate}[(i)]
    \item If $\Gamma,\,x:B \types M : A$ and $\Gamma \types N : B$, then $\Gamma \types M[N/x] : A$.
    \item If $\Gamma,\,x:B,\,M:A \types \Delta$ and $\Gamma \types N : B$, then $\Gamma,\, M[N/x] : A \types \Delta$.
  \end{enumerate}
\end{lemma}
\begin{proof}
The proof is by induction on $M$.  Note that, in (ii), it is always possible that $\Delta$ is of shape $z:\okty$ for some variable $z$.  Then the conclusion follows immediately by \rlnm{VarK}, so we exclude this from the case analysis below.
\begin{itemize}
  \item When $M$ is $x$, we reason as follows.  In (i), it follows by inversion that $\Gamma \types B \subtype A$. Since $M[N/x] = N$, we have $\Gamma \types M[N/x] : B$ and by \rlnm{SubR}, therefore $\Gamma \types M[N/x] : A$.  In (ii) it follows from inversion that $\Delta$ must contain $y:\okty$ for some variable $y$, which we have dealt with above.
  \item When $M$ is some variable $y$ distinct from $x$, we reason as follows.  In (i), it follows from the assumptions and inversion that $y:A' \in \Gamma$ with $\Gamma \types A' \subtype A$.  Since $M[N/x] = M = y$, we have by \rlnm{LVar} that $\Gamma \types M[N/x] : A'$ and the conclusion follows from \rlnm{SubR}.  In (ii), $\Delta$ is $y:A'$ with $A \subtype A'$.  Then the conclusion follows from \rlnm{LVar} and \rlnm{SubL}.
  \item When $M$ is an application $P\,Q$, we reason as follows.  In (i), by inversion we have that there is some $B'$ such that $\Gamma,\,x:B \types P : B' \to A$ and $\Gamma,\,x:B \types Q : B'$.  It follows from the induction hypotheses that $\Gamma \types P[N/x] : B' \to A$ and $\Gamma \types Q[N/x] : B'$ and so the conclusion follows from \rlnm{AppR}.  In (ii), inversion gives us that either: (1) there is a type $B'$ and $\Gamma,\,x:B \types P : B' \from A$ and $\Gamma,\,x:B,\,Q:B' \types \Delta$, or (2) $\Gamma,\,x:B,\,P:\okty \from A \types \Delta$.  Suppose the former, then it follows from the induction hypothesis part (i) that $\Gamma \types P[N/x] : B' \from A$ and from part (ii) that $\Gamma,\,Q[N/x] :B' \types \Delta$.  Hence the conclusion follows from \rlnm{AppL}.  Suppose the latter, then it follows from the induction hypothesis part (ii) that $\Gamma,\,P[N/x] : \okty \from A$, and the conclusion follows from \rlnm{FunK}.  
  \item \ch{When $M$ is a fixpoint abstraction $\fixtm{f(y)}{P}$, we may assume by the variable convention that $y$ and $f$ do not occur outside of $P$.  In (i), by inversion we have that there are types $B_1$ and $B_2$ such that either (1) $\Gamma,\,x:B,\,f:B_1 \to B_2,\,y:B_1 \types P:B_2$ and $\Gamma \types B_1 \to B_2 \subtype A$ or (2) $\Gamma,\,x:B,\,f:B_1 \from B_2,\,P:B_2 \types y:B_1$ and $\Gamma \types B_1 \from B_2 \subtype A$.  In case (1), it follows from the induction hypothesis part (i) that $\Gamma,\,y:B_1 \types P[N/x] : B_2$ and thus the conclusion follows from \rlnm{FixsR} and \rlnm{SubR}.  In case (2) it follows from the induction hypothesis part (ii) that $\Gamma,\,P[N/x]:B_2 \types y:B_1$ and hence the conclusion follows from \rlnm{FixnR} and \rlnm{SubR}.  In (ii), by inversion we have that there is some arrow type or sum type $A'$ and $\Gamma \types A \subtype A'$ and $\Gamma,\,x:B,\,\fixtm{f(y)}{P} : A' \types \Delta$.  In either case, the conclusion follows from \rlnm{FixDL} and \rlnm{SubL}.}
  \item When $M$ is a constructor term $c(P_1,\ldots,P_n)$, we reason as follows.  In (i) it follows from inversion that there are types $B_1,\ldots,B_n$ and $\Gamma,\,x:B \types P_i:B_i$ for each $i$ and $\Gamma \types c(B_1,\ldots,B_n) \subtype A$.  It follows from the induction hypothesis that $\Gamma \types c(P_1[N/x],\ldots,P_n[N/x]) : c(B_1,\ldots,B_n)$ and then the conclusion follows from \rlnm{SubR}.  In (ii), it follows from inversion that either (1) there are types $B_1,\ldots,B_n$ and $i$ such that $\Gamma \types A \subtype c(B_1,\ldots,B_n)$ and $\Gamma,\,x:B,\,P_i:B_i \types \Delta$, or (2) there is an arrow type $A'$ such that $\Gamma \types A \subtype A'$ and $\Gamma,\,x:B,\,c(P_1,\ldots,P_n) : A' \types \Delta$.  In the former case, the conclusion follows from the induction hypothesis and \rlnm{SubL}, in the latter it follows from \rlnm{CnsDL} and \rlnm{SubL}.
  \item When $M$ is a match expression $\matchtm{Q}{\mid_{i=1}^k (p_i \mapsto P_i)}$, we reason as follows.  In (i), by inversion we have a family of types $B_x$ (for each $i$ and $x\in\fv(p_i)$) such that $\Gamma,\,x:B \types Q : \Sigma_{i=1}^k\,p_i[B_x/x\mid x \in \fv(p_i)]$ and $\Gamma,\,x:B \cup \{x:B_x \mid x \in \fv(p_i)\} \types P_i : A$ for each $i$.  The conclusion follows from the induction hypothesis, part (i).  In (ii), by inversion we have a family of types $B_y$ (for each pattern-bound variable $y$) such that $\Gamma,\,x:B,\,P_i:A \types y:B_y$ and either $\Gamma,\,Q : p_i[B_y/y \mid y\in\fv(p_i)] \types \Delta$ or $\Gamma,\,P_i : A \types \Delta$.  In both cases, it follows from the induction hypothesis and \rlnm{CnsL} that $\Gamma,\,(Q[N/x],P_i[N/x]) : (p_i[B_y/y \mid y\in\fv(p_i)],A) \types \Delta$.  Then we can obtain the conclusion by applying the induction hypothesis part (ii) to each of the former judgements and concluding by \rlnm{MchL}.
\end{itemize}
\end{proof}

\begin{lemma}[Trivial Substitution]\label{lem:trivial-subst}
  Assume $\Gamma$ is a consistent variable environment disjoint from $\Delta$ and neither contain $x$.  Assume $\Gamma \types V : B$ is a closed value.  Then: 
  \begin{enumerate}[(i)]
    \item if $\Gamma \types M : A$ then $\Gamma \types M[V/x] : A$
    \item if $\Gamma,\,M:A \types \Delta$ then $\Gamma,\,M[V/x] : A \types \Delta$
  \end{enumerate}
\end{lemma}
\begin{proof}
  The proof is by induction on $M$.  Note, in case (ii), there is always the possibility that $\Delta$ has shape $z:\okty$ for some variable $z$ and then the conclusion follows immediately from \rlnm{VarK}.  Moreover, if $M$ does not contain $x$, then the conclusion is immediate.  Hence, we will exclude these cases from consideration below.
  \begin{itemize}
    \item If $M$ is $x$, we reason as follows.  In (i), by inversion, it must be that $A=\okty$ and the conclusion follows because all closed values are well-typed.  In (ii), by inversion, it must be that $\Delta$ is of shape $z:\okty$ as above.
    \item If $M$ is of shape $PQ$ we reason as follows.  In (i) by inversion it must be that there is some type $B$ such that $\Gamma \types P : B \to A$ and $\Gamma \types Q : B$.  The conclusion follows from the induction hypotheses, part (i).  In (ii) by inversion it must be that either (1) there is a type $B$ such that $\Gamma \types P : B \from A$ and $\Gamma,\,Q:B \types \Delta$, or (2) $\Gamma,\,P:\okty \from A \types \Delta$.  In the former case, the result follows from the induction hypotheses and \rlnm{AppL}.  In the latter case, the result follows from the induction hypothesis and \rlnm{FunK}.
    \item If $M$ is of shape $c(P_1,\ldots,P_n)$ we reason as follows.  In (i), by inversion, it must be that there are types $B_1,\ldots,B_n$ and $\Gamma \types c(B_1,\ldots,B_n) \subtype A$ and $\Gamma \types P_i : B_i$ for each $i$.  The result follows from the induction hypothesis, \rlnm{CnsR} and \rlnm{SubR}.  In (ii), by inversion, it follows that either (1) there are types $B_1,\ldots,B_n$ and $i$ such that $\Gamma \types A \subtype c(B_1,\ldots,B_n)$ and $\Gamma,\,P_i:B_i \types \Delta$, or (2) there is an arrow type $A'$ such that $\Gamma \types A \subtype A'$ and $\Gamma,\,c(P_1,\ldots,P_n) : A' \types \Delta$.  In the former case, the result follows from the induction hypotheses, \rlnm{CnsL} and \rlnm{SubL}.  In the latter case, the result follows from \rlnm{CnsDL}.  
    \item \ch{If $M$ is of shape $\fixtm{f(y)}{P}$, we may assume that $y$ and $f$ do not occur outside of $P$.  In (i), by inversion, it must be that there are types $B_1$ and $B_2$ such that, either: (1) $\Gamma,\,f:B_1 \to B_2,\,y:B_1 \types P:B_2$ and $\Gamma \types B_1 \to B_2 \subtype A$ or (2) $\Gamma,\,f:B_1 \from B_2,\,P:B_2 \types y:B_1$ and $\Gamma \types B_1 \from B_2 \subtype A$.  In (1), the result follows from the induction hypothesis, \rlnm{FixsR} and \rlnm{SubR}.  In (2), the result follows from the induction hypothesis, \rlnm{FixnR} and \rlnm{SubR}.  In (ii), by inversion, it must be that there is some sum type $A'$ such that $\Gamma \types A \subtype A'$ and $\Gamma \types \fixtm{f(y)}{P} : A' \types \Delta$.  Then the result follows from \rlnm{FixDL}.}
    \item If $M$ is of shape $\matchtm{Q}{\mid_{i=1}^k (p_i \mapsto P_i)}$, we reason as follows.  In (i), it follows from inversion that there is a family of types $B_y$ (for each pattern-bound variable $y$) and $\Gamma \types Q : \Sigma_{i=1}^k p_i[B_y/y \mid x\in\fv(p_i)]$ and $\Gamma \cup \{y:B_y \mid y\in\fv(p_i)\} \types P_i : A$ for each $i$.  Then the conclusion follows from the induction hypothesis and \rlnm{MchR}.  In (ii), by inversion we have a family of types $B_y$ (one for each pattern-bound variable $y$) such that $\Gamma,\,P_i:A \types y:B_y$ (for each $i$ and each $y\in\fv(p_i)$) and either $\Gamma,\,Q : p_i[B_y/y \mid y\in\fv(p_i)] \types \Delta$ or $\Gamma,\,P_i : A \types \Delta$.  In both cases, it follows from the induction hypothesis and \rlnm{CnsL} that $\Gamma,\,(Q[V/x],P_i[V/x]) : (p_i[B_y/y \mid y\in\fv(p_i)],A) \types \Delta$.  Then the conclusion follows from the induction hypothesis and \rlnm{MchL}.
  \end{itemize}
\end{proof}

\begin{lemma}[Substitution on the Left]\label{lem:subst-left}
  Assume $\Gamma$ is a consistent identifier environment disjoint from $\Delta$ and neither contains $x$.  If $\Gamma,\,M:A \types x:B$ and $\Gamma,\,V:B \types \Delta$ then $\Gamma,\,M[V/x] : A \types \Delta$.
\end{lemma}
\begin{proof}
  The proof is by induction on $M$.  Note, there is always the possibility that $B=\okty$, in which case we have $\Gamma,\,V:\okty \types \Delta$ for closed value $V$ and so it follows from inversion that $\Delta$ must be of shape $z:\okty$ for some variable $z$.  Hence, the desired conclusion follows in this case immediately from \rlnm{VarK}.
  \begin{itemize}
    \item If $M$ is $x$, then it follows from inversion that $\Gamma \types A \subtype B$.  Hence, the result follows from the assumed $\Gamma,\,V:B \types \Delta$ and \rlnm{SubL}, since $M[V/x] = V$.
    \item If $M$ is $f$, then it follows from inversion that $B$ must be $\okty$, and we have proven this case above. 
    \item If $M$ is a variable $y$ distinct from $x$, the reasoning is as in the previous case.
    \item If $M$ is an application $PQ$, then it follows from inversion that either (i) there is a type $B'$ and $\Gamma \types P : B' \from A$ and $\Gamma,\,Q:B' \types x:B$, or (ii) $\Gamma,\,P:\okty \from A \types x:B$.  In (i), it follows from Lemma~\ref{lem:trivial-subst} that $\Gamma \types P[V/x] : B' \from A$ and it follows from the induction hypothesis that $\Gamma,\,Q[V/x] : B' \types \Delta$.  Therefore, the conclusion follows from \rlnm{AppL}. In (ii), the result follows from the induction hypothesis and \rlnm{FunK}.
    \item \ch{If $M$ is a fixpoint abstraction $\fixtm{f(y)}{P}$ we may assume $y$ and $f$ are fresh for the context.  Then it follows from inversion that there is some sum type $A'$ such that $\Gamma \types A \subtype A'$ and $\Gamma,\,\fixtm{f(y)}{P}:A' \types x:B$.  Then the result follows from \rlnm{FixDL}.  }
    \item If $M$ is a constructor $c(P_1,\ldots,P_n)$ then it follows from inversion that either (1) there are types $B_1,\ldots,B_n$ and $\Gamma \types A \subtype c(B_1,\ldots,B_n)$ and $\Gamma,\,P_i:B_i \types x:B$ for some $i$, or (2) there is some arrow type $A'$ such that $\Gamma \types A \subtype A'$ and $\Gamma,\,c(P_1,\ldots,P_n) : A' \types x:B$.  In case (1), the result follows from the induction hypothesis and \rlnm{CnsL} and \rlnm{SubL}.  In case (2), the result follows immediately from \rlnm{CnsDL} and \rlnm{SubL}. 
    \item If $M$ is a match expression $\matchtm{Q}{\mid_{i=1}^k (p_i \mapsto P_i)}$ then it follows from inversion that there is some family $B_y$ (for each pattern-bound variable $y$) and $\Gamma,\,P_i : A \types y: B_y$ (for each $i$ and $y \in \fv(p_i)$) and either $\Gamma,\,Q : p_i[B_y/y \mid y\in\fv(p_i)] \types x:B$ or $\Gamma,\,P_i : A \types x:B$.  In both cases, it follows from the induction hypothesis and \rlnm{CnsL} that $\Gamma,\,(Q[V/x],P_i[V/x]) : (p_i[B_y/y \mid y\in\fv(p_i)],A) \types x:B$.  Then the result follows from the induction hypothesis and \rlnm{MchL}.
  \end{itemize}
\end{proof}

\begin{theorem}[Preservation on the Right]\label{thm:preservation-right}
  Suppose $\Gamma$ and $\Delta$ are disjoint, consistent variable environments with $\types \mathcal{M} : \Gamma$.
    If $\Gamma \types M : A$ and $M \ped N$ then $\Gamma \types N : A$.
\end{theorem}
\begin{proof}
  We prove the result for all $\calE$, redexes $P$ and contractions $Q$ such that $M = \calE[P] \ped \calE[Q] = N$, by induction on $\calE$.
  \begin{itemize}
    \item If the context is just a hole, then we proceed by cases on the redex.
    \begin{description}
      \item[\rlnm{Fix$\beta$}] \ch{In this case, $P$ has shape $(\fixtm{f(x)}{P'})\,V$ and $Q = P'[V/x,\fixtm{f(x)}{P'}/f]$.  It follows from inversion that there is a type $B$ such that $\Gamma \types \fixtm{f(x)}{P'} : B \to A$ and $\Gamma \types V : B$.  By inversion and the consistency of $\Gamma$, we have some $B_1$ and $B_2$ such that $\Gamma,\,f:B_1 \to B_2,\,x:B_1 \types P' : B_2$ and $\Gamma \types B_1 \to B_2 \subtype B \to A$ (the other possibility would involve a type constraint that is not syntactically consistent). By the closedness of the constraints in $\Gamma$, we have $\Gamma \types B_2 \subtype A$.  Then it follows from the substitution lemma (twice) that $\Gamma \types P'[V/x][\fixtm{f(x)}{P'}/f] : B_2$ and hence $P'[V/x,\fixtm{f(x)}{P'}/f] : A$ by \rlnm{SubR}.}
      \item[\rlnm{Delta}] In this case, it follows from inversion that there is some type $B$ and family $\bar{B}$ such that $f:\forall \bar{a}.\,D \Rightarrow B$ and $\Gamma \types B[\bar{B/a}] \subtype A$ and $\Gamma \types D[\bar{B/a}]$.  Then it follows from $\types \calM : \Gamma$ that $\Gamma \cup D \types M : B$ and therefore $\Gamma[\bar{B/a}] \cup D[\bar{B/a}] \types M : B[\bar{B/a}]$.  However, $\Gamma$ is guaranteed to be closed with respect to type variables, and so we have $\Gamma \cup D[\bar{B/a}] \types M: B[\bar{B/a}]$.  Since also $\Gamma \types D[\bar{B/a}]$, it follows that $\Gamma \types M : B[\bar{B/a}]$, as required.
      \item[\rlnm{Match}] In this case, $P$ has shape $\matchtm{c(V_1,\ldots,V_n)}{|_{i=1}^n (p_i \mapsto P_i)}$ and one of the alternatives is of shape $c(x_1,\ldots,x_n) \mapsto P$ and $Q=P[V_i/x_i \mid 1 \leq i \leq n]$  It follows from inversion that there is a family of types $B_x$ (one for each pattern-bound variable $x$) such that $\Gamma \types c(V_1,\ldots,V_n) : \Sigma_{i=1}^k p_i[B_x/x \mid x \in \fv(p_i)]$ and $\Gamma \cup \{x:B_x \mid x \in \fv(p_i)\} \types P : A$.  Therefore, it follows by value inversion that $\Gamma \types V_i : B_i$.  Moreover, it follows by induction on $n$ and the substitution lemma (since the $x_i$ are disjoint and the $V_i$ closed) that $\Gamma \types P[V_i/x_i \mid 1 \leq i \leq n] : A$, as required.
    \end{description}
  \item If the context is of shape $\calE\,M'$ then it follows by inversion that there is some $B$ such that $\Gamma \types \calE[P] : B \to A$ and $\Gamma \types M':B$.  Then it follows from the induction hypothesis that $\Gamma \types \calE[Q] : B \to A$ hence the result follows by \rlnm{AppR}.
  \item If the context is of shape $(\abs{x}{M'})\,\calE$ then it follows by inversion that there is some $B$ such that $\Gamma \types \abs{x}{M'} : B \to A$ and $\Gamma \types \calE[P] : B$.  Then the result follows from the induction hypothesis and \rlnm{AppR}.
  \item If the context is of shape $\matchtm{\calE}{|_{i=1}^k (p_i \mapsto P_i)}$ then it follows from inversion that there is some family $B_x$ and $\Gamma \types \calE[P] : \Sigma_{i=1}^k p_i[B_x/x \mid x \in \fv(p_i)]$ for each $i$ and $\Gamma \cup \{x:B_x \mid x\in\fv(p_i)\} \types P_i : A$ for all $i$.  Then the result follows from the induction hypothesis and \rlnm{MchR}.
  \end{itemize} 
\end{proof}

\begin{theorem}[Preservation on the Left]\label{thm:preservation-left}
  Suppose $\Gamma$ and $\Delta$ are disjoint, consistent variable environments.  If $M \ped N$ and $\Gamma,\,M:A \types \Delta$ then $\Gamma,\,N:A \types \Delta$.
\end{theorem}
\begin{proof}
  We prove the result for all $\calE$, redexes $P$ and contractions $Q$ such that $M = \calE[P] \ped \calE[Q] = N$, by induction on $\calE$.
  \begin{itemize}
    \item When the context is just a hole, we proceed to analyse the redex.
      \begin{description}
        \item[\rlnm{Fix$\beta$}] \ch{In this case, $P$ has shape $(\fixtm{f(x)}{P'})\,V$ and $Q=P'[V/x,\fixtm{f(x)}{P'}/f]$.  It follows from inversion that either (1) there is a type $B$ such that $\Gamma \types \abs{x}{P'} : B \from A$ and $\Gamma,\,V:B \types \Delta$, or (2) $\Gamma,\,\abs{x}{P'} : \okty \from A \types \Delta$.  In case (1), it follows by inversion that there are types $B_1$ and $B_2$ such that $\Gamma,\,f:B_1 \from B_2,\,P':B_2 \types x:B_1$ and $\Gamma \types B_1 \from B_2 \subtype B \from A$ (the other possibility is excluded by the consistency of the constraints).  Then it follows from the substitution lemma (twice) that $\Gamma,\,P'[V/x,\fixtm{f(x)}{P'}/f]:B_2 \types \Delta$.  By the closedness of the constraints, $\Gamma \types A \subtype B_2$ and so the desired conclusion follows by \rlnm{SubL}.  By the Progress theorem, case (2) is impossible unless $\Delta$ is of shape $z:\okty$ for some $z$.  In this case, the conclusion follows immediately by \rlnm{VarK}.}
        \item[\rlnm{Delta}] By inversion, these cases are impossible unless $\Delta$ has shape $z:\okty$ for some variable $z$ and then the conclusion follows immediately by \rlnm{VarK}.
        \item[\rlnm{Match}] In this case, $P$ has shape $\matchtm{c(V_1,\ldots,V_n)}{|_{i=1}^k (p_i \mapsto P_i)}$ and one of the alternatives has shape $c(x_1,\ldots,x_n) \mapsto P$ and $Q=P[V_i/x_i \mid 1 \leq i \leq k]$.  It follows from inversion that either $\Delta$ is of shape $z:\okty$, in which case the conclusion is immediate by \rlnm{VarK}, or there are a family of types $B_x$, one for each pattern-bound variable $x$ such that $\Gamma,\,P : A \types x:B_x$ (for each $x \in \{x_1,\ldots,x_n\}$) and, either $\Gamma,\,c(V_1,\ldots,V_n) : c(B_{x_1},\ldots,B_{x_n}) \types \Delta$ or $\Gamma,\,P:A \types \Delta$.  In the former case, it follows from inversion and the consistency of the constraints in $\Gamma$ that there is $j$ such that $\Gamma,\,V_j : B_{x_j} \types \Delta$.  Then is follows from the substitution lemma that $\Gamma,\,P[V_j/x_j] \types \Delta$.  Then it follows from the trivial substitution lemma, by induction on $n$, that $\Gamma,\,P[V_i/x_i \mid 1 \leq i \leq n] : A \types \Delta$ as required.
      \end{description}
    \item When the context is of shape $\calE\,M'$, it follows from inversion that either (0) $\Delta$ has shape $z:\okty$ for some variable $z$, (1) there is a type $B$ and $\Gamma \types \calE[P] : B \from A$ and $\Gamma,\,M':B \types \Delta$, or (2) $\Gamma,\,\calE[P] : \okty \from A \types \Delta$. In case (0) the result is immediate by \rlnm{VarK}.  In case (1), we obtain $\Gamma \types \calE[Q] : B \from A$ by preservation on the right and the result follows by \rlnm{AppL}.  In case (2), we obtain the result from the induction hypothesis and \rlnm{FunK}.
    \item \ch{When the context has shape $(\fixtm{f(x)}{M'})\,\calE$, it follows from inversion that either (0) $\Delta$ has shape $z:\okty$, or (1) there is a type $B$ and $\Gamma \types \abs{f(x)}{M'} : B \from A$ and $\Gamma,\,\calE[P]:B \types \Delta$, or (2) $\Gamma \types \fixtm{f(x)}{M'} : \okty \from A \types \Delta$. In case (0) the result is immediate by \rlnm{VarK}.  In case (2), the result follows by the induction hypothesis and \rlnm{AppL}.  In case (3), the result follows from the induction hypothesis and \rlnm{FunK}.}
    \item When the context has shape $\matchtm{\calE}{|_{i=1}^k (p_i \mapsto P_i)}$ it follows from inversion that either $\Delta$ has shape $z:\okty$ or there is a family of types $B_x$ (one for each pattern-bound variable $x$) such that $\Gamma,\,P_i:A \types x:B_x$ and, for each $i$, either $\Gamma,\,\calE[P] : p_i[B_x/x \mid x \in \fv(p_i)] \types \Delta$ or $\Gamma,\,P_i:A \types \Delta$.  Then the result follows from the induction hypothesis and \rlnm{MchL}.
  \end{itemize}
\end{proof}

Finally, we can prove Theorem~\ref{thm:weak-soundness}.

\begin{proof}
  Suppose $\Gamma \types M : \okty$ converges to a normal form $N$.  By preservation, $\Gamma \types N : \okty$ and, by progress, $N$ is a value.  Suppose $\Gamma,\,M:\okty \types$ converges to a normal form $N$, then, by preservation, $\Gamma,\,N:\okty \types$ and, by progress, $\Gamma \not\types N : \okty$.  Since closed values are well typed, $N$ is not a value.
\end{proof}

%% file: arxiv.bbl

\begin{thebibliography}{31}


\ifx \showCODEN    \undefined \def \showCODEN     #1{\unskip}     \fi
\ifx \showDOI      \undefined \def \showDOI       #1{#1}\fi
\ifx \showISBNx    \undefined \def \showISBNx     #1{\unskip}     \fi
\ifx \showISBNxiii \undefined \def \showISBNxiii  #1{\unskip}     \fi
\ifx \showISSN     \undefined \def \showISSN      #1{\unskip}     \fi
\ifx \showLCCN     \undefined \def \showLCCN      #1{\unskip}     \fi
\ifx \shownote     \undefined \def \shownote      #1{#1}          \fi
\ifx \showarticletitle \undefined \def \showarticletitle #1{#1}   \fi
\ifx \showURL      \undefined \def \showURL       {\relax}        \fi
\providecommand\bibfield[2]{#2}
\providecommand\bibinfo[2]{#2}
\providecommand\natexlab[1]{#1}
\providecommand\showeprint[2][]{arXiv:#2}

\bibitem[Aiken et~al\mbox{.}(1994)]%
        {aiken-et-al-popl1994}
\bibfield{author}{\bibinfo{person}{Alexander Aiken}, \bibinfo{person}{Edward~L.
  Wimmers}, {and} \bibinfo{person}{T.~K. Lakshman}.}
  \bibinfo{year}{1994}\natexlab{}.
\newblock \showarticletitle{Soft Typing with Conditional Types}. In
  \bibinfo{booktitle}{\emph{Conference Record of POPL'94: 21st {ACM}
  {SIGPLAN-SIGACT} Symposium on Principles of Programming Languages, Portland,
  Oregon, USA, January 17-21, 1994}}. \bibinfo{pages}{163--173}.
\newblock
\urldef\tempurl%
\url{https://doi.org/10.1145/174675.177847}
\showDOI{\tempurl}


\bibitem[Bradfield and Walukiewicz(2018)]%
        {bradfield-walukiewicz-2018}
\bibfield{author}{\bibinfo{person}{Julian Bradfield} {and}
  \bibinfo{person}{Igor Walukiewicz}.} \bibinfo{year}{2018}\natexlab{}.
\newblock \bibinfo{booktitle}{\emph{Handbook of Model Checking}}.
\newblock \bibinfo{publisher}{Springer International Publishing}, Chapter The
  mu-calculus and model checking, \bibinfo{pages}{871--919}.
\newblock
\urldef\tempurl%
\url{https://doi.org/10.1007/978-3-319-10575-8_26}
\showDOI{\tempurl}


\bibitem[Cousot et~al\mbox{.}(2013)]%
        {coutsot-et-al-vmcai2013}
\bibfield{author}{\bibinfo{person}{Patrick Cousot}, \bibinfo{person}{Radhia
  Cousot}, \bibinfo{person}{Manuel F{\"a}hndrich}, {and}
  \bibinfo{person}{Francesco Logozzo}.} \bibinfo{year}{2013}\natexlab{}.
\newblock \showarticletitle{Automatic Inference of Necessary Preconditions}. In
  \bibinfo{booktitle}{\emph{Verification, Model Checking, and Abstract
  Interpretation}}, \bibfield{editor}{\bibinfo{person}{Roberto Giacobazzi},
  \bibinfo{person}{Josh Berdine}, {and} \bibinfo{person}{Isabella Mastroeni}}
  (Eds.). \bibinfo{publisher}{Springer Berlin Heidelberg},
  \bibinfo{address}{Berlin, Heidelberg}, \bibinfo{pages}{128--148}.
\newblock
\showISBNx{978-3-642-35873-9}
\urldef\tempurl%
\url{https://doi.org/10.1007/978-3-642-35873-9_10}
\showDOI{\tempurl}


\bibitem[Cousot et~al\mbox{.}(2011)]%
        {cousot-et-al-vmcai2011}
\bibfield{author}{\bibinfo{person}{Patrick Cousot}, \bibinfo{person}{Radhia
  Cousot}, {and} \bibinfo{person}{Francesco Logozzo}.}
  \bibinfo{year}{2011}\natexlab{}.
\newblock \showarticletitle{Precondition Inference from Intermittent Assertions
  and Application to Contracts on Collections}. In
  \bibinfo{booktitle}{\emph{Verification, Model Checking, and Abstract
  Interpretation}}, \bibfield{editor}{\bibinfo{person}{Ranjit Jhala} {and}
  \bibinfo{person}{David Schmidt}} (Eds.). \bibinfo{publisher}{Springer Berlin
  Heidelberg}, \bibinfo{address}{Berlin, Heidelberg},
  \bibinfo{pages}{150--168}.
\newblock
\showISBNx{978-3-642-18275-4}
\urldef\tempurl%
\url{https://doi.org/10.1007/978-3-642-18275-4_12}
\showDOI{\tempurl}


\bibitem[Curien and Herbelin(2000)]%
        {curien-herbelin-icfp2000}
\bibfield{author}{\bibinfo{person}{Pierre{-}Louis Curien} {and}
  \bibinfo{person}{Hugo Herbelin}.} \bibinfo{year}{2000}\natexlab{}.
\newblock \showarticletitle{The duality of computation}. In
  \bibinfo{booktitle}{\emph{Proceedings of the Fifth {ACM} {SIGPLAN}
  International Conference on Functional Programming {(ICFP} '00), Montreal,
  Canada, September 18-21, 2000}}, \bibfield{editor}{\bibinfo{person}{Martin
  Odersky} {and} \bibinfo{person}{Philip Wadler}} (Eds.).
  \bibinfo{publisher}{{ACM}}, \bibinfo{pages}{233--243}.
\newblock
\urldef\tempurl%
\url{https://doi.org/10.1145/351240.351262}
\showDOI{\tempurl}


\bibitem[Dolan and Mycroft(2017)]%
        {dolan-mycroft-popl}
\bibfield{author}{\bibinfo{person}{Stephen Dolan} {and} \bibinfo{person}{Alan
  Mycroft}.} \bibinfo{year}{2017}\natexlab{}.
\newblock \showarticletitle{Polymorphism, Subtyping, and Type Inference in
  MLsub}. In \bibinfo{booktitle}{\emph{Proceedings of the 44th ACM SIGPLAN
  Symposium on Principles of Programming Languages}} (Paris, France)
  \emph{(\bibinfo{series}{POPL '17})}. \bibinfo{publisher}{Association for
  Computing Machinery}, \bibinfo{address}{New York, NY, USA},
  \bibinfo{pages}{60–72}.
\newblock
\showISBNx{9781450346603}
\urldef\tempurl%
\url{https://doi.org/10.1145/3009837.3009882}
\showDOI{\tempurl}


\bibitem[Ebbinghaus et~al\mbox{.}(1994)]%
        {ebbinhaus-flum-thomas-1994}
\bibfield{author}{\bibinfo{person}{H.-D. Ebbinghaus}, \bibinfo{person}{J.
  Flum}, {and} \bibinfo{person}{W. Thomas}.} \bibinfo{year}{1994}\natexlab{}.
\newblock \bibinfo{booktitle}{\emph{Mathematical Logic}}.
\newblock \bibinfo{publisher}{Springer}.
\newblock
\urldef\tempurl%
\url{https://doi.org/10.1007/978-1-4757-2355-7}
\showDOI{\tempurl}


\bibitem[Eifrig et~al\mbox{.}(1995)]%
        {eifrig-et-al-oopsla1995}
\bibfield{author}{\bibinfo{person}{Jonathan Eifrig}, \bibinfo{person}{Scott
  Smith}, {and} \bibinfo{person}{Valery Trifonov}.}
  \bibinfo{year}{1995}\natexlab{}.
\newblock \showarticletitle{Sound Polymorphic Type Inference for Objects}. In
  \bibinfo{booktitle}{\emph{Proceedings of the Tenth Annual Conference on
  Object-Oriented Programming Systems, Languages, and Applications}} (Austin,
  Texas, USA) \emph{(\bibinfo{series}{OOPSLA '95})}.
  \bibinfo{publisher}{Association for Computing Machinery},
  \bibinfo{address}{New York, NY, USA}, \bibinfo{pages}{169–184}.
\newblock
\showISBNx{0897917030}
\urldef\tempurl%
\url{https://doi.org/10.1145/217838.217858}
\showDOI{\tempurl}


\bibitem[Jakob and Thiemann(2015)]%
        {jakob-thiemann-nasa2015}
\bibfield{author}{\bibinfo{person}{Robert Jakob} {and} \bibinfo{person}{Peter
  Thiemann}.} \bibinfo{year}{2015}\natexlab{}.
\newblock \showarticletitle{A Falsification View of Success Typing}. In
  \bibinfo{booktitle}{\emph{NASA Formal Methods}},
  \bibfield{editor}{\bibinfo{person}{Klaus Havelund}, \bibinfo{person}{Gerard
  Holzmann}, {and} \bibinfo{person}{Rajeev Joshi}} (Eds.).
  \bibinfo{publisher}{Springer International Publishing},
  \bibinfo{address}{Cham}, \bibinfo{pages}{234--247}.
\newblock
\showISBNx{978-3-319-17524-9}
\urldef\tempurl%
\url{https://doi.org/10.1007/978-3-319-17524-9_17}
\showDOI{\tempurl}


\bibitem[Jones and Ramsay(2021)]%
        {jones-ramsay-popl2021}
\bibfield{author}{\bibinfo{person}{Eddie Jones} {and} \bibinfo{person}{Steven
  Ramsay}.} \bibinfo{year}{2021}\natexlab{}.
\newblock \showarticletitle{Intensional Datatype Refinement: With Application
  to Scalable Verification of Pattern-Match Safety}.
\newblock \bibinfo{journal}{\emph{Proc. ACM Program. Lang.}}
  \bibinfo{volume}{5}, \bibinfo{number}{POPL}, Article \bibinfo{articleno}{55}
  (\bibinfo{date}{jan} \bibinfo{year}{2021}), \bibinfo{numpages}{29}~pages.
\newblock
\urldef\tempurl%
\url{https://doi.org/10.1145/3434336}
\showDOI{\tempurl}


\bibitem[Le et~al\mbox{.}(2022)]%
        {le-et-al-oopsla2022}
\bibfield{author}{\bibinfo{person}{Quang~Loc Le}, \bibinfo{person}{Azalea
  Raad}, \bibinfo{person}{Jules Villard}, \bibinfo{person}{Josh Berdine},
  \bibinfo{person}{Derek Dreyer}, {and} \bibinfo{person}{Peter~W. O'Hearn}.}
  \bibinfo{year}{2022}\natexlab{}.
\newblock \showarticletitle{Finding Real Bugs in Big Programs with
  Incorrectness Logic}.
\newblock \bibinfo{journal}{\emph{Proc. ACM Program. Lang.}}
  \bibinfo{volume}{6}, \bibinfo{number}{OOPSLA1}, Article
  \bibinfo{articleno}{81} (\bibinfo{date}{apr} \bibinfo{year}{2022}),
  \bibinfo{numpages}{27}~pages.
\newblock
\urldef\tempurl%
\url{https://doi.org/10.1145/3527325}
\showDOI{\tempurl}


\bibitem[Lindahl and Sagonas(2006)]%
        {lindahl-sagonas-ppdp2006}
\bibfield{author}{\bibinfo{person}{Tobias Lindahl} {and}
  \bibinfo{person}{Konstantinos Sagonas}.} \bibinfo{year}{2006}\natexlab{}.
\newblock \showarticletitle{Practical Type Inference Based on Success Typings}.
  In \bibinfo{booktitle}{\emph{Proceedings of the 8th ACM SIGPLAN International
  Conference on Principles and Practice of Declarative Programming}} (Venice,
  Italy) \emph{(\bibinfo{series}{PPDP '06})}. \bibinfo{publisher}{Association
  for Computing Machinery}, \bibinfo{address}{New York, NY, USA},
  \bibinfo{pages}{167–178}.
\newblock
\showISBNx{1595933883}
\urldef\tempurl%
\url{https://doi.org/10.1145/1140335.1140356}
\showDOI{\tempurl}


\bibitem[L\'opez-Fraguas et~al\mbox{.}(2018)]%
        {Lopez-Fraguas-et-al-lpar2018}
\bibfield{author}{\bibinfo{person}{Francisco~J. L\'opez-Fraguas},
  \bibinfo{person}{Manuel Montenegro}, {and} \bibinfo{person}{Gorka
  Su\textbackslash{}'arez-Garc\textbackslash{}'ia}.}
  \bibinfo{year}{2018}\natexlab{}.
\newblock \showarticletitle{Polymorphic success types for Erlang}. In
  \bibinfo{booktitle}{\emph{LPAR-22. 22nd International Conference on Logic for
  Programming, Artificial Intelligence and Reasoning}}
  \emph{(\bibinfo{series}{EPiC Series in Computing},
  Vol.~\bibinfo{volume}{57})}, \bibfield{editor}{\bibinfo{person}{Gilles
  Barthe}, \bibinfo{person}{Geoff Sutcliffe}, {and} \bibinfo{person}{Margus
  Veanes}} (Eds.). \bibinfo{publisher}{EasyChair}, \bibinfo{pages}{515--533}.
\newblock
\showISSN{2398-7340}
\urldef\tempurl%
\url{https://doi.org/10.29007/w2m2}
\showDOI{\tempurl}


\bibitem[Mackay et~al\mbox{.}(2022)]%
        {mackay-et-al-oopsla2022}
\bibfield{author}{\bibinfo{person}{Julian Mackay}, \bibinfo{person}{Susan
  Eisenbach}, \bibinfo{person}{James Noble}, {and} \bibinfo{person}{Sophia
  Drossopoulou}.} \bibinfo{year}{2022}\natexlab{}.
\newblock \showarticletitle{Necessity Specifications for Robustness}.
\newblock \bibinfo{journal}{\emph{Proc. ACM Program. Lang.}}
  \bibinfo{volume}{6}, \bibinfo{number}{OOPSLA2}, Article
  \bibinfo{articleno}{154} (\bibinfo{date}{oct} \bibinfo{year}{2022}),
  \bibinfo{numpages}{30}~pages.
\newblock
\urldef\tempurl%
\url{https://doi.org/10.1145/3563317}
\showDOI{\tempurl}


\bibitem[Marlow and Wadler(1997)]%
        {marlow-wadler-icfp1997}
\bibfield{author}{\bibinfo{person}{Simon Marlow} {and} \bibinfo{person}{Philip
  Wadler}.} \bibinfo{year}{1997}\natexlab{}.
\newblock \showarticletitle{A Practical Subtyping System for Erlang}. In
  \bibinfo{booktitle}{\emph{Proceedings of the Second ACM SIGPLAN International
  Conference on Functional Programming}} (Amsterdam, The Netherlands)
  \emph{(\bibinfo{series}{ICFP '97})}. \bibinfo{publisher}{Association for
  Computing Machinery}, \bibinfo{address}{New York, NY, USA},
  \bibinfo{pages}{136–149}.
\newblock
\showISBNx{0897919181}
\urldef\tempurl%
\url{https://doi.org/10.1145/258948.258962}
\showDOI{\tempurl}


\bibitem[Mitchell(1984)]%
        {mitchell-popl1984}
\bibfield{author}{\bibinfo{person}{John~C. Mitchell}.}
  \bibinfo{year}{1984}\natexlab{}.
\newblock \showarticletitle{Coercion and Type Inference}. In
  \bibinfo{booktitle}{\emph{Proceedings of the 11th ACM SIGACT-SIGPLAN
  Symposium on Principles of Programming Languages}} (Salt Lake City, Utah,
  USA) \emph{(\bibinfo{series}{POPL '84})}. \bibinfo{publisher}{Association for
  Computing Machinery}, \bibinfo{address}{New York, NY, USA},
  \bibinfo{pages}{175–185}.
\newblock
\showISBNx{0897911253}
\urldef\tempurl%
\url{https://doi.org/10.1145/800017.800529}
\showDOI{\tempurl}


\bibitem[Odersky et~al\mbox{.}(1999)]%
        {odersky-sulzmann-wehr-TSPOS1999}
\bibfield{author}{\bibinfo{person}{Martin Odersky}, \bibinfo{person}{Martin
  Sulzmann}, {and} \bibinfo{person}{Martin Wehr}.}
  \bibinfo{year}{1999}\natexlab{}.
\newblock \showarticletitle{Type Inference with Constrained Types}.
\newblock \bibinfo{journal}{\emph{{TAPOS}}} \bibinfo{volume}{5},
  \bibinfo{number}{1} (\bibinfo{year}{1999}), \bibinfo{pages}{35--55}.
\newblock
\urldef\tempurl%
\url{https://doi.org/10.1002/(SICI)1096-9942(199901/03)5:1<35::AID-TAPO4>3.0.CO;2-4}
\showDOI{\tempurl}


\bibitem[O'Hearn(2019)]%
        {ohearn-popl2019}
\bibfield{author}{\bibinfo{person}{Peter~W. O'Hearn}.}
  \bibinfo{year}{2019}\natexlab{}.
\newblock \showarticletitle{Incorrectness Logic}.
\newblock \bibinfo{journal}{\emph{Proc. ACM Program. Lang.}}
  \bibinfo{volume}{4}, \bibinfo{number}{POPL}, Article \bibinfo{articleno}{10}
  (\bibinfo{date}{dec} \bibinfo{year}{2019}), \bibinfo{numpages}{32}~pages.
\newblock
\urldef\tempurl%
\url{https://doi.org/10.1145/3371078}
\showDOI{\tempurl}


\bibitem[Parreaux and Chau(2022)]%
        {parreaux-et-al-oopsla2022}
\bibfield{author}{\bibinfo{person}{Lionel Parreaux} {and}
  \bibinfo{person}{Chun~Yin Chau}.} \bibinfo{year}{2022}\natexlab{}.
\newblock \showarticletitle{MLstruct: Principal Type Inference in a Boolean
  Algebra of Structural Types}.
\newblock \bibinfo{journal}{\emph{Proc. ACM Program. Lang.}}
  \bibinfo{volume}{6}, \bibinfo{number}{OOPSLA2}, Article
  \bibinfo{articleno}{141} (\bibinfo{date}{oct} \bibinfo{year}{2022}),
  \bibinfo{numpages}{30}~pages.
\newblock
\urldef\tempurl%
\url{https://doi.org/10.1145/3563304}
\showDOI{\tempurl}


\bibitem[Raad et~al\mbox{.}(2022)]%
        {raad-et-al-popl2022}
\bibfield{author}{\bibinfo{person}{Azalea Raad}, \bibinfo{person}{Josh
  Berdine}, \bibinfo{person}{Derek Dreyer}, {and} \bibinfo{person}{Peter~W.
  O'Hearn}.} \bibinfo{year}{2022}\natexlab{}.
\newblock \showarticletitle{Concurrent Incorrectness Separation Logic}.
\newblock \bibinfo{journal}{\emph{Proc. ACM Program. Lang.}}
  \bibinfo{volume}{6}, \bibinfo{number}{POPL}, Article \bibinfo{articleno}{34}
  (\bibinfo{date}{jan} \bibinfo{year}{2022}), \bibinfo{numpages}{29}~pages.
\newblock
\urldef\tempurl%
\url{https://doi.org/10.1145/3498695}
\showDOI{\tempurl}


\bibitem[Rondon et~al\mbox{.}(2008)]%
        {rondon-et-al-pldi2008}
\bibfield{author}{\bibinfo{person}{Patrick~Maxim Rondon}, \bibinfo{person}{Ming
  Kawaguchi}, {and} \bibinfo{person}{Ranjit Jhala}.}
  \bibinfo{year}{2008}\natexlab{}.
\newblock \showarticletitle{Liquid types}. In
  \bibinfo{booktitle}{\emph{Proceedings of the {ACM} {SIGPLAN} 2008 Conference
  on Programming Language Design and Implementation, Tucson, AZ, USA, June
  7-13, 2008}}. \bibinfo{pages}{159--169}.
\newblock
\urldef\tempurl%
\url{https://doi.org/10.1145/1375581.1375602}
\showDOI{\tempurl}


\bibitem[Sagonas(2010)]%
        {sagonas-flp2010}
\bibfield{author}{\bibinfo{person}{Konstantinos Sagonas}.}
  \bibinfo{year}{2010}\natexlab{}.
\newblock \showarticletitle{Using Static Analysis to Detect Type Errors and
  Concurrency Defects in Erlang Programs}. In
  \bibinfo{booktitle}{\emph{Functional and Logic Programming}},
  \bibfield{editor}{\bibinfo{person}{Matthias Blume}, \bibinfo{person}{Naoki
  Kobayashi}, {and} \bibinfo{person}{Germ{\'a}n Vidal}} (Eds.).
  \bibinfo{publisher}{Springer Berlin Heidelberg}, \bibinfo{address}{Berlin,
  Heidelberg}, \bibinfo{pages}{13--18}.
\newblock
\showISBNx{978-3-642-12251-4}
\urldef\tempurl%
\url{https://doi.org/10.1007/978-3-642-12251-4_2}
\showDOI{\tempurl}


\bibitem[Terauchi(2010)]%
        {terauchi-popl2010}
\bibfield{author}{\bibinfo{person}{Tachio Terauchi}.}
  \bibinfo{year}{2010}\natexlab{}.
\newblock \showarticletitle{Dependent types from counterexamples}. In
  \bibinfo{booktitle}{\emph{Proceedings of the 37th {ACM} {SIGPLAN-SIGACT}
  Symposium on Principles of Programming Languages, {POPL} 2010, Madrid, Spain,
  January 17-23, 2010}}. \bibinfo{pages}{119--130}.
\newblock
\urldef\tempurl%
\url{https://doi.org/10.1145/1706299.1706315}
\showDOI{\tempurl}


\bibitem[Unno and Kobayashi(2009)]%
        {unno-kobayashi-PPDP2009}
\bibfield{author}{\bibinfo{person}{Hiroshi Unno} {and} \bibinfo{person}{Naoki
  Kobayashi}.} \bibinfo{year}{2009}\natexlab{}.
\newblock \showarticletitle{Dependent type inference with interpolants}. In
  \bibinfo{booktitle}{\emph{Proceedings of the 11th International {ACM}
  {SIGPLAN} Conference on Principles and Practice of Declarative Programming,
  September 7-9, 2009, Coimbra, Portugal}}. \bibinfo{pages}{277--288}.
\newblock
\urldef\tempurl%
\url{https://doi.org/10.1145/1599410.1599445}
\showDOI{\tempurl}


\bibitem[Unno et~al\mbox{.}(2017)]%
        {unno-et-al-popl2018}
\bibfield{author}{\bibinfo{person}{Hiroshi Unno}, \bibinfo{person}{Yuki
  Satake}, {and} \bibinfo{person}{Tachio Terauchi}.}
  \bibinfo{year}{2017}\natexlab{}.
\newblock \showarticletitle{Relatively Complete Refinement Type System for
  Verification of Higher-Order Non-Deterministic Programs}.
\newblock \bibinfo{journal}{\emph{Proc. ACM Program. Lang.}}
  \bibinfo{volume}{2}, \bibinfo{number}{POPL}, Article \bibinfo{articleno}{12}
  (\bibinfo{date}{dec} \bibinfo{year}{2017}), \bibinfo{numpages}{29}~pages.
\newblock
\urldef\tempurl%
\url{https://doi.org/10.1145/3158100}
\showDOI{\tempurl}


\bibitem[Vazou et~al\mbox{.}(2015)]%
        {vazou-et-al-icfp2015}
\bibfield{author}{\bibinfo{person}{Niki Vazou}, \bibinfo{person}{Alexander
  Bakst}, {and} \bibinfo{person}{Ranjit Jhala}.}
  \bibinfo{year}{2015}\natexlab{}.
\newblock \showarticletitle{Bounded refinement types}. In
  \bibinfo{booktitle}{\emph{Proceedings of the 20th {ACM} {SIGPLAN}
  International Conference on Functional Programming, {ICFP} 2015, Vancouver,
  BC, Canada, September 1-3, 2015}}. \bibinfo{pages}{48--61}.
\newblock
\urldef\tempurl%
\url{https://doi.org/10.1145/2784731.2784745}
\showDOI{\tempurl}


\bibitem[Vazou et~al\mbox{.}(2013)]%
        {vazou-et-al-esop2013}
\bibfield{author}{\bibinfo{person}{Niki Vazou}, \bibinfo{person}{Patrick~Maxim
  Rondon}, {and} \bibinfo{person}{Ranjit Jhala}.}
  \bibinfo{year}{2013}\natexlab{}.
\newblock \showarticletitle{Abstract Refinement Types}. In
  \bibinfo{booktitle}{\emph{Programming Languages and Systems - 22nd European
  Symposium on Programming, {ESOP} 2013, Held as Part of the European Joint
  Conferences on Theory and Practice of Software, {ETAPS} 2013, Rome, Italy,
  March 16-24, 2013. Proceedings}}. \bibinfo{pages}{209--228}.
\newblock
\urldef\tempurl%
\url{https://doi.org/10.1007/978-3-642-37036-6_13}
\showDOI{\tempurl}


\bibitem[Vazou et~al\mbox{.}(2014)]%
        {Vazou:2014:RTH:2628136.2628161}
\bibfield{author}{\bibinfo{person}{Niki Vazou}, \bibinfo{person}{Eric~L.
  Seidel}, \bibinfo{person}{Ranjit Jhala}, \bibinfo{person}{Dimitrios
  Vytiniotis}, {and} \bibinfo{person}{Simon Peyton-Jones}.}
  \bibinfo{year}{2014}\natexlab{}.
\newblock \showarticletitle{Refinement Types for Haskell}. In
  \bibinfo{booktitle}{\emph{Proceedings of the 19th ACM SIGPLAN International
  Conference on Functional Programming}} (Gothenburg, Sweden)
  \emph{(\bibinfo{series}{ICFP '14})}. \bibinfo{pages}{269--282}.
\newblock
\showISBNx{978-1-4503-2873-9}


\bibitem[Winskel(1993)]%
        {winskel-1993}
\bibfield{author}{\bibinfo{person}{Glynn Winskel}.}
  \bibinfo{year}{1993}\natexlab{}.
\newblock \bibinfo{booktitle}{\emph{The Formal Semantics of Programming
  Languages: An Introduction}}.
\newblock \bibinfo{publisher}{MIT Press}.
\newblock
\urldef\tempurl%
\url{https://doi.org/10.7551/mitpress/3054.001.0001}
\showDOI{\tempurl}


\bibitem[Wright and Felleisen(1994)]%
        {wright-felleisen-infcomp1994}
\bibfield{author}{\bibinfo{person}{A.K. Wright} {and} \bibinfo{person}{M.
  Felleisen}.} \bibinfo{year}{1994}\natexlab{}.
\newblock \showarticletitle{A Syntactic Approach to Type Soundness}.
\newblock \bibinfo{journal}{\emph{Information and Computation}}
  \bibinfo{volume}{115}, \bibinfo{number}{1} (\bibinfo{year}{1994}),
  \bibinfo{pages}{38--94}.
\newblock
\urldef\tempurl%
\url{https://doi.org/10.1006/inco.1994.1093}
\showDOI{\tempurl}


\bibitem[Zilberstein et~al\mbox{.}(2023)]%
        {zilberstein-et-al-oopsla2023}
\bibfield{author}{\bibinfo{person}{Noam Zilberstein}, \bibinfo{person}{Derek
  Dreyer}, {and} \bibinfo{person}{Alexandra Silva}.}
  \bibinfo{year}{2023}\natexlab{}.
\newblock \showarticletitle{Outcome Logic: A Unifying Foundation for
  Correctness and Incorrectness Reasoning}.
\newblock \bibinfo{journal}{\emph{Proc. ACM Program. Lang.}}
  \bibinfo{volume}{7}, \bibinfo{number}{OOPSLA1}, Article
  \bibinfo{articleno}{93} (\bibinfo{date}{apr} \bibinfo{year}{2023}),
  \bibinfo{numpages}{29}~pages.
\newblock
\urldef\tempurl%
\url{https://doi.org/10.1145/3586045}
\showDOI{\tempurl}


\end{thebibliography}
